\documentclass[11pt,a4paper,twoside,openright,titlepage]{book}
\def\arxivmode{1}

\usepackage{amsmath}
\usepackage{amssymb}
\usepackage{amsthm}
\usepackage{array}
\usepackage[german,english]{babel}
\usepackage[backend=biber,style=bibstyle]{biblatex}
\usepackage{bm}
\usepackage{caption}
\usepackage{csquotes}
\usepackage{dsfont}
\usepackage[shortlabels]{enumitem}
\usepackage{fancyhdr}
\usepackage[a4paper,twoside,outer=2.5cm,inner=4.0cm,top=4.0cm,bottom=4.0cm,headheight=0.5cm]{geometry}
\usepackage{longtable}
\usepackage{mathrsfs}
\usepackage{mathtools}
\usepackage[sc]{mathpazo}
\usepackage{nicefrac}
\usepackage{subcaption}
\usepackage[most]{tcolorbox}
\usepackage{xcolor}
\usepackage{hyperref}
\usepackage[capitalize]{cleveref}
\usepackage{algpseudocode}
\usepackage{algorithm}

\DeclareFieldFormat[article]{citetitle}{\mkbibemph{#1}}
\DeclareFieldFormat[article]{title}{\mkbibemph{#1}}

\algdef{SE}[REPEAT]{Repeat}{EndRepeat}[1]{\algorithmicrepeat\ #1 \textbf{times}}{\algorithmicend}

\newtcbtheorem[number within=chapter]{theorem}{Theorem}
{
    label type=theorem,
    arc=0.5mm, 
    bottomtitle=0.5mm,
    boxrule=0mm,
    colbacktitle=black!20!white, 
    colback=black!5!white, 
    coltitle=black, 
    fonttitle=\bfseries, 
    left=2.5mm,
    right=3.5mm,
    toptitle=0.75mm, 
}{thm}
\newtcbtheorem[use counter from=theorem]{proposition}{Proposition}
{
    label type=proposition,
    arc=0.5mm, 
    bottomtitle=0.5mm,
    boxrule=0mm,
    colbacktitle=black!20!white, 
    colback=black!5!white, 
    coltitle=black, 
    fonttitle=\bfseries, 
    left=2.5mm,
    right=3.5mm,
    toptitle=0.75mm, 
}{prop}
\newtcbtheorem[use counter from=theorem]{lemma}{Lemma}
{
    label type=lemma,
    theorem style=plain,
    enhanced,
    arc=0.5mm, 
    boxrule=0mm,
    colframe=blue!50!black,
    colback=black!5!white, 
    coltitle=black, 
    fonttitle=\bfseries,
    left=2.5mm,
    right=3.5mm
}{lem}
\newtcbtheorem[use counter from=theorem]{corollary}{Corollary}
{
    label type=corollary,
    theorem style=plain,
    enhanced,
    arc=0.5mm, 
    boxrule=0mm,
    colframe=blue!50!black,
    colback=black!5!white, 
    coltitle=black, 
    fonttitle=\bfseries,
    left=2.5mm,
    right=3.5mm
}{cor}
\newtcbtheorem[use counter from=theorem]{definition}{Definition}
{
    label type=definition,
    theorem style=plain,
    enhanced,
    arc=0.5mm, 
    boxrule=0mm,
    colframe=blue!50!black,
    colback=black!5!white, 
    coltitle=black, 
    fonttitle=\bfseries,
    left=2.5mm,
    right=3.5mm
}{def}
\newtcbtheorem[use counter from=theorem]{remark}{Remark}
{
    label type=remark,
    theorem style=plain,
    breakable,
    enhanced,
    arc=0.5mm, 
    boxrule=0mm,
    colframe=blue!50!black,
    colback=black!5!white, 
    coltitle=black, 
    fonttitle=\bfseries,
    left=2.5mm,
    right=3.5mm
}{rem}

\theoremstyle{plain}

\theoremstyle{definition}
\newtheorem{example}[tcb@cnt@theorem]{Example}
\NewCommandCopy{\proofqedsymbol}{\qedsymbol}

\AtBeginEnvironment{proof}{\renewcommand{\qedsymbol}{\proofqedsymbol}}
\AtBeginEnvironment{example}{%
  \pushQED{\qed}\renewcommand{\qedsymbol}{$\triangle$}%
}
\AtEndEnvironment{example}{\popQED\endexample}

\makeatletter
\newcommand*{\resetnamehighlights}{\let\nhcbx@highlightlist\@empty}
\resetnamehighlights
\newcommand*{\highlightnames}{%
  \@ifstar{\resetnamehighlights\highlightnames@add}{\highlightnames@add}}
\newcommand*{\highlightnames@add@inner}[2]{%
  \listeadd{#1}{\the\numexpr#2\relax}}
\newcommand*{\highlightnames@add}{%
  \forcsvlist{\highlightnames@add@inner\nhcbx@highlightlist}}
\newcommand*{\mkbibhighlightnthname}[1]{%
  \xifinlist{\the\value{listcount}}{\nhcbx@highlightlist}
    {\mkbibbold{#1}}
    {#1}}
\makeatother
\DeclareNameWrapperFormat{highlightnthname}{%
  #1}
\DeclareCiteCommand{\fullcite}
  {\usebibmacro{prenote}}
  {\usedriver
     {\DeclareNameWrapperAlias{author}{highlightnthname}%
      \DeclareNameWrapperAlias{editor}{highlightnthname}%
      \DeclareNameWrapperAlias{translator}{highlightnthname}%
      \DeclareNameWrapperAlias{sortname}{highlightnthname}}
     {\thefield{entrytype}}}
  {\multicitedelim}
  {\usebibmacro{postnote}}

\newcommand*{\cA}{\mathcal{A}}
\newcommand*{\cB}{\mathcal{B}}
\newcommand*{\cD}{\mathcal{D}}
\newcommand*{\cE}{\mathcal{E}}
\newcommand*{\cF}{\mathcal{F}}
\newcommand*{\cG}{\mathcal{G}}
\newcommand*{\cH}{\mathcal{H}}
\newcommand*{\cI}{\mathcal{I}}
\newcommand*{\cJ}{\mathcal{J}}

\newcommand*{\cN}{\mathcal{N}}

\newcommand*{\cU}{\mathcal{U}}

\newcommand*{\lino}{\mathrm{L}} 
\newcommand*{\supo}{\mathrm{T}} 
\newcommand*{\herm}{\mathrm{H}} 
\newcommand*{\pos}{\mathrm{P}} 
\newcommand*{\uni}{\mathrm{U}} 
\newcommand*{\dops}{\mathrm{D}} 
\newcommand*{\hp}{\mathrm{HP}} 
\newcommand*{\cp}{\mathrm{CP}} 
\newcommand*{\tp}{\mathrm{TP}} 
\newcommand*{\hptp}{\mathrm{HPTP}} 
\newcommand*{\cptp}{\mathrm{CPTP}} 
\newcommand*{\cptni}{\mathrm{CPTN}} 
\newcommand*{\nonsig}{\mathrm{NS}} 
\newcommand*{\schan}{\mathrm{SC}} 
\newcommand*{\qi}{\mathrm{QI}} 
\newcommand*{\povm}{\mathrm{POVM}} 
\newcommand*{\ds}{S} 
\newcommand*{\fs}{\mathfrak{F}} 
\newcommand*{\qrt}{\mathfrak{O}} 
\newcommand*{\fqi}{\mathfrak{I}} 
\newcommand*{\fsc}{\mathfrak{S}} 
\newcommand*{\qids}{\mathrm{chan}^*} 
\newcommand*{\qichan}{\mathrm{chan}} 
\newcommand{\instr}[1]{\mathfrak{I}_{#1}} 
\newcommand*{\sep}{\mathrm{SEP}} 
\newcommand*{\ppt}{\mathrm{PPT}} 
\newcommand*{\lo}{\mathrm{LO}} 
\newcommand*{\locc}{\mathrm{LOCC}} 
\newcommand*{\sepc}{\mathrm{SEPC}} 
\newcommand*{\sepp}{\mathrm{SEPP}} 
\newcommand*{\pptc}{\mathrm{PPTC}} 
\newcommand*{\seppsc}{\mathrm{SEPPSC}} 
\newcommand*{\pptsc}{\mathrm{PPTSC}} 
\newcommand*{\ebc}{\mathrm{EBC}} 
\newcommand*{\mpc}{\mathrm{MPC}} 
\newcommand*{\ebpsc}{\mathrm{EBPSC}} 
\newcommand*{\stab}{\mathrm{STAB}} 
\newcommand*{\stabops}{\mathrm{SO}} 
\newcommand*{\csp}{\mathrm{CSP}} 
\newcommand*{\gammareg}{\overline{\gamma}} 
\newcommand*{\gammasreg}{\gamma^{\infty}} 
\newcommand*{\gammabb}{\gamma^{\mathrm{BB}}} 
\newcommand*{\maxlogneg}{LN_{\mathrm{max}}} 
\DeclareMathOperator{\im}{Im}
\DeclareMathOperator{\re}{Re}
\DeclareMathOperator{\conv}{conv}
\DeclareMathOperator{\aconv}{aconv}
\DeclareMathOperator{\spn}{span}
\DeclareMathOperator{\sgn}{sgn}
\DeclareMathOperator{\tr}{tr}
\DeclareMathOperator{\dimension}{dim}
\DeclareMathOperator{\diag}{diag}
\DeclareMathOperator{\cnot}{CNOT}
\DeclareMathOperator{\swap}{SWAP}
\DeclareMathOperator{\iswap}{iSWAP}
\DeclareMathOperator{\toffoli}{Toffoli}
\DeclareMathOperator{\unot}{UNOT}
\newcommand{\ceil}[1]{\lceil#1\rceil}

\newcommand{\id}{\mathds{1}}
\newcommand{\idchan}{\mathrm{id}}
\newcommand{\abs}[1]{\left\lvert#1\right\rvert}
\newcommand{\norm}[1]{\left\lVert#1\right\rVert}
\newcommand{\dnorm}[1]{\left\lVert#1\right\rVert_{\Diamond}}
\newcommand{\opnorm}[1]{\left\lVert#1\right\rVert_{\infty}}
\newcommand{\indsupo}[1]{\mathrm{Ad}\left[#1\right]} 
\newcommand{\choi}[1]{J_{#1}}

\providecommand{\ket}[1]{\left\lvert#1\right\rangle}
\providecommand{\bra}[1]{\left\langle#1\right\rvert}
\providecommand{\braket}[2]{\left\langle\smash{#1}\middle\vert\smash{#2}\right\rangle}
\providecommand{\braopket}[3]{\left\langle\smash{#1}\middle\vert\smash{#2}\middle\vert\smash{#3}\right\rangle}
\providecommand{\ketbra}[2]{\left\lvert \smash{#1}\middle\rangle\!\middle\langle\smash{#2}\right\rvert}
\providecommand{\proj}[1]{\ketbra{#1}{#1}}
\providecommand{\symplip}[2]{[\![ #1 , #2 ]\!]}
\newcommand{\transpose}{T}
\newcommand{\etal}{\emph{et al.}}
\newcommand{\loewnergeq}{\succcurlyeq}
\newcommand{\loewnerleq}{\preccurlyeq}
\newcommand{\walshhad}{\bm{\mathrm{W}}}
\newcommand{\reference}{Ref.}

\definecolor{CPlightgreen}{rgb}{0.55, 0.71, 0.0}
\definecolor{CPdarkgreen}{rgb}{0.0, 0.5, 0.0}
\definecolor{CPlightblue}{rgb}{0.0, 0.5, 1.0}
\definecolor{CPdarkblue}{rgb}{0.01, 0.28, 1.0}
\definecolor{CPlightred}{rgb}{1.0, 0.13, 0.32}
\definecolor{CPdarkred}{rgb}{0.8, 0.0, 0.0}
\definecolor{CPviolet}{rgb}{0.16, 0.11, 0.57}
\definecolor{CPyellow}{rgb}{1.0, 0.88, 0.21}
\definecolor{CPorange}{rgb}{0.93, 0.57, 0.13}

\begin{document}

\begin{titlepage}
  \vspace*{-2.5cm}
  \begin{center}
    {\large DISS.~ETH NO.~30914 \par}
    \vspace*{2.0cm}
    {\huge {\textbf{Simulating quantum circuits with restricted quantum computers}} \par}
     \vspace*{2cm}
     {\large A thesis submitted to attain the degree of \par}
     \vspace*{1cm}
     {\large DOCTOR OF SCIENCES \\ (Dr. sc. ETH Zurich) \par}
     \vspace*{2cm}
     {\large presented by \par}
     \vspace*{1cm}
     {\large Christophe Piveteau \par}
     \vspace*{1cm}
     {\large born on 23.10.1996 \par}
     \vspace*{2cm}
     {\large accepted on the recommendation of \par}
     \vspace*{1.0cm}
     {\large
     Prof. Dr. Renato Renner, examiner \\
     Dr. Joseph M. Renes, co-examiner \\
     Prof. Dr. David Gross, co-examiner \\
     Prof. Dr. Aram Harrow, co-examiner
     \par}
     \vspace*{2.6cm}
     {\large 2025 \par}
  \end{center}
\end{titlepage}

\chapter*{Acknowledgements}
I would like to thank my doctoral advisor Renato Renner for granting me the outstanding opportunity to perform my doctoral studies in the quantum information theory group at ETH Zurich.
The extraordinary and highly stimulating environment has been instrumental for my academic development and has sharpened my scientific skills over the last four years.
I especially appreciate the freedom I was given to work on a wide variety of different topics of my choice.

I would like to convey special thanks to Joe Renes who has been a phenomenal supervisor.
No matter the circumstances, he always found time for me, and he always supported me in every way necessary. 
I enjoyed our plentiful discussions about science and non-scientific topics, which have made my doctoral studies so much more enjoyable.
Joe is an amazing researcher to collaborate with, and he is a great inspiration for the kind of scientist I aspire to become.
Furthermore, I would like to express my deepest gratitude to David Sutter and Christopher Chubb, who both served as mentors to me. 
I am truly thankful for their guidance that has helped me navigate the scientific and non-scientific aspects of academic life.

I also want to thank my other collaborators Lukas Schmitt, Lukas Brenner, Tho\-mas Dubach, Stefan Wörner, Kristan Temme, Sergey Bravyi and Jay Gambetta for the many fascinating research projects that I could partake in.
My gratitude also extends to the present and past members of the quantum information theory group for the supportive and friendly atmosphere.
I especially appreciate the numerous and highly productive discussions with Lukas Schmitt and David Sutter that were instrumental in the development of this thesis.

I also want to thank both of my sisters as well as my parents for their constant support and encouragement.
Finally, and most importantly, I would like to thank my partner Anna for her unwavering support and the willingness to put up with the ups and downs that are an inevitable part the PhD journey.
I am immensely grateful to have shared this period of my life with her.
\vspace*{\fill}
\begin{flushright}
Christophe Piveteau\\
Zurich, January 2025
\end{flushright}

\chapter*{Abstract}
It is one of the most fundamental objectives in quantum information science to understand the boundary between the computational power of classical and quantum computers.
One possible avenue to explore this boundary is to identify classes of quantum circuits that can be efficiently simulated on a classical computer.
In recent years, this idea has been extended one step further:
Instead of simulating a general quantum circuit with a classical device, new schemes have emerged to simulate them on a \emph{quantum} device that is restricted in some manner.
As such, these techniques allow us to study how the restrictions impact the computational power of the device.
One such technique is called \emph{quasiprobability simulation} (QPS) and it estimates the result of a quantum circuit with a Monte Carlo procedure that randomly replaces circuit elements with ones that can be executed on the restricted quantum device.

The main focus of this thesis is dedicated to the QPS-based simulation of nonlocal quantum computation using local quantum operations.
On the practical side, this enables the simulation of large quantum circuits using multiple smaller quantum devices --- a procedure that is sometimes called \emph{circuit knitting}.
We uncover a rich mathematical formalism with many connections to the resource theory of entanglement.
We characterize the optimal simulation overhead for a broad range of practically relevant nonlocal states and channels and we explicitly provide achieving protocols.
Moreover, we also investigate the utility of classical communication between the local parties.
Our results address both the single-shot and asymptotic regime.

Furthermore, this thesis also presents a comprehensive overview of recent developments of QPS.
Besides the simulation of nonlocal computation, we also investigate the simulation of magic computation with a Clifford device, the simulation noise-free computation with a noisy device as well as the simulation of nonphysical operations.
We frame QPS in a quantum resource theoretic framework, which highlights similarities that arise in the different instantiations of the technique.
Furthermore, we study the importance of classical side information in the QPS procedure and how it impacts the overhead and expressibility of QPS.

Overall, this thesis provides a self-contained introduction for researchers interested in learning about circuit knitting and QPS.

\begin{otherlanguage}{german}
\chapter*{Zusammenfassung}
Es ist eines der grundlegendsten Ziele der Quanteninformationswissenschaft, die Grenze zwischen der Rechenleistung klassischer und Quantencomputer zu verstehen. Ein möglicher Ansatz, diese Grenze zu erforschen, besteht darin, Klassen von Quantenschaltkreisen zu identifizieren, die effizient auf einem klassischen Computer simuliert werden können. In den letzten Jahren wurde diese Idee um einen Schritt erweitert: Anstatt einen allgemeinen Quantenschaltkreis mit einem klassischen Gerät zu simulieren, sind neue Ansätze entstanden, um diese auf einem \emph{eingeschränkten Quantencomputer} zu simulieren. Solche Techniken erlauben es uns, zu untersuchen, wie diese Einschränkungen die Rechenleistung des Geräts beeinflussen. Eine solche Technik nennt sich \emph{quasiprobability simulation} (QPS), bei der das Ergebnis eines Quantenschaltkreises mithilfe eines Monte-Carlo-Verfahrens geschätzt wird, das Schaltelemente zufällig durch solche ersetzt, die auf dem eingeschränkten Quantencomputer ausführbar sind.

Der Hauptfokus dieser Arbeit liegt auf der QPS-basierten Simulation nichtlokaler Quantenrechnungen unter Verwendung lokaler Quantenoperationen. Auf der praktischen Seite ermöglicht dies die Simulation grosser Quantenschaltkreise mithilfe mehrerer kleinerer Quantencomputer – ein Verfahren, das manchmal als \emph{circuit knitting} bezeichnet wird. Wir entdecken einen reichen mathematische Formalismus mit zahlreichen Verbindungen zur Ressourcentheorie der Verschränkung. Wir charakterisieren den optimalen Simulationsaufwand für eine breite Palette praktisch relevanter nichtlokaler Zustände und Kanäle und stellen explizite Protokolle vor, die diesen Simulationsaufwand erreichen. Darüber hinaus untersuchen wir den Nutzen klassischer Kommunikation zwischen den lokalen Parteien. Unsere Ergebnisse behandeln sowohl das single-shot als auch das asymptotische Regime.

Des Weiteren bietet diese Arbeit einen umfassenden Überblick über aktuelle Entwicklungen im Bereich der QPS. Neben der Simulation nichtlokaler Berechnungen untersuchen wir auch die Simulation von ``magischen'' Berechnungen mit einem Clifford-Computer, die Simulation rauschfreier Berechnungen mit einem fehlerbehafteten Gerät sowie die Simulation nichtphysikalischer Operationen. Wir rahmen QPS in einem ressourcentheoretischen Framework ein, das die Gemeinsamkeiten zwischen den verschiedenen Varianten dieser Technik hervorhebt. Ausserdem untersuchen wir die Bedeutung klassischer Zusatzinformationen im QPS-Verfahren und deren Einfluss auf den Simulationsaufwand und die Ausdruckskraft von QPS.

Insgesamt bietet diese Arbeit eine Einführung für Forschende, die sich für circuit knitting und QPS interessieren.
\end{otherlanguage}

\setcounter{tocdepth}{3}
\tableofcontents

\chapter{Introduction}
Over the last few decades, the study of quantum information theory has produced remarkable insights into how the physical laws of quantum mechanics impact the fundamental limits of information storage, processing and transmission.
The perhaps most prominent instance is the concept of quantum computing, which harnesses the quantum mechanical nature of quantum bits (qubits) to perform calculations in a fundamentally different way than a conventional classical computer.
The construction of a large-scale fault-tolerant quantum computer is one of the greatest scientific and technological challenges of our time, with vast efforts across academia and industry currently dedicated to this goal.

It is a widely held belief that quantum computers can solve certain problems exponentially faster than their classical counterparts.
Most famously, Peter Shor discovered a quantum algorithm in 1994 that efficiently computes the prime factorization of a number~\cite{shor1997_polynomial}.
This task is believed to be so difficult for classical computers that many cryptographic protocols fundamentally rely on its hardness.
Since Shor's discovery, many other tasks have been found with presumed polynomial or superpolynomial speedups.

To better understand the boundary between the computational power of classical and quantum computing, there are broadly two approaches that can be taken.
On one hand, one can attempt to find evidence for the hardness of simulating certain classes of quantum algorithms.
Efforts in this direction include results on boson sampling~\cite{aaronson2011_computational} and instantaneous quantum polynomial (IQP) circuits~\cite{bremner2011_classical,bremner2016_average}.
On the converse side, one can identify classes of quantum circuits that can be efficiently simulated on classical computers.
The prime example thereof is the famous Gottesman-Knill theorem, which states that any quantum circuit involving only Clifford operations can be efficiently classically simulated~\cite{gottesman1998_heisenberg,aaronson2004_improved}.
This result can be extended to quantum circuits containing only a small (i.e. logarithmic in system size) number of non-Clifford gates~\cite{pashayan2015_estimating,bravyi2016_improved}.
Other efficiently simulatable classes include, for example, circuits with low entanglement~\cite{vidal2003_efficient} and matchgate circuits~\cite{terhal2002_classical,valiant2001_quantum,josza2008_matchgates}.
These various insights provide us with an understanding of certain features that must be present in a quantum algorithm to be out of reach for classical computers.

In this thesis, we will consider an extension of this idea that takes it one step further.
Instead of simulating a quantum computer with a classical computer, we will consider the task of simulating it with some notion of a ``restricted'' quantum computer.
The utility of this line of work is twofold.
On one hand, it provides us with a more fine-grained understanding of the separation of classical and quantum computers by studying intermediate computational models that lie between the two.
On the other hand, these simulation algorithms can potentially be of practical interest from a technological point of view.
Near-term quantum devices will likely suffer from severe limitations, such as high noise levels and a limited number of (logical) qubits, for the foreseeable future.
Understanding how these restrictions affect their computational power is of significant interest.
More concretely, these simulation algorithms can identify classes of quantum circuits that can already be simulated on near-term devices.
Quantum circuits in this class, which also lie outside the reach of classical computers, constitute potential candidates for demonstrating a quantum advantage.

Our main workhorse in this thesis will be a technique which we refer to as \emph{quasiprobability simulation} (QPS).
Its roots lie in the simulation of near-Clifford quantum circuits with classical computers, but in recent years, it has been increasingly realized that it can also be extended to simulate quantum circuits with quantum computers.
Consider a restricted quantum computer that can only realize some limited non-universal ``instruction set'' $\ds$, i.e., the circuits that our device can execute may only contain operations that lie in $\ds$.
To simulate a circuit that contains a non-supported operation $\cE\notin\ds$, we first need to find a so-called \emph{quasiprobability decomposition} (QPD)
\begin{equation}\label{eq:qpd}
  \cE = \sum_{i=1}^m a_i \cF_i \quad\quad\text{for some}\quad a\in\mathbb{R}^m, \cF_i\in\ds \, .
\end{equation}
Given this QPD, we can simulate $\cE$ by probabilistically replacing it with the operations $\cF_i$ and appropriately post-processing the final measurement outcome of the circuit, in a fashion reminiscent of quantum Monte Carlo algorithms.
This allows for an unbiased estimation of the expectation value of an observable at the end of the circuit, albeit with an increased sampling overhead.
The increased number of required shots is characterized by the 1-norm $\norm{a}_1=\sum_i \abs{a_i}$ of the QPD.
For a circuit with $m$ instances of the unsupported gate $\cE$, the number of circuit executions is increased by $\norm{a}_1^{2m}$.
Clearly, when performing QPS, one should use the best possible QPD that exhibits the smallest possible overhead.
This smallest achievable 1-norm is such an important quantity, that we endow it with its own symbol $\gamma_{\ds}(\cE)$, and we call it the \emph{quasiprobability extent} of $\cE$ w.r.t. $\ds$.

In this thesis, we will explore four different notions of a restricted quantum computer, which essentially amount to choosing different decomposition sets $\ds$.
In each setting, one of the central objectives will be to characterize the corresponding quasiprobability extent $\gamma_{\ds}$.
\begin{description}
  \item[Simulating nonlocal computation with local operations]
  In this instance of QPS, we consider the simulation of operations $\cE$, which act nonlocally across two parties, using only local operations (and possibly classical communication).
  This can be used to simulate nonlocal entangling gates across two isolated quantum computers that cannot exchange any quantum information.
  Put differently, we will investigate the overhead for simulating large-scale quantum computation using smaller computers of bounded size.
  The sampling overhead is exponential in the amount of nonlocal interaction between the individual parties.

  On a more practical note, this setting has potential technological implications.
  Near-term quantum computers are likely to suffer from a small number of available (logical) qubits.
  If a desired quantum circuit is almost local, we can simulate it using two or more smaller quantum devices in a process sometimes called \emph{circuit knitting}.
  This can be practically useful, provided that the smaller local quantum computations are out of reach of classical simulation.

  \item[Simulating noise-free computation with a noisy quantum computer]
  The most significant obstacle in experimentally realizing large-scale robust quantum computers is noise, which leads to erroneous results in the computation.
  Due to decoherence from interactions with the environment, quantum information is inherently more prone to noise than its classical counterpart.
  The universally accepted approach to remedy effects of noise and decoherence in quantum computing is the use of \emph{quantum error correction} (QEC)~\cite{shor1995_scheme,steane1996_error,calderbank1996_good}.
  The celebrated threshold theorem guarantees that errors in a quantum computation can be suppressed efficiently to any level for arbitrarily long circuits, provided that the physical noise rate lies below some threshold value~\cite{shor1996_faulttolerant,knill1998_resilient,aharonov2008_faulttolerant}.
  The low noise thresholds and the qubit overhead needed to implement error correction protocols are extremely challenging.

  As an alternative approach to remove noise on an imperfect computer that doesn't fulfill the stringent requirements of QEC, one can utilize QPS to simulate ideal noise-free gates\footnote{In fact, one could also view QEC itself as a simulation algorithm of noise-free computation using a noisy quantum computer.}.
  In this case, the decomposition set $\ds$ contains noisy operations that the device can physically execute.
  This is an instance of a \emph{quantum error mitigation} technique, a family of recently proposed methods to remove the effect of noise on near-term devices in a hardware-friendly manner.
  QPS-based error mitigation has been a central component for some of the most impressive hardware demonstrations in recent times~\cite{kim2023_evidence}.

  \item[Simulating magic computation with a Clifford quantum computer]
  In this setting, we consider the decomposition set $\ds$ which only contains Clifford operations.
  Note that such a quantum computer is essentially equally powerful to a classical computer, due to the Gottesman-Knill theorem.
  Clifford operations are not universal and must be complemented with additional non-classical (also called ``magic'') states or channels to achieve universality.

  Non-Clifford operations, such as the T gate, can be simulated using QPS.
  This results in a classical simulation algorithm for Clifford+T circuits with a runtime that is polynomial in the circuit size, but exponential in the number of T gates.

  \item[Simulating non-physical computation]
  Finally, we will also consider QPS w.r.t. the ``maximal'' decomposition set $\ds$ which contains all physically implementable (i.e. CPTP) maps.
  We will characterize the quasiprobability extent $\gamma_{\cptp}(\cE)$ of non-physical maps $\cE$ that may not be completely positive or trace-preserving.
  In a sense, this can be understood as simulating a computer that is even more powerful than a full-blown ideal quantum computer.
  Interestingly, this can have certain practical applications, which we will briefly explore.
\end{description}
The four settings above amount to choosing different decomposition sets $\ds$, and for each we will strive to evalute the quasiprobability extent for some target operation $\cE$ that lies outside $\ds$.
This setup is very reminiscent of a \emph{quantum resource theory} (QRT) which aims to characterize a certain quantum phenomenon by dividing the set of all quantum states and channels into ``free'' and ``resourceful'' ones.
If a QRT captures some quantum phenomenon that is difficult to realize for a quantum computer, then it is very natural to pick $\ds$ to be the free channels of that theory.
In fact, the simulation of nonlocal and magic computation that we outlined above precisely correspond to the QRTs of entanglement and magic.
Arguably, the simulation of non-physical operations can also be attributed to the trivial resource theory where all states and channels are free.
In this thesis, we will explore this connection between QPS and QRTs.
We will observe how certain techniques and properties involving the quasiprobability extent can be traced back to more general features of the underlying QRT.

One might naturally assume that the decomposition set $\ds$ should always be chosen to be a subset of the physical quantum channels, i.e., only contain CPTP maps.
However, this cannot capture the possibility of having intermediate measurements that produce classical side information which is then included in the QPS post-processing.
For instance, it is possible to use completely positive trace-non-increasing maps as part of $\ds$ by weighting a sample with zero if a certain measurement outcome occurs.
More generally, it is even possible to simulate non-completely positive maps using negative weights.
In much of the QPS literature, this trick is not utilized, or at least not to its full extent.
In many instances, the implications of using classical side information were previously poorly understood.
We aim to rectify this shortcoming in this thesis by introducing a general framework that captures the utilization of classical side information in its full generality.
Particular emphasis is placed on understanding in which settings the side information provides an advantage in terms of the sampling overhead of QPS.

Another aspect that we will investigate in this thesis is the asymptotic cost of performing QPS.
As we will see, the quasiprobability extent is sub-multiplicative $\gamma_{\ds}(\cE\otimes \cF)\leq \gamma_{\ds}(\cE)\gamma_{\ds}(\cF)$ and in many instances even \emph{strictly} sub-multiplicative.
This has direct practical consequences, as simulating multiple parallel operations together can be strictly cheaper than simulating them individually.
As such, we will be interested in the regularized extent $\smash{\scriptstyle\gammareg_{\ds}(\cE) \coloneqq \lim\limits_{n\rightarrow\infty} \gamma_{\ds}(\cE^{\otimes n})^{1/n}}$.
Interestingly, an even lower asymptotic cost can be achieved by considering some \emph{approximate} notion of QPS with a vanishing asymptotic error $\smash{\scriptstyle\gammasreg_{\ds}(\cE) \coloneqq \lim\limits_{\epsilon\rightarrow 0^+}\liminf\limits_{n\rightarrow\infty} \gamma_{\ds}^{\epsilon}(\cE^{\otimes n})^{1/n}}$.
Here, $\gamma_{\ds}^{\epsilon}$ is the \emph{smoothed} quasiprobability extent which captures the minimal sampling overhead if we allow for an error of at most $\epsilon$ in the QPS.
Evaluating $\gammareg_{\ds}$ is generally very difficult, and $\gammasreg_{\ds}$ even more so.
Nevertheless, we will remarkably find some explicit examples where $\gammareg_{\ds}$ is strictly larger than $\gammasreg_{\ds}$, meaning that approximate QPS is asymptotically more efficient.

We briefly highlight that the largest focus of this thesis will be on the QPS of non-local operations.
This theory exhibits a rich mathematical structure and we establish many connections to previously known results in the study of entanglement.
We investigate the importance of classical communication between the smaller quantum devices for the task of circuit knitting and explicitly characterize the quasiprobability extent for a wide range of nonlocal states and channels, both in the single-shot and asymptotic regime.

\section{A brief history of quasiprobability simulation}
The notion of quasiprobability distribution has a long history in quantum physics.
The Wigner function, introduced by Eugene Wigner in 1932, is a widely used representation of quantum states of continuous-variable systems.
It can be understood as a distribution on the phase space of the system, though it can also exhibit negative values (hence the \emph{quasi} in quasiprobability).
The negativity in the Wigner function is strongly linked to non-classical properties, such as nonlocality~\cite{anatole2004_negativity}.
As such, it is a natural question to ask which states are ``classical'' in the sense that they exhibit a non-negative Wigner function.
This problem is solved by Hudson's theorem, which states that a pure state has a non-negative Wigner function if and only if it is a Gaussian state.

Following the rise of quantum information theory, considerable interest has been directed towards studying the finite-dimensional counterparts of the Wigner function.
In particular, odd-dimensional qudit systems exhibit a well-behaved Wigner function which is non-negative for a pure state if and only if it is a stabilizer state~\cite{gross2006_hudson}.
Stabilizer states are thus endowed with some notion of being ``classical'', analogous to Gaussian states.
This insight can be translated to the setting of computation, as it was soon realized that any odd-dimensional qudit quantum circuit can be efficiently classically simulated, as long as its elements (initial states, gates and measurements) exhibit non-negative Wigner functions~\cite{veitch2012_negative,mari2012_positive}.
This can be considered a generalization of the Gottesman-Knill theorem~\cite{gottesman1998_heisenberg,aaronson2004_improved}.

In a seminal work in \reference~\cite{pashayan2015_estimating}, Pashayan, Wallman and Bartlett extended this result and demonstrated a classical simulation algorithm for quantum circuits where the runtime is governed by the amount of negativity in the Wigner function of the circuit elements\footnote{This negativity is related to the magic monotone called \emph{mana}, introduced in \reference~\cite{veitch2014_resource}.}.
This can be considered the first instance of QPS, though it relies on the Wigner representation of the circuit elements (or more generally a representation w.r.t. to some fixed frame and dual frame), which has no well-behaved analog for qubits.
Only later work by Howard and Campbell applied QPS in the qubit setting~\cite{howard2017_application} by representing magic states as a QPD of stabilizer states.
As such, they defined and investigated the quasiprobability extent of states for the first time.
Shortly thereafter, Bennink \etal~proposed a generalization to simulate non-Clifford qubit channel with stabilizer channels~\cite{bennink2017_unbiased}, providing the first qubit instance of QPS of channels (rather than states).

Until that point, all discussed references exclusively focused on \emph{classical} simulation of quantum circuits.
In \reference~\cite{temme2017_errormitigation}, Temme, Bravyi and Gambetta came to the pivotal insight that the QPS technique can also be used to simulate quantum circuits using a quantum computer itself.
More precisely, they considered using QPS to simulate the noise-free execution of a circuit with an imperfect noisy device.
Since then, other such applications have been proposed and explored, such as the simulation of non-local computation~\cite{mitarai2021_constructing,piveteau2024_circuit}
and unphysical computation~\cite{jiang2021_physical,regula2021_operational}.
They will be introduced in extended details throughout the thesis.

\ifdefined\arxivmode
\section{Outline}
The preliminaries in~\Cref{chap:preliminaries} can be read as needed and are mostly intended to specify the notation and definitions used throughout the thesis.
In~\Cref{chap:qpsim}, we formally introduce QPS and the general framework for including classical side information.
We also study various properties of the quasiprobability extent, especially in the context of general QRTs.
While~\Cref{chap:qpsim} lays the crucial foundation for this thesis, the subsequent ~\Cref{chap:nonphysical,chap:nonlocal,chap:magic,chap:noisefree} can in principle be read independently and in any order, as they treat four separate applications of QPS.
We summarize some important insights from this thesis as well as open questions in~\Cref{chap:conclusion}.

In the author's opinion, the most important contributions of this thesis lie in~\Cref{chap:nonlocal}.

\else 

\section{Outline and suggested reading order}
The preliminaries in~\Cref{chap:preliminaries} can be read as needed and are mostly intended to specify the notation and definitions used throughout the thesis.
In~\Cref{chap:qpsim}, we formally introduce QPS and the general framework for including classical side information.
We also study various properties of the quasiprobability extent, especially in the context of general QRTs.
While~\Cref{chap:qpsim} lays the crucial foundation for this thesis, the subsequent ~\Cref{chap:nonphysical,chap:nonlocal,chap:magic,chap:noisefree} can in principle be read independently and in any order, as they treat four separate applications of QPS.
We summarize some important insights from this thesis as well as open questions in~\Cref{chap:conclusion}.

In the author's opinion, the strongest contributions of this thesis lie in~\Cref{chap:nonlocal}.
As such, we suggest following reading order for a reader that is interested in the most important results.
\begin{itemize}[noitemsep]
  \item \Cref{sec:motivational_example,sec:qpsim} for an introduction to QPS (\Cref{sec:motivational_example} can be skipped if the reader has previous experience with QPS).
  \item \Cref{sec:qps_intermediate_measurements}, which is crucial to establish the formalism of QPS with classical side information.
  \item The introduction of~\Cref{sec:gamma_qrt} as well as the paragraph \emph{Asymptotic simulation cost}.

  In general, it might also be useful to briefly skim through the~\Cref{sec:qps_basic_properties,sec:gamma_qrt}.
  These results will be regularly referred to from subsequent chapters and can also be read when needed.
  \item The introduction of~\Cref{chap:nonlocal}.
  \item \Cref{sec:gate_cut_side_info}.
  \item \Cref{sec:gamma_nonlocal_states}.
  \item \Cref{sec:gamma_nonlocal_channels} up to (and including) \Cref{sec:gamma_kaklike}.
  \item \Cref{sec:submult}.
  \item The introduction of \Cref{sec:timelike_cuts}.
  \item \Cref{sec:wire_cut_side_info}.
  \item \Cref{sec:id_wire_cut}.
\end{itemize}
After~\Cref{chap:nonlocal}, we suggest continuing to~\Cref{chap:nonphysical}.

\fi 

\section{Overview of contributions and published papers}
To provide a cohesive treatment of the subject, this thesis contains not solely contributions by the author, but also some results from other researchers.
Every main chapter in this thesis includes an ``overview of contribution'' section which precisely outlines which parts are due to the author which are not.
The current section serves as an additional brief overview of where the main contributions are located.
We also highlight that this thesis contains a large amount of novel contributions which have not been published before, since they were either deemed not interesting enough to warrant writing a paper or because they will be included in an upcoming publication.

In~\Cref{chap:qpsim}, the main novelty lies in the establishment of a general framework to describe QPS with classical side information.
We briefly summarize the contributions in the subsequent four chapters, in decreasing order of importance.
\begin{itemize}[noitemsep]
  \item Almost the entirety of~\Cref{chap:nonlocal} is novel work by the author, with a large part being previously unpublished.
  \item QPDs of nonphysical operations, which are treated in~\Cref{chap:nonphysical}, were originally introduced by the author.
  However, the chapter also contains more advanced results which were subsequently found by other researchers.
  \item The general discussion of QEM in~\Cref{chap:noisefree} summarizes various techniques in the field that were not introduced by the author.
  However, the~\Cref{sec:qemqec} is dedicated to the author's publication in \reference~\cite{piveteau2021_errormitigation} discussing the interplay between QEM and QEC.
  \item Finally, \Cref{chap:magic} contains almost no contributions by the authors and is mostly included for completeness.
\end{itemize}
This thesis includes the content of following publications by the author.
\begin{itemize}[noitemsep]
  \item \cite{piveteau2024_circuit}: \highlightnames*{1}\fullcite{piveteau2024_circuit}
  \item \cite{brenner2023_optimal}: \highlightnames*{2}\fullcite{brenner2023_optimal}
  \item \cite{schmitt2024_cutting}: \highlightnames*{2}\fullcite{schmitt2024_cutting}
  \item \cite{piveteau2021_errormitigation}: \highlightnames*{1}\fullcite{piveteau2021_errormitigation}
\end{itemize}
The following publication is also about QPS and its application to QEM.
\begin{itemize}[noitemsep]
  \item \cite{piveteau2022_quasiprobability}: \highlightnames*{1}\fullcite{piveteau2022_quasiprobability}
\end{itemize}
As the main result is somewhat heuristic in nature and doesn't fit nicely with the rest of the thesis, we refrain from introducing it in detail.
However, some of its technical results will be included where relevant.

Finally, we note three further publications by the author in the domain of quantum error correction.
They are not covered in this thesis in any form.
\begin{itemize}[noitemsep]
  \item \cite{piveteau2022_bpqm}: \highlightnames*{1}\fullcite{piveteau2022_bpqm}
  \item \cite{pfister2022_bpqm}: \highlightnames*{2}\fullcite{pfister2022_bpqm}
  \item \cite{piveteau2024_tnd}: \highlightnames*{1}\fullcite{piveteau2024_tnd}
\end{itemize}
\resetnamehighlights


\chapter{Preliminaries}\label{chap:preliminaries}
In this chapter, we briefly summarize the notation and conventions that we use throughout this manuscript.
Furthermore, we review some crucial prerequisites from basic quantum theory, quantum computing and quantum resource theory.
An experienced reader may directly skip to~\Cref{chap:qpsim}.

\section{Mathematical notation}
Throughout this manuscript, we will always restrict ourselves to finite-dimensional quantum systems.
A list of used mathematical symbols is depicted in~\Cref{tab:notation}.
\begin{center}
\begin{longtable*}{p{3.0cm}p{10.5cm}}
  \multicolumn{2}{l}{\textbf{General}} \\
  \hline
  $\mathbb{C},\mathbb{R},\mathbb{N}$& complex, real and natural numbers where $\mathbb{N}$ starts with $1$ \\
  $A,B,C,\dots$         & typical names for finite-dimensional quantum systems \\
  $\cH_A,\cH_B,\cH_C,\dots$ & Hilbert spaces corresponding to quantum systems $A,B,C,\dots$ \\
  $AB$                  & joint quantum system corresponding to $\cH_A\otimes\cH_B$ \\
  $\dimension(A)$       & dimension of Hilbert space $\cH_A$ \\
  \footnotesize{$\spn_{\mathbb{R}}(X),\spn_{\mathbb{C}}(X)$}& span of set $X$ over the real and complex number fields \\
  $\sgn(x)$             & sign function $\sgn(x)=1$ for $x\geq 0$ and $\sgn(x)=-1$ for $x<0$ \\
  $\log(x)$             & logarithm of $x$ to the base $2$ \\
  $\ln(x)$              & natural logarithm of $x$ \\
  $\re[z],\im[z]$       & real and imaginary parts of complex number $z\in\mathbb{C}$ \\
  $\conv(X)$            & convex hull of the set $X$ \\
  $\aconv(X)$           & absolute convex hull of the set $X$ \\
  $\diag(\lambda)$      & diagonal $n\times n$ matrix with entries $\lambda\in\mathbb{C}^n$ \\
 \hline
\end{longtable*}
\begin{longtable*}{p{3.0cm}p{10.5cm}}
  \multicolumn{2}{l}{\textbf{Operators}} \\
  \hline
  $\lino(A,B)$          & set of linear operators from $A$ to $B$ \\
  $\herm(A)$            & set of Hermitian operators on $A$ \\
  $\pos(A)$             & set of positive semi-definite operators on $A$ \\
  $\uni(A)$             & set of unitary operators on $A$ \\
  $\dops(A)$            & set of density operators on $A$ \\
  $\rho,\sigma,\tau$    & typical names for density operators \\
  $\supo(A\rightarrow B)$& set of linear superoperators $\lino(\lino(A),\lino(B))$ \\
  $\cE,\cF,\cG$         & typical names for superoperators \\
  $\hp(A\rightarrow B)$ & set of Hermitian-preserving superoperators from $A$ to $B$ \\
  $\cp(A\rightarrow B)$ & set of completely positive superoperators from $A$ to $B$ \\
  $\tp(A\rightarrow B)$ & set of trace-preserving superoperators from $A$ to $B$ \\
  $\hptp(A\rightarrow B)$ & set of Hermitian-preserving trace-preserving superoperators from $A$ to $B$ \\
  $\cptp(A\rightarrow B)$ & set of quantum channels from $A$ to $B$ \\
  $\cptni(A\rightarrow B)$ & set of completely positive trace-non-increasing maps from $A$ to $B$ \\
  $\qi(A\rightarrow B)$ & set of quantum instruments from $A$ to $B$ \\
  $\povm(A)$            & set of POVMs on $A$ \\
  $\cI,\cJ$             & typical names for quantum instruments \\
  $\schan(A,B\rightarrow C,D)$     & set of superchannels from $\cptp(A\rightarrow B)$ to $\cptp(C\rightarrow D)$ \\
  $\id_A$               & identity map from $A$ to $A$ \\
  $\idchan_A$           & identity channel on $A$ (map from $\lino(A)$ to $\lino(A)$) \\
  $\choi{\cE}$          & Choi operator of superoperator $\cE\in\supo(A\rightarrow B)$ \\
  $O^{\dagger}$         & conjugate transpose of operator $O$ \\
  $\abs{O}$             & absolute value $\sqrt{O^{\dagger}O}$ of operator $O$ \\
  $O_1\otimes O_2$      & tensor product of operators $O_1$ and $O_2$ \\
  $\tr[O]$              & trace of operator $O$ \\
  $\tr_A[O]$            & partial trace of operator $O$ over subsystem $A$ \\
  $\transpose_A$        & transpose map on system $A$ \\
  $\indsupo{O}$         & linear superoperator $X\mapsto OXO^{\dagger}$ induced by operator $O$ \\
  $\langle O_1,O_2\rangle$ & Hilbert-Schmidt inner product $\frac{1}{\dimension(A)}\tr[O_1^{\dagger}O_2]$ for $O_i\in\lino(A)$ \\
  $O_1\loewnergeq O_2$  & Loewner order on Hermitian operators $O_1,O_2\in\herm(A)$ \\
  \hline
\end{longtable*}
\begin{longtable*}{p{3.0cm}p{10.5cm}}
  \multicolumn{2}{l}{\textbf{Norms}} \\
  \hline
  $\norm{x}_p$          & $p$-norm $\left(\sum_i\abs{x_i}^p\right)^{1/p}$ of vector $x\in\mathbb{C}^n$ \\
  $\norm{O}_p$          & Schatten $p$-norm of operator $O\in\lino(A,B)$ \\
  $\opnorm{O}$          & spectral norm of operator $O\in\lino(A,B)$ \\
  $\dnorm{\cE}$         & diamond norm of superoperator $\cE\in\supo(A\rightarrow B)$ \\
  \hline
\end{longtable*}
\begin{longtable*}{p{3.0cm}p{10.5cm}}
  \multicolumn{2}{l}{\textbf{Quantum resource theories}} \\
  \hline
  $\fs(A)$              & set of free states on $A$ \\
  $\qrt(A\rightarrow B)$ & set of free channels from $A$ to $B$ \\
  $\fqi(A\rightarrow B)$ & set of free quantum instruments from $A$ to $B$ \\
  $\fsc(A,B\rightarrow C,D)$ & {\small set of free superchannels from $\cptp(A\rightarrow B)$ to $\cptp(C\rightarrow D)$} \\
  \hline
\end{longtable*}
\begin{longtable*}{p{3.5cm}p{10.0cm}}
  \multicolumn{2}{l}{\textbf{Entanglement theory}} \\
  \hline
  $\sep(A;B)$           & set of separable states across $A$ and $B$ \\
  $\ppt(A;B)$           & set of states across $A$ and $B$ with positive partial transpose \\
  $\lo(A;B\rightarrow A';B')$      & set of local operations \\
  \small{$\locc(A;B\rightarrow A';B')$} & set of local operations with classical communication \\
  \small{$\sepc(A;B\rightarrow A';B')$} & set of separable channels \\
  \small{$\pptc(A;B\rightarrow A';B')$} & set of PPT channels \\
  $\ebc(A\rightarrow A')$& set of entanglement breaking channels \\
  \hline
\end{longtable*}
\begin{longtable}{p{3.0cm}p{10.5cm}}
  \multicolumn{2}{l}{\textbf{Abbreviations}} \\
  \hline
  PPT                   & positive partial transpose \\
  SDP                   & semidefinite program \\
  QPS                   & quasiprobability simulation \\
  QPD                   & quasiprobability decomposition \\
  QRT                   & quantum resource theory \\
  QEM                   & quantum error mitigation \\
  QEC                   & quantum error correction \\
  \hline
  \caption{Summary of symbols used throughout this thesis.}
  \label{tab:notation}
\end{longtable}
\end{center}
For the above sets of operators, whenever the input and output systems are identical, we only specify one system as the argument.
For instance, the set of endomorphisms is $\lino(A)\coloneqq\lino(A,A)$ and similarly $\cptp(A)\coloneqq\cptp(A\rightarrow A)$, $\qi(A)\coloneqq\qi(A\rightarrow A)$ and so forth.

Whenever we write the composition $\cE_2\circ\cE_1$ of two superoperators which do not have matching image and support spaces, the operations are implicitly meant to be tensored with the identity channel.
For instance, if $\cE_1\in\supo(A\rightarrow BC)$ and $\cE_2\in\supo(BD\rightarrow E)$, then
\begin{equation}
  \cE_2\circ\cE_1\coloneqq (\cE_2\otimes\idchan_C)(\cE_1\otimes\idchan_D) \in \supo(AD\rightarrow CE) \, .
\end{equation}

\section{Elementary quantum gates}
We assume basic familiarity with the quantum circuit formalism and refer the reader to~\cite{nielsen2010_quantum} for a detailed introduction.
Here, we briefly summarize the most important gates that we will encounter.

Each of the three Pauli matrices
\begin{equation}
  X = \begin{pmatrix}0 & 1 \\ 1 & 0\end{pmatrix} \quad
  Y = \begin{pmatrix}0 & -i \\ i & 0\end{pmatrix} \quad
  Z = \begin{pmatrix}1 & 0 \\ 0 & -1\end{pmatrix} \, .
\end{equation}
induces a corresponding single-qubit gate.
The $n$-fold tensor products of the Pauli matrices and the identity matrix generate the $n$-qubit \emph{Pauli group}
\begin{equation}
  \mathrm{P}_n \coloneqq \{ i^c Q_1\otimes Q_2\otimes \dots \otimes Q_n | c\in\{0,1,2,3\}, Q_i\in \{\id,X,Y,Z\} \} \, .
\end{equation}
Note that the phase $i^c$ is of no importance when considering an $n$-qubit Pauli operator as a gate.
Based on the Pauli group, we can define the \emph{Clifford hierarchy}, an increasing sequence of sets of $n$-qubit gates recursively defined as
\begin{equation}
  \mathrm{CH}_n^{(1)} \coloneqq \mathrm{P}_n , \quad
  \mathrm{CH}_n^{(k)} \coloneqq \{U\in \uni(\mathbb{C}^{2^n}) | \forall Q\in \mathrm{P}_n: UQU^{\dagger}\in \mathrm{CH}_n^{(k-1)} \} \, .
\end{equation}

The second level of the hierarchy $\mathrm{CH}_n^{(2)}$ is called the $n$-qubit \emph{Clifford group} and plays a central role in the theory of quantum computation.
Important single-qubit Clifford gates include the Hadamard and phase gates
\begin{equation}
  \mathrm{H} = \frac{1}{\sqrt{2}}\begin{pmatrix}1 & 1 \\ 1 & -1\end{pmatrix} \quad
  \mathrm{S} = \begin{pmatrix}1 & 0 \\ 0 & i\end{pmatrix}
\end{equation}
and all two-qubit Clifford gates are equivalent to either the $\cnot$, $\swap$ or $\iswap$ gate
{\small
\begin{equation}
  \cnot = \begin{pmatrix}1&0&0&0\\ 0&1&0&0 \\ 0&0&0&1 \\ 0&0&1&0\end{pmatrix} \quad 
  \swap = \begin{pmatrix}1&0&0&0\\ 0&0&1&0 \\ 0&1&0&0 \\ 0&0&0&1\end{pmatrix} \quad
  \iswap = \begin{pmatrix}1&0&0&0\\ 0&0&i&0 \\ 0&i&0&0 \\ 0&0&0&1\end{pmatrix}
\end{equation}
}
up to local Clifford unitaries.
The Clifford unitaries are generated by the H, S and CNOT gates.
The well-known Gottesman-Knill theorem~\cite{gottesman1998_heisenberg,aaronson2004_improved} states that there exists a polynomial-time classical simulation algorithm for circuits consisting only of Clifford gates as well as preparations and measurements of qubits in the computational basis.
Unsurprisingly, the Clifford gates are not universal, i.e., it is not possible to approximate every unitary arbitrarily well using only Clifford gates.
To reach universality, we need at least one gate from $\mathrm{CH}_n^{(3)}$, as the third level does not exhibit a group structure as the first two levels do.
Two typical examples of third-level gates that can each achieve universality when paired with the Clifford gates are the $\mathrm{T}$ gate and the $\toffoli$ (controlled-controlled-not) gates
\begin{equation}
  \mathrm{T} = \begin{pmatrix}1 & 0 \\ 0 & e^{i\pi / 4}\end{pmatrix} \quad
  \toffoli = \proj{00}\otimes\id + \proj{01}\otimes\id + \proj{10}\otimes\id + \proj{11}\otimes X \, .
\end{equation}

\section{Quantum channels and instruments}
A state of a quantum system $A$ can most generally be described by a density operator $\rho\in\dops(A)$, i.e., a non-negative operator with unit trace $\tr[\rho]=1$.
One of the fundamental postulates of quantum mechanics states that the time evolution of an isolated system is described by $\rho \mapsto U\rho U^{\dagger}$ for some unitary operator $U\in\uni(A)$.
When the system can also interact with the environment, then the evolution can take a more general form of a quantum channel.
\begin{definition}{}{}
  A superoperator $\cE\in\supo(A\rightarrow B)$ is called \emph{trace-preserving} if
  \begin{equation}
    \forall \omega\in\lino(A) : \tr[\cE(\omega)] = \tr[\omega]
  \end{equation}
  and it is called \emph{completely positive} if $\cE\otimes\idchan_E \in \supo(AE\rightarrow BE)$ is positive for any system $E$.
  We say $\cE$ is a \emph{quantum channel} if it is both completely positive and trace-preserving.
\end{definition}
\sloppy The sets of completely positive superoperators, trace-preserving superoperators and quantum channels from $A$ to $B$ are denoted by $\cp(A\rightarrow B)$, $\tp(A\rightarrow B)$ and $\cptp(A\rightarrow B)$.
Quantum channels precisely represent the most general class of superoperators which always map density operators to density operators, even when only applied to a subsystem.

Note that there exist superoperators which are positive, but not completely positive.
A canonical example is the transpose operator $T_A\in\supo(A)$ which is defined as
\begin{equation}
  T\left(\sum_{i,j} \alpha_{i,j} \ketbra{b_i}{b_j} \right) \coloneqq \sum_{i,j} \alpha_{i,j} \ketbra{b_j}{b_i}
\end{equation}
with respect to some basis $\{\ket{b_i}\}_i$ of $A$.

Whenever the input Hilbert space is of unit dimension $\cH_A=\mathbb{C}$, the quantum channels in $\cptp(A\rightarrow B)$ can be understood as a deterministic preparation of some state on $B$.
Therefore, we can identify the set of quantum channels $\cptp(A\rightarrow B)$ with the set of quantum states $\dops(B)$.

There exists a variety of representations of quantum channels.
We summarize the ones that will be relevant throughout this manuscript.
\begin{proposition}{Choi-Jamiołkowski representation}{}
  The map from $\supo(A\rightarrow B)$ to $\lino(AB)$ defined by
  \begin{equation}
    \cE \mapsto \choi{\cE}\coloneqq (\idchan_A\otimes\cE)(\proj{\Psi}_{AA'})
  \end{equation}
  is an isomorphism, where $A'$ is a copy of $A$, $\cE$ is understood to be acting on $A'$ and $\ket{\Psi}\coloneqq \dimension(A)^{-1/2}\sum_{i=1}^{\dimension(A)}\ket{i}\otimes\ket{i}$ is the maximally entangled state w.r.t. some fixed orthonormal basis.
\end{proposition}
A major advantage of the Choi representation is that it directly captures some important properties of the superoperator $\cE$.
For instance,
\begin{equation}
  \cE \text{ is Hermitian-preserving } \Longleftrightarrow \choi{\cE} \text{ is Hermitian} \, ,
\end{equation}
\begin{equation}
  \cE \text{ is completely positive } \Longleftrightarrow \choi{\cE} \loewnergeq 0
\end{equation}
and
\begin{equation}
  \cE \text{ is trace-preserving } \Longleftrightarrow \tr_{B}[\choi{\cE}] = \frac{\id_A}{\dimension(A)} \, .
\end{equation}

Another important representation is provided by the Stinespring dilation theorem.
This result states that any quantum channel can be considered to be a unitary evolution by sufficiently enlarging the space.
\begin{proposition}{Stinespring representation}{}
  Let $\cE\in \cptp(A\rightarrow B)$.
  Then there exists a system $E$, and an isometry $V\in\lino(A,BE)$ such that
  \begin{equation}
    \forall\omega\in\lino(A): \cE(\omega) = \tr_E[V\omega V^{\dagger}] \, .
  \end{equation}
\end{proposition}

Next, we introduce the representation of completely positive linear operators based on Kraus operators.
\begin{proposition}{Operator-sum representation}{}
  Let $\cE\in \cp(A\rightarrow B)$.
  Then there exists a finite family of operators $E_i\in\lino(A,B)$, $i=1,\dots,r$ such that
  \begin{equation}
    \forall\omega\in\lino(A): \cE(\omega) = \sum_{i=1}^r E_i\omega E_i^{\dagger} \, .
  \end{equation}
  The map $\cE$ is trace-preserving if and only if $\sum_{i=1}^r E_i^{\dagger}E_i = \id_A$.
\end{proposition}

Sometimes it can be useful to study completely positive maps that are trace-non-increasing, i.e., $\tr[\cE(\rho)]\leq \tr[\rho]$.
These are characterized by their Choi operator fulfilling $\tr_B[\choi{\cE}]\loewnerleq\frac{\id_A}{\dimension(A)}$, respectively by their Kraus operators fulfilling $\sum_{i=1}^r E_i^{\dagger}E_i \loewnerleq \id_A$.

Finally, we mention one last representations that is applicable to $n$-qubit superoperators $\cE\in\supo(A)$ for $\cH_A = (\mathbb{C}^{2})^{\otimes n}$.
The \emph{Pauli transfer matrix} $M\in\mathbb{C}^{4^n\times 4^n}$ of $\cE$ is defined as 
\begin{equation}
  M_{ij}\coloneqq \langle Q_i, \cE(Q_j)\rangle = \frac{1}{2^n}\tr\left[Q_i\cE(Q_j)\right]
\end{equation}
where $\{Q_1,Q_2,\dots,Q_{4^n}\}$ are the $n$-qubit Pauli strings $\{\id,X,Y,Z\}^{\otimes n}$ in some fixed order.
As such, the Pauli transfer matrix of a superoperator can be understood as its matrix representation in terms of the Pauli strings, which form an orthonormal basis of the space of $2^n\times 2^n$ matrices.

In quantum information theory we are often interested in measurements, i.e., processes which produce classical information from a quantum state.
The most general measurement process can be modelled by a quantum channel $\cE_{A\rightarrow BX}\in\cptp(A\rightarrow BX)$ which maps any state $\rho_A\in\dops(A)$ to a so-called \emph{classical-quantum} state of the form
\begin{equation}
  \sum_i p_i \sigma^{(i)}_B\otimes \proj{i}_X
\end{equation}
where $p_i$ is a probability distribution $p_i\geq 0,\sum_i p_i=1$, $\ket{i}$ is an orthonormal basis of $X$ and $\sigma^{(i)}_B\in\dops(B)$.
As such, a generalized measurement can always be written in the form
\begin{equation}
  \cE_{A\rightarrow BX}(\rho) = \sum_i p_i(\rho) \sigma^{(i)}(\rho) \otimes \proj{i}
\end{equation}
where both $p_i(\rho)$ and $\sigma^{(i)}(\rho)$ are functions of $\rho$.
If we define $\cE_i(\rho)\coloneqq p_i(\rho)\sigma^{(i)}(\rho)$ we can equivalently write
\begin{equation}
  \cE_{A\rightarrow BX}(\rho) = \sum_i \cE_i(\rho) \otimes \proj{i} \, .
\end{equation}
Since $\cE\in\cp(A\rightarrow BX)$, the maps $\cE_i$ must also be completely positive.
Furthermore, $\sum_i\cE_i = \tr_X\circ \cE_{A\rightarrow BX}$ which is trace-preserving.
Any generalized quantum measurement is fully characterized by a family $(\cE_i)_i$ that fulfills these two properties.
\begin{definition}{}{}
  A quantum instrument from system $A$ to system $B$ is a countable family $\cI=(\cE_i)_i$ of non-zero operators $\cE_i\in\cp(A\rightarrow B)$ such that $\sum_i\cE_i\in\tp(A\rightarrow B)$.
  We denote the set of such quantum instruments by $\qi(A\rightarrow B)$.
\end{definition}
A quantum instrument is therefore a mathematical abstraction of a general quantum measurement process.
The defined objects are sometimes also called \emph{discrete} quantum instruments due to their countable nature.
In this manuscript, we will never consider non-discrete instruments.

We say that a quantum instrument $\cI'$ is a coarse-graining of another instrument $\cI$ if it corresponds to the same physical measurement process, except that we merge some of the classical measurement outcomes (i.e., we consider them to be the same outcome).
More mathematically, this is captured as follows.
\begin{definition}{}{qi_coarse_graining}
  Let $\cI=(\cE_i)_{i\in\Theta}$ and $\cI'=(\cF_j)_{j\in\Theta'}$ be quantum instruments from $A$ to $B$.
  We say that $\cI'$ is a coarse-graining of $\cI$ if there exists a map $f:\Theta\rightarrow\Theta'$ such that
  \begin{equation}
    \cF_j = \sum_{i:f(i)=j} \cE_i\, .
  \end{equation}
\end{definition}

For a quantum instrument $\cI=(\cE_i)_i\in\qi(A\rightarrow B)$, we define its associated marginal channel as
\begin{equation}
  \qichan[\cI] \coloneqq \sum_i \cE_i \in \cptp(A\rightarrow B) \, .
\end{equation}
This channel essentially corresponds to $\tr_X\circ\cE_{A\rightarrow BX}$ from before and can also be understood as the ``fully coarse-grained'' instrument of $\cI$.
Similarly, if we consider a set $Q\subset\qi(A\rightarrow B)$ of quantum instruments, we define
\begin{equation}\label{eq:qichandef}
  \qichan[Q] \coloneqq \{\qichan[\cI] | \cI \in Q\} \subset \cptp(A\rightarrow B) \, .
\end{equation}

Oftentimes, one is purely interested in the distribution of the outcomes of a measurement process and not in the precise post-measurement state.
Put differently, the only relevant quantity in this case is the quantum-to-classical channel
\begin{align}
  (\tr_B\circ \cE_{A\rightarrow BX})(\rho)
  &= \sum_i \tr[\cE_i(\rho)] \proj{i} \\
  &= \sum_i \tr[\sum_j E_j^{(i)}\rho (E_j^{(i)})^{\dagger}] \proj{i} \\
  &= \sum_i \tr[M_i\rho] \proj{i}
\end{align}
where we inserted the operator-sum representation $\cE_i(\rho) = \sum_j E_j^{(i)}\rho (E_j^{(i)})^{\dagger}$ and defined the positive operators $M_i\coloneqq \sum_j(E_j^{(i)})^{\dagger}E_j^{(i)}\in\pos(A)$ which fulfill $\sum_i M_i=\id_A$.
So to compute the statistics of the measurement outcomes, is suffices to only specify the operators $M_i$.
\begin{definition}{}{}
  A positive operator-valued measure (POVM) on the system $A$ is a countable family $(M_i)_i$ of positive operators $M_i\in\pos(A)$ that fulfills $\sum_i M_i=\id_A$.
  We denote the set of POVMs on $A$ by $\povm(A)$.
\end{definition}

\section{Quantum resource theories}\label{sec:qrt}
The framework of quantum resource theories (QRT) is used to study and characterize various phenomena in quantum mechanics, such as entanglement, coherence, magic and many more.
The idea is to consider a restricted setting where operations can only be accomplished if they do not incur the quantum phenomena in question.
As such, every QRT has a characteristic set of ``free channels'', which constitutes the $\cptp$ maps which can be realized under the restrictions of the resource theory.
Similarly, states are categorized into ``free states'' and ``resource states''.
Some of the typical goals of a QRT include the quantification of the amount of resource contained in a state as well as the convertibility between states under free operations.

For a detailed introduction to the framework of QRTs, we refer the reader to \reference~\cite{chitambar2019_quantum,gour2024_resources,gour2020_dynamical}.
Note that there are some technical variations in how resource theories are defined in literature.
We will mostly follow the conventions from~\cite{gour2024_resources}.

Mathematically speaking, a QRT is in essence an object that specifies which quantum channels between two systems $A$ and $B$ are considered ``free''.
\begin{definition}{}{qrt}
  A \emph{quantum resource theory} is a mapping $\qrt$ that takes two quantum systems $A,B$ to a set of quantum channels $\qrt(A\rightarrow B)\subset \cptp(A\rightarrow B)$ that fulfills following condition:
  \begin{enumerate}
    \item \emph{Doing nothing is free:} For any system $A$, $\idchan_A\in\qrt(A\rightarrow A)$.
    \item \emph{Composition is free:} For any $\cE\in\qrt(A\rightarrow B)$ and $\cF\in\qrt(B\rightarrow C)$ we also have $\cF\circ\cE\in \qrt(A\rightarrow C)$.
    \item \emph{Discarding is free:} For any system $A$, $\tr\in\qrt(A\rightarrow \mathbb{C})$.
  \end{enumerate}
\end{definition}
Every QRT has an associated set of free states which do not contain any resource.
\begin{definition}{}{}
  The free states of a QRT $\qrt$ on a system $A$ are $\fs(A)\coloneqq \qrt(\mathbb{C}\rightarrow A)$.
\end{definition}
The following two properties are direct consequences of properties 2 and 3 of a QRT.
\begin{lemma}{Golden rule of quantum resource theories}{qrt_golden_rule}
  Consider a QRT $\qrt$.
  If $\cE\in\qrt(A\rightarrow B)$ and $\rho\in\fs(A)$, then $\cE(\rho)\in\fs(B)$.
\end{lemma}
\begin{lemma}{}{qrt_free_state_prep}
  Consider a QRT $\qrt$ and let $\rho\in\fs(B)$.
  For any physical system $A$ the map $\cE\in\cptp(A\rightarrow B)$ defined by $\cE(\sigma)\coloneqq \tr[\sigma]\rho$ is in $\qrt(A\rightarrow B)$.
\end{lemma}

Most QRTs studied in practice also fulfill additional properties when considered on multipartite systems.
\begin{definition}{}{qrt_tp_structure}
  We say a QRT $\qrt$ has a \emph{tensor product structure} if
  \begin{enumerate}
    \item \emph{Completely free operations:} For any $\cE\in\qrt(A\rightarrow B)$ and any system $C$, one has $\cE\otimes\idchan_C\in\qrt(AC\rightarrow BC)$.
    \item \emph{Freedom of relabelling:} Let $A=A_1A_2\dots A_n$ and $B=B_1B_2\dots B_n$ be $n$-fold copies of identical systems and let $P_{\pi}^A\in\cptp(A)$, $P_{\pi^{-1}}^B\in\cptp(B)$ be permutation channels corresponding to a permutation $\pi$ of $n$ elements.
    Then $\forall\cE\in\qrt(A\rightarrow B)$ we have $P^{B}_{\pi^{-1}}\circ\cE\circ P^{A}_{\pi}\in \qrt(A\rightarrow B)$.
  \end{enumerate}
\end{definition}

Every QRT naturally induces a set of free quantum instruments.
\begin{definition}{}{free_qi}
  The set of free instruments $\fqi(A\rightarrow B)\subset\qi(A\rightarrow B)$ associated to a QRT $\qrt$ is defined to be all quantum instruments $(\cE_i)_i$ which either consist of a single element in $\qrt(A\rightarrow B)$ or for which the map
  \begin{equation}
    \rho_A \mapsto \sum_i \cE_i(\rho_A)_B \otimes \proj{i}_X
  \end{equation}
  is a free channel in $\qrt(A\rightarrow BX)$ for some system $X$ with orthonormal basis $\{\ket{i} | i=1,2,\dots\}$.
\end{definition}
A QRT with tensor product structure is fully characterized by its set of free instruments.
\begin{lemma}{}{free_chan_from_qi}
    For a QRT $\qrt$ with tensor product structure, one has
    \begin{equation}
        \qichan[\fqi(A\rightarrow B)]=\qrt(A\rightarrow B) \, .
    \end{equation}
\end{lemma}
Recall the definition of $\qichan$ in~\Cref{eq:qichandef}.
\begin{proof}
    For any $\cE\in\qrt(A\rightarrow B)$, the quantum instrument with a single component $(\cE)$ is an element of $\fqi(A\rightarrow B)$.
    Thus, $\cE=\qichan[(\cE)]\in \qichan[\fqi(A\rightarrow B)]$.
    Conversely, let $\cI=(\cE_i)_i$ be a free instrument in $\fqi(A\rightarrow B)$.
    There exists a free channel $\tilde{\cE}(\rho) = \sum_i \cE_i(\rho)\otimes \proj{i}_X$, $\tilde{\cE}\in\qrt(A\rightarrow BX)$.  
    Therefore, $\qichan[\cI]=\sum_i\cE_i = (\idchan_B\otimes tr_X)\circ \tilde{\cE} \in \qrt(A\rightarrow B)$.
\end{proof}
Consider a general quantum measurement $\cE$ of the form
\begin{equation}
  \rho \mapsto \sum_i \cE_i(\rho) \otimes \proj{i}_X
\end{equation}
which happens to be a free channel in some QRT $\qrt$, i.e.,  $\cE\in\qrt(A\rightarrow BX)$.
The golden rule of QRTs implies that for any free state $\rho\in\fs(A)$ the output $\cE(\rho)$ is also free.
However, it doesn't rule out the possibility that a post-selected quantum state $\frac{1}{\tr[\cE_i(\rho)]}\cE_i(\rho)$ could possibly be non-free.
Put differently, $\cE$ can in general create resource states in a probabilistic fashion.

It turns out that this is not possible for the vast majority of QRTs that are studied in literature, i.e., measurement maps that probabilistically produce resource states are never free in these theories.
This insight is captured by following axiom.
\begin{definition}{Axiom of free instruments}{axiom_fqi}
  We say a QRT $\qrt$ fulfills the \emph{axiom of free instruments} if for all free states $\rho\in\fs(A)$ and for all free instruments $\cI=(\cE_i)_i\in\fqi(A\rightarrow B)$ one has 
  \begin{equation}
    \forall i: \cE_i(\rho) / \tr[\cE_i(\rho)] \in \fs(B) \, .
  \end{equation}
\end{definition}

Now we turn our attention to measuring the amount of resource contained in a state or channel.
Functions that quantify a resource are typically called \emph{resource measures} or \emph{resource monotones}.
The exact meanings of these two terms in literature can vary, and in the following we present the definitions that we will use throughout this manuscript.

\begin{definition}{}{}
  A \emph{resource measure} for a QRT $\qrt$ is a map $M: \bigcup_{\cH}\dops(\cH) \rightarrow \mathbb{R}_{\geq 0}$ fulfilling the two properties:
  \begin{enumerate}
    \item \emph{Weak monotonicity:} $\forall\cE\in\qrt(A\rightarrow B), \forall \rho\in\dops(A): M(\cE(\rho))\leq M(\rho)$.
    \item \emph{Normalization:} $M(1)=0$ where $1$ is the only state of the Hilbert space $\mathbb{C}$.
  \end{enumerate}
\end{definition}
Here, $\cup_{\cH}$ denotes the union over all Hilbert spaces.
Note that any resource measure $M$ fulfills $\rho\in\fs(A) \Rightarrow M(\rho)=0$ for any $\rho\in\dops(A)$.
\begin{definition}{}{}
  A resource measure $M$ for a QRT $\qrt$ is called a \emph{resource monotone} if it additionally fulfills
  \begin{enumerate}
    \item \emph{Strong monotonicity:} For any $\rho\in\dops(A)$ and $(\cE_i)_i\in\fqi(A\rightarrow B)$
    \begin{equation}
      M(\rho) \geq \sum_i \tr[\cE_i(\rho)] M\left(\frac{\cE_i(\rho)}{\tr[\cE_i(\rho)]}\right) \, .
    \end{equation}
    \item \emph{Convexity:} For any finite ensemble of states $(p_i,\rho_i)_i$ with $\rho_i\in\dops(A)$
    \begin{equation}
      M\left(\sum_i p_i\rho_i\right) \leq \sum_i p_i M(\rho_i) \, .
    \end{equation}
  \end{enumerate}
\end{definition}
Beyond the above two definitions, we list some further properties that a desirable measure or monotone might fulfill.
\begin{itemize}
  \item \emph{Faithfulness:} For $\rho\in\dops(A): \rho\in\fs(A) \Leftrightarrow M(\rho)=0$
  \item \emph{Sub-additivity:} $\forall\rho\in \dops(A),\sigma\in\dops(B): M(\rho\otimes\sigma) \leq M(\rho) + M(\sigma)$.
  \item \emph{Additivity:} $\forall\rho\in \dops(A),\sigma\in\dops(B): M(\rho\otimes\sigma) = M(\rho) + M(\sigma)$.
\end{itemize}

While \emph{(static) resource theories} content themselves with the study of non-free states, the broader field of \emph{dynamic resource theories} goes beyond that by also quantifying the resources of general non-free channels.
In that case, we additionally require some notion of free transformations that map free channels back to free channels.

We call a \emph{superchannel} $\Theta$ a linear map from $\supo(A\rightarrow B)$ to $\supo(C\rightarrow D)$ that maps quantum channels back to quantum channels, i.e., $\Theta(\cptp(A\rightarrow B))\subset \cptp(C\rightarrow D)$.
We denote the set of all such superchannels $\schan(A,B\rightarrow C,D)$.
It was shown in \reference~\cite{chiribella2008_transforming} that any superchannel can be written in the form
\begin{equation}
  \Theta(\cE_{A\rightarrow B}) = \cF^{\mathrm{post}}_{BE\rightarrow D}\circ \left(\cE_{A\rightarrow B}\otimes \idchan_E\right)\circ \cF^{\mathrm{pre}}_{C\rightarrow AE}
\end{equation}
for some ancillary system $E$ and pre- and post-processing maps $\cF^{\mathrm{pre}}_{C\rightarrow AE}\in\cptp(C\rightarrow AE)$ and $\cF^{\mathrm{post}}_{BE\rightarrow D}\in\cptp(BE\rightarrow D)$.

One possible choice of free superchannels would be all such maps $\Theta$ where the pre- and post-processing can be realized using free channels.
In order to remain more general, we will instead choose to characterize the set of free superchannels in a more axiomatic fashion (following the approach in \reference~\cite{gour2020_dynamical}), more akin to how we introduced the free channels of a QRT.
\begin{definition}{}{}
  A \emph{dynamical quantum resource theory} is a mapping $\fsc$ that takes four quantum systems $A,B,C,D$ to a set of superchannels $\fsc(A,B\rightarrow C,D)\subset \schan(A,B\rightarrow C,D)$ and fulfills following conditions
  \begin{enumerate}
    \item $\fsc$ induces a QRT $\qrt$ with tensor product structure $\qrt(A\rightarrow B)\coloneqq \fsc(\mathbb{C},A\rightarrow \mathbb{C},B)$.
    \item $\id_{\supo(A\rightarrow B)}\in \fsc(A,B\rightarrow A,B)$ for all systems $A,B$.
    \item If $\Theta_1\in\fsc(A,B\rightarrow C,D)$ and $\Theta_2\in\fsc(C,D\rightarrow E,F)$ then $\Theta_2\circ\Theta_1\in\fsc(A,B\rightarrow E,F)$.
    \item $\forall\Theta\in\fsc(A,B\rightarrow C,D)$ and any systems $E,F$: $\Theta\otimes\id_{\supo(E\rightarrow F)} \in \fsc(AE,BF\rightarrow CE,DF)$.
  \end{enumerate}
\end{definition}
The golden rule can straightforwardly be translated to dynamical QRTs.
\begin{lemma}{}{dynamical_qrt_golden_rule}
  Consider a dynamical QRT $\fsc$.
  Then
  \begin{equation}
    \forall\cE\in\qrt(A\rightarrow B),\Theta\in\fsc(A,B\rightarrow C,D): \Theta(\cE)\in\qrt(C\rightarrow D) \, .
  \end{equation}
\end{lemma}
Finally, we also introduce the analog notion of a resource measure for a dynamical QRT.
\begin{definition}{}{}
  A \emph{dynamical resource measure} for a dynamical QRT $\fsc$ is a map $M: \bigcup_{(\cH_A,\cH_B)}\cptp(A\rightarrow B) \rightarrow \mathbb{R}_{\geq 0}$ fulfilling the following two properties:
  \begin{enumerate}
    \item \emph{Monotonicity:} $\forall\Theta\in\fsc(A,B\rightarrow C,D), \forall\cE\in\qrt(A\rightarrow B): M(\Theta(\cE)) \leq M(\cE) $. 
    \item \emph{Normalization:} $M(\idchan_A)=0$ for all systems $A$.
  \end{enumerate}
\end{definition}
Here, $\cup_{(\cH_A,\cH_B)}$ denotes the union over all pairs of two Hilbert spaces.
Note that these two properties again imply that $M(\cE)=0$ for any free channel $\cE$.

\chapter{Quasiprobability simulation framework}\label{chap:qpsim}
\section{Motivational example: non-Clifford simulation}\label{sec:motivational_example}
We now explore a first application of the quasiprobability simulation (QPS) technique as an illustrative example.
Consider the following setup: suppose we have access to a quantum computer that is only capable of executing Clifford operations.
Specifically, this means the computer can take the description of some quantum circuit containing only Clifford gates along with qubit preparations and measurements in the computational basis, execute the circuit, and generate samples from its output distribution.
Such a computer is not universal, i.e., it cannot approximate any unitary with arbitrary precision.
In fact, the celebrated Gottesman-Knill theorem states that any such circuit can be efficiently simulated classically~\cite{gottesman1998_heisenberg,aaronson2004_improved}.

\begin{figure}
  \centering
  \includegraphics{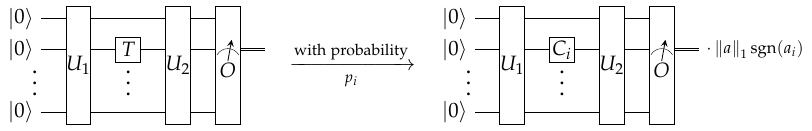}
  \caption{QPS of a Clifford circuit with one T gate. The T gate is probabilistically replaced with a Clifford gate $C_i$ and the measurement outcome of the circuits is correspondingly multiplied by the factor $\norm{a}_1\sgn(a_i)$.}
  \label{fig:qps_example_single_t}
\end{figure}

Now, the objective is to use this restricted quantum computer to simulate a general quantum circuit that may include both Clifford and non-Clifford gates.
For simplicity, we will assume that the only non-Clifford operations are T gates which are enough to elevate the Clifford group to a universal gate set.
We first illustrate the QPS technique on the simplest case which is a circuit containing only a single T gate, as depicted on the left of~\Cref{fig:qps_example_single_t}.
The state at the end of the $n$-qubit circuit is
\begin{equation}\label{eq:qps_example_circ}
  \omega \coloneqq \left(\indsupo{U_2}\circ \indsupo{\mathrm{T}}\circ \indsupo{U_1}\right)(\proj{0}^{\otimes n})
\end{equation}
where $U_1$ and $U_2$ are the Clifford unitaries representing all Clifford gates that occur before and after the T gate.
Recall that $\indsupo{U}(\rho)\coloneqq U\rho U^\dagger$ is the channel induced by the unitary $U$.
Let us assume that we are interested in estimating the expectation value $\tr[O\omega]$ of some observable $O$.
If we know a decomposition of the T gate channel in terms of Clifford gates, i.e.,
\begin{equation}\label{eq:T_qpd}
  \indsupo{\mathrm{T}} = \sum_{i=1}^m a_i \indsupo{C_i}
\end{equation}
for some coefficients $a\in\mathbb{R}^m$ and one-qubit Clifford gates $C_i$, then we can simply rewrite the expectation value as
\begin{align}
  \tr[O\omega]
  &= \sum_{i=1}^m a_i \tr\left[O \left(\indsupo{U_2}\circ \indsupo{C_i} \circ \indsupo{U_1}\right)(\proj{0})^{\otimes n}\right] \\
  \label{eq:expval_monte_carlo}
  &= \sum_{i=1}^m p_i \tr\left[O \left(\indsupo{U_2}\circ \indsupo{C_i} \circ \indsupo{U_1}\right)(\proj{0})^{\otimes n}\right] \cdot \sgn(a_i)\norm{a}_1
\end{align}
where $p_i\coloneqq \abs{a_i}/\norm{a}_1$ is a probability distribution $p_i\geq 0$ and $\sum_i p_i = 1$.
As depicted in~\Cref{fig:qps_example_single_t}, this probabilistic mixture provides us with an explicit Monte Carlo method to evaluate $\tr[O\omega]$ while only having access to a restricted Clifford quantum computer.
For every execution of the circuit, the T gate is randomly replaced by $C_i$ according to the probability $p_i$ and one keeps track of the sampled value of $i$.
At the end of the circuit, the measurement outcome of the observable is weighted by $\norm{a}_1\sgn(a_i)$.
Repeating this sampling procedure many times and then averaging the (weighted) measurement outcomes, \Cref{eq:expval_monte_carlo} tells us that this procedure will exactly retrieve $\tr[O\omega]$ in the limit of infinitely many shots.

Note that decompositions as in~\Cref{eq:T_qpd} do exist, for example
\begin{equation}\label{eq:optimal_T_qpd}
  \indsupo{\mathrm{T}} = \frac{1}{2}\idchan + \frac{1}{\sqrt{2}}\indsupo{S} - \frac{(\sqrt{2}-1)}{2}\indsupo{Z} \, .
\end{equation}
While the procedure outlined above allows us to retrieve the expectation value of the observable, it will typically require more shots to reach the same confidence compared to natively running the original circuit on a quantum computer that does support T gates.
The reason for this is that the circuit output is scaled up by a factor of $\pm\norm{a}_1$ (note that generally $\norm{a}_1$ will be $\geq 1$), so the variance increases by $\norm{a}_1^2$.
The number of shots required to reach a certain confidence thus also roughly increases with $\norm{a}_1^2$.
A rigorous statement on this sampling overhead will be made in~\Cref{sec:qpsim}.

To finish this exposition, let us consider how this technique can be applied to a circuit that contains more than one T gate.
The main idea here is that the quasiprobability decompositions of the individual T gates can be combined into one large quasiprobability decomposition of the total circuit.
For instance, if we have 3 T gates, then the channel describing the total circuit is\footnote{Note that our notation here is ambiguous about which qubits the T gates act on. They can in principle act on different qubits.}
\begin{equation}
  \indsupo{U_4}\circ \indsupo{\mathrm{T}} \circ \indsupo{U_3} \circ \indsupo{\mathrm{T}} \circ \indsupo{U_2} \circ \indsupo{\mathrm{T}} \circ \indsupo{U_1}
\end{equation}
where the $U_i$ are Clifford unitaries.
Inserting the decomposition from~\Cref{eq:T_qpd}, the channel of the total circuit can be decomposed into Clifford channels
\begin{equation}
  \sum\limits_{i_1,i_2,i_3} p_{i_1}p_{i_2}p_{i_3} \cdot \indsupo{U_4C_{i_3}U_3C_{i_2}U_2C_{i_1}U_1} \cdot \sgn(a_1a_2a_3)\norm{a}_1^3 \, .
\end{equation}
The Monte Carlo sampling procedure in this case works as follows: For every instance of the T gate in the circuit, we separately and independently sample a random Clifford gate $C_i$ to replace it.
The sampled indices ($i_1$, $i_2$ and $i_3$ in the example) are stored and at the end of the circuit we weight the measurement outcome correspondingly by either $+\norm{a}_1^3$ or $-\norm{a}_1^3$.

Notice how the 1-norms of the individual decompositions combine in a multiplicative fashion, so the 1-norm of the total circuit QPD is given by $\norm{a}_1^t$ where $t$ is the number of T gates.
The sampling overhead of the technique thus scales as $\mathcal{O}(\norm{a}_1^{2t})$.

Using the above QPS technique on top of a Gottesman-Knill simulator for Clifford circuits, one obtains a classical algorithm to simulate Clifford+T circuits with a runtime that remains efficient for a small (logarithmic) number of T gates.
It is not too surprising that QPS entails an exponential sampling overhead --- a polynomial overhead would have implied that quantum computers could be simulated efficiently with classical computers.
We also note that there exist asymptotically faster (but still exponential) classical simulation algorithms for Clifford+T circuits.
Though they usually utilize strong simulation of Clifford circuits and as such cannot be used to simulate Clifford+T circuits on a Clifford \emph{quantum} computer.
A more in-depth discussion will be presented in~\Cref{chap:magic}.

\section{Quasiprobability simulation}\label{sec:qpsim}
Consider a quantum computation described by some quantum circuit $C$ followed by the measurement of an observable $O\in\herm(A)$, which is assumed to be normalized $\opnorm{O}\leq 1$ without loss of generality.
Let us denote by $\omega\in\herm(A)$ the state that $C$ prepares.
Our ultimate goal is to estimate the expectation value $\tr[O\omega]$.

Furthermore, assume that we have access to a restricted quantum computer that cannot directly execute the circuit $C$, because it contains operations that it cannot physically implement\footnote{Note however, that the restricted quantum computer must be able to realize the measurement of the observable $O$.}.
Instead, assume that we know a certain family of circuits $C_1,\dots,C_m$ that the computer can execute, and let us denote by $\sigma_i$ the state prepared by $C_i$.
These circuits are carefully chosen such that they provide a quasiprobability decomposition (QPD) of the desired target state
\begin{equation}
  \omega = \sum_{i=1}^m a_i \sigma_i
\end{equation}
for some $a\in\mathbb{R}^m$.
\Cref{alg:qpsim} provides an unbiased estimate $\hat{O}$ of $\tr[O\omega]$ with an error that can be made arbitrarily small by increasing the number of circuit executions $n_s\in\mathbb{N}$ (also called number of \emph{shots}).

\begin{algorithm}
  \begin{algorithmic}
    \Require descriptions of $O$ and $C_i$, $a\in\mathbb{R}^m$, number of shots $n_s$
    \Ensure estimate $\hat{O}$ of $\tr[O\omega]$
    \State $r \gets $ empty list
    \Repeat{$n_s$}
      \State $i \gets$ sample random number in $\{1,\dots,m\}$ with probability $p_i=\abs{a_i}/\norm{a}_1$
      \State $o \gets$ execute circuit $C_i$ and measure $O$
      \State $o' \gets o\cdot \sgn(a_i)\norm{a}_1$
      \State append $o'$ to $r$
    \EndRepeat
    \State \textbf{return} $\mathrm{mean}(r)$
  \end{algorithmic}
  \caption{Quasiprobability simulation algorithm}
  \label{alg:qpsim}
\end{algorithm}

\begin{proposition}{QPS sampling overhead}{qps_overhead}
  Consider the setup described above for~\Cref{alg:qpsim}.
  Let $\epsilon > 0$ and $\delta\in (0,1]$.
  If $n_s \geq 2\left(\frac{\norm{a}_1}{\epsilon}\right)^2\ln\left(\frac{2}{\delta}\right)$ then $\abs{\hat{O}-\tr[O\omega]}\leq \epsilon$ with probability at least $1-\delta$.
\end{proposition}
This is a direct consequence of Hoeffding's inequality (see~\Cref{app:hoeffding}).
\begin{proof}
  Denote by $X_1,\dots,X_{n_s}$ the i.i.d. random variables describing the weighted outcome of the $n_s$ different shots.
  The output of the QPS algorithms is thus $\hat{O}=\frac{1}{n_s}\sum_j X_j$.
  First, notice that
  \begin{equation}
    \mathbb{E}[\hat{O}] = \mathbb{E}[X_j] = \sum_i p_i \tr[O\sigma_i]\sgn(a_i)\norm{a}_1 = \tr[O\omega] \, .
  \end{equation}
  Since the $|X_i|$ are bounded by $\norm{a}_1$, Hoeffding's inequality (\Cref{thm:hoeffding}) states
  \begin{equation}
    \mathbb{P}\left[\abs{\hat{O}-\tr[O\omega]}\geq \epsilon\right] \leq 2\exp\left(-\frac{n_s\epsilon^2}{2\norm{a}_1^2}\right) \, .
  \end{equation}
  The right-hand side is $\leq \delta$ if and only if
  \begin{equation}
    n_s \geq \frac{2\norm{a}_1^2}{\epsilon^2}\ln\left(\frac{2}{\delta}\right) \, .
  \end{equation}
\end{proof}

\begin{remark}{Tightness of~\Cref{prop:qps_overhead}}{qps_overhead_tightness}
  Notice that~\Cref{prop:qps_overhead} only provides us with an upper bound for the required number of shots to achieve a certain confidence on the result.
  Indeed, for some circuits and decompositions, one might get lucky and the required number of shots could be much smaller.
  Consider for instance the extreme example where the observable $O$ is zero - clearly, a single execution of the circuit is sufficient to retrieve the desired expectation value with certainty.

  Still, there is a notion in which the statement is tight, because there do exist instances for which the number of shots described in~\Cref{prop:qps_overhead} is required, at least up to constant factors.
  More concretely, consider an observable $O$ with eigenvalues $\pm 1$ as well as a state $\omega$ such that $\tr[O\omega]=0$.
  In this case, the QPS estimator outputs one of the values $+\norm{a}_1$ or $-\norm{a}_1$ every shot, with probability $1/2$ each.
  By~\Cref{lem:hoeffding_tightness}, the required number of shots to estimate the true expectation value, up to error $\epsilon$ and with probability at least $1-\delta$, must be at least $\frac{1}{4}\left(\frac{\norm{a}_1}{\epsilon}\right)^2\ln\left(\frac{2}{15\delta}\right)$.

  We note that lower sampling overheads can sometimes be achieved when additional structure about the problem is given.
  For example, in \reference~\cite{harrow2024_optimal} a quadratic improvement is demonstrated for the task of circuit knitting under the condition that $O$ fulfills some sparsity condition.
\end{remark}

To minimize the sampling overhead, one should use a QPD of the circuit $C$ which exhibits the smallest possible 1-norm $\smash{\norm{a}_1}$.
Finding this optimal QPD is very difficult and generally intractable.
Instead, the standard approach consists of constructing the overall circuit QPD from decompositions of the individual components in the circuit\footnote{Interestingly, we will encounter some instances where we will go beyond this standard approach later in the thesis.}.
We illustrated this in the example from~\Cref{sec:motivational_example} where QPDs of individual T gates were combined.
To describe the technique more generally, let us write the circuit $C$ as a sequence of $n$ superoperators, which could e.g. represent gates or state preparations\footnote{Recall that states can be seen as channels (superoperators) with trivial input space.}
\begin{equation}
  \omega = \cE_n\circ \cE_{n-1} \circ \dots \circ \cE_2 \circ \cE_1 \, .
\end{equation}
Each circuit element $\cE_i$ is then decomposed into operations $\cE_i = \sum_{j} a_j^{(i)} \cF_j^{(i)}$ that our quantum computer can physically execute.
This in turn provides us with a valid QPD of the total circuit
\begin{align}\label{eq:qpd_circuit_elements}
  \omega 
  &= \sum_{j_1,j_2,\dots,j_n} a_{j_1}^{(1)}a_{j_2}^{(2)}\cdots a_{j_n}^{(n)} \left( \cF^{(n)}_{j_n} \circ \cF^{(n-1)}_{j_{n-1}} \circ \dots \circ \cF^{(2)}_{j_2} \circ \cF^{(1)}_{j_1} \right) \\
  &= \sum_{j_1,j_2,\dots,j_n} p_{j_1}^{(1)}p_{j_2}^{(2)}\cdots p_{j_n}^{(n)} \left( \cF^{(n)}_{j_n} \circ \cF^{(n-1)}_{j_{n-1}} \circ \dots \circ \cF^{(2)}_{j_2} \circ \cF^{(1)}_{j_1} \right) \nonumber\\
  & \quad\quad\quad\quad\quad \cdot \norm{a^{(1)}}_1\norm{a^{(2)}}_1\cdots \norm{a^{(n)}}_1 \sgn\left(\prod_k a_{j_k}^{(k)}\right)
\end{align}
where $p_{j}^{(i)}\coloneqq \abs{a_{j}^{(i)}} / \norm{a^{(i)}}_1$.

Note that \Cref{prop:qps_overhead} only makes a statement about the number of circuit executions in QPS.
There are a few other considerations that also impact the overall runtime of the technique, such as the size of the circuits $C_i$ and the complexity of sampling from the probability distribution.
If we construct the circuit QPD from individual circuit elements, as previously described, then both of these issues can be addressed.
First, the complexity of the circuits $C_i$ is then identical to the complexity of $C$ (assuming that the operations $\smash{\cF_j^{(i)}}$ can be implemented with constant overhead).
Furthermore, the probability distribution factorizes into separate distributions $\smash{(p_{j}^{(i)})_j}$ for every circuit element, which can be independently (and hence efficiently) sampled.
This justifies our focus on the number of circuit executions to measure the complexity of the QPS method.

As illustrated in~\Cref{eq:qpd_circuit_elements}, the 1-norm of the QPD we constructed for our circuit is given by the product of the 1-norms of the decompositions $\smash{\cE_i = \sum_{j} a_j^{(i)} \cF_j^{(i)}}$ of the individual circuit elements.
It is therefore advantageous to use QPDs of these circuit elements that have the smallest possible 1-norm.
Finding such optimal QPDs for small circuit elements is more tractable than for the entire circuit itself (albeit often still very difficult).
Given a fixed set $\ds\subset\hp(A\rightarrow B)$ that represents the operations that our restricted computer can perform, we define the \emph{quasiprobability extent} of $\cE\in\hp(A\rightarrow B)$ to be the smallest achievable 1-norm.
\begin{definition}{}{gamma}
  For a decomposition set $\ds\subset\hp(A\rightarrow B)$ and a target operation $\cE\in\hp(A\rightarrow B)$, the quasiprobability extent of $\cE$ w.r.t. $\ds$ is defined as
  \begin{equation}\label{eq:gamma}
    \gamma_{\ds}(\cE) \coloneqq \inf \left\{ \sum_{i=1}^m\abs{a_i} \,\middle\vert\, \cE = \sum\limits_{i=1}^m a_i \cF_i  \text{ where } m\geq 1, \cF_i\in \ds \textnormal{ and } a_i \in \mathbb{R} \right\} \, .
  \end{equation}
\end{definition}
If no valid decomposition exists, $\gamma_{\ds}$ is defined to be infinite.
\begin{example}
  If we pick $\ds$ to be all Clifford operations and $\cE=\indsupo{T}$ to be the T gate channel, then the QPD in~\Cref{eq:optimal_T_qpd} happens to be optimal, i.e., $\gamma_{\ds}(\indsupo{T}) = \sqrt{2}$.
  A proof of this fact will be presented in~\Cref{chap:magic}.
\end{example}

One aspect of~\Cref{def:gamma} that immediately stands out is that we allow for both $\cE$ and the elements of $\ds$ to possibly be non-CPTP maps.
Indeed, the quasiprobability simulation algorithm works with arbitrary linear superoperators and can be used to estimate expectation values of circuits $C$ that contain nonphysical operations in them.
We will encounter multiple settings throughout this thesis where simulating nonphysical target operations $\cE$ can be very useful.
Regarding the nonphysicalness of the decomposition set $\ds$, we will elaborate in~\Cref{sec:qps_intermediate_measurements} that more general Hermitian-preserving maps can capture processes that include intermediate measurements and post-processing, such as completely positive trace-non-increasing maps.

\begin{remark}{real versus complex coefficients}{}
  The QPS technique can straightforwardly be extended to also work with QPDs that include complex coefficients $a_i$ instead of real ones.
  The main advantage of QPDs with complex coefficients is that they are more expressive: For a Hermitian-preserving decomposition set $\ds\subset\hp(A\rightarrow B)$, $\spn_{\mathbb{R}}(\ds)$ is itself a subset of $\hp(A\rightarrow B)$.
  In comparison, $\spn_{\mathbb{C}}(\ds)$ can generally contain all of $\supo(A,B)$.

  In this work, we refrain from discussing complex QPDs due to two observations:
  \begin{enumerate}
    \item There is no physically meaningful situation known to us where we would be interested in a target operation $\cE\in \supo(A\rightarrow B)\setminus \hp(A\rightarrow B)$.
    \item For a target operation $\cE\in\hp(A\rightarrow B)$, the quasiprobability extent remains unaffected by the choice of allowing complex coefficients.
  \end{enumerate}
  To see why the second point holds, consider a QPD $\cE=\sum_ja_j\cF_j$ for $\cE,\cF_j\in\hp(A\rightarrow B)$ and $a_j\in\mathbb{C}$.
  By considering the QPD in the Choi representation $\choi{\cE}=\sum_j \re[a_j]\choi{\cF_j} + i\sum_j\im[a_j]\choi{\cF_j}$, we see that the skew-Hermitian component $i\sum_j\im[a_j]\choi{\cF_j}$ must be zero as $\choi{\cE}$ is Hermitian.
  Therefore, $\cE=\sum_i\re[a_i]\cF_i$ is also a valid QPD with at most the same 1-norm.
\end{remark}

Note that QPS can also be used to estimate the outcome probabilities of a circuit, since they can be phrased as the expectation value of some observable.
For instance, if at the end of a circuit $n$ qubits are measured in the computational basis yielding an outcome $x\in \{0,1\}^n$, then the probability $p(x)$ of obtaining $x$ is given by the expectation value of $O=\proj{x}$.
Unfortunately, when the distribution $p(x)$ is very flat (i.e. all probabilities are exponentially small in $n$) this is only of little use, as QPS only efficiently provides estimates of expectation values up to an \emph{absolute error} $\epsilon$, which would need to be exponentially small in $n$ to properly assess the structure of the distribution.
This is in stark contrast to some classical simulation algorithms which can efficiently estimate probability amplitudes up to a \emph{multiplicative error}, which can then be used to sample from the output distribution using the chain rule for conditional probabilities\footnote{The idea here is to sample one bit at a time in a sequential fashion, where each bit is sampled according to its (binary) marginal that is conditioned on the values of the previous bits. This technique is sometimes called the ``qubit-by-qubit'' algorithm in the context of stabilizer rank simulators~\cite{bravyi2022_simulate,bravyi2019_simulation}.}

There are still cases where the output distribution (or at least important aspects thereof) can be fully captured with QPS.
For instance, when the output distribution is sufficiently sparse, it becomes possible to accurately sample from it using QPS~\cite{pashayan2020_estimation}.
Similarly, when $n$ is very small (asymptotically speaking, i.e. constant or logarithmic in the problem size), then estimates up to an absolute error are sufficient to retrieve the distribution.
The extreme example thereof are decision problems, where $n=1$.
Therefore, QPS is applicable to circuits corresponding to languages in the complexity class BQP.
More practically speaking, many applications of quantum computing are anyway primarily concerned with evaluating expectation values.
This is the case for variational algorithms or for estimating physical observables in the context of Hamiltonian simulation of physical systems.
Furthermore, other quantum algorithms, like Shors's factoring algorithm and quantum phase estimation, can be decomposed in a series of decision problems~\cite[Appendix H]{suzuki2022_qem}.

In general, there are many notions of what it means to ``simulate'' a quantum circuit, and the differences between them are often subtle.
We refer readers that are further interested in this matter to the comprehensive treatment in \reference~\cite{pashayan2020_estimation}.
In the language of these authors, QPS constitutes a ``poly-box''.
Under certain complexity theoretic assumptions, it can be proven that poly-boxes provide a strictly weaker simulation compared to ``$\epsilon$-simulators'', which arguably constitute the operationally most meaningful notion of simulation\footnote{It is statistically impossible to distinguish an agent that can execute a quantum circuit from another agent which only possesses an $\epsilon$-simulator thereof.}.


\smallskip
We end this section by briefly highlighting a special instance where the input space of a circuit element $\cE\in\hp(A\rightarrow B)$ is trivial $A=\mathbb{C}$.
In this case, $\cE$ is a state preparation channel and can be directly identified with a state $\rho\in\herm(B)$.
Put differently, QPS cannot only be used to simulate certain gates in a circuit, but also to simulate the preparation of states that are used in the circuit.
The quasiproability extent of a state $\rho\in\herm(B)$ w.r.t. some decomposition set $\ds\subset\herm(B)$ is analogously defined as
\begin{equation}
  \gamma_{\ds}(\rho) \coloneqq \inf \left\{ \sum_{i=1}^m\abs{a_i} \,\middle\vert\, \rho = \sum\limits_{i=1}^m a_i \tau_i  \text{ where } m\geq 1, \tau_i\in \ds \textnormal{ and } a_i \in \mathbb{R} \right\} \, .
\end{equation}
\begin{example}\label{ex:state_qpd_nonclifford}
  We briefly illustrate the concept of state QPDs with the Clifford+T setting from~\Cref{sec:motivational_example}.
  It is a widely known trick that a T gate can be realized using a Clifford circuit that consumes a single-qubit magic state
  \begin{equation}
    \ket{H} \coloneqq \frac{1}{\sqrt{2}}\left(\ket{0} + e^{i\pi / 4} \ket{1} \right) 
  \end{equation}
  as depicted in~\Cref{fig:magic_state_injection}.
  Using this gadget, the action of an arbitrary Clifford+T circuit $\mathcal{C}$ with $m$ T gates can be rewritten as 
  \begin{equation}
    \mathcal{C}(\proj{0}^{\otimes n}) = \mathcal{C}'(\proj{0}^{\otimes n}\otimes\proj{H}^m) 
  \end{equation}
  for some appropriately chosen Clifford circuit $\mathcal{C}'$.

  Using a QPD of the magic state
  \begin{equation}\label{eq:magic_state_qpd}
    \proj{H} = \frac{1}{2}\proj{+} + \frac{1}{\sqrt{2}}\proj{i+} - \frac{(\sqrt{2}-1)}{2}\proj{-}
  \end{equation}
  into stabilizer states (i.e. states which are reachable from a computational basis state using Clifford gates), we can write the original circuit $\mathcal{C}$ as a quasiprobabilistic mixture of Clifford circuits and hence apply QPS.

  In fact, this state-based approach for Clifford+T simulation is historically more commonly considered in literature compared to the channel-based approach~\cite{heinrich2019_robustness,howard2017_application}.
  Note also that the QPD in~\Cref{eq:magic_state_qpd} has the same 1-norm of $\sqrt{2}$ as the channel-based one in~\Cref{eq:optimal_T_qpd}.
\end{example}
\begin{figure}
  \centering
  \includegraphics{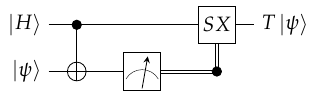}
  \caption{The magic state injection circuit realizes a T gate using a magic state and Clifford operations. The $X$ and $S$ gate on the upper qubit are only performed if the measurement outcome is $1$.}
  \label{fig:magic_state_injection}
\end{figure}

\section{QPS with classical side information}\label{sec:qps_intermediate_measurements}
The optimal simulation overhead for some target operation $\cE\in\supo(A\rightarrow B)$ strongly depends on the decomposition set $\ds$ with respect to which we consider its QPD.
At first, one might expect that $\ds$ always has to be chosen as some subset of $\cptp(A\rightarrow B)$, since these are precisely the physically realizable maps.
In this section, we will see that it can be very useful to allow $\ds$ to also contain certain trace-non-increasing maps and even certain non-completely positive maps.
This should be understood as a mathematical method to capture the idea that intermediate measurements can produce classical side information during the circuit execution, as will be explained below.
We will first consider two explanatory examples before explaining the trick in full generality.

As an example, suppose we want to simulate the one-qubit trace-non-increasing superoperator $\cE=\indsupo{\proj{0}}$ which maps $\rho\mapsto \braopket{0}{\rho}{0}\cdot \proj{0}$.
Clearly, no valid QPD of $\indsupo{\proj{0}}$ w.r.t. a decomposition set $\ds\subset\cptp(A\rightarrow B)$ can exist, because any linear combination of trace-preserving maps is itself proportional to a trace-preserving map.
Yet still, there is a way to simulate $\indsupo{\proj{0}}$ using CPTP maps.
To see this, notice that the action of $\indsupo{\proj{0}}$ can be written as
\begin{equation}
  \indsupo{\proj{0}}(\rho) = (\id\otimes \bra{0}_X)\cdot \mathcal{M}(\rho) \cdot (\id\otimes \ket{0}_X) \, .
\end{equation}
where the CPTP map $\mathcal{M}$ denotes a measurement in the computational basis
\begin{equation}
  \mathcal{M}(\rho) \coloneqq \sum_{i=0,1} \indsupo{\proj{i}}(\rho) \otimes \proj{i}_X
\end{equation}
that stores the outcome in a classical ancilla register $X$.
Consider now that we want to estimate the expectation value of an observable $O$ on $\indsupo{\proj{0}}(\rho_{\mathrm{in}})$ for some input state $\rho_{\mathrm{in}}$
\begin{equation}
  \tr\left[O\cdot \indsupo{\proj{0}}(\rho_{\mathrm{in}})\right]
  = 
  \tr\left[(O\otimes\proj{0}_X)\cdot \mathcal{M}(\rho_{\mathrm{in}})\right] \, .
\end{equation}
Notice how we replaced $\indsupo{\proj{0}}$ with $\mathcal{M}$ by extending the observable to also act on the classical register $X$.
Put differently, our simulation procedure operates as follows: replace $\indsupo{\proj{0}}$ by a computational basis measurement, and in case the measurement result corresponds to $\proj{1}$, then weight the final measurement outcome by zero.
Note that this is not post-selection, as the ``discarded'' runs are still taken into account when averaging over multiple runs, they just have value zero.

It is illustrative to briefly consider a second example.
This time, we want to simulate a target operation $\cE=\indsupo{\proj{0}}-\indsupo{\proj{1}}$ which is not even completely positive.
Here again, if we only care about the expectation value of some final observable $O$, we can simulate $\cE$ by replacing it with a computational basis measurement
\begin{equation}
  \tr\left[O\cdot (\indsupo{\proj{0}}-\indsupo{\proj{1}})(\rho_{\mathrm{in}})\right]
  = 
  \tr\left[(O\otimes(\proj{0}_X-\proj{1}_X))\cdot \mathcal{M}(\rho_{\mathrm{in}})\right]
  \, .
\end{equation}
and by weighting the circuit outcome by either $+1$ or $-1$ depending on the intermediate measurement outcome.

The main insights of the previous two examples is that \emph{for the purpose of estimating some expectation value} one can use classical side information provided by intermediate measurements to simulate non-CPTP maps by weighting the final observable measurement by some factor depending on the intermediate measurement outcome.
The most general operation that produces classical information is described by a quantum instrument.
Given a quantum instrument $(\cG_i)_i$ that our computer can physically realize, this technique allows us to effectively perform any operation $\sum_i \beta_i \cG_i$ for some weighting factors $\beta_i\in\mathbb{R}$.

For completeness, we state the modified quasiprobability simulation algorithm with classical side information.
Consider a given target circuit $C$ preparing a state $\omega\in\herm(A)$.
Furthermore, consider that we have a family of circuits that can be executed on the restricted quantum computer $C_1,\dots,C_m$ where $C_i$ prepares the classical-quantum state $\sigma_i=\sum_j \tau_j^{(i)}\otimes \proj{j}$.
These circuits shall form a valid QPD in the sense that $\omega=\sum_i a_i (\sum_j \beta_j^{(i)}\tau_j)$ for some $a\in\mathbb{R}^m$ and $\beta_{j}^{(i)}\in\mathbb{R}$.
\Cref{alg:qpsim_intermediate_measurements} produces an estimate $\hat{O}$ of $\tr[O\omega]$.
\begin{algorithm}
  \begin{algorithmic}
    \Require descriptions of $O$ and $C_i$, $a\in\mathbb{R}^m$, $\beta_j^{(i)}\in\mathbb{R}$, number of shots $n_s$
    \Ensure estimate $\hat{O}$ of $\tr[O\omega]$
    \State $r \gets $ empty list
    \Repeat{$n_s$}
      \State $i \gets$ sample random number in $\{1,\dots,m\}$ with probability $p_i=\abs{a_i}/\norm{a}_1$
      \State execute Circuit $C_i$ to prepare CQ state $\sum_j\tau_j^{(i)}\otimes\proj{j}_X$
      \State $j\gets$ value of classical register $X$
      \State $o \gets$ measurement outcome of observable $O$ on $\tau_j^{(i)}$
      \State $o' \gets o\cdot \sgn(a_i)\norm{a}_1\beta_j^{(i)}$
      \State append $o'$ to $r$
    \EndRepeat
    \State \textbf{return} $\mathrm{mean}(r)$
  \end{algorithmic}
  \caption{QPS algorithm with classical side information}
  \label{alg:qpsim_intermediate_measurements}
\end{algorithm}

It is easy to check that $\hat{O}$ is an unbiased estimator for the true expectation value as
\begin{align}
  \tr[O\cdot\omega]
  &= \sum_i a_i \tr\left[ O\cdot \left(\sum_j\beta_j^{(i)}\tau_j^{(i)}\right) \right]  \\
  &= \sum_i p_i \tr\left[ \left( O\otimes \sum_k \beta_k^{(i)}\proj{k}\right)\cdot \left(\sum_j \tau_j^{(i)}\otimes \proj{j}\right) \right] \norm{a}_1\sgn(a_i) \, .
\end{align}
In order for~\Cref{prop:qps_overhead} to still apply, we must also make sure that we don't scale up the outcome of the estimator to a magnitude larger than $\norm{a}_1$.
For this reason, we restrict ourselves coefficients of magnitude $\abs{\smash{\beta_j^{(i)}}}\leq 1$.

In summary, given a quantum computer that can realize some set of quantum instruments $Q$, the associated decomposition set which captures the capabilities of that computer is given as follows.
\begin{definition}{}{}
  Let $Q\subset\qi(A\rightarrow B)$ be a set of quantum instruments.
  We define
  \begin{equation}
    \qids[Q] \coloneqq \{ \sum_i \beta_i \cE_i | (\cE_i)_i\in Q, \beta_i\in [-1,1] \} \, .
  \end{equation}
\end{definition}
The notation is meant to mirror the definition in~\Cref{eq:qichandef}.

The two introductory examples in this section have shown that classical side information can enable the simulation of circuits containing nonphysical superoperators that would be out of reach otherwise, i.e., 
\begin{equation}
  \gamma_{\qids[\qi]}(\indsupo{\proj{0}})=\gamma_{\qids[\qi]}(\indsupo{\proj{0}}-\indsupo{\proj{1}})\leq 1
\end{equation}
whereas
\begin{equation}
\gamma_{\cptp}(\indsupo{\proj{0}}) = \gamma_{\cptp}(\indsupo{\proj{0}}-\indsupo{\proj{1}}) = \infty \, .
\end{equation}
An arguably even more interesting question is, whether the utilization of classical side information can also help to lower the quasiprobability extent to simulate certain physical CPTP maps.
This is one of the key questions that we will investigate in this thesis, and interestingly, the answer will be different depending on the considered setting.

We finish this chapter by providing a simple mathematical characterization of $\qids[Q]$.
For this, we introduce following notions for sets of quantum instruments.
\begin{definition}{}{}
  For a set of quantum instruments $Q\subset \qi(A\rightarrow B)$ we say that
  \begin{itemize}
    \item $Q$ is \emph{coarse-grainable} if $\forall \cI\in Q$, any coarse-graining of $\cI$ is also in $Q$.
    \item $Q$ is \emph{trivially fine-grainable} if $\forall (\cE_i)_i\in Q$ we have $((1-p_i)\cE_i)_i\cup (p_i\cE_i)_i \in Q$ for any choice of probabilities $p_i\in [0,1]$.
    \item $Q$ is \emph{closed under mixture} if for any two instruments $(\cE_i)_i$ and $(\cF_j)_j$ in $Q$, their mixture $(p\cE_i)_i\cup ((1-p)\cF_j)_j$ is also in $Q$, for any $p\in [0,1]$.
  \end{itemize}
\end{definition}
These properties are extremely natural when describing the set of quantum instruments $Q$ that a hypothetical quantum computer can implement, and all three properties should hold in most physically relevant settings.
Indeed, $Q$ being coarse-grainable simply means that one can discard classical information and treat different classical outcomes as the same.
Trivial fine-grainability essentially means that one can generate a random bit with the probability depending on the observed measurement outcome.
Similarly, when $Q$ is closed under mixture, then randomly choosing between two quantum instruments is itself also a quantum instrument in $Q$.

At first glance, the trick of using classical side information might seem to add a lot of complexity to the process of finding QPDs, as elements in the decomposition set can now consist of a possibly infinite sum of quantum instrument maps.
Fortunately, for a sufficiently nice underlying set of quantum instruments $Q$, it turns out that we can restrict ourselves to only considering the two-element instruments in $Q$ and coefficients $\beta_i=\pm 1$.
\begin{lemma}{}{characterization_expanded_decomposition_set}
  If $Q\subset\qi(A\rightarrow B)$ is coarse-grainable and trivially fine-grainable, then 
  \begin{equation}
    \qids[Q] = \{ \cE^+ - \cE^- | (\cE^+,\cE^-)\in Q\} \, .
  \end{equation}
\end{lemma}
\begin{proof}
  Consider an element $\cE = \sum_i \beta_i \cE_i \in \qids[Q]$ where $(\cE_i)_i\in Q$, $\abs{\beta_i}\leq 1$.
  By the trivially fine-grainable property of $Q$, we can also express $\cE$ as a weighted sum of another instrument in $Q$ using only the coefficients $+1,-1$ and $0$:
  \begin{equation}
    \cE = \sum_i \sgn(\beta_i) \left(\abs{\beta_i} \cE_i\right) + \sum_i 0\cdot \left((1-\abs{\beta_i}) \cE_i\right)
  \end{equation}
  We now coarse-grain this instrument into a three-element instrument $(\cE^+,\cE^-,\cE^0)$
  \begin{align}
    \cE^+ &\coloneqq \sum_{i : \beta_i \geq 0} \abs{\beta_i}\cE_i \\
    \cE^- &\coloneqq \sum_{i : \beta_i < 0} \abs{\beta_i}\cE_i \\
    \cE^0 &\coloneqq \sum_{i} (1-\abs{\beta_i})\cE_i
  \end{align}
  which is still in $Q$ by the coarse-grainability property.
  We thus have
  \begin{equation}
    \cE = \cE^+ - \cE^- + 0\cdot \cE^0 \, .
  \end{equation}
  Finally, we use the fine-grainability and coarse-grainability properties once more to argue that $(\cE^+,\cE^-,\frac{1}{2}\cE^0,\frac{1}{2}\cE^0)$ and thus $(\cE^++\frac{1}{2}\cE^0,\cE^-+\frac{1}{2}\cE^0)$ are quantum instruments in $Q$.
  The desired statement thus follows from
  \begin{equation}
    \cE = (\cE^++\frac{1}{2}\cE^0) - (\cE^-+\frac{1}{2}\cE^0) \, .
  \end{equation}
\end{proof}

\section{Basic properties of the quasiprobability extent}\label{sec:qps_basic_properties}
In this section, we explore some central mathematical properties of the quasiprobability extent.
They will be repeatedly used in later chapters, so we collect them here for convenient reference.
The first one is straightforward, but provides an important tool for finding lower and upper bounds.
\begin{lemma}{}{gamma_ds_bound}
  If $\ds\subset\ds'\subset\hp(A\rightarrow B)$ then $\forall\cE\in\hp(A\rightarrow B): \gamma_{\ds}(\cE) \geq \gamma_{\ds'}(\cE)$.
\end{lemma}
\begin{proof}
  This follows directly from the definition of $\gamma$, as any QPD of $\cE$ w.r.t. $\ds$ is also a valid QPD w.r.t. $\ds'$.
\end{proof}

\paragraph{Convex and absolutely convex decomposition sets}
In most practical instances, the decomposition set $\ds$ is convex.
Clearly, convexity is a very physical property: if our computer has the capabilities to physically realize the operations $\cF_1$ and $\cF_2$, then it will likely be able to also realize $(1-p)\cF_1+p\cF_2$ for some $p\in [0,1]$ by randomly choosing one of the two with the according probability.
Convex decomposition sets have the useful property that we can without loss of generality only consider QPDs with at most two elements.
\begin{lemma}{}{two_element_qpd}
  For a convex decomposition set $\ds\subset\hp(A\rightarrow B)$ and $\cE\in\hp(A\rightarrow B)$ one has
  \begin{equation}\label{eq:two_element_qpd}
    \gamma_{\ds}(\cE) = \inf \Big\{ a^+ + a^- \Big\vert \cE = a^+\cF^+ - a^-\cF^- \text{ where } \cF^{\pm}\in\ds \textnormal{ and } a^{\pm}\geq 0 \Big\} \, .
  \end{equation}
\end{lemma}
\begin{proof}
  Any QPD $\cE=\sum_ia_i\cF_i$ with $\cF_i\in\ds$ induces a two-element QPD with the same 1-norm by grouping together all positive and negative terms
  \begin{equation}
    \cE = a^+\cF^+ - a^-\cF^-
  \end{equation}
  where 
  \begin{equation}
    a^{+}\coloneqq \sum_{i:a_i\geq 0}a_i \, ,
    \quad\quad
    a^{-}\coloneqq -\sum_{i:a_i<0}a_i
  \end{equation}
  and
  \begin{equation}
    \cF^{+}\coloneqq \frac{\sum_{i:a_i\geq 0}a_i\cF_i}{a^+} \, ,
    \quad\quad
    \cF^{-}\coloneqq -\frac{\sum_{i:a_i<0}a_i\cF_i}{a^-} \, ,
  \end{equation}
  where $\cF^{\pm}\in\ds$ due to convexity and $a^++a^- = \sum_i\abs{a_i}$.
\end{proof}
In~\Cref{sec:qps_intermediate_measurements} we saw that it can be meaningful to think of the set of quantum instruments $Q$ that a computer can realize, in order to capture the notion of classical side information.
When $Q$ is closed under mixture, then the induced decomposition set has an even stronger property than convexity.
\begin{definition}{}{}
  A subset $X$ of a real vector space $V$ is called \emph{absolutely convex} if it is convex and balanced (i.e. $\forall x\in X,\forall \lambda\in [-1,1]: \lambda x\in X$).
\end{definition}
It is an easy exercise to verify that absolute convexity is equivalent to being closed under \emph{absolutely convex mixtures}: $\sum \lambda_i x_i$ for $x_i\in X$ and $\sum_i\abs{\lambda_i}\leq 1$.
\begin{lemma}{}{decompset_absolutely_convex}
  If $Q\subset\qi(A\rightarrow B)$ is closed under mixture, then $\qichan[Q]$ is convex and $\qids[Q]$ is absolutely convex.
\end{lemma}
\begin{proof}
  We start with the convexity of $\qichan[Q]$.
  For any two $\cE=\qichan[\cI],\cF=\qichan[\cI']$ with $\cI=(\cE_i)_i$, $\cI'=(\cF_j)_j$, we also have $(1-p)\cE+p\cF = \qichan[\cJ]$ where $\cJ\coloneqq ((1-p)\cE_i)_i \cup (p\cF_j)_j$ is also in $Q$.
  The convexity of $\qids[Q]$ follows from a similar argument.
  It remains to show that $\qids[Q]$ is balanced:
  \begin{equation}
    \lambda\cdot\sum_i\beta_i\cE_i = \sum_i(\lambda\beta_i)\cE_i \in \qids[Q]
  \end{equation}
  for $\abs{\lambda}\leq 1$, $\abs{\beta_i}\leq 1$ and $(\cE_i)_i\in Q$.
\end{proof}

When a decomposition set is absolutely convex, we can even assume without loss of generality that our QPDs only have a single term in them.
\begin{lemma}{}{one_element_qpd}
  For an absolutely convex decomposition set $\ds\subset\hp(A\rightarrow B)$ and $\cE\in\hp(A\rightarrow B)$
  \begin{equation}\label{eq:one_element_qpd}
    \gamma_{\ds}(\cE) = \inf\{\kappa\geq 0 : \cE = \kappa\cF \text{ for some } \cF\in\ds \} \, .
  \end{equation}
\end{lemma}
\begin{proof}
  Similar to the proof in~\Cref{lem:two_element_qpd}, we only need to show that any QPD $\cE=\sum_ia_i\cF_i$ with $\cF_i\in\ds$ induces another QPD with only one element and the same 1-norm.
  Define $\kappa\coloneqq \sum_i\abs{a_i}$ and $\cF\coloneqq \sum_i \frac{a_i}{\kappa}\cF_i$.
  Clearly, $\cE=\kappa\cF$ and $\cF\in\ds$ by the absolute convexity.
\end{proof}
Besides the previous physical motivation, there is another reason why in many instances one can assume without loss of generality that the decomposition set $\ds$ is either convex or absolutely convex.
If we only care about computing the quasiprobability extent, we can replace $\ds$ by its convex hull or its absolute convex hull.
\begin{lemma}{}{gamma_hulls}
  For any $S\subset\hp(A\rightarrow B)$ and $\cE\in\hp(A\rightarrow B)$ one has
  \begin{equation}
    \gamma_{\ds}(\cE) = \gamma_{\conv(\ds)}(\cE) = \gamma_{\aconv(\ds)}(\cE)
  \end{equation}
  where the \emph{convex hull} and \emph{absolute convex hull} of $\ds$ are defined as
  \begin{align}
    \conv(\ds) &\coloneqq \left\{\sum_{i=1}^n \lambda_i x_i \middle\vert n\in\mathbb{N}, x_i\in \ds, \lambda_i\in [0,1], \sum_{i=1}^m\lambda_i = 1 \right\} \\
    \aconv(\ds) &\coloneqq \left\{\sum_{i=1}^n \lambda_i x_i \middle\vert n\in\mathbb{N}, x_i\in \ds, \lambda_i\in [-1,1], \sum_{i=1}^m\abs{\lambda_i} \leq 1 \right\} \, .
  \end{align}
\end{lemma}
Note that this statement does not rule out the possibility that classical side information can reduce the quasiprobability extent, as $\aconv(\qichan[Q])$ is generally a strict subset of $\qids[Q]$ for some $Q\subset\qi(A\rightarrow B)$.
\begin{proof}
  By~\Cref{lem:gamma_ds_bound} we know $\gamma_{\ds}(\cE) \geq \gamma_{\conv(\ds)}(\cE) \geq \gamma_{\aconv(\ds)}(\cE)$ so it suffices to show that $\gamma_{\aconv(\ds)}(\cE) \geq \gamma_{\ds}(\cE)$.
  We argue this by showing that any QPD of $\cE$ w.r.t. $\aconv(\ds)$ induces a QPD of $\cE$ w.r.t. $\ds$ with at most the same 1-norm of its coefficients.
  Let $\cE=\sum_i a_i\cF_i$ be such a QPD with $\smash{\cF_i=\sum_j\lambda^{(i)}_j\cG^{(i)}_j}$ where $\smash{\cG^{(i)}_j\in\ds}$.
  We thus have 
  \begin{equation}
    \cE = \sum_{i,j} (a_i\lambda^{(i)}_j) \cG^{(i)}_j
  \end{equation}
  which is a valid QPD with 1-norm
  \begin{equation}
    \sum_{i,j} \abs{a_i\lambda^{(i)}_j} = \sum_i \abs{a_i}\sum_j\abs{\lambda^{(i)}_j} \leq \sum_i \abs{a_i} .
  \end{equation}
\end{proof}

\paragraph{Relation to robustness}
Motivated by a class of resource measures studied in many QRTs, we define the \emph{robustness} (sometimes also called \emph{total robustness} to differentiate it from related measures) of a given operation $\cE$ w.r.t. a decomposition set $\ds$ as follows.
\begin{definition}{}{robustness}
  The robustness of $\cE\in\hp(A\rightarrow B)$ w.r.t. $\ds\subset\hp(A\rightarrow B)$ is defined as
  \begin{equation}\label{eq:robustness}
    R_{\ds}(\cE) \coloneqq \inf \{ t\geq 0 | \exists\cF\in\ds \text{ s.t. } \frac{\cE+t\cF}{\tr[\choi{\cE}]+t}\in\ds \} \, .
  \end{equation}
\end{definition}
Intuitively, the robustness quantifies the distance of $\cE$ to $\ds$ by measuring how much of a ``free operation'' (where we call operations from $\ds$ free) can be mixed with $\cE$ before it becomes itself ``free''.
For convex decomposition sets consisting of quantum channels, and more generally trace-preserving maps, the robustness is directly related to the quasiprobability extent.
\begin{lemma}{}{robustness1}
  Consider $\ds\subset \hptp(A,B)$ convex and $\cE\in\hp(A\rightarrow B)$. Then
  \begin{equation}
    \gamma_{\ds}(\cE) = \tr[\choi{\cE}] + 2 R_{\ds}(\cE)
  \end{equation}
\end{lemma}
Note that in this setting, $\gamma_{\ds}(\cE)=R_{\ds}(\cE)=\infty$ if $\cE$ is not itself proportional to a trace-preserving map.
\begin{proof}
  We prove this by showing that any valid two-element (see~\Cref{lem:two_element_qpd}) QPD $(a^{\pm},\cF^{\pm})$ induces a valid pair $(t,\cF)$ with $\tr[\choi{\cE}]+2t=a^++a^-$ and vice versa.

  First, let $t,\cF$ be valid.
  They thus fulfill $(\tr[\choi{\cE}]+t)\cG = \cE + t\cF$ for some $\cG\in\ds$.
  Therefore, we get a valid QPD $\cE = (\tr[\choi{\cE}]+t)\cG - t\cF$.

  Conversely, let us assume $\cE=a^+\cF^+ - a^-\cF^-$ is a valid QPD.
  Applying the trace on the Choi representation of both sides of this equation yields $\tr[\choi{\cE}]=a^+-a^-$ or $a^+=\tr[\choi{\cE}]+a^-$.
  Therefore, $\frac{\cE+a^-\cF^-}{\tr[\choi{\cE}]+a^-}=\cF^+\in\ds$.
\end{proof}
Notice that~\Cref{lem:robustness1} doesn't apply for (non-trivial) absolutely convex decomposition sets $\ds$, as they always contain non-trace-preserving maps.
Fortunately, the relation to the robustness can still be established in physically relevant settings.
\begin{lemma}{}{robustness2}
  Let $Q\subset\qi(A\rightarrow B)$ be closed under mixture and $\cE\in\hp(A\rightarrow B)$.
  Then
  \begin{equation}
    \gamma_{\qids[Q]}(\cE) = \tr[\choi{\cE}] + 2 R_{\qids[Q]}(\cE) \, .
  \end{equation}
\end{lemma}
\begin{proof}
  First, we note that for any $\cG=\sum_j\beta_j\cG_j\in\qids[Q]$ we have
  \begin{equation}
    \norm{\choi{\cG}}_1 = \norm{\sum_j\beta_j\choi{\cG_j}}_1 \leq \norm{\sum_j\abs{\beta_j}\choi{\cG_j}}_1 = \tr\left[ \sum_j\abs{\beta_j}\choi{\cG_j} \right] \leq  \tr\left[ \sum_j\choi{\cG_j} \right] = 1
  \end{equation}
  where we used that $\norm{X-Y}_1\leq\norm{X}_1+\norm{Y}_1 = \norm{X+Y}_1$ for two positive operators $X,Y$.
  Therefore, for any valid QPD $\cE=\sum_ia_i\cF_i$ with $\cF_i\in \qids[Q]$ one has
  \begin{align}
    \norm{a}_1
    &\geq \sum_i \abs{a_i}\norm{\choi{\cF_i}}_1 \\
    &= \sum_i \norm{a_i\choi{\cF_i}}_1 \\
    &\geq \norm{\sum_i a_i\choi{\cF_i}}_1 \\
    &= \norm{\choi{\cE}}_1 \\
    &\geq \tr[\choi{\cE}]
  \end{align}
  and hence $\gamma_{\qids[Q]}(\cE)\geq \tr[\choi{\cE}]$.
  We can thus write $\gamma_{\qids[Q]}(\cE) = \tr[\choi{\cE}] + 2r$ for some $r\geq 0$.
  It remains to show that this $r$ precisely corresponds to $R_{\qids[Q]}(\cE)$.

  The proof proceeds analogously to the one of~\Cref{lem:robustness1}.
  On one side, consider an arbitrary valid QPD $\cE=\kappa\cF$ with $\cF\in\qids[Q]$ (recall~\Cref{lem:one_element_qpd}).
  We can write $\kappa = \tr[\choi{\cE}] + 2r$ for some $r\geq 0$.
  Hence, $\smash{\frac{\cE + r(-\cF)}{\tr[\choi{\cE}]+r}=\cF\in\qids[Q]}$ where $(-\cF)$ is also in $\qids[Q]$ due to the absolute convexity (see~\Cref{lem:decompset_absolutely_convex}).
  This means that $(r,(-\cF))$ is a feasible point for the robustness.

  Conversely, let $t,\cF$ be an arbitrary feasible point for the definition of the robustness.
  They thus fulfill $(\tr[\choi{\cE}]+t)\cG = \cE + t\cF$ for some $\cG\in\qids[Q]$.
  Therefore, we get a valid QPD $\cE = (\tr[\choi{\cE}]+t)\cG - t\cF$ with 1-norm $\tr[\choi{\cE}]+2t$.
\end{proof}

\paragraph{Function properties of the quasiprobability extent}
We now explore a few properties of $\gamma_{\cE}$ as a function.
\begin{lemma}{}{gamma_convexity}
  For any $\ds\subset\hp(A\rightarrow B)$, $\cE_i\in\hp(A\rightarrow B)$ for $i=1,\dots,m$ and $a\in\mathbb{R}^m$ one has
  \begin{equation}\label{eq:gamma_convexity}
    \gamma_{\ds}\left(\sum_{i=1}^m a_i\cE_i\right) \leq \sum_{i=1}^m \abs{a_i}\gamma_{\ds}(\cE_i)
  \end{equation}
\end{lemma}
In particular, this result shows that $\gamma_{\ds}$ is a convex function on the span of the decomposition set $\spn_{\mathbb{R}}(\ds)$.
\begin{proof}
  If the left-hand side is infinite, then there must be at least one $\cE_i$ for which $\gamma_{\ds}(\cE_i)=\infty$ (as $\sum_i a_i\cE_i$ lies in the span of the $\cE_i$), so the statement trivially holds.
  We can thus assume without loss of generality that all involved quantities exhibit a finite quasiprobability extent.

  The proof follows from the insight that a collection of QPDs for each of the $\cE_i$
  \begin{equation}
     \cE_i = \sum_{j=1}^{m_i} b_j^{(i)} \cF_j^{(i)}
  \end{equation}
  induces a QPD of $\sum_ia_i\cE_i$
  \begin{equation}
    \sum_i a_i\cE_i = \sum_{i=1}^m\sum_{j=1}^{m_i} a_ib_j^{(i)} \cF_j^{(i)}
  \end{equation}
  which has coefficients with a 1-norm of $\sum_{i}\abs{a_i}\cdot\left(\sum_j\abs{b_j^{(i)}}\right)$.
\end{proof}
\begin{lemma}{}{gamma_homogeneous}
  For any $\ds\subset\hp(A\rightarrow B)$ the quasiprobability extent $\gamma_{\ds}$ is absolutely homogeneous, i.e.,
  \begin{equation}
    \forall\lambda\in\mathbb{R},\forall \cE\in\hp(A\rightarrow B): \gamma_{\ds}(\lambda \cE) = \abs{\lambda}\gamma_{\ds}(\cE) \, .
  \end{equation}
\end{lemma}
\begin{proof}
  If $\lambda=0$ or $\cE\notin\spn_{\mathbb{R}}(\ds)$, then the statement is trivial.
  For the remaining case, it suffices to consider $\lambda>0$ as clearly $\gamma_{\ds}(\cE)=\gamma_{\ds}(-\cE)$.
  Showing $\forall\lambda>0,\forall\cE: \gamma_{\ds}(\lambda \cE)\leq \lambda\gamma_{\ds}(\cE)$ is sufficient in that regard, as this then implies
  \begin{equation}
    \lambda \gamma_{\ds}(\cE) = \lambda \gamma_{\ds}(\frac{1}{\lambda}\lambda \cE) \leq \gamma_{\ds}(\lambda \cE) \, .
  \end{equation}
  We can easily see $\gamma_{\ds}(\lambda \cE)\leq \lambda\gamma_{\ds}(\cE)$ by verifying that every QPD $\cE=\sum_i a_i\cF_i$ induces a QPD of $\lambda \cE$ with a 1-norm of the coefficients that is exactly $\lambda$ multiplied by the 1-norm of the original QPD.
\end{proof}
\begin{corollary}{}{}
  For any bounded decomposition set $\ds\subset\hp(A\rightarrow B)$, the quasiprobability extent $\gamma_{\ds}$ is a norm on $\spn_{\mathbb{R}}(\ds)$.
\end{corollary}
As an example, we will see in~\Cref{chap:nonphysical} that choosing $\ds$ to be the set of CPTP maps results into the quasiprobability extent being the well-known diamond norm.
\begin{proof}
  A norm must fulfill three properties: It is \emph{sub-additive} ($\gamma_{\ds}(\cE + \cF) \leq \gamma_{\ds}(\cE) + \gamma_{\ds}(\cF)$, which follows from~\Cref{lem:gamma_convexity}), \emph{absolutely homogeneous} (see~\Cref{lem:gamma_homogeneous}) and \emph{point-separating} ($\gamma_{\ds}(\cE)=0 \Leftrightarrow \cE=0$).
  It remains to show the last property.
  Consider some $\cE\in\hp(A\rightarrow B)$ such that $\gamma_{\ds}(\cE)=0$.
  Therefore, there exists a sequence of QPDs $\cE=\sum_i a_i^{(j)}\cF_i^{(j)}$ for $j\in\mathbb{N}$ such that $\lim\limits_{j\rightarrow\infty}\norm{a^{(j)}}_1=0$.
  Since for all $j$ we have
  \begin{equation}
    \norm{\cE}\leq \sum_i\abs{a^{(j)}_i}\norm{\cF_i^{(j)}}\leq C\norm{a^{(j)}}_1
  \end{equation}
  for some constant $C\geq 0$, we have $\norm{\cE}=0$ and thus $\cE=0$ (here $\norm{\cdot}$ denotes some arbitrary norm).
\end{proof}

The following statements will depend on a choice of basis $\{\mathcal{B}_1,\mathcal{B}_2,\dots,\mathcal{B}_m\}\subset\ds$ of $\spn_{\mathbb{R}}(\ds)$.
The gram matrix $G$ of such a basis is defined w.r.t. the Hilbert-Schmidt inner product of the associated Choi operators $G_{ij}\coloneqq\langle\choi{\mathcal{B}_i},\choi{\mathcal{B}_j}\rangle$.
\begin{lemma}{}{gamma_upper_bound}
  Consider some $\ds\subset\hp(A\rightarrow B)$ and a basis of $\spn_{\mathbb{R}}(\ds)$ contained in $\ds$ with Gram matrix $G$.
  For any $\cE\in\hp(A\rightarrow B)$ one has
  \begin{equation}
    \gamma_{\ds}(\cE)\leq \sqrt{\frac{\dimension(A)\dimension(B)}{\sigma_{\mathrm{min}}(G)}}\norm{\choi{\cE}}_2
  \end{equation}
  where $\sigma_{\mathrm{min}}(G)$ is the smallest eigenvalue of $G$.
\end{lemma}
For superoperators derived from a set of quantum instruments $\cE\in\qids[Q]$, we always have $\smash{\norm{\choi{\cE}}_2\leq\norm{\choi{\cE}}_1\leq 1}$.
\begin{proof}
  Let us write the basis as $\{\mathcal{B}_1,\mathcal{B}_2,\dots,\mathcal{B}_m\}$ and denote the expansion of $\cE$ into the basis by $\cE=\sum_i \lambda_i \mathcal{B}_i$ where $\lambda\in\mathbb{R}^m$.
  Notice that
  \begin{equation}
    \frac{\norm{\choi{\cE}}_2^2}{\dimension(AB)} = \frac{\tr[\choi{\cE}^{\dagger}\choi{\cE}]}{\dimension(AB)} = \langle \choi{\cE},\choi{\cE}\rangle = \sum_{i,j}\lambda_i\lambda_j G_{ij} \geq \sigma_{\mathrm{min}}(G)\norm{\lambda}_2^2 \, .
  \end{equation}
  The desired statement follows from
  \begin{equation}
    \gamma_{\ds}(\cE) \leq \norm{\lambda}_1 \leq \sqrt{m}\norm{\lambda}_2 \leq \frac{\sqrt{m}}{\sqrt{\dimension(AB)\sigma_{\mathrm{min}}(G)}}\norm{\choi{\cE}}_2 \leq \sqrt{\frac{\dimension(AB)}{\sigma_{\mathrm{min}}(G)}}\norm{\choi{\cE}}_2
  \end{equation}
  where we used the Cauchy-Schwarz inequality.
\end{proof}
As a simple consequence, this allows us to prove the Lipschitz continuity of the quasiprobability extent.
\begin{lemma}{}{gamma_lipschitz}
  Consider some $\ds\subset\hp(A\rightarrow B)$ and a basis of $\spn_{\mathbb{R}}(\ds)$ contained in $\ds$ with Gram matrix $G$.
  Then $\gamma_{\ds}$ is Lipschitz continuous on $\spn_{\mathbb{R}}(\ds)$
  \begin{equation}
	\forall \cF,\cG\in \spn_{\mathbb{R}}(S): 
    \abs{\gamma_{\ds}(\cE)-\gamma_{\ds}(\cF)}\leq \sqrt{\frac{\dimension(A)\dimension(B)}{\sigma_{\mathrm{min}}(G)}} \norm{\choi{\cE}-\choi{\cF}}_2
  \end{equation}
  where $\sigma_{\mathrm{min}}(G)$ is the smallest eigenvalue of $G$.
\end{lemma}
\begin{proof}
  Define $\Delta\coloneqq\cF - \cE$.
  By~\Cref{lem:gamma_convexity} we have
  \begin{equation}
    \gamma_{\ds}(\cF) = \gamma_{\ds}(\cE + \Delta) \leq \gamma_{\ds}(\cE) + \gamma_{\ds}(\Delta)
    \quad\text{and}\quad
    \gamma_{\ds}(\cE) \leq \gamma_{\ds}(\cF) +  \gamma_{\ds}(\Delta) \, .
  \end{equation}
  It thus suffices to apply~\Cref{lem:gamma_upper_bound} and obtain
  \begin{equation}
    \gamma_{\ds}(\Delta) \leq \sqrt{\frac{\dimension(AB)}{\sigma_{\mathrm{min}}(G)}}\norm{\choi{\Delta}}_2 \, .
  \end{equation}
\end{proof}

\begin{example}\label{ex:modifiedendobasis}
  For single-qubit Pauli operators $Q_1,Q_2\in \mathrm{P}_1$ we define
  \begin{equation}
    R_{Q_1} \coloneqq \frac{1}{\sqrt{2}}\left( \id + iQ_1 \right)
    \, \quad
    R_{Q_1,Q_2} \coloneqq \frac{1}{\sqrt{2}}\left( Q_1 + Q_2 \right)
  \end{equation}
  \begin{equation}
    \pi_{Q_1} \coloneqq \frac{1}{2}\left( \id + Q_1 \right)
    \, \quad
    \pi_{Q_1,Q_2} \coloneqq \frac{1}{2}\left( Q_1 + iQ_2 \right) \, .
  \end{equation}
  Consider the basis $B=\{\mathcal{B}_1,\dots,\mathcal{B}_{16}\}$ of the 16-dimensional space of single-qubit Hermitian-preserving operators depicted in~\Cref{tab:modifiedendobasis}.
  Note that $R_{Q_1}$ is a Clifford unitary, since for $Q_3\in P_1$ we have
  \begin{equation}
    R_{Q_1}Q_3R_{Q_1}^{\dagger} = \begin{cases}Q_3 & \text{ if $[Q_1,Q_3]=0$} \\ iQ_1Q_3 & \text{ else}\end{cases} \, .
  \end{equation}
  Similarly, $R_{Q_1,Q_2}$ is also a Clifford unitary whenever $Q_1$ and $Q_2$ anti-commute
  \begin{equation}
    R_{Q_1,Q_2}Q_3R_{Q_1,Q_2}^{\dagger} = \begin{cases}
        Q_3 & \text{ if $[Q_1,Q_3]=0$ and $[Q_2,Q_3]=0$} \\
        -Q_3 & \text{ if $[Q_1,Q_3]\neq 0$ and $[Q_2,Q_3]\neq 0$} \\
        Q_1Q_2Q_3 & \text{ if $[Q_1,Q_3]\neq 0$ and $[Q_2,Q_3]=0$} \\
        -Q_1Q_2Q_3 & \text{ if $[Q_1,Q_3]=0$ and $[Q_2,Q_3]\neq 0$}
        \end{cases} \, .
  \end{equation}
  Therefore, the $\mathcal{B}_i$ can all be implemented using single-qubit Clifford gates, measurements in the computational basis, classical randomness and classical side information (recall that $\mathcal{M}\coloneqq\indsupo{\proj{0}}-\indsupo{\proj{1}}$ can be realized with classical side information as discussed in~\Cref{sec:qps_intermediate_measurements}).
  
  Remarkably, the Choi matrices of these 16 operators are pairwise orthogonal.
  The Frobenius norm of the Choi matrices is $1$ for the first four elements and $1/\sqrt{2}$ for the remaining 12 elements.
  The associated Gram matrix is therefore diagonal with smallest eigenvalue $1/8$.

  The basis $B$ can be extended to an orthogonal basis of an arbitrary $n$-qubit Hilbert space by simply taking the tensor products of the basis elements.
  The smallest eigenvalue of the Gram matrix is then $2^{-3n}$.
  As such,~\Cref{lem:gamma_lipschitz} implies that for any decomposition set $\ds$ containing $B^{\otimes n}$, the quasiprobability extent $\gamma_{\ds}$ is Lipschitz continuous with Lipschitz constant $2^{5n/2}$. 
  \begin{table}[ht]
  \centering
  \begin{tabular}{| r | l l l|} 
    \hline
     1 & $\idchan$ & & \\  
     2 & $\indsupo{X}$ & & \\
     3 & $\indsupo{Y}$ & & \\
     4 & $\indsupo{Z}$ & & \\
     5 & $\frac{1}{2}(\indsupo{R_X}-\indsupo{R_{-X}})$ & & \\
     6 & $\frac{1}{2}(\indsupo{R_Y}-\indsupo{R_{-Y}})$ & & \\
     7 & $\frac{1}{2}(\indsupo{R_Z}-\indsupo{R_{-Z}})$ & & \\
     8 & $\frac{1}{2}(\indsupo{R_{Y,Z}}-\indsupo{R_{Y,-Z}})$ & & \\
     9 & $\frac{1}{2}(\indsupo{R_{Z,X}}-\indsupo{R_{Z,-X}})$ & & \\
     10 & $\frac{1}{2}(\indsupo{R_{X,Y}}-\indsupo{R_{X,-Y}})$ & & \\
     11 & $\indsupo{\pi_{X}}-\indsupo{\pi_{-X}}$ & = & $\indsupo{R_Z^3R_X^3}\circ\mathcal{M}\circ\indsupo{R_XR_Z}$ \\
     12 & $\indsupo{\pi_{Y}}-\indsupo{\pi_{-Y}}$ & = & $\indsupo{R_X}\circ\mathcal{M}\circ\indsupo{R_X^3}$ \\
     13 & $\indsupo{\pi_{Z}}-\indsupo{\pi_{-Z}}$ & = & $\mathcal{M}$ \\
     14 & $\indsupo{\pi_{Y,Z}}-\indsupo{\pi_{Z,Y}}$ & = & $\indsupo{R_Z^3R_X^3}\circ\mathcal{M}\circ\indsupo{R_X^3R_Z}$ \\
     15 & $\indsupo{\pi_{Z,X}}-\indsupo{\pi_{X,Z}}$ & = & $\indsupo{R_X}\circ\mathcal{M}\circ\indsupo{R_X^3R_Z^2}$ \\
     16 & $\indsupo{\pi_{X,Y}}-\indsupo{\pi_{Y,X}}$ & = & $\mathcal{M}\circ\indsupo{R_X^2}$ \\
     \hline
  \end{tabular}
  \caption{16 Clifford operations constituting a basis of the 16-dimensional space of Hermitian-preserving single-qubit superoperators. Here, $\mathcal{M}\coloneqq\indsupo{\proj{0}}-\indsupo{\proj{1}}$ as in~\Cref{sec:qps_intermediate_measurements}.}
  \label{tab:modifiedendobasis}
  \end{table}
\end{example}

\paragraph{Compact decomposition sets}
We will encounter many settings, where the infimum of the quasiprobability extent is always achieved.
\begin{lemma}{}{compact_ds}
  If $\ds$ is compact, then the infimum in the quasiprobability extent (see~\Cref{eq:gamma,eq:one_element_qpd,eq:two_element_qpd}) is achieved on $\spn_{\mathbb{R}}(\ds)$.
\end{lemma}
\begin{proof}
  Without loss of generality, we can assume $\ds$ is absolutely convex, otherwise we can consider its absolute convex hull (see~\Cref{lem:gamma_hulls}, the quasiprobability extent is clearly achieved on $\ds$ if and only if it is achieved on $\aconv(\ds)$).
  For any $\cE\in\spn_{\mathbb{R}}(\ds)$ there must exist a QPD $\cE=\tau\cF$ for some $\cF\in\aconv(\ds)$.
  Hence, we can write
  \begin{equation}
    \gamma_{\ds}(\cE) = \inf\limits_{\kappa\in [0,\tau],\cG\in\aconv(\ds)}\{\kappa | \cE = \kappa\cG \} \, .
  \end{equation}
  Since the absolute convex hull of a compact set is again compact, we are optimizing a continuous function over a compact set, and as such the infimum must be achieved.
\end{proof}

The next property shows that the quasiprobability extent can provide a faithful criterion to determine whether an element is part of the decomposition set or not.
\begin{lemma}{}{gamma_faithful}
  Let $\ds\subset\hp(A\rightarrow B)$ be compact and $\cE\in\hp(A\rightarrow B)$.
  Then $\ds$ is absolutely convex if and only if $\gamma_{\ds}$ fulfills
  \begin{equation}
    \gamma_S(\cE) \leq 1 \Longleftrightarrow \cE\in S \, .
  \end{equation}
\end{lemma}
\begin{proof}
  Consider $\ds$ to be absolutely convex.
  Clearly $\cE\in\ds\Rightarrow\gamma_{\ds}(\cE)\leq 1$ and $\gamma_{\ds}(\cE)\leq 1$ implies by~\Cref{lem:one_element_qpd} and~\Cref{lem:compact_ds} that $\cE\in\ds$.

  Conversely, consider a $\ds$ such that $\gamma_{\ds}$ fulfills the faithfulness criterion.
  Then, for any absolutely convex combination $\sum_i a_i\cF_i$ with $\sum_i\abs{a_i}\leq 1$ and $\cF_i\in\ds$ we have $\gamma_{\ds}(\sum_ia_i\cF_i)\leq \sum_i\abs{a_i}\gamma_{\ds}(\cF_i)\leq 1$ by~\Cref{lem:gamma_convexity}.
  Hence, $\sum_ia_i\cF_i$ is an element of $\ds$.
\end{proof}
If we can additionally assume $\cE$ to be trace-preserving, then a tighter statement can be made.
\begin{lemma}{}{gamma_faithful2}
  Let $\ds\subset\hp(A\rightarrow B)$ be convex and compact.
  Furthermore, let either $\ds\subset\tp(A\rightarrow B)$ or $\ds=\qids[Q]$ for some $Q\subset\qi(A\rightarrow B)$ that is closed under mixture.
  For any $\cE\in\hptp(A\rightarrow B)$ we have
  \begin{equation}
    \gamma_{\ds}(\cE) = 1 \Longleftrightarrow \cE\in\ds \, .
  \end{equation}
\end{lemma}
\begin{proof}
  If $\cE\in\ds$ then $\gamma_{\ds}(\cE)\leq 1$.
  The other inequality $\gamma_{\ds}(\cE)\geq 1$ follows directly from~\Cref{lem:robustness1} or~\Cref{lem:robustness2} and the non-negativity of the robustness.
  
  Conversely, let $\gamma_{\ds}(\cE)=1$.
  If $\ds=\qids[Q]$, then~\Cref{lem:gamma_faithful} applies.
  Otherwise, by~\Cref{lem:two_element_qpd,lem:compact_ds} there exists a QPD $\cE=a^+\cF^+-a^-\cF^-$ with $a^++a^-=1$.
  Applying the trace on the Choi representation of both sides, we get $a^+-a^-=1$ and therefore $a^+=1,a^-=0$.
  Hence, $\cE\in\ds$.
\end{proof}

\section{The quasiprobability extent in quantum resource theories}\label{sec:gamma_qrt}
QPS lends itself very well to being used in a resource theoretic context.
A QRT separates all channels into a set of ``free'' and ``non-free'' operations.
As such, it is natural to study how expensive it is to simulate a non-free operation using free operations.
Most applications of QPS we will consider in this thesis can be framed in the context of some QRT.
For instance, the near-Clifford simulation procedure introduced in~\Cref{sec:motivational_example} is tightly connected to the resource theory of magic.
The non-free resources in this case are the magic states and the non-Clifford operations.

More mathematically, for some QRT $\qrt$ we will now study the quasiprobability extent $\gamma_{\qrt^{\star}}$ for
\begin{equation}\label{eq:qrt_extended_ds}
  \qrt^{\star}(A\rightarrow B) \coloneqq \qids[\fqi(A\rightarrow B)]
\end{equation}
where $\fqi$ are the free quantum instruments associated to $\qrt$.
This choice of decomposition set allows us to capture QPS with classical side information, as discussed in~\Cref{sec:qps_intermediate_measurements}.

Generally, we will consider convex QRTs, i.e., resource theories where the set of free channels forms a convex set.
This property can also be characterized in different ways.
\begin{lemma}{}{qrt_convexity}
  Let $\qrt$ be a QRT in which classical processing is free.
  The following are equivalent.
  \begin{enumerate}[(a)]
    \item $\qrt(A\rightarrow B)$ is convex for all $A,B$.
    \item $\fqi(A\rightarrow B)$ is closed under mixture for all $A,B$.
    \item $\qrt^{\star}(A\rightarrow B)$ is absolutely convex for all $A,B$.
  \end{enumerate}
\end{lemma}
\begin{proof}
  To show (a) $\Rightarrow$ (b), recall that $(\cE_i)_i,(\cF_j)_j\in\fqi(A\rightarrow B)$ implies that there exist free channels of the form $\sum_i\cE_i\otimes\proj{i}_{X_1}\in\qrt(A\rightarrow BX_1)$ and $\sum_j\cF_j\otimes\proj{j}_{X_2}\in\qrt(A\rightarrow BX_2)$.
  Taking a convex mixture of these two measurement channels corresponds exactly to taking a mixture of the two original instruments.
  To achieve this, appropriate classical processing with padding is required on $X_1$ and $X_2$ such that the two measurement channels operate on the same spaces.

  The statement (b) $\Rightarrow$ (c) was already proven in~\Cref{lem:decompset_absolutely_convex}.

  For the final direction (c) $\Rightarrow$ (a), consider two channels $\cE,\cF\in\qrt(A\rightarrow B)$.
  We know that any convex mixture $(1-p)\cE +p\cF$ lies in $\qrt^{\star}(A\rightarrow B)$, and we now need to show that it also lies in $\qrt(A\rightarrow B)$.
  By the definition of $\qrt^{\star}(A\rightarrow B)$, we can write
  \begin{equation}
    (1-p)\cE + p\cF = \sum_i \beta_i \cG_i
  \end{equation}
  for some $\beta_i\in [-1,1]$ and $(\cG_i)_i\in\fqi(A\rightarrow B)$.
  Taking the Choi isomorphism on both sides, as well as the partial trace $\tr_B$ thereof, we obtain
  \begin{equation}
    \frac{1}{\dimension(A)}\id = \sum_i \beta_i \tr_B[\choi{\cG_i}] \, .
  \end{equation}
  At the same time, we know that $\sum_i\cG_i\in\cptp(A\rightarrow B)$, so
  \begin{equation}
    \frac{1}{\dimension(A)}\id = \sum_i \beta_i \tr_B[\choi{\cG_i}] \loewnerleq \sum_i \tr_B[\choi{\cG_i}] = \frac{1}{\dimension(A)}\id \, .
  \end{equation}
  This implies that for all $i$ we have $\beta_i=1$, unless $\cG_i$ is itself zero.
  Therefore, $(1-p)\cE + p\cF\in\qichan[\fqi(A\rightarrow B)]$, which is equal to $\qrt(A\rightarrow B)$ by~\Cref{lem:free_chan_from_qi}.
\end{proof}
In the context of QPS, convexity of a QRT is very natural, as the QPS procedure itself requires the ability to randomly pick and apply a channel during the circuit execution.

The set of free states of a QRT is denoted by $\fs(A)$ for some given system $A$.
It can be understood to be the free channels from the trivial Hilbert space $\mathbb{C}$ to $A$, i.e., $\fs(A)\coloneqq\qrt(\mathbb{C}\rightarrow A)\subset\dops(A)$.
As such, one might a priori expect that the simulation of non-free states $\omega\in\dops(A)$ should involve QPDs w.r.t. $\qrt^{\star}(\mathbb{C},A)$ to fully harness the utility of classical side information.
However, it turns out that classical side information does not provide any benefit for simulating states.
\begin{lemma}{}{negativity_state_qpd}
  Let $\qrt$ be a QRT that fulfills the axiom of free instruments and $A$ be some system.
  Then $\aconv(\qrt^{\star}(\mathbb{C},A)) = \aconv(\qrt(\mathbb{C},A))$ and hence
  \begin{equation}
    \forall\omega\in\herm(A): \gamma_{\qrt^{\star}(\mathbb{C},A)}(\omega) = \gamma_{\qrt(\mathbb{C},A)}(\omega) = \gamma_{\fs(A)}(\omega) \, .
  \end{equation}
\end{lemma}
\begin{proof}
  The second statement follows from~\Cref{lem:gamma_hulls}.
  Since $\qrt(\mathbb{C},A)\subset\qrt^{\star}(\mathbb{C},A)$, it suffices to show that $\qrt^{\star}(\mathbb{C},A)\subset \aconv(\qrt(\mathbb{C},A))$.
  Consider an element $\sum_i\beta_i\cE_i \in \qrt^{\star}(\mathbb{C},A)$.
  The $\cE_i$ can be identified with positive operators $\rho_i \in \herm(A)$ which fulfill $\tr\left[\sum_i \rho_i\right]=1$.
  By the axiom of free instruments, all $\rho_i / \tr[\rho_i]$ must be free states $\rho_i / \tr[\rho_i] \in \fs(A)$, respectively $\cE_i / \tr[\rho_i]\in \qrt(\mathbb{C},A)$, which implies that $\sum_i\beta_i\tr[\rho_i](\cE_i/\tr[\rho_i]) \in \aconv(\qrt(\mathbb{C},A))$.
  Note that this proof would not work for channels (i.e. $A\neq \mathbb{C}$) as the $\cE_i$ can generally not be normalized by a scalar in order to obtain a valid quantum channel.
\end{proof}
For notational convenience, we will usually omit the involved systems in the decomposition set when they are clear from context.
For instance, we simply write $\gamma_{\qrt^{\star}}(\cE)$ and $\gamma_{\fs}(\omega)$, as the systems are implicitly specified by $\cE\in\hp(A\rightarrow B)$ and $\omega\in\herm(A)$.
The rest of this section is dedicated to studying various properties of the quasiprobability extent $\gamma_{\qrt^{\star}}$.
While we will focus on proving these properties for $\gamma_{\qrt^{\star}}$, many of them do also hold for $\gamma_{\qrt}$\footnote{In fact, these properties for $\gamma_{\qrt}$ are rather direct consequences of the definition of a QRT. Showing them for $\gamma_{\qrt^{\star}}$ requires a bit more work.}.

\paragraph{Chaining property}
Recall that the free operations of a QRT are by definition closed under composition.
We first show that the same holds for $\qrt^{\star}(A\rightarrow B)$ if the QRT additionally has a tensor product structure.
\begin{lemma}{}{ds_closed_under_concat}
  Let $\qrt$ be a QRT with tensor product structure and $\cE\in\qrt^{\star}(A\rightarrow B$), $\cF\in\qrt^{\star}(B\rightarrow C)$.
  Then $\cF\circ\cE\in\qrt^{\star}(A\rightarrow C)$.
\end{lemma}
\begin{proof}
  We can write $\cE=\sum_i\alpha_i\cE_i$ and $\cF=\sum_j\beta_j\cF_j$ for $(\cE_i)_i\in\fqi(A\rightarrow B)$, $(\cF_j)_j\in\fqi(B\rightarrow C)$ and $\alpha_i,\beta_j\in [-1,1]$.
  By the definition of $\fqi$, there exist free channels $\tilde{\cE}\in\qrt(A\rightarrow BX_1)$ and $\tilde{\cF}\in\qrt(B\rightarrow CX_2)$ given by
  \begin{equation}
    \tilde{\cE}(\rho)=\sum_i\cE_i(\rho)\otimes\proj{i}
    \, , \quad
    \tilde{\cF}(\rho)=\sum_j\cF_j(\rho)\otimes\proj{j} \, .
  \end{equation}
  By the tensor product structure, $(\tilde{\cF}\otimes\idchan_{X_1})\circ \tilde{\cE}$ is also a free operation in $\qrt(A\rightarrow CX_1X_2)$, which implies that the quantum instrument $(\cF_j\circ \cE_i)_{i,j}$ is free.
  Therefore, $\cF\circ\cE = \sum_{i,j}\alpha_i\beta_j \cF_j\circ\cE_i$ is an element of $\qrt^{\star}(A\rightarrow C)$.
\end{proof}
As a direct consequence, the quasiprobability extent fulfills the \emph{chaining property}.
\begin{lemma}{}{gamma_chaining}
  For any QRT $\qrt$ with tensor product structure and any $\cE\in\hp(A\rightarrow B),\cF\in\hp(B\rightarrow C)$ one has
  \begin{equation}
    \gamma_{\qrt^{\star}}(\cF\circ \cE) \leq \gamma_{\qrt^{\star}}(\cE) \gamma_{\qrt^{\star}}(\cF) \, .
  \end{equation}
\end{lemma}
To prove this, notice that any two QPDs of $\cE$ and $\cF$ can be merged into a QPD of $\cF\circ\cE$.
As an important special case, one can choose $A$ to be the trivial Hilbert space, which gives us a relation between the quasiprobability extent of states and of channels.
\begin{corollary}{}{gamma_state_from_channel}
  Let $\qrt$ be a QRT with tensor product structure, $\omega\in\herm(A)$ and $\cE\in\hp(A\rightarrow B)$.
  Then
  \begin{equation}
    \gamma_{\qrt^{\star}}(\cE(\omega)) \leq \gamma_{\qrt^{\star}}(\cE) \gamma_{\qrt^{\star}}(\omega) \, .
  \end{equation}
\end{corollary}
Another consequence is that the quasiprobability extent is invariant under reversible free operations.
\begin{corollary}{}{gamma_invariant_reversible}
  Consider a QRT $\qrt$ with tensor product structure.
  Let $\cU_1\in\qrt(A),\cU_2\in\qrt(B)$ be unitary channels s.t. their inverses $\cU_1^{-1}$, $\cU_2^{-1}$ are also free.
  For any $\cE\in\hp(A\rightarrow B)$ one has
  \begin{equation}
    \gamma_{\qrt^{\star}}(\cU_2\circ \cE) = \gamma_{\qrt^{\star}}(\cE) = \gamma_{\qrt^{\star}}(\cE\circ \cU_1) \, .
  \end{equation}
\end{corollary}
\begin{proof}
  We only show the first equality, the other is analogous.
  On one hand
  \begin{equation}
    \gamma_{\qrt^{\star}}( \cU_2\circ\cE ) \leq \gamma_{\qrt^{\star}}(\cU_2) \gamma_{\qrt^{\star}}(\cE) \leq \gamma_{\qrt^{\star}}(\cE)
  \end{equation}
  and conversely
  \begin{equation}
    \gamma_{\qrt^{\star}}(\cE) = \gamma_{\qrt^{\star}}( \cU_2^{-1} \circ \cU_2 \circ \cE) \leq \gamma_{\qrt^{\star}}(\cU_2^{-1}) \gamma_{\qrt^{\star}}(\cU_2 \circ \cE) \leq \gamma_{\qrt^{\star}}(\cU_2 \circ \cE) \, .
  \end{equation}
\end{proof}

%
%

\paragraph{Properties under tensor product}
Here, we investigate properties of the quasiprobability extent that are induced by the tensor product structure of the QRT.
First, we note that the decomposition set is closed under tensor product with the identity.
\begin{lemma}{}{ds_closed_under_tp}
  Let $\qrt$ be a QRT with tensor product structure, $\cE\in\qrt^{\star}(A\rightarrow B)$ and $C$ some system.
  Then $(\cE\otimes\idchan_C) \in \qrt^{\star}(AC\rightarrow BC)$.
\end{lemma}
\begin{proof}
  Let $\cE=\sum_j\beta_j\cE_j$ for $(\cE_j)_j\in\fqi(A\rightarrow B)$ and $\beta_j\in[-1,1]$.
  Therefore, the channel $\rho\mapsto \sum_j\cE_j(\rho)\otimes\proj{j}$ is free.
  Taking the tensor product of this channel with $\idchan_C$ yields another free channel (by the tensor product structure) and therefore $(\cE_j\otimes\idchan_C)_j\in\fqi(AC\rightarrow BC)$.
  Hence, $(\cE\otimes\idchan_C)=\sum_j\beta_j(\cE_j\otimes\idchan_C)$ is in $\qrt^{\star}(AC\rightarrow BC)$.
\end{proof}
This result can be directly applied to the quasiprobability extent.
\begin{lemma}{}{gamma_tp_idchan}
  Let $\qrt$ be a QRT with tensor product structure and $\cE\in\hp(A\rightarrow B)$.
  For any system $C$
  \begin{equation}
    \gamma_{\qrt^{\star}}(\cE\otimes\idchan_C) = \gamma_{\qrt^{\star}}(\cE) \, .
  \end{equation}
\end{lemma}
\begin{proof}
  For the side $\leq$, \Cref{lem:ds_closed_under_tp} implies that any QPD $\cE=\sum_ia_i\cF_i$ directly induces a valid QPD $\cE\otimes\idchan_C = \sum_ia_i \cF_i\otimes\idchan_C$.
  For the converse side $\geq$, we use~\Cref{lem:gamma_chaining} and
  \begin{equation}
    \cE = (\idchan_B\otimes \tr_E) \circ (\cE\otimes\idchan_E) \circ (\idchan_A\otimes \cG)
  \end{equation}
  where $\cG\in\cptp(\mathbb{C}\rightarrow E)$ is a state preparation channel that prepares some arbitrary free state.
  Since tracing out and preparing free states are free operations (see~\Cref{def:qrt,lem:qrt_free_state_prep}), so are $(\idchan_B\otimes \tr_E)$ and $(\idchan_A\otimes \cF)$ by the completely free operations property.
\end{proof}
\begin{remark}{}{combining_qpds}
  In~\Cref{sec:qpsim,sec:qps_intermediate_measurements}, we showed that the standard strategy for finding QPDs of large circuits consists of combining QPDs of the constant-sized building blocks that make up that circuit.
  The results in~\Cref{lem:ds_closed_under_concat,lem:ds_closed_under_tp} justify why this is generally possible.
  More precisely, they imply that if we each have a QPD of two superoperators $\cE\in\hp(A\rightarrow BC)$ and $\cF\in\hp(CD\rightarrow E)$, then they can be combined into a QPD of $(\cF\circ\cE)\in\hp(AD\rightarrow BE)$.
  By repeatedly applying this procedure on two subsequent building blocks in a circuit, we can combine all individual QPDs into one large QPD of the total circuit.
\end{remark}

Together,~\Cref{lem:gamma_chaining,lem:gamma_tp_idchan} imply that the quasiprobability extent is sub-multiplicative.
\begin{corollary}{}{gamma_submult}
  Let $\qrt$ be a QRT with tensor product structure and $\cE\in\hp(A\rightarrow B),\cF\in\hp(C\rightarrow D)$.
  Then
  \begin{equation}
    \gamma_{\qrt^{\star}}(\cE\otimes\cF) \leq \gamma_{\qrt^{\star}}(\cE)\gamma_{\qrt^{\star}}(\cF)
  \end{equation}
\end{corollary}

\paragraph{Properties from free classical computation}
Whenever the classical output of a measurement can be further processed using free operations, the induced decomposition set of the QRT can be written in a simpler form.
\begin{corollary}{}{free_classical_computation}
  Let $\qrt$ be a QRT with tensor product structure where classical computation channels (possibly involving classical randomness) are free.
  Then, for all systems $A$ and $B$ one has
  \begin{equation}
    \qrt^{\star}(A\rightarrow B) = \{ \cE^+ - \cE^- | (\cE^+,\cE^-)\in\fqi(A\rightarrow B)\} \, .
  \end{equation}
\end{corollary}
This is a direct consequence from~\Cref{lem:characterization_expanded_decomposition_set}.
Indeed, any resource theory in which classical computation is free, the coarse-graining map (called $f$ in~\Cref{def:qi_coarse_graining}) can be freely implemented on the classical measurement outcome register.
As such, a coarse-graining of a free instrument is again itself a free instrument.
Similarly, $\fqi(A\rightarrow B)$ is also trivially fine-grainable, since one can freely process a measurement result by extending the classical register with a random classical bit.

\paragraph{The quasiprobability extent as a resource measure}\label{par:gamma_static_monotone}
Even though the definition of the quasiprobability extent is operationally motivated by the simulation overhead of QPS, we have remarkably encountered many properties of it that are desirable for a measure to quantify resources in a QRT.
In this paragraph, we briefly summarize these insights and formally frame the quasiprobability extent as a resource measure.

Consider a QRT that fulfills the axiom of free instruments.
We start by considering $\gamma_{\fs}(\rho)$ (which is equal to $\gamma_{\qrt^{\star}}(\rho)$ by~\Cref{lem:negativity_state_qpd}) as a measure to quantify the amount of non-free resource in a state $\rho$.
To satisfy the commonly used conventions for normalization and sub-additivity (as opposed to sub-multiplicativity), we will consider the \emph{log-quasiprobability extent} $\log(\gamma_{\fs}(\rho))$.
\begin{proposition}{}{log_gamma_monotone}
  Let $\qrt$ be a QRT with tensor product structure that fulfills the axiom of free instruments.
  The log-quasiprobability extent $\log\gamma_{\fs}$ is a resource measure that is additionally
  \begin{itemize}[noitemsep]
    \item quasi-convex
    \item sub-additive
    \item strongly monotonous if $\qrt$ is convex and closed
    \item faithful if $\qrt$ is convex and closed.
  \end{itemize}
\end{proposition}
\begin{proof}
  The log-quasiprobability extent is clearly non-negative: For any $\rho\in\dops(A)$ with QPD $\rho=\sum_ia_i\sigma_i$ one has
  \begin{equation}
    1 = \norm{\rho}_1 \leq \sum_i \abs{a_i} \norm{\sigma_i}_1 = \norm{a}_1 \, .
  \end{equation}
  Normalization follows directly from $1\in\fs(\mathbb{C})$.
  We have already proven weak monotonicity (\Cref{cor:gamma_state_from_channel}), quasi-convexity (\Cref{lem:gamma_convexity}), sub-additivity (\Cref{cor:gamma_submult}) and faithfulness (\Cref{lem:gamma_faithful2}).
  It only remains to show strong monotonicity.

  Let $\rho\in\dops(A)$ and $(\cE_i)_i\in\fqi(A\rightarrow B)$ be a free quantum instrument.
  Let us write the optimal QPD of $\rho$ as $\rho=(1+t)\rho^+ - t\rho^-$ (recall~\Cref{lem:one_element_qpd}).
  This allows us to write
  \begin{equation}
    \frac{\cE_i(\rho)}{\tr[\cE_i(\rho)]}
    = \frac{1}{\tr[\cE_i(\rho)]}\left( (1+t)\cE_i(\rho^+) - t\cE_i(\rho^-) \right)
    = (1+\tilde{t}_i)\frac{\cE_i(\rho^+)}{\tr[\cE_i(\rho^+)]} - \tilde{t}_i\frac{\cE_i(\rho^-)}{\tr[\cE_i(\rho^-)]}
  \end{equation}
  where $\tilde{t}_i\coloneqq \frac{\tr[\cE_i(\rho^-)]}{\tr[\cE_i(\rho)]}t$.
  Hence, we have 
  \begin{equation}
    \sum_i \tr[\cE_i(\rho)]\gamma_{\fs}\left(\frac{\cE_i(\rho)}{\tr[\cE_i(\rho)]}\right)
    \leq \sum_i \tr[\cE_i(\rho)] (1 + 2\tilde{t}_i)
    \leq 1 + 2\sum_i \tr[\cE_i(\rho^-)] t
    = \gamma_{\fs}(\rho) \, .
  \end{equation}
\end{proof}
While it is remarkable that the quasiprobability extent has strong enough properties to be an interesting resource measure, one should not forget that these properties even hold for general Hermitian operators and not just states.

For a convex resource theory, the quasiprobability extent of $\omega\in\herm(A)$ is directly related to the robustness $\gamma_{\fs}(\rho)=1+2R_{\fs}(\omega)$ (see~\Cref{lem:robustness1}).
In QRT literature, it is more common to instead consider the robustness as a resource measure.
Most of the properties of the quasiprobability extent directly translate.
The only major difference is that the robustness lacks the sub-additive behavior, but it is convex instead of quasi-convex.
\begin{corollary}{}{robustness_monotone}
  Let $\qrt$ be a convex QRT with tensor product structure that fulfills the axiom of free instruments.
  The robustness $R_{\fs}$ is a resource monotone that is additionally faithful if $\qrt$ is closed.
\end{corollary}
\begin{remark}{}{}
  Our definition of the robustness (sometimes called \emph{total robustness}) should not be confounded with the so-called \emph{global robustness}~\cite{steiner2003_generalized,gour2024_resources} which characterizes the robustness w.r.t. arbitrary (possibly non-free) noise.
  The global robustness happens to be equivalent to the max-relative entropy~\cite{datta2009_max},~\cite[Lemma 10.1]{gour2024_resources}.
\end{remark}
\begin{remark}{}{}
  Note that in literature, one sometimes encounters a resource measure called \emph{log-robustness} which is defined as $\log(1+R_{\fs})$ as opposed to $\log(\gamma_{\fs})=\log(1+2R_{\fs})$\cite{gour2024_resources}.
  Our quantity $\log(\gamma_{\fs})$ can in some sense be argued to be more natural, as its sub-additivity directly mirrors the sub-multiplicativity of $\gamma_{\fs}$.
\end{remark}

We now turn our attention to quantifying the amount of non-free resource in a channel instead of a state.
In complete analogy with~\Cref{prop:log_gamma_monotone}, we summarize as follows.
\begin{proposition}{}{log_gamma_dynamic_monotone}
  For any dynamical QRT, the log-quasiprobability extent $\log(\gamma_{\qrt^{\star}})$ is a dynamical resource measure that is additionally quasi-convex, sub-additive and faithful if $\qrt^{\star}$ is convex and closed.
\end{proposition}
\begin{proof}
  This result is an amalgamation of~\Cref{lem:gamma_convexity,cor:gamma_submult,lem:gamma_faithful2}.
  The only new property to show is weak monotonicity, which directly follows from proving that for any $\Theta\in\fsc(A,B\rightarrow C,D)$ one has $\Theta(\qrt^{\star}(A\rightarrow B)) \subset \qrt^{\star}(C\rightarrow D)$.
  For this purpose, let $\cE=\sum_j\beta_j\cE_j\in\qrt^{\star}(A\rightarrow B)$.
  Denote the associated free measurement channel by $\tilde{\cE}=\sum_j\cE_j\otimes\proj{j}\in\qrt(A\rightarrow BX)$.
  Since $\Theta\otimes\id_{\supo(\mathbb{C}\rightarrow X)}$ is itself also free,~\Cref{lem:dynamical_qrt_golden_rule} implies that
  \begin{equation}
    (\Theta\otimes\id_{\supo(\mathbb{C}\rightarrow X)})(\tilde{\cE}) = \sum_j \Theta(\cE_j)\otimes\proj{j}
  \end{equation}
  is a free channel in $\qrt(C\rightarrow DX)$.
  Therefore, $(\Theta(\cE_j))_j$ is a free instrument in $\fqi(C\rightarrow D)$ and $\Theta(\cE)=\sum_j\beta_j\Theta(\cE_j)$ lies in $\qrt^{\star}(C\rightarrow D)$.
\end{proof}
For a convex QRT, the quasiprobability extent of a channel $\cE\in\cptp(A\rightarrow B)$ is related to the robustness by $\gamma_{\qrt^{\star}}(\cE)=1+2R_{\qrt^{\star}}(\cE)$ (see~\Cref{lem:robustness2}).
As such, $R_{\qrt^{\star}}(\cE)$ is also a dynamical resource measure that inherits the properties of $\log(\gamma_{\qrt^{\star}})$.

\paragraph{Asymptotic simulation cost}
In~\Cref{cor:gamma_submult}, we observed that the quasiprobability extent exhibits sub-multiplicative behavior.
In fact, we will later see that it is oftentimes even \emph{strictly} sub-multiplicative.
This has drastic implications for the sampling overhead of QPS: It can be significantly more efficient to jointly simulate multiple operations in parallel instead of optimally simulating them individually.

When simulating the $n$-fold copy $\cE^{\otimes n}$ of the superoperator $\cE$, we can capture the asymptotic \emph{effective} sampling overhead per instance by regularizing the quasiprobability extent.
\begin{definition}{}{}
  Let $\qrt$ be a QRT with tensor product structure.
  We define the regularized quasiprobability extent of $\cE\in\hp(A\rightarrow B)$ as 
  \begin{equation}
    \gammareg_{\qrt^{\star}}(\cE) \coloneqq \lim\limits_{n\rightarrow\infty} \gamma_{\qrt^{\star}}(\cE^{\otimes n})^{1/n} \, .
  \end{equation}
\end{definition}
Note that this regularization is equivalent to the ``additive'' regularization of the log-quasiprobability extent
  \begin{equation}
    \log\gammareg_{\qrt^{\star}}(\cE) = \lim\limits_{n\rightarrow\infty} \frac{1}{n} \log\gamma_{\qrt^{\star}}(\cE^{\otimes n}) \, .
  \end{equation}

\begin{example}\label{ex:magic_state_submult}
  We again return to the Clifford+T setting from~\Cref{sec:motivational_example}.
  In~\Cref{ex:state_qpd_nonclifford}, we have already shown how using the magic state injection circuit allows an arbitrary Clifford+T circuit $\mathcal{C}$ can be rewritten as
  \begin{equation}
    \mathcal{C}(\proj{0}^{\otimes n}) = \mathcal{C}'(\proj{0}^{\otimes n}\otimes\proj{H}^m) 
  \end{equation}
  for some appropriately chosen Clifford circuit $\mathcal{C}'$.
  Notice that this reduction from gate to state allowed us to assume that all magic states are prepared at the very beginning of the circuit, and more importantly, in parallel.
  The quasiprobability extent of magic states is strictly sub-multiplicative: As previously mentioned $\gamma_{\fs}(\proj{H})=\sqrt{2}=1.414$ whereas the regularized quasiprobability extent is estimated to be $\gammareg_{\fs}(\proj{H})\approx 1.283\pm 0.002$~\cite{heinrich2019_robustness}.
  As a consequence, the simulation cost for a Clifford+T circuit can be decreased from $\mathcal{O}(\gamma_{\fs}(\proj{H})^{2m})$ to $\mathcal{O}(\gammareg_{\fs}(\proj{H})^{2m})$.

  A more detailed discussion, including a precise definition of the decomposition set, will follow in~\Cref{chap:magic}.
\end{example}

Interestingly, $\gammareg$ is arguably not the most meaningful quantity to operationally quantify the asymptotic simulation cost of QPS.
In general, one can also consider \emph{approximate} QPS with an asymptotically vanishing error.
To capture the minimum overhead for the simulation of some superoperator $\cE$ when we allow for some error of magnitude $\epsilon$, we introduce the \emph{smoothed} quasiprobability extent.
\begin{definition}{}{smoothed_gamma}
  Let $\qrt$ be a QRT, $\cE\in\cptp(A\rightarrow B)$ and $\epsilon\geq 0$.
  The smoothed quasiprobability extent of $\cE$ is defined as
  \begin{equation}
    \gamma_{\qrt^{\star}}^{\epsilon}(\cE) \coloneqq \min\limits_{\cF\in B_{\epsilon}(\cE)} \gamma_{\qrt^{\star}}(\cF)
  \end{equation}
  where $B_{\epsilon}(\cE) \coloneqq \{\cF\in\cptp(A\rightarrow B) | \frac{1}{2}\dnorm{\cE-\cF}\leq \epsilon\}$.
\end{definition}

\begin{remark}{}{}
  Note that for technical reasons, we chose to smooth only over CPTP maps instead of general HPTP operations.
  Therefore, we also restricted $\cE$ to be CPTP itself.
  In later sections, this will allow us to directly relate $\gamma_{\qrt^{\star}}^{\epsilon}$ to other smoothed resource measures of QRTs, which are typically also defined by smoothing over physical operations.
  In principle, one could also define the smoothed quasiprobability extent without the CPTP constraint.
\end{remark}
The asymptotic overhead with a vanishing error is thus captured by following quantity.
\begin{definition}{}{}
  Let $\qrt$ be a QRT with tensor product structure.
  We define the \emph{asymptotic quasiprobability extent} of some $\cE\in\hp(A\rightarrow B)$ as 
  \begin{equation}
    \gammasreg_{\qrt^{\star}}(\cE) \coloneqq \lim\limits_{\epsilon\rightarrow 0^+}\liminf\limits_{n\rightarrow\infty} \gamma_{\qrt^{\star}}^{\epsilon}(\cE^{\otimes n})^{1/n} \, .
  \end{equation}
\end{definition}
Notice that we use a $\liminf$ since it is a priori unclear if the sequence $\gamma_{\qrt^{\star}}^{\epsilon}(\cE^{\otimes n})^{1/n}$ converges.
It is however bounded by 
\begin{equation}
  \gammasreg_{\qrt^{\star}}(\cE) \leq \lim\limits_{\epsilon\rightarrow 0^+}\liminf\limits_{n\rightarrow\infty} \gamma_{\qrt^{\star}}(\cE^{\otimes n})^{1/n} = \gammareg_{\qrt^{\star}}(\cE) \, .
\end{equation}
In the study of QPS for various QRTs, it is an important task to understand the relation between $\gamma_{\qrt^{\star}}$, $\gammareg_{\qrt^{\star}}$ and $\gammasreg_{\qrt^{\star}}$.
We note that for many QRTs, computing the quasiprobability extent of an operation is a very difficult task and often only possible in restricted cases.
In extension, computing the regularized quasiprobability extent is typically harder, and for the asymptotic quasiprobability extent it is even harder.

For the resource theory of entanglement in~\Cref{chap:nonlocal}, we will find an explicit expression for $\gammareg_{\qrt^{\star}}$ and $\gammasreg_{\qrt^{\star}}$ for pure entangled states.
Remarkably, we will observe that $\gammasreg_{\qrt^{\star}}$ can be strictly smaller, meaning that erroneous QPS with vanishingly small error is asymptotically more efficient than exact QPS.

\section{Convex programming for finite decomposition sets}\label{sec:qps_conv_programming}
In this section, we will consider decomposition sets that are finite $\ds=\{\cF_1,\dots,\cF_M\}$ (or equivalently convex polytopes, by~\Cref{lem:gamma_hulls}).
Such situations occur naturally in some settings, for example in the resource theory of magic (see~\Cref{chap:magic}) or for probabilistic error cancellation (see~\Cref{chap:noisefree}).
Even when the decomposition set is infinite, it can be useful to pick a finite subset to obtain upper bounds to the quasiprobability extent through~\Cref{lem:gamma_ds_bound}.

For a finite decomposition set $\ds\subset\hp(A\rightarrow B)$ with $M$ elements, the quasiprobability extent of $\cE\in\hp(A\rightarrow B)$ is given by
\begin{equation}
  \gamma_{\ds}(\cE) = \left \lbrace
  \begin{array}{r l}
    \min \limits_{a\in \mathbb{R}^M} &
    \norm{a}_1 \\
    \textnormal{s.t.} &
    \cE = \sum_i a_i\cF_i
  \end{array} \right . \, .
\end{equation}
Through a standard trick, this optimization problem can be brought into the form of a linear program by introducing a slack variable $b\in\mathbb{R}^M$
\begin{equation}\label{eq:gamma_lp}
  \gamma_{\ds}(\cE) = \left \lbrace
  \begin{array}{r l}
    \min \limits_{a,b\in \mathbb{R}^M} &
    \sum_i b_i \\
    \textnormal{s.t.} &
    \cE = \sum_i a_i\cF_i \quad\text{and}\quad \forall i: -a_i \leq b_i, a_i\leq b_i
  \end{array} \right .
\end{equation}
as the optimizer $(a^*,b^*)$ must clearly fulfill $b^*_i=\abs{a^*_i}$.
~\Cref{eq:gamma_lp} can be solved numerically through the usage of readily-available linear program solvers.

The smoothed quasiprobability extent of a channel $\cE\in\cptp(A\rightarrow B)$ w.r.t. a finite decomposition set
\begin{equation}\label{eq:gamma_smoothed}
  \gamma_{\ds}^{\epsilon}(\cE) = \left \lbrace
  \begin{array}{r l}
    \min \limits_{a\in \mathbb{R}^M} &
    \norm{a}_1 \\
    \textnormal{s.t.} &
    \frac{1}{2}\dnorm{\cE - \sum_i a_i\cF_i}\leq \epsilon \\
    & \sum_i a_i\choi{\cF_i} \loewnergeq 0, \\
    & \tr_{B}\left[\sum_i a_i\choi{\cF_i}\right] = \frac{\id_A}{\dimension(A)} \\
  \end{array} \right .
\end{equation}
can also straightforwardly be cast into a convex program, because the diamond norm can be expressed as a SDP~\cite{watrous2009_sdp,watrous2013_simplersdp}.
For this purpose, we use following formulation of the channel distinguishability.
\begin{proposition}{\cite{watrous2009_sdp}}{diamond_sdp}
  Let $\Delta\coloneqq\cG_1-\cG_2$ be the difference of two quantum channels $\cG_1,\cG_2\in\cptp(A\rightarrow B)$.
  Then the diamond norm of $\Delta$ is expressible as the semidefinite program
  \begin{equation}
    \dnorm{\Delta} = \left \lbrace
    \begin{array}{r l}
      \min \limits_{\lambda\in\mathbb{R},X\in\pos(AB)} &
      2\lambda \\
      \textnormal{s.t.} &
      \lambda\id_A \loewnergeq \tr_{B}[X] , \\
      & X\loewnergeq \choi{\Delta}\dimension(A)
    \end{array} \right . \, .
  \end{equation}
\end{proposition}
As such, we obtain the semidefinite program 
\begin{equation}
  \gamma_{\ds}^{\epsilon}(\cE) = \left \lbrace
  \begin{array}{r l}
    \min \limits_{a\in \mathbb{R}^M,\lambda\in\mathbb{R},X\in\pos(AB)} &
    \norm{a}_1 \\
    \textnormal{s.t.} &
    \lambda \leq \epsilon , \\
    & \lambda\id_A \loewnergeq \tr_{B}[X] , \\
    & X\loewnergeq \dimension(A)(\choi{\cE} - \sum_i a_i \choi{\cF_i}) , \\
    & \sum_i a_i\choi{\cF_i} \loewnergeq 0, \\
    & \tr_{B}\left[\sum_i a_i\choi{\cF_i}\right] = \frac{\id_A}{\dimension(A)} \\
  \end{array} \right .
\end{equation}

Another situation that can occur in practice is that we are given a maximal \emph{budget} for the sampling overhead, i.e., a maximum allowed 1-norm $\kappa\geq 0$ for the QPD.
Under this restriction, the approximation error of the best possible QPD is given by
\begin{equation}\label{eq:gamma_budgeted}
  \left \lbrace
  \begin{array}{r l}
    \min \limits_{a\in\mathbb{R}^M} & \norm{\cE - \sum_i a_i\cF_i}_{\diamond} \\
    \textnormal{s.t.} & \sum_i a_i\cF_i\in\cptp(A\rightarrow B), \\
    & \sum\limits_i{|a_i|}\leq \kappa \, .
  \end{array} \right .
\end{equation}
which again can be turned into a semidefinite program using~\Cref{prop:diamond_sdp} and by introducing a slack variable
\begin{equation}
  \begin{array}{r l}
    \min \limits_{a,b\in\mathbb{R}^M,\lambda\in\mathbb{R},X\in\pos(AB)} & 2\lambda \\
    \textnormal{s.t.} 
    & \sum_i b_i \leq \kappa , \\
    & -a_i \leq b_i, a_i\leq b_i , \\
    & \sum_i a_i\choi{\cF_i} \loewnergeq 0 , \\
    & \sum_i a_i \tr_{B}[\choi{\cF_i}] = \frac{\id_A}{\dimension(A)} , \\
    & \lambda\id_A \loewnergeq \tr_{B}[X] , \\
    & X\loewnergeq \dimension(A)(\choi{\cE} - \sum_i a_i \choi{\cF_i})
  \end{array} \, .
\end{equation}

Note that the way how we set up the optimization problems in~\Cref{eq:gamma_smoothed} and~\Cref{eq:gamma_budgeted}, we explicitly required the approximation $\sum_ia_i\cF_i$ of $\cE$ to be a physical $\cptp$ map.
This constraint could also be dropped to allow for potentially cheaper or better approximations, respectively.

\section{Further reading}
QPS was originally introduced by Pashayan, Wallman and Bartlett in the context of classical simulation of non-Clifford quantum circuits~\cite{pashayan2015_estimating}.
However, this original work was not based on general QPDs, but rather on the discrete Wigner representation for odd-dimensional qudit systems (and more general quasiprobability representations w.r.t. some frame and dual frame).
Later work by Howard and Campbell~\cite{howard2017_application} generalized the algorithm to the qubit setting and applied it to the QPS of Clifford+T circuits by simulating magic states with stabilizer states.
General QPS of non-Clifford channels was first treated in \reference~\cite{bennink2017_unbiased}.

Initially, the utility of classical side information was not considered in any of the applications of QPS, be that the simulation of non-Clifford channels~\cite{bennink2017_unbiased}, error mitigation~\cite{temme2017_errormitigation} or nonphysical simulation~\cite{piveteau2022_quasiprobability,jiang2021_physical}.
To our knowledge, the first inclusions of classical side information are due to Endo \etal~\cite{endo2018_practical} in the context of error mitigation, as well as Regula \etal~\cite{regula2021_operational} in the context on nonphysical simulation.
However, both of these works only considered augmenting the decomposition set with completely positive trace-non-increasing maps.
This corresponds to using classical side-information, but with only non-negative weights $\beta_j$ of the individual quantum instruments.
Only later work by Mitarai and Fujii~\cite{mitarai2021_constructing,mitarai2021_overhead} on circuit knitting realized that negative weights were also admissible, though they still didn't properly formulate and characterize the full associated decomposition set.
This was done in later work by Piveteau \etal~\cite{piveteau2024_circuit} and Lowe \etal~\cite{lowe2023_fast}, though both lack the quantum instrument formulation presented here.
The full usage of classical side information, including negative weights and a formalism based on quantum instruments, was later also applied to the setting of nonphysical simulation in \reference~\cite{zhao2023_power}.
The authors consider the decomposition set of \emph{twisted channels} of the form $\sum_i\beta_i\cE_i$ where $\beta_i\in\{+1,-1\}$ and $(\cE_i)_i$ is a general quantum instrument.

For an overview of the history of each individual application of QPS, we refer to the ``further reading'' section of the corresponding chapter.

The (total) robustness was first introduced in the context of entanglement theory by Vidal and Tarrach~\cite{vidal1999_robustness} and later generalized to the context of magic theory~\cite{howard2017_application}, nonphysical operations~\cite{jiang2021_physical,regula2021_operational} and probabilistic error cancellation~\cite{takagi2021_optimal}.
For a comprehensive mathematical treatment of convex resource theories and associated robustness measures, we refer to reader to \reference~\cite{regula2018_convex}.

The linear program for the optimal quasiprobability extent w.r.t. a finite decomposition set is due to \reference~\cite{howard2017_application,temme2017_errormitigation}.
The SDP for the optimal approximation with budgeted 1-norm is due to Piveteau~\etal~\cite{piveteau2022_quasiprobability}.

\section{Contributions}
We briefly summarize the author's contributions in this chapter.
\begin{itemize}
  \item The exposition of QPS in \Cref{sec:motivational_example,sec:qpsim} summarizes a standard technique in literature and is devoid of new contributions besides a few remarks.

  \item The framework in~\Cref{sec:qps_intermediate_measurements}, together with results in subsequent sections, is the first comprehensive treatment of classical side information in general QPS and has not been published before.
  We note that the idea of basing the decomposition set on quantum instruments was proposed in~\reference~\cite{zhao2023_power}.

  \item In~\Cref{sec:qps_basic_properties}, the results on absolutely convex decomposition sets induced by quantum instruments (see~\Cref{lem:decompset_absolutely_convex,lem:one_element_qpd,lem:robustness2,lem:gamma_faithful}) are novel and previously unpublished.
  Similarly, the Lipschitz continuity result (\Cref{lem:gamma_lipschitz}) is new to our knowledge (though it might possibly be related to results in \reference~\cite{schluck2023_continuity}).
  We also highlight that the basis in~\Cref{ex:modifiedendobasis} was found by Thomas Dubach in his master thesis~\cite{dubach2023_masterthesis} which was supervised by the author.

  The remaining results in the section are mostly covered explicitly or implicitly by other authors in previous literature.

  \item Almost all results in~\Cref{sec:gamma_qrt} are new and previously unpublished.
  Though it should be noted that many of them are adaptations of previously know properties of the robustness to the decomposition set $\qrt^{\star}$ (as opposed to $\qrt$).
  We note the results in~\Cref{prop:log_gamma_monotone,cor:robustness_monotone} are well-established in QRT literature.
  We also note that the regularized quasiprobability extent has been previously studied~\cite{heinrich2019_robustness,piveteau2024_circuit}, but the asymptotic quasiprobability extent has to our knowledge never been considered before.

  \item The linear program in~\Cref{sec:qps_conv_programming} is due to \reference~\cite{howard2017_application,temme2017_errormitigation}.
  The semidefinite program for a budgeted QPD was introduced by the author in \reference~\cite{piveteau2022_quasiprobability}.
  The semidefinite program for the smooth quasiprobability extent has previously not been published.
\end{itemize}

\chapter{Simulating non-physical computation}\label{chap:nonphysical}
The set of completely positive trace-preserving (CPTP) superoperators describes all physically admissible ways in which a quantum system can evolve in time.
Nonetheless, in various fields of quantum information, the necessity to study non-CPTP maps occasionally arises.
For instance, positive non-CP maps play an important role in entanglement detection~\cite{peres1996_separability}.
Other such examples include the study of non-Markovian dynamics of open quantum systems~\cite{pechukas1994_reduced,shaji2005_who} and quantum error mitigation~\cite{temme2017_errormitigation,jiang2021_physical}.

This chapter is dedicated to the study of the QPS overhead for simulating non-CPTP maps given one has access to a quantum computer that is capable of running any CPTP operation.
From the point of view of resource theories (as discussed in~\Cref{sec:gamma_qrt}), we will essentially consider the most powerful QRT where the free channels from $A$ to $B$ are precisely all CPTP maps $\qrt(A\rightarrow B)\coloneqq \cptp(A\rightarrow B)$ and the set of free states on $A$ is the set of all density operators $\dops(A)$.
Clearly, the quasiprobability extent of any physical quantum channel or quantum state is $1$, so our focus purely lies on the study of the quasiprobability extent of non-physical (super-)operators.
More concretely, we will consider the quasiprobability extent $\gamma_{\cptp^{\star}}$ w.r.t. the decomposition set $\cptp^{\star}(A\rightarrow B)\coloneqq \qids[\qi(A\rightarrow B)]$ which captures the utilization of classical side information through intermediate measurements (recall the discussion in~\Cref{sec:qps_intermediate_measurements} and~\Cref{eq:qrt_extended_ds}) and how it compares to the quasiprobability extent $\gamma_{\cptp}$ where intermediate measurements are not considered.

The utility of studying this setting is twofold.
On one hand, this ``maximal resource theory'' can be a useful technical tool to bound the behavior of QPS in other resource theories.
For example, for any QRT $\qrt$, $\gamma_{\qrt^{\star}}$ is lower bounded by $\gamma_{\cptp^{\star}}$ due to~\Cref{lem:gamma_ds_bound}.
On the other hand, there are direct practical applications involving the simulation of quantum circuits containing non-physical operations.
We will touch upon a few of these applications towards the end of the chapter.

\section{Utility of classical side information}
In this section, we aim to study the difference between $\gamma_{\cptp^{\star}}$ and $\gamma_{\cptp}$ in order to evaluate how useful it is to consider classical side information in this setting.
We start by mathematically characterizing the decomposition sets.
Clearly, by definition the set of all quantum instruments $\qi(A\rightarrow B)$ is coarse-grainable, trivially fine-grainable and closed under mixture.
So by~\Cref{lem:characterization_expanded_decomposition_set} we can characterize
\begin{equation}\label{eq:char_cptpstar}
  \cptp^{\star}(A\rightarrow B) = \{ \cE^+ - \cE^- | \cE^{\pm}\in\cp(A\rightarrow B), \cE^++\cE^-\in\cptp(A\rightarrow B) \} \, .
\end{equation}
In fact, the trace-preserving condition can be relaxed to trace-non-increasing
\begin{equation}\label{eq:char_cptpstar2}
  \cptp^{\star}(A\rightarrow B) = \{ \cE^+ - \cE^- | \cE^{\pm}\in\cp(A\rightarrow B), \cE^++\cE^-\in\cptni(A\rightarrow B) \}
\end{equation}
since any $\cE^++\cE^-\in\cptni(A\rightarrow B)$ can be completed to a trace-preserving map $(\cE^++\cE^-)+\cG\in\cptp(A\rightarrow B)$ and hence $\cE^+-\cE^- = (\cE^++\frac{1}{2}\cG) - (\cE^-+\frac{1}{2}\cG)$.

We now turn to the question of which superoperators $\cE\in\hp(A\rightarrow B)$ possess a QPD into the sets $\cptp(A\rightarrow B)$ and $\cptp^{\star}(A\rightarrow B)$.
\begin{lemma}{}{cptp_span}
  For any systems $A$ and $B$,
  \begin{equation}\label{eq:cptp_span_1}
    \spn_{\mathbb{R}}\left(\cptp(A\rightarrow B)\right) = \mathbb{R}\hptp(A\rightarrow B) \coloneqq \{r\cdot \cE | r\in\mathbb{R}, \cE\in\hptp(A\rightarrow B)\}
  \end{equation}
  and
  \begin{equation}\label{eq:cptp_span_2}
    \spn_{\mathbb{R}}\left(\cptp^{\star}(A\rightarrow B)\right) = \hp(A\rightarrow B) \, .
  \end{equation}
\end{lemma}
\begin{proof}
  We start by proving~\Cref{eq:cptp_span_1}.
  The inclusion $\subset$ follows directly from the fact that any linear combination of trace-preserving maps is itself proportional to a trace-preserving map.
  Conversely, consider some $r\in\mathbb{R}$ and $\cE\in \hptp(A\rightarrow B)$.
  We will now find a decomposition of $r\cE$ in terms of CPTP maps.
  We can assume without loss of generality that $r=1$, otherwise the decomposition just needs to be scaled accordingly.
  Let's denote by $d_A\coloneqq\dimension(A)$, $d_B\coloneqq\dimension(B)$ and by $\smash{\cD(\rho)\coloneqq\frac{\tr[\rho]}{d_B}\id_B}$ the fully depolarizing channel on $B$.
  Furthermore, define $\lambda\geq \sigma_{\mathrm{max}}(\choi{\cE})$ to be at least the maximum eigenvalue of $\choi{\cE}$.
  We can decompose
  \begin{equation}
    \cE = d_Ad_B\lambda \cdot \cD -(d_Ad_B\lambda -1)\cdot\left( \frac{d_Ad_B\lambda}{d_Ad_B\lambda-1}\cD - \frac{1}{d_Ad_B\lambda-1}\cE \right) \, .
  \end{equation}
  It remains to show that $\left( \frac{d_Ad_B\lambda}{d_Ad_B\lambda-1}\cD - \frac{1}{d_Ad_B\lambda-1}\cE \right)$ is CPTP.
  This is straightforward to check as
  \begin{equation}
    \tr_B\left[ \frac{d_Ad_B\lambda}{d_Ad_B\lambda-1}\choi{\cD} - \frac{1}{d_Ad_B\lambda-1}\choi{\cE} \right] = \left(\frac{d_Ad_B\lambda}{d_Ad_B\lambda-1}-\frac{1}{d_Ad_B\lambda-1}\right)\frac{1}{d_A}\id_A = \frac{1}{d_A}\id_A
  \end{equation}
  and
  \begin{align}
    & \left( \frac{d_Ad_B\lambda}{d_Ad_B\lambda-1}\choi{\cD} - \frac{1}{d_Ad_B\lambda-1}\choi{\cE} \right) \loewnergeq 0 \\
    \Longleftrightarrow & d_Ad_B\lambda \frac{1}{d_Ad_B}\id_{AB} \loewnergeq \choi{\cE} \\
    \Longleftrightarrow & \lambda \geq \sigma_{\mathrm{max}}(\choi{\cE}) \, .
  \end{align}

  We now turn to proving~\Cref{eq:cptp_span_2}.
  The direction $\subset$ again directly follows from the fact that any linear combination of Hermitian-preserving maps is again Hermitian-preserving.
  For the converse direction, consider a superoperator $\cE\in\hp(A\rightarrow B)$.
  We can split up its Choi representation into the difference of two positive operators $\Lambda^{\pm}\loewnergeq 0$
  \begin{equation}
    \choi{\cE} = \Lambda^+ - \Lambda^- = \lambda^+ \frac{\Lambda^+}{\lambda^+} - \lambda^- \frac{\Lambda^-}{\lambda^-}
  \end{equation}
  where we introduced the maximal eigenvalues $\lambda^{\pm}$ of $\dimension(A)\tr_B[\Lambda^{\pm}]$.
  We now observe that $\Lambda^{\pm} / \lambda^{\pm}$ are Choi operators representing completely positive trace-non-increasing maps, as they are positive and $\smash{\tr_B[\frac{\Lambda^{\pm}}{\lambda^{\pm}}] \loewnerleq \frac{1}{\dimension(A)}\id_A}$.
  The proof follows from $\cptni(A\rightarrow B)\subset\cptp^{\star}(A\rightarrow B)$ which is a consequence of~\Cref{eq:char_cptpstar2}.
  Note that we implicitly assumed $\lambda^{\pm}>0$, but the proof can be straightforwardly adapted for the case where $\choi{\cE}$ is positive or negative.
\end{proof}

\Cref{lem:cptp_span} implies that classical side information is crucial in order to simulate Hermitian-preserving maps that do not lie in $\mathbb{R}\hptp(A\rightarrow B)$.
It is natural to ask, whether using intermediate measurements improves the quasiprobability extent for maps that do lie in $\mathbb{R}\hptp(A\rightarrow B)$.
The answer to this question is negative.
\begin{proposition}{}{utility_cptp_interm_mmts}
  For all $\cE\in\mathbb{R}\hptp(A\rightarrow B)$ one has $\gamma_{\cptp}(\cE)=\gamma_{\cptp^{\star}}(\cE)$.
\end{proposition}
\begin{proof}
  The proof consists of showing that any QPD of $\cE$ with respect to $\cptp^{\star}$ can be reformulated as a QPD with respect to $\cptp$ without increasing the 1-norm of the coefficients.
  By~\Cref{lem:decompset_absolutely_convex,lem:one_element_qpd} we can without loss of generality assume that the QPD is of the form $\cE = \kappa\cF$ for some $\cF\in\cptp^{\star}(A\rightarrow B)$.
  By~\Cref{eq:char_cptpstar}, we can write $\cF = \cF^+-\cF^-$ for $\cF^{\pm}\in\cp(A\rightarrow B)$ such that $\cF^++\cF^-\in\cptp(A\rightarrow B)$.
  If either $\cF^+=0$ or $\cF^-=0$, then we already have a QPD into $\cptp(A\rightarrow B)$.
  Otherwise, observe that both $\cF^++\cF^-$ and $\cF^+-\cF^-$ are in $\mathbb{R}\hptp(A\rightarrow B)$.
  Clearly, this is only possible if both $\cF^+$ and $\cF^-$ are themselves in $\mathbb{R}\hptp(A\rightarrow B)$.
  The following is therefore a valid QPD of $\cE$ w.r.t. $\cptp(A\rightarrow B)$
  \begin{equation}
    \cE = \kappa\tr[\choi{\cF^+}] \frac{\cF^+}{\tr[\choi{\cF^+}]} - \kappa\tr[\choi{\cF^-}] \frac{\cF^-}{\tr[\choi{\cF^-}]}
  \end{equation}
  and the 1-norm of its coefficients is given by $\kappa(\tr[\choi{\cF^+}] + \tr[\choi{\cF^-}]) = \kappa$.
\end{proof}

\section{Characterization of quasiprobability extent}
In this section, we will derive a simple characterization of $\gamma_{\cptp^{\star}}(\cE)$ for Hermitian-preserving operations $\cE$.
Recall that since $\cptp$ and $\cptp^{\star}$ are convex and absolutely convex (see~\Cref{lem:decompset_absolutely_convex}), we can restrict ourselves to QPDs with one and two elements respectively by~\Cref{lem:two_element_qpd,lem:one_element_qpd}.
Together with the characterization of $\cptp^{\star}$ in~\Cref{eq:char_cptpstar}, we can write the quasiprobability extent of some $\cE\in\mathbb{R}\hptp(A\rightarrow B)$ as
\begin{align}
  \gamma_{\cptp}(\cE) = \left \lbrace
  \begin{array}{r l}
    \min\limits_{a^{\pm}\geq 0, \Lambda^{\pm}\in\pos(AB)} & a^+ + a^- \\
    \textnormal{s.t.} & \choi{\cE} = a^+\Lambda^+ - a^-\Lambda^- \\
    & \tr_B\Lambda^{\pm} = \frac{1}{\dimension(A)}\id_A
  \end{array} \right . 
\end{align}
and similarly for $\cE\in\hp(A\rightarrow B)$ we have
\begin{align}
  \gamma_{\cptp^{\star}}(\cE) = \left \lbrace
  \begin{array}{r l}
    \min\limits_{\kappa\geq 0, \Lambda^{\pm}\in\pos(AB)} & \kappa \\
    \textnormal{s.t.} & \choi{\cE} = \kappa(\Lambda^+ - \Lambda^-) \\
    & \tr_B[\Lambda^+ + \Lambda^-] = \frac{1}{\dimension(A)}\id_A
  \end{array} \right . \, .
\end{align}
Note that we used a minimum instead of an infimum since both $\cptp$ and $\cptp^{\star}$ are closed (recall~\Cref{lem:compact_ds}).
Using a simple variable substitution $a^{\pm}\Lambda^{\pm}\rightarrow \Lambda^{\pm}$, respectively $\kappa\Lambda^{\pm}\rightarrow \Lambda^{\pm}$, the constraints in the above two optimization problems can be recast into the form of linear matrix equalities 
\begin{align}\label{eq:gamma_cptp_sdp}
  \gamma_{\cptp}(\cE) = \left \lbrace
  \begin{array}{r l}
    \min\limits_{a^{\pm}\geq 0, \Lambda^{\pm}\in\pos(AB)} & a^+ + a^- \\
    \textnormal{s.t.} & \choi{\cE} = \Lambda^+ - \Lambda^- \\
    & \tr_B\Lambda^{\pm} = \frac{a^{\pm}}{\dimension(A)}\id_A
  \end{array} \right .
\end{align}
\begin{align}\label{eq:gamma_cptpstar_sdp}
  \gamma_{\cptp^{\star}}(\cE) = \left \lbrace
  \begin{array}{r l}
    \min\limits_{\kappa\geq 0, \Lambda^{\pm}\in\pos(AB)} & \kappa \\
    \textnormal{s.t.} & \choi{\cE} = \Lambda^+ - \Lambda^- \\
    & \tr_B[\Lambda^+ + \Lambda^-] = \frac{\kappa}{\dimension(A)}\id_A
  \end{array} \right . \, .
\end{align}
The resulting optimization problems are now semidefinite programs, which will prove to be very useful.
On one hand, it allows for the efficient numerical evaluation of the quasiprobability extent, at least for small Hilbert spaces.
On the other hand, it allows us to prove various properties of the optimal solution analytically.
We will usually not have the same luxury for other resource theories.

Is will be useful to determine the dual forms to the above SDPs.
\begin{lemma}{}{gamma_cptp_sdp_dual}
  The following SDPs are valid dual SDPs to~\Cref{eq:gamma_cptp_sdp,eq:gamma_cptpstar_sdp}:
  \begin{equation}
  \gamma_{\cptp}(\cE) = \left \lbrace
  \begin{array}{r l}
    \max\limits_{X\in\herm(AB),Y^{\pm}\in\herm(A)} & \tr[\choi{\cE}X] \\
    \textnormal{s.t.} & -Y^-\otimes\id_B\loewnerleq X\loewnerleq Y^+\otimes\id_B \\
	& \tr[Y^{\pm}] = \dimension(A)
  \end{array} \right.
  \end{equation}
  \begin{equation}
  \gamma_{\cptp^\star}(\cE) = \left \lbrace
  \begin{array}{r l}
    \max\limits_{X\in\herm(AB),Y\in\herm(A)} & \tr[\choi{\cE}X] \\
    \textnormal{s.t.} & -Y\otimes\id_B\loewnerleq X\loewnerleq Y\otimes\id_B \\
	& \tr[Y]= \dimension(A)
  \end{array} \right. \, .
  \end{equation}
\end{lemma}
\begin{proof}
  In the following, we will denote the unnormalized Hilbert-Schmidt inner product by $\langle A,B\rangle\coloneqq\tr[A^{\dagger}B]$.
  We start with the SDP in~\Cref{eq:gamma_cptp_sdp}.
  The Lagrangian in this case is given by
  \begin{align}
    L(a^{\pm},\Lambda^{\pm};X,Y^{\pm})
    &= (a^++a^-) - \langle X, \choi{\cE}-\Lambda^++\Lambda^-\rangle \nonumber\\
    &- \langle Y^+,\tr_B[\Lambda^+]-\frac{a^+}{\dimension(A)}\id_A\rangle
    - \langle Y^-,\tr_B[\Lambda^-]-\frac{a^-}{\dimension(A)}\id_A\rangle \, .
  \end{align}
  Using the relation $\langle R_{A},\tr_B[S_{AB}]\rangle = \langle R_{A}\otimes\mathds{1}_B,S_{AB}\rangle$, we can rewrite this as
  \begin{align}
    L(a^{\pm},\Lambda^{\pm};X,Y^{\pm})
    &= -\langle X,\choi{\cE}\rangle
    + a^+(1 + \frac{1}{\dimension(A)}\tr[Y^+])
    + a^-(1 + \frac{1}{\dimension(A)}\tr[Y^-]) \nonumber\\
    &+ \langle \Lambda^+, X-Y^+\otimes\id_B \rangle
    + \langle \Lambda^-, -X-Y^-\otimes\id_B \rangle \, .
  \end{align}
  The dual function $g(X,Y^{\pm})\coloneqq\inf_{a^{\pm},\Lambda^{\pm}}L(a^{\pm},\Lambda^{\pm};X,Y^{\pm})$ can clearly only be greater than $-\infty$ if $\pm X\loewnergeq Y^{\pm}\otimes\id_B$ and $\tr[Y^{\pm}]\geq -\dimension(A)$.
  We thus get the dual SDP
  \begin{equation}
  \gamma_{\cptp}(\cE) = \left \lbrace
  \begin{array}{r l}
    \max\limits_{X\in\herm(AB),Y^{\pm}\in\herm(A)} & -\langle X,\choi{\cE}\rangle \\
    \textnormal{s.t.} & Y^+\otimes\id_B\loewnerleq X\loewnerleq -Y^-\otimes\id_B \\
	& \tr[Y^{\pm}]\geq -\dimension(A)
  \end{array} \right.
  \end{equation}
  The desired form of this SDP can be obtained by substituting $X,Y^+,Y^-$ by $-X,-Y^+,-Y^-$ and by realizing that the inequality in the trace can without loss of generality be assumed to be achieved.

  We now turn to the second SDP in~\Cref{eq:gamma_cptp_sdp}.
  The Lagrangian here is given by
  \begin{align}
    L(\kappa,\Lambda^{\pm};X,Y)
    &= \kappa - \langle X, \choi{\cE}-\Lambda^++\Lambda^-\rangle \nonumber\\
    &\quad- \langle Y,\tr_B[\Lambda^++\Lambda^-]-\kappa\frac{1}{\dimension(A)}\id_A\rangle \\
    &= - \langle X,\choi{\cE}\rangle
    + \kappa(1 + \frac{1}{\dimension(A)}\tr[Y]) \nonumber\\
    &\quad+ \langle \Lambda^+, X - Y\otimes\id_B \rangle 
    + \langle \Lambda^-, -X - Y\otimes\id_B \rangle \, .
  \end{align}
  By the same argument as before, we get a dual SDP
  \begin{equation}
  \gamma_{\cptp^\star}(\cE) = \left \lbrace
  \begin{array}{r l}
    \max\limits_{X\in\herm(AB),Y\in\herm(A)} & -\langle X,\choi{\cE}\rangle \\
    \textnormal{s.t.} & -Y\otimes\id_B\loewnerleq X\loewnerleq Y\otimes\id_B \\
	& \tr[Y]\geq -\dimension(A)
  \end{array} \right.
  \end{equation}
  Again, with a suitable substitution and realizing the inequality $\tr[Y]\leq\dimension(A)$ is achieved, we get the desired result.
\end{proof}

In fact, it turns out that the above semidefinite program for $\gamma_{\cptp^{\star}}$ is equivalent to the well-known diamond norm.
\begin{theorem}{\cite[Theorem 3]{zhao2023_power}}{gamma_diamond_norm}
  For any $\cE\in\hp(A\rightarrow B)$ one has $\gamma_{\cptp^{\star}}(\cE)=\dnorm{\cE}$.
\end{theorem}
In turn, this also means that $\gamma_{\cptp}(\cE)=\dnorm{\cE}$ for any $\cE\in\mathbb{R}\hptp(A\rightarrow B)$ by~\Cref{prop:utility_cptp_interm_mmts}.
The standard operational interpretation of the diamond norm is typically stated in terms of the single-shot distinguishability of two quantum channels.
Our result can be considered an alternative and new operational interpretation: The diamond norm measures the non-physicalness of a certain operation by quantifying the cost of simulating it using physical operations.

When specialized to states, \Cref{thm:gamma_diamond_norm} implies 
\begin{equation}
  \forall \omega\in\herm(A) : \gamma_{\dops^{\star}(A)}(\omega) = \norm{\omega}_1 \, .
\end{equation}
This is perhaps not too surprising: By~\Cref{lem:negativity_state_qpd} we know that $\gamma_{\dops^{\star}}=\gamma_{\dops}$ and clearly the optimal QPD w.r.t. $\dops(A)$ of $\omega$ is given by separating its positive and negative eigenspace projectors.

Remarkably,~\Cref{thm:gamma_diamond_norm} thus implies that the quasiprobability extent of a superoperator $\gamma_{\cptp^{\star}}(\cE)$ is equal to the largest extent $\gamma_{\dops}(\cE(\rho))$ optimized over all Hermitian operators $\rho$ (on a possibly extended Hilbert space) with quasiprobability extent $\gamma_{\dops}(\rho)\leq 1$.

\begin{proof}
  This proof essentially follows~\cite[Lemma 1]{regula2021_operational} which could also be derived from~\cite[Lemma 4, Theorem 2]{jencova2014_basenorms}.
  We start with the dual formulation of $\gamma_{\cptp^{\star}}$ from~\Cref{lem:gamma_cptp_sdp_dual}.
  In the following, we will denote the unnormalized Hilbert-Schmidt inner product by $\langle A,B\rangle\coloneqq\tr[A^{\dagger}B]$.
  \begin{align}
    \gamma_{\cptp^{\star}}(\cE)
    &= \max \{ \langle \choi{\cE},X \rangle | -Y\otimes\id_B\loewnerleq X \loewnerleq Y\otimes\id_B, \tr[Y]=\dimension(A) \} \\
    &= \max \{ \langle \choi{\cE},X \rangle | -Y\otimes\id_B\loewnerleq \frac{X}{\dimension(A)} \loewnerleq Y\otimes\id_B, Y\in\dops(A) \} \\
    &= \sup \{ \langle \choi{\cE},X \rangle | -Y\otimes\id_B\loewnerleq \frac{X}{\dimension(A)} \loewnerleq Y\otimes\id_B, Y\in\dops_{>0}(A) \} \\
    &= \sup \{ \langle \sqrt{Y\otimes\id_B}\choi{\cE}\sqrt{Y\otimes\id_B},X \rangle | \id_{AB}\loewnerleq \frac{X}{\dimension(A)} \loewnerleq \id_{AB}, Y\in\dops_{>0}(A) \}
  \end{align}
  where $\dops_{>0}(A)$ denotes the set of full-rank density operators and the last equality used the variable substitution $X\mapsto \sqrt{Y^{-1}\otimes\id_B}X\sqrt{Y^{-1}\otimes\id_B}$.
  The trace norm of a Hermitian matrix $M$ can be expressed as
  \begin{equation}
    \norm{M}_1 = \max\limits_{\opnorm{N}\leq 1} \tr[MN] \, .
  \end{equation}
  Thus, we get
  \begin{align}
    \gamma_{\cptp^{\star}}(\cE)
    &= \sup\limits_{Y\in\dops_{>0}(A)} \dimension(A) \norm{\sqrt{Y\otimes\id_B}\choi{\cE}\sqrt{Y\otimes\id_B}}_1 \\
    &= \max\limits_{Y\in\dops(A)} \dimension(A) \norm{\sqrt{Y\otimes\id_B}\choi{\cE}\sqrt{Y\otimes\id_B}}_1 \\
    &= \max\limits_{Y\in\dops(A)} \norm{(\idchan_A\otimes\cE)\left[ \sqrt{Y}\otimes\id_{A'}\proj{\Psi}\sqrt{Y}\otimes\id_{A'} \right] }_1 \\
    &= \max\limits_{Y\in\dops(A)} \norm{(\idchan_A\otimes\cE)\left[ \id_A\otimes\sqrt{Y}\proj{\Psi}\id_A\otimes\sqrt{Y} \right] }_1 \\
    &= \max\limits_{\phi} \norm{(\idchan_A\otimes\cE)(\phi) }_1 \\
    &= \max\limits_{\rho\in\dops(AA')} \norm{(\idchan_A\otimes\cE)(\rho) }_1 \\
    &= \dnorm{\cE}
  \end{align}
  where we used $(\id_A\otimes Z_{A'})\ket{\Psi} = (Z_A^T\otimes\id_{A'})\ket{\Psi}$ for the unnormalized maximally entangled state $\ket{\Psi}_{AA'}=\sum_i\ket{i}_{A}\otimes\ket{i}_{A'}$ and the fact that any pure state can be seen as the canonical purification $(\id\otimes\sqrt{\rho})\ket{\Psi}$ of some mixed state $\rho$.
\end{proof}
We summarize the results from this section in~\Cref{fig:nonphysical_set_diagram}. 
Note that the diamond norm is multiplicative under the tensor product $\dnorm{\cE\otimes\cF}=\dnorm{\cE}\dnorm{\cF}$, see for example \reference~\cite[Theorem 3.49]{watrous2018_theory}.
As such, the regularized quasiprobability extent remains unchanged $\gammareg_{\cptp^{\star}}=\gamma_{\cptp^{\star}}$.

\begin{figure}
  \centering
  \includegraphics{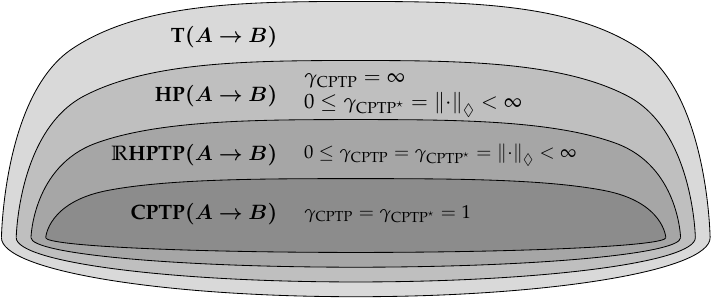}
  \caption{Set diagram representing the relations between the sets of linear superoperators, Hermitian-preserving maps, maps proportional to a trace-preserving map and quantum channels. For each shaded region, the possible ranges of values taken by the quasiprobability extents $\gamma_{\cptp}$ and $\gamma_{\cptp^\star}$ are depicted.}
  \label{fig:nonphysical_set_diagram}
\end{figure}

\section{Applications}\label{sec:nonphysical_applications}
In this section, we investigate a few practical applications of QPS of non-physical operations.
Before we get to the applications themselves, we briefly introduce a technical result on the quasiprobability extent of superoperators with diagonal Pauli transfer matrices.
In the following, we will denote by $\{Q_1,\dots,Q_{4^n}\}$ the Pauli strings $\{\id,X,Y,Z\}^{\otimes n}$ endowed with some arbitrary order.
Furthermore, $\symplip{Q_{\alpha}}{Q_{\beta}}$ denotes the symplectic inner product which is $=0$ if $Q_{\alpha}$ and $Q_{\beta}$ commute and $=1$ if they anti-commute.
\begin{definition}{}{}
  The $n$-qubit symplectic Walsh-Hadamard transformation is the linear map $\walshhad_n:\mathbb{C}^{4^n}\rightarrow\mathbb{C}^{4^n}$ defined by
  \begin{equation}
    \left(\walshhad_n(x)\right)_i \coloneqq \sum_{j=1}^{4^n} (-1)^{\symplip{Q_i}{Q_j}} x_j \, .
  \end{equation}
\end{definition}
By definition, the Walsh-Hadamard transform is a linear and symmetric operator.
Furthermore, it is invertible, and its inverse is proportional to itself $\walshhad_n^2 = 4^n\id$ which can be readily verified from its definition
\begin{align}
  \left( \walshhad_n^2(x) \right)_i 
  &= \sum_{j,k} (-1)^{\symplip{Q_i}{Q_j}}(-1)^{\symplip{Q_j}{Q_k}} x_k \\
  &= \sum_{j,k} (-1)^{\symplip{Q_iQ_k}{Q_j}} x_k \\
  &= \sum_{k} 4^n\delta_{i,k} x_k \\
  &= 4^nx_i
\end{align}
where we used that any Pauli string (except the identity) precisely commutes with half of all Pauli strings.
We refer readers interested in an exposition with more mathematical details to \reference~\cite{flammia2020_efficient}.
\begin{proposition}{}{gamma_cptp_paulisupo}
  Let $\cE\in\hp(A)$ be an $n$-qubit superoperator with diagonal Pauli transfer matrix representation $M=\diag(m_1,\dots,m_{4^n})$ for some $m\in\mathbb{R}^{4^n}$.
  Then
  \begin{equation}
    \gamma_{\cptp^{\star}}(\cE) = \gamma_{\mathcal{P}_n}(\cE) = \frac{1}{4^n}\norm{\walshhad_n(m)}_1
  \end{equation}
  where $\mathcal{P}_n$ denotes the $n$-qubit Pauli channels.
\end{proposition}
\begin{proof}
  Consider an $n$-qubit superoperator of the form
  \begin{equation}\label{eq:diagonal_chi_matrix}
    \cF(\rho) \coloneqq \sum_i r_i Q_i\rho Q_i
  \end{equation}
  for some $r\in\mathbb{R}^{4^n}$.
  One can show that $\dnorm{\cF} = \norm{r}_1$ by using the triangle inequality
  \begin{equation}
    \dnorm{\cF} \leq \sum_i \abs{r}_i \dnorm{\indsupo{Q_i}} = \norm{r}_1
  \end{equation}
  and conversely by using the definition of the diamond norm
  \begin{align}
    \dnorm{\cF}
    &\geq \norm{(\cF\otimes\idchan)(\proj{\Psi})}_1 \\
    &= \norm{\sum_i r_i (Q_i\otimes\id)\proj{\Psi}(Q_i\otimes\id)}_1 \\
    &= \sum_i \abs{r_i} \norm{(Q_i\otimes\id)\proj{\Psi}(Q_i\otimes\id)}_1 \\
    &= \norm{r}_1
  \end{align}
  where $\ket{\Psi}$ is the maximally-entangled $2n$-qubit state and we used the fact that the states $\{(Q_i\otimes\id)\ket{\Psi}\}_i$ form an orthonormal basis.
  So our desired statement follows from~\Cref{thm:gamma_diamond_norm} if we can show that $\cE$ must be of the form of~\Cref{eq:diagonal_chi_matrix} with $r=\walshhad_n^{-1}(m)=\walshhad_n(m)/4^n$.

  To this end, we note that there is a one-to-one correspondence between superoperators of the form in~\Cref{eq:diagonal_chi_matrix} and superoperators with diagonal Pauli transfer matrix.
  On one hand, consider $\cF$ as in~\Cref{eq:diagonal_chi_matrix}.
  The associated Pauli transfer matrix reads
  \begin{align}
    M^{(\cF)}_{\alpha,\beta}
    & = \frac{1}{2^n} \tr\left[ Q_{\alpha}\cF(Q_{\beta}) \right] \\
    & = \frac{1}{2^n}\sum_i r_i \tr\left[ Q_{\alpha}Q_i Q_{\beta} Q_i\right] \\
    & = \frac{1}{2^n}\sum_i r_i (-1)^{\symplip{Q_i}{Q_{\beta}} }\tr\left[ Q_{\alpha}Q_{\beta}\right] \\
    & = \delta_{\alpha,\beta} \sum_i r_i (-1)^{\symplip{Q_i}{Q_{\beta}} }
  \end{align}
  and the diagonal elements are precisely given by $\walshhad_n(r)$.
  Conversely, thanks to the invertibility of the Walsh-Hadamard transformation, for every superoperator with diagonal Pauli transfer matrix, one can find a vector $r$ such that it can be written in the form of~\Cref{eq:diagonal_chi_matrix}.
\end{proof}
We now present a few applications of non-physical QPS.

\paragraph{Noise inversion}
Consider following situation: In some part of a quantum computation, a subset of the qubits, denoted by system $A$, undergoes some noise described by a quantum channel $\cN\in\cptp(A)$.
For example, at this point of the circuit the qubits could be sent from one computational device to another through a faulty transmission channel.
To undo the effect of the noise, one would ideally perform the operation $\cN^{-1}$ to reverse the effect $\cN^{-1}\circ\cN(\rho)=\rho$ for any arbitrary input state $\rho$\footnote{Note that here we make an implicit assumption that a (possibly non-$\cptp$) inverse of the channel $\cN$ exists in the first place. This is not necessarily the case, as exemplified by choosing $\cN$ to be the fully depolarizing channel or by choosing differing input and output spaces. For non-invertible channels one could utilize a pseudo-inverse like the Moore-Penrose inverse or the Drazin inverse. See~\cite{cao2022_nisq} for more discussion about pseudoinverses in noise inversion.}.
Unfortunately, $\cN^{-1}$ is not in $\cptp(A)$ unless $\cN$ is unitary, making it impossible to fully recover the quantum information in physically realistic scenarios.
Using QPS, we can however simulate the channel $\cN^{-1}$ with a sampling overhead of $\gamma_{\cptp^{\star}}(\cN^{-1})$.

First, we note that the inverse of a $\cptp$ map is itself trace-preserving.
\begin{lemma}{}{inv_is_tp}
  Let $\cN\in\cptp(A)$ be an invertible quantum channel.
  Then $\cN^{-1}$ is Hermitian-preserving and trace-preserving.
\end{lemma}
As such, the quasiprobability extent of $\cN^{-1}$ is always well-defined and
\begin{equation}
  \gamma_{\cptp^{\star}}(\cN^{-1}) = \gamma_{\cptp}(\cN^{-1}) = \dnorm{\cN^{-1}} \, .
\end{equation}
\begin{proof}
  The only nontrivial step is to show that a Hermitian-preserving superoperator cannot map a non-Hermitian map onto a Hermitian map.
  An arbitrary operator $O\in\lino(A)$ can always be uniquely decomposed into a Hermitian and skew-Hermitian part $O=X+iY$ where $X,Y\in\herm(A)$.
  By linearity, a Hermitian-preserving superoperator $\cE\in\hp(A)$ will preserve this decomposition
  \begin{equation}
    O \mapsto \cE(X) + i\cE(Y)
  \end{equation}
  as $\cE(X)$ and $i\cE(Y)$ are still Hermitian and skew-Hermitian respectively.
  Due to the invertibility of $\cE$ we have $\cE(Y)=0\Rightarrow Y=0$.
\end{proof}

\begin{example}\label{ex:noise_inv_depol}
  The one-qubit depolarizing channel
  \begin{equation}
    \cN(\rho) = (1-p)\rho + p \tr[\rho]\frac{\id}{2}
  \end{equation}
  has the Pauli transfer matrix representation\footnote{The Pauli order is chosen to be $\id,X,Y,Z$.} $\diag(1, 1-p, 1-p, 1-p)$.
  By using the inverse Walsh-Hadamard transformation we obtain
  \begin{equation}
    \cN^{-1} = \frac{4-p}{4(1-p)}\idchan - \frac{p}{4(1-p)}\indsupo{X} - \frac{p}{4(1-p)}\indsupo{Y} - \frac{p}{4(1-p)}\indsupo{Z} \, .
  \end{equation}
  By~\Cref{prop:gamma_cptp_paulisupo}, the quasiprobability extent for the noise inversion is given by
  \begin{equation}
    \gamma_{\cptp^{\star}}(\cN^{-1}) = \frac{4-p}{4(1-p)} + 3\frac{p}{4(1-p)} = \frac{1+\frac{p}{2}}{1-p}  \, .
  \end{equation}
  Note that the quasiprobability extent converges to $1$ as $p$ approaches zero.
  Conversely, towards $p\rightarrow 1$ the extent diverges, as $\cN$ is not invertible at $p=1$.
\end{example}
\begin{example}
  Consider the one-qubit amplitude damping noise channel
  \begin{equation}
    \cN(\rho)\coloneqq K_0\rho K_0^{\dagger} + K_1\rho K_1^{\dagger}
    \quad\text{where}\quad
    K_0=\begin{pmatrix}1&0\\0&\sqrt{1-p}\end{pmatrix} \, ,
    K_1=\begin{pmatrix}0&\sqrt{p}\\0&0\end{pmatrix}
  \end{equation}
  and $p\in[0,1]$.
  The quasiprobability extent of its inverse can be shown to be
  \begin{equation}
    \gamma_{\cptp^{\star}}(\cN^{-1}) = \dnorm{\cN^{-1}} = \frac{1+p}{1-p} \, .
  \end{equation}
  For a proof, we refer the reader to~\cite{regula2021_operational}.
\end{example}

\paragraph{Entanglement detection}
Positive but not completely positive maps $\cE\in\hp(A)$ provide a simple criterion to detect entanglement in a bipartite state $\rho_{AB}$.
It is easy to show that $(\cE\otimes\idchan_B)(\rho_{AB})$ can only contain negative eigenvalues if $\rho_{AB}$ is entangled.
The canonical example of such an entanglement criterion is provided by choosing $\cE$ to be the transpose map
\begin{equation}
  \transpose_A\left(\sum_{i,j} \alpha_{i,j} \ketbra{b_i}{b_j} \right) \coloneqq \sum_{i,j} \alpha_{j,i} \ketbra{b_i}{b_j} \, .
\end{equation}
which is defined w.r.t. some fixed orthonormal basis $\{\ket{b_i}\}_i$ of $A$.

While useful in theory, it is difficult to translate such an entanglement criterion into an experimentally implementable protocol to detect entanglement.
Using QPS, we can address this problem by simulating the transpose map $\transpose_A$ and then estimate the smallest eigenvalue of the state.
The associated sampling overhead is as follows.
\begin{proposition}{}{}
  For any system $A$ one has $\gamma_{\cptp^{\star}}(\transpose_A) = \dimension(A)$. 
\end{proposition}
Note that for qubit systems, $\transpose_A$ has a diagonal Pauli transfer matrix and hence the decomposition set can be chosen to only contain efficiently implementable Pauli channels (recall~\Cref{prop:gamma_cptp_paulisupo}).
\begin{proof}
  Let's denote $d_A\coloneqq\dimension(A)$.
  By~\Cref{thm:gamma_diamond_norm}, we need to evaluate $\dnorm{\transpose_A}$.
  It is insightful to realize that the Choi representation of $\transpose_A$ is simply the qudit swap operator
  \begin{equation}
    \choi{\transpose_A}
    = \frac{1}{d_A}(\idchan\otimes\transpose)\left(\sum_{i,j}\ketbra{ii}{jj}\right)
    = \frac{1}{d_A}\sum_{i,j}\ket{ij}\bra{ji}
    = \frac{1}{d_A}\swap_{AA'} \, .
  \end{equation}
  Since $\swap_{AA'}$ is both Hermitian and unitary, its eigenvalues are precisely $\pm 1$ and
  \begin{equation}
    \dnorm{\transpose_A} \geq \norm{\choi{\transpose_A}}_1 = \norm{\frac{1}{d_A}\swap_{AA'}}_1 = d_A \, .
  \end{equation}
  For the converse, we briefly discuss the spectral decomposition of $\swap_{AA'}$.
  The $+1$ eigenstates are $\ket{ii}$ and $\frac{\ket{ij}+\ket{ji}}{\sqrt{2}}$ for $i<j$ and the $-1$ eigenstates are $\frac{\ket{ij}-\ket{ji}}{\sqrt{2}}$ for $i<j$.
  The decomposition of $\choi{\transpose_A}$ into the difference of the $+1$ and $-1$ eigenspace projectors
  \begin{equation}
    \choi{\transpose_A}
    = \Lambda^+ - \Lambda^-
    = \frac{d_A+1}{2}\frac{\Lambda^+}{\tr[\Lambda^+]} - \frac{d_A-1}{2}\frac{\Lambda^-}{\tr[\Lambda^-]}
  \end{equation}
  induces a valid QPD of $\transpose_E$ w.r.t. $\cptp$ because
  \begin{align}
    \tr_{A'}\left[\frac{\Lambda^+}{\tr[\Lambda^+]}\right]
    &= \frac{2}{d_A(d_A+1)}\tr_{A'}\left[\sum_{i}\proj{ii}\right] \\
    &\quad + \frac{2}{d_A(d_A+1)}\tr_{A'}\left[\sum_{i<j}\frac{(\ket{ij}+\ket{ji})(\bra{ij}+\bra{ji})}{2}\right] \\
    &= \frac{2}{d_A(d_A+1)} \sum_i \proj{i}
    + \frac{2}{d_A(d_A+1)} \sum_{i<j}\frac{\proj{i}+\proj{j}}{2}  \\
    &= \frac{2}{d_A(d_A+1)}\id_A 
    + \frac{d_A-1}{d_A(d_A+1)} \id_A \\
    &= \frac{\id_{A}}{d_A}
  \end{align}
  and analogously $\tr_{A'}\left[\frac{\Lambda^-}{\tr[\Lambda^-]}\right] = \frac{\id_{A}}{d_A}$.
  Hence, we have
  \begin{equation}
    \gamma_{\cptp^{\star}}\leq \frac{d_A+1}{2} + \frac{d_A-1}{2} = d_A \, .
  \end{equation}
\end{proof}

\paragraph{Quantum universal-NOT gate}
The universal-NOT ($\unot)$ gate is an anti-unitary single-qubit operation that maps a vector on the Bloch sphere to its opposite point 
\begin{equation}
  \alpha\ket{0} + \beta\ket{1} \mapsto \beta^*\ket{0} - \alpha^*\ket{1} \, .
\end{equation}
As the basis-independent quantum analog of the classical NOT gate, the $\unot$ gate and its best CPTP approximations have been extensively studied~\cite{buzek1999_optimal,vanenk_2005,}.
In fact, it has been shown to be closely related to the transpose superoperator and quantum cloning.
Its action as a linear superoperator
\begin{equation}\label{eq:unot_supo}
  q\id + \vec{r}\cdot\vec{\sigma} \mapsto q\id - \vec{r}\cdot\vec{\sigma}
\end{equation}
for $q\in\mathbb{R}$, $\vec{r}\in\mathbb{R}^3$ is trace-preserving but not completely positive.
Here, $\vec{r}\cdot\vec{\sigma}$ is short-hand notation for $r_x X + r_y Y + r_z Z$.
Still, it is possible to simulate a $\unot$ gate using QPS.
In fact,~\Cref{eq:unot_supo} precisely provides us with its Pauli transfer representation $\diag(1,-1,-1,-1)$.
By~\Cref{prop:gamma_cptp_paulisupo}, this implies that $\gamma_{\cptp^{\star}}(\unot) = 2$ with the achieving QPD being
\begin{equation}
  \unot = -\frac{1}{2}\idchan + \frac{1}{2}\indsupo{X} + \frac{1}{2}\indsupo{Y} + \frac{1}{2}\indsupo{Z} \, .
\end{equation}

\paragraph{Quantum cloning}
The famous quantum no-cloning theorem states that there exists no physical quantum process that maps every state $\ket{\psi}$ onto two copies of itself $\ket{\psi}\mapsto\ket{\psi}\otimes\ket{\psi}$.
As such, much research has been devoted to finding the $\cptp$ maps which best approximate the cloning operation.
We refer the reader to~\cite{scarani2005_qcloning} for a comprehensive overview of results on approximate cloning.

Since QPS allows for the simulation of nonphysical quantum operations, it stands that it could be used to effectively realize improved or even perfect cloning machines.
Here, we investigate the arguably simplest setting of state-dependent cloning of two pure states.
More precisely, we consider two pure states $\ket{\psi_0}$ and $\ket{\psi_1}$ and investigate cloning machines $\cE$ which aim to maximize the average fidelity
\begin{equation*}
  \frac{1}{2}\left(\bra{\psi_0}\otimes\bra{\psi_0}\right) \cE\left(\proj{\psi_0}\right) \left( \ket{\psi_0}\otimes\ket{\psi_0}\right)
  +\frac{1}{2}\left(\bra{\psi_1}\otimes\bra{\psi_1}\right) \cE\left(\proj{\psi_1}\right) \left( \ket{\psi_1}\otimes\ket{\psi_1}\right)
  .
\end{equation*}
Without loss of generality, we can assume that the two states are of the form
\begin{equation}
  \ket{\psi_0} = \cos\frac{\theta}{2}\ket{0} + \sin\frac{\theta}{2}\ket{1}
  \, , \quad
  \ket{\psi_1} = \cos\frac{\theta}{2}\ket{0} - \sin\frac{\theta}{2}\ket{1}
\end{equation}
for $\theta\in[0,\pi/2]$ such that $\cos(\theta)$ characterizes the overlap $\braket{\psi_0}{\psi_1}$.

Let us denote by $A$ and $B$ the input 1-qubit system and the output 2-qubit system.
The maximum achievable fidelity $F_{\mathrm{max}}$ with a physical $\cptp$ cloner can be expressed in terms of a simple SDP
\begin{equation}\label{eq:sdp_fmax}
  F_{\mathrm{max}}(\theta) = \left \lbrace
  \begin{array}{r l}
    \max\limits_{\choi{\cE}\in\pos(AB)}
    & \frac{1}{2}\left(\bra{\psi_0}\otimes\bra{\psi_0}\right) \cE\left(\proj{\psi_0}\right) \left( \ket{\psi_0}\otimes\ket{\psi_0}\right) \\
    & +\frac{1}{2}\left(\bra{\psi_1}\otimes\bra{\psi_1}\right) \cE\left(\proj{\psi_1}\right) \left( \ket{\psi_1}\otimes\ket{\psi_1}\right) \\
    \textnormal{s.t.} & \tr_{B}[\choi{\cE}] = \frac{1}{2}\id_A
  \end{array} \right .
\end{equation}
where $\cE(\sigma)$ should be understood as the evolution in terms of the Choi matrix, i.e., $\cE(\sigma)=\dimension(A)\tr_{A}\left[\choi{\cE}(\sigma^T\otimes\id_{B})\right]$.

While no perfect $\cptp$ cloner exists (unless $\theta=0$ or $\theta=\pi /2$), it is easy to see that there always exists a Hermitian-preserving superoperator which achieves optimal cloning.
At this point, it is natural to ask what is the smallest quasiprobability extent $\gamma_{\cptp^{\star}}(\cE)$ over all the maps $\cE\in\hp(A\rightarrow B)$ which achieve perfect cloning.
Again, this can be expressed as a SDP
\begin{equation}\label{eq:sdp_gammamin}
  \gamma_{\mathrm{min}}(\theta) = \left \lbrace
  \begin{array}{r l}
    \max\limits_{\substack{\choi{\cE}\in\herm(AB) \\ Y_0,Y_1\in\pos(AB)}} & \frac{1}{2}\opnorm{\tr_{B}[Y_0]} + \frac{1}{2}\opnorm{\tr_{B}[Y_1]} \\
    \textnormal{s.t.} & \cE(\proj{\psi_0}) = \proj{\psi_0}\otimes\proj{\psi_0}, \\
    & \cE(\proj{\psi_1}) = \proj{\psi_1}\otimes\proj{\psi_1}, \\
    & \begin{pmatrix}Y_0 & -\choi{\cE} \\ - \choi{\cE} & Y_1 \end{pmatrix} \loewnergeq 0
  \end{array} \right .
\end{equation}
where we used~\Cref{thm:gamma_diamond_norm} as well as the SDP characterization of the diamond norm from \reference~\cite{watrous2013_simplersdp}.

The numerical evaluation of the SDPs in~\Cref{eq:sdp_fmax,eq:sdp_gammamin} is depicted in~\Cref{fig:approximate_cloning}.
Interestingly, the minimal quasiprobability extent does not directly correlate with the maximum fidelity of the optimal $\cptp$ cloner.
Indeed, as the overlap between the states approaches $1$, the approximate $\cptp$ cloner becomes more accurate whereas the minimal quasiprobability extent seems to converge towards a value of $2$.
Notably, $\gamma_{\mathrm{min}}(\theta)$ is discontinuous at $\theta=0$, as cloning is trivial for identical states.
Our results show that \emph{exact} cloning is most expensive to simulate for $\theta$ close to zero, even though \emph{approximate} cloning is very easy in this regime (by essentially doing almost nothing).

\begin{figure}
  \centering
  \includegraphics{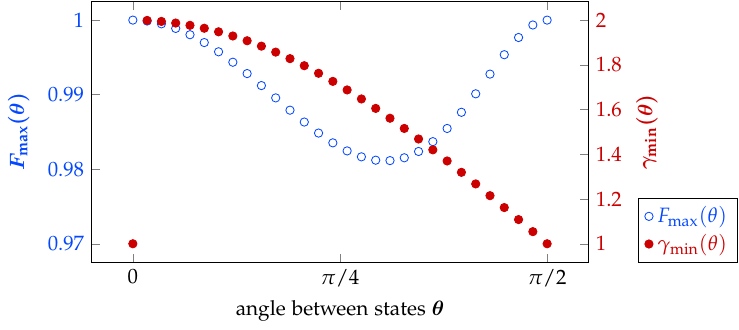}
  \caption{Numerics on state-dependent cloning machines for two pure input states with overlap $\cos\theta$.
  $F_{\mathrm{max}}$ denotes the maximum cloning fidelity for a physical ($\cptp$) cloning machine and $\gamma_{\mathrm{min}}$ denotes the minimal quasiprobability extent of a non-physical optimal cloning machine. The limits $\theta=0$ and $\theta=\pi/2$ correspond to two identical states, respectively two orthogonal states.}
  \label{fig:approximate_cloning}
\end{figure}

\paragraph{Imaginary time evolution}
\textit{Disclaimer:
The idea described in this paragraph was orally conveyed to the author by Ray \etal~\cite{ray2025_imaginary}.
The original authors plan to publish a detailed analysis in the near future.
The current paragraph was written with their explicit permission and all credit for the idea is due to them.}

Imaginary time evolution is a standard technique in computational quantum mechanics to find the ground state of a given Hamiltonian $H$.
Consider evolving a random initial state $\ket{\psi}$ with the operator $e^{-\tau H}$ for some $\tau\geq 0$.
This evolution can be interpreted as ``imaginary time evolution'', since it corresponds to substituting imaginary time $\tau=it$ into the time evolution operator $e^{-iHt}$ derived from the Schrödinger equation (note that we are fixing $\hbar=1$).
By choosing $\tau$ large enough, the ground state component becomes dominant.

More precisely, let $H\ket{b_i}=E_i\ket{b_i}$ be the eigenvectors and eigenvalues of $H$.
The imaginary time evolution expressed in the eigenbasis reads
\begin{equation}
  \ket{\psi} = \sum_i c_i \ket{b_i} \mapsto e^{-\tau H}\ket{\psi} = \sum_i c_ie^{-\tau E_i}\ket{b_i}
\end{equation}
so the ground state component decays exponentially slower than the other components.
The convergence time $\tau_{\mathrm{conv}}$ is governed by the inverse gap $\tau_{\mathrm{conv}}=\mathcal{O}(\frac{1}{\Delta})$ where $\Delta\coloneqq E_1-E_0$.
Using QPS, we can simulate the imaginary time evolution operator and hence estimate the ground state energy
\begin{equation}
  E_0 \approx \frac{\braopket{\psi'}{H}{\psi'}}{\braket{\psi'}{\psi'}} \text{ where } \ket{\psi'}\coloneqq e^{-\tau H}\ket{\psi}
\end{equation}
by estimating the two expectation values $\braopket{\psi'}{H}{\psi'}$ and $\braket{\psi'}{\psi'}$.
Recalling~\Cref{prop:qps_overhead,thm:gamma_diamond_norm}, we can estimate both to an additive error $\epsilon$ with a sampling overhead of $\mathcal{O}\left( \dnorm{\indsupo{e^{-\tau H}}}^2 \epsilon^{-2} \right)$.
As $\tau$ increases, the quasiprobability extent $\dnorm{\indsupo{e^{-\tau H}}}$ decreases, but so does also the required additive error $\epsilon$.
We refrain from a full analysis of the scheme and instead refer to \reference~\cite{ray2025_imaginary}.


\section{Further reading}
Two separate works by Piveteau \etal~\cite{piveteau2022_quasiprobability,piveteau2020_advanced} and Jiang \etal~\cite{jiang2021_physical} independently introduced the idea of decomposing $\mathbb{R}\hptp$ linear maps into CPTP operations.
The former calls the optimal QPD the \emph{channel difference decomposition} and uses it as an intermediate technical results for an error mitigation protocol.
The latter calls the log-quasiprobability extent the \emph{physical implementability}, and shows some basic properties (faithfulness, additivity, unitary invariance, relation to robustness) as well as weak monotonicity under superchannels.
Both express the quasiprobability extent in terms of a semidefinite program, but it was only later work by Regula \etal~\cite{regula2021_operational} that realized that this SDP is equivalent to the diamond norm (at least for $\mathbb{R}\hptp$ maps).
This realization essentially implies that many proofs in the previous papers just re-stated known properties of the diamond norm.
The authors also relate the quasiprobability extent to some advantage that non-physical maps can provide in discrimination-based quantum games.
In addition, this third work also realized that all Hermitian-preserving maps can be simulated if the decomposition set is enlarged from $\cptp$ to the set of all completely positive trace-non-increasing maps $\cptni$.
Note that this does not capture the full extent of using classical side information, as it effectively only allows for non-negative weighting of the quantum instruments.
Indeed, there exist explicit non-TP Hermitian-preserving maps $\cE$ for which $\gamma_{\cptni}(\cE)$ is strictly greater than $\gamma_{\cptp^{\star}}(\cE)$, but still finite.
Only further work by Xuanqiang \etal~\cite{zhao2023_power} incorporated the full power of classical side information and realized that the diamond norm characterization also extends to $\gamma_{\cptp^{\star}}$.

The study of QPDs of non-physical maps is related to the technique of \emph{structural physical approximation} (SPA)~\cite{horodecki2003_limits}.
The underlying motivation is that positive, but not completely positive, maps $\cE\in\hp(A)$ can serve as entanglement witnesses: $\cE\otimes\id_B(\rho_{AB})$ can only have negative eigenvalues if $\rho_{AB}$ is itself entangled.
The most famous instance is the Peres-Horodecki criterion, where $\cE$ is chosen to be the transpose map~\cite{peres1996_separability}.
Unfortunately, such positive non-CP maps are not physically realizable, making it difficult to directly translate these mathematical criteria into experimentally realizable protocols to detect entanglement.
It can be shown that taking a mixture of a non-CP map with the fully-depolarizing channel $\cD(\rho)\coloneqq \id_A / \dimension(A)$ will produce a CPTP map
\begin{equation}
  \tilde{\cE} \coloneqq (1-p)\cE + p\cD \in \cptp(A)
\end{equation}
provided the mixing probability $p$ is chosen large enough.
This map $\tilde{\cE}$ for the smallest admissible $p$ is called a structural physical approximation of $\cE$ and it can be used to physically realize entanglement detection protocols.
We refer the reader to~\cite{horodecki2002_method,korbicz2008_structural,bae2017_designing} for more details.
Clearly, any SPA can be seen as a QPD
\begin{equation}
  \cE = \frac{1}{1-p}\tilde{\cE} - \frac{p}{1-p}\cD \, .
\end{equation}
Therefore, QPDs can be considered a generalization of SPAs which allows for mixing with operations beyond the depolarizing channel to achieve complete positivity.
The minimal admissible probability $p$ which gives a valid SPA is thus an upper bound for the quasiprobability extent and can easily be shown to be related to a variant of the robustness measure from~\Cref{def:robustness}, called the \emph{random robustness}~\cite{vidal1999_robustness,chitambar2019_quantum}.

\section{Contributions}
The idea of decomposing and simulating general $\mathbb{R}\hptp$ maps with CPTP maps was proposed by the author in \reference~\cite{piveteau2020_advanced,piveteau2022_quasiprobability}.
Later work in \reference~\cite{jiang2021_physical} indpendently proposed the same idea.
The SDP formulation of $\gamma_{\cptp}$ is due to the author.
However, the connection to the diamond norm as well as the generalization to $\gamma_{\cptp^{\star}}$ were first introduced by other researchers.

The applications in~\Cref{sec:nonphysical_applications}, including~\Cref{prop:gamma_cptp_paulisupo}, are novel work and have not been published before.
The only exception is the paragraph \emph{imaginary time evolution}, which is due to \reference~\cite{ray2025_imaginary}.

\chapter{Simulating nonlocal computation}\label{chap:nonlocal}
In this chapter, we will study QPS in the context of the QRT of entanglement, where the free operations across two parties are given by local operations with classical communication (LOCC).
In the broadest sense, this means that we will study the quasiprobability extent $\gamma_{\locc^{\star}}(\cE)$ of bipartite nonlocal channels $\cE\in\cptp(AB)$ across two parties $A$ and $B$.
As before, the set $\locc^{\star}$ (as opposed to $\locc$) captures the possibility of utilizing classical side information from intermediate measurements in the sampling procedure (recall the discussion in~\Cref{sec:qps_intermediate_measurements} and~\Cref{eq:qrt_extended_ds}).
In contrast to the QRT of physical operations we encountered in~\Cref{chap:nonphysical}, the characterization of $\gamma_{\locc^{\star}}(\cE)$ will turn out to be extremely difficult, and we will only find closed-form expressions for a few simple cases.
This can be attributed to the very rich structure of the QRT, and we will establish many connections to the study of entanglement.

The main practical application of QPS of nonlocal operations is for the purpose of realizing \emph{circuit knitting} (also called \emph{circuit cutting}) to simulate a \emph{large} quantum computer using \emph{small} quantum computers.
Suppose one wants to run a quantum circuit on a given quantum device, but its number of available qubits is too small to fit the circuit.
The premise of circuit knitting is to subdivide the circuit into smaller partitions which are each small enough to fit on the available quantum computer.
The goal is to reconstruct the outcome of the original circuit by executing the smaller partitions on the quantum computer and performing appropriate classical post-processing.

It is very natural to realize circuit knitting through QPS.
Consider for instance the setup depicted in~\Cref{fig:example_gate_cut}.
The qubits in a circuit are grouped into two systems $A$ and $B$.
Every gate that acts nonlocally across $A$ and $B$ is simulated using LOCC operations via QPS.
This effectively means that we express the original large circuit as a quasiprobabilistic mixture of two sub-circuits which interact with each other only through classical communication.
As such, we can simulate the original quantum circuit with two smaller quantum computers, one executing $A$ and the other executing $B$.
Note that this approach also works when the circuit is divided into more than 2 partitions.

\begin{figure}
  \centering
  \includegraphics{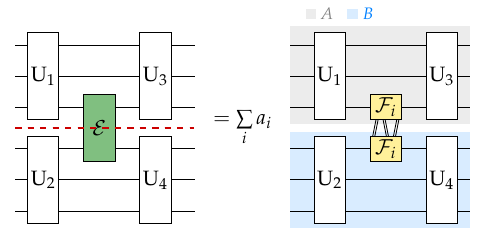}
  \caption{Example of a space-like cut (red dashed line) which partitions a large $6$-qubit circuit into two $3$-qubit sub-circuits. Using a QPD $\cE=\sum_ia_i\cF_i$ of $\cE$ into LOCC operations $\cF_i$, the total circuit can be written as a probabilistic mixture of two separate $3$-qubit circuits $A$ and $B$ which only interact through classical communication.}
  \label{fig:example_gate_cut}
\end{figure}

The partitioning of a circuit in circuit knitting can generally be achieved using two kinds of ``cuts''.
The previous example illustrated \emph{space-like cuts} (sometimes also called \emph{gate cuts}) which horizontally split nonlocal gates into local operations.
There is also a second kind of cut, called \emph{time-like cuts} (or \emph{wire cuts}), which vertically split a qubit wire as depicted on the left side of~\Cref{fig:example_wire_cut}.
A wire cut can be achieved through QPS by decomposing the qubit identity channel into a quasiprobabilistic mixture of \emph{entanglement breaking channels}\footnote{The reason for the name ``entanglement breaking channel'' will become more apparent in~\Cref{sec:timelike_cuts}.} which consist of a measurement followed by a state preparation which may depend on the measurement outcome.
As depicted in~\Cref{fig:example_wire_cut}, this again allows us to write the total circuit as a quasiprobabilistic mixture of two sub-circuits which only communicate via classical communication, and hence can be executed on two smaller quantum devices.

\begin{figure}
  \centering
  \includegraphics{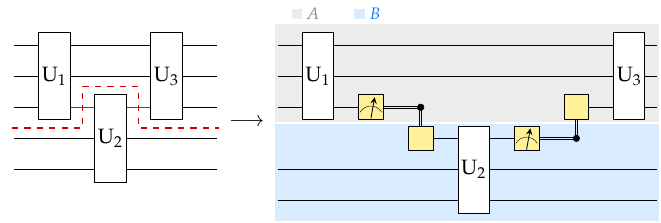}
  \caption{Example of a wire cut (red dashed line) which partitions a large $5$-qubit circuit into two $3$-qubit sub-circuits. Using a QPD of the identity channel into entanglement breaking channels, the total circuit can be written as a probabilistic mixture of two separate $3$-qubit circuits $A$ and $B$ which only interact through classical communication.}
  \label{fig:example_wire_cut}
\end{figure}

In this chapter, we will study both space-like cuts and time-like cuts and their associated quasiprobability extents $\gamma_{\locc^{\star}}$ and $\gamma_{\ebc^{\star}}$ where EBC stands for entanglement breaking channels.
Correspondingly, the chapter is split up into~\Cref{sec:spacelike_cuts,sec:timelike_cuts} treating the two cases.
Below, we briefly outline some specific questions that will be considered in detail.
\begin{description}
  \item[Utility of classical communication]
  In the above descriptions of space-like and time-like cutting, we assumed that the individual sub-circuits can exchange classical communication with each other.
  Interestingly, circuit knitting remains possible even without this capability.
  We will investigate whether classical communication between the sub-circuits can allow for a lower simulation overhead.
  Mathematically speaking, for space-like cuts this amounts to comparing $\gamma_{\locc^{\star}}$ with $\gamma_{\lo^{\star}}$ where LO stands for local operations \emph{without} classical communication.
  Similarly, for time-like cuts, we will also introduce a decomposition set which is based on channels consisting of measurements and state preparations, but in constrast to EBC, the state may not depend on the measurement outcome.

  While the importance of classical communication is academically very interesting, it also has some important practical ramifications for circuit knitting.
  Allowing for classical communication could a priori lead to a lower sampling overhead, since we are extending the decomposition set to a larger set.
  However, this comes at an additional price in terms of practical implementation costs: Without classical communication, the individual sub-circuits could be run sequentially and in principle on the same quantum device.
  However, once classical communication is involved, one will generally require multiple small quantum processors executing the sub-circuits simultaneously.

  Somewhat surprisingly, we will observe quite different results for the settings of space-like and time-like cutting.
  We will show that for a large range of nonlocal gates, their quasiprobability extents $\gamma_{\locc^{\star}}$ and $\gamma_{\lo^{\star}}$ are identical.
  In comparison, time-like cuts are significantly less expensive with classical communication.

  Note that when we are referring to the utilization of classical communication here, we explicitly mean classical communication \emph{during} the circuit execution.
  Circuit knitting inherently requires some form of communication to combine the final measurement results of the individual sub-circuits.
  However, this final communication can be considered part of the post-processing and doesn't entail the same technological requirements as mid-circuit classical communication.

  \item[Utility of classical side information]
  Following the discussion in~\Cref{sec:qps_intermediate_measurements}, we will consider the decomposition sets that capture the utilization of classical side information, e.g., $\locc^{\star}$ instead of $\locc$ and $\ebc^{\star}$ instead of $\ebc$ (recall~\Cref{eq:qrt_extended_ds}).
  It is a natural question to ask whether this is really useful and if we could achieve the same overhead without classical side information.

  Intriguingly, the answer seems to be very different in the settings with and without classical communication.
  For instance, $\gamma_{\locc^{\star}}(\cU)=\gamma_{\locc}(\cU)$ when $\cU$ is a Clifford gate, so classical side information is generally not always required to reach the lowest extent.
  In contrast, we will show in~\Cref{sec:gate_cut_side_info} that intermediate measurements are absolutely crucial in the absence of classical communication.
  Mathematically, this manifests as $\gamma_{\lo}(\cE)=\infty$ for any signalling operation $\cE$ whereas $\gamma_{\lo^{\star}}(\cE)$ is finite for all Hermitian-preserving superoperators $\cE$.
  An analogous result also holds for time-like cuts, as discussed in~\Cref{sec:wire_cut_side_info}.

  \item[Space-like cuts of states]
  Besides just studying space-like cuts of quantum channels, we will also spend significant effort in~\Cref{sec:gamma_nonlocal_states} considering space-like cuts of quantum states, i.e., QPDs of bipartite nonlocal states into separable states (or equivalently product states).
  This can not only be useful as a technical tool for bounding the quasiprobability extent of channels, but also on its own for cutting circuits that contain nonlocal states.

  To illustrate the importance, we highlight that there exist quantum gates that can be realized from a pre-shared entangled state using only LOCC.
  A great example for this is the CNOT gate: The gate-teleportation protocol~\cite{gottesman1999_demonstrating} depicted in~\Cref{fig:cnot_teleportation} allows for the realization of a CNOT gate using only a pre-shared Bell pair $\ket{\Psi}$ and LOCC operations.
  As such, any QPD of a Bell pair into product states $\proj{\Psi}=\sum_i a_i \rho_i\otimes\sigma_i$ naturally induces a QPD of the $\cnot$ gate into $\locc$ channels.
  We note that this is completely analogous to the magic state injection procedure discussed in~\Cref{ex:state_qpd_nonclifford}.

  \item[Scaling of quasiprobability extent under tensor product]
  Another important aspect that we will study in this chapter is the scaling of the quasiprobability extent under the tensor product.
  In many instances, the quasiprobability extent is strictly sub-multiplicative, which has direct ramifications for the asymptotic overhead of circuit knitting.
  Consider for instance the extent of a nonlocal Bell-pair.
  We will later see that it is given by\footnote{We are being intentionally imprecise here by omitting the decomposition set. We defer the formal mathematical discussion to later in the chapter.} $\gamma(\proj{\Psi})=3$, while $\gamma(\proj{\Psi}^{\otimes n}) = 2^{n+1}-1<3^n$ and $\gammareg(\proj{\Psi})=\gammasreg(\proj{\Psi})=2$.
  It is therefore cheaper to quasiprobabilistically simulate the joint preparation of $n$ nonlocal Bell pairs instead of optimally simulating them separately.
  When simulating $c$ nonlocal $\cnot$ gates through the gate teleportation approach (as discussed in~\Cref{fig:cnot_teleportation}), this joint simulation strategy constitutes a large reduction in sampling overhead from $\mathcal{O}(9^c)$ to $\mathcal{O}(4^c)$.

  While the regularized and asymptotic quasiprobability extent for the Bell state $\ket{\Psi}$ are identical, we will show that this is surprisingly not the case for almost all other pure states.
  More precisely, we will show that for a pure bipartite state $\ket{\phi}_{AB}$ one has $\log\gammareg(\proj{\phi})=H_{1/2}(\tr_B[\proj{\phi}])$ and $\log\gammasreg(\proj{\phi})=H(\tr_B[\proj{\phi}])$ where $H$ and $H_{1/2}$ denote the von Neumann entropy and the quantum Rényi-$\nicefrac{1}{2}$ entropy.

  The quasiprobability extent of an $n$-fold copy of some gate $\gamma(\cU^{\otimes n})$ captures the optimal overhead for simulating $n$ instances of the nonlocal gate $\cU$ in \emph{parallel}, i.e., in the same time slice of the circuit.
  This is a strong assumption, as typically the nonlocal gates in a circuit can occur far apart from each other.
  We will see that in some instances, it remains possible to harness the submultiplicative behavior of the quasiprobability extent for distant gates and reduce the sampling overhead.
  We refer to this as the \emph{black box setting}, and it will be discussed in~\Cref{sec:submult}.

\end{description}
\begin{figure}
  \centering
  \includegraphics{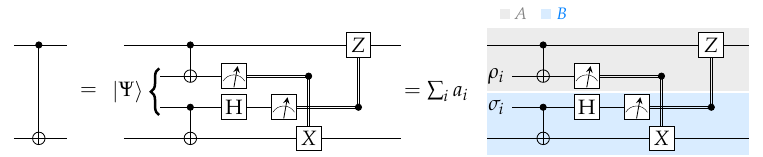}
  \caption{The first equality shows how a nonlocal $\cnot$ gate between two parties can be realized using a $\locc$ protocol that consumes a pre-shared Bell state $\ket{\Psi}=(\ket{00}+\ket{11})/\sqrt{2}$. By writing the Bell pair in terms of a quasiprobabilistic mixture of product state $\proj{\Psi}=\sum_i a_i \rho_i\otimes \sigma_i$, we directly obtain a QPD of the $\cnot$ gate into $\locc$ protocols. The equalities should be understood in terms of superoperators.}
  \label{fig:cnot_teleportation}
\end{figure}
Generally speaking, there is no efficient methodology for evaluating the quasiprobability extent for space-like and time-like cuts of arbitrary channels.
Instead, we will have to content ourselves with various techniques to find lower and upper bounds. 
Fortunately, in some instances, these bounds match.
The most important such examples that we will encounter are
\begin{itemize}
  \item space-like cuts of pure states (see~\Cref{thm:optimal_qpd_pure} for the single-shot result and~\Cref{cor:gamma_reg_pure,cor:gamma_sep_asymptotic} for the asymptotic results), 
  \item space-like cuts of Clifford gates (see~\Cref{cor:gamma_locc_clifford}),
  \item space-like cuts of parallel two-qubit gates (see~\Cref{thm:gamma_kaklike_unitary}) and 
  \item time-like cuts of the identity channel (see~\Cref{thm:gamma_mpc_id,thm:gamma_ebc_id}).
\end{itemize}
Generally, we will find upper bounds by explicitly constructing a QPD with the corresponding 1-norm of its coefficients.
To find lower bounds, we will often (but not always) proceed by relaxing the decomposition set to a larger set, for which the quasiprobability extent is easier to evaluate.
For instance, LOCC has a notoriously complicated structure.
To avoid dealing with its difficulties, we will often consider the relaxations to the sets of separable channels (SEPC) and positive partial transpose channels (PPTC).
Especially the latter exhibits a much simpler characterization of the quasiprobability extent in terms of a semidefinite program.

In the context of space-like cuts, we briefly note that the set LOCC contains protocols that incur an unbounded number of communication rounds.
As such, the time complexity for realizing arbitrary LOCC protocols can be exceedingly high, possibly overshadowing the sampling overhead of QPS in the first place.
A practically more meaningful decomposition set would be based on the set $\locc_r$ with at most $r$ rounds, for some appropriately chosen $r$.
However, in this thesis, we will only consider the unbounded round LOCC setting for two reasons.
First, it is already exceedingly difficult to make statements about $\gamma_{\locc^{\star}}$ without having to worry about the number of communication rounds.
Furthermore, we are fortunate that all protocols that we will encounter (like the ones mentioned in the list above) anyway only involve a small number of rounds.

\section{Space-like cuts}\label{sec:spacelike_cuts}
\subsection{Decomposition sets}\label{sec:spacelike_ds}
In this section, we start by formally introducing all relevant decomposition sets that will play a major role for our study of space-like cuts.
We start by characterizing the sets of states that are important for the resource theory of entanglement.

\begin{definition}{}{}
  For two quantum systems $A,B$, we define the set of separable states $\sep(A;B)$ to be all bipartite states $\rho_{AB}\in\dops(AB)$ of the form
  \begin{equation}
    \rho_{AB} = \sum\limits_{i=1}^m p_i \rho_A^{(i)}\otimes \sigma_B^{(i)}
  \end{equation}
  where $m\in\mathbb{N}$, $p\in\mathbb{R}^m$ is a probability vector and $\rho^{(i)}_A\in\dops(A),\sigma^{(i)}_B\in\dops(B)$.
\end{definition}
The set $\sep(A;B)$ precisely contains all states that are reachable from an initial product state using local operations and classical communication.
By definition, they are the states that do not possess any entanglement, as their correlations can be thought to be of purely classical origin.

Since separability is not expressible as a semidefinite constraint, one often encounters the challenge that optimizing over $\sep(A;B)$ is difficult.
To address this, a common approach is to relax such problems by instead optimizing over the larger set of \emph{positive partial transpose states}.
For this purpose, we define the transpose map $\transpose_A$ for some system $A$ w.r.t. an orthonormal basis $\{\ket{b_1},\dots,\ket{b_n}\}$ by
\begin{equation}
  \transpose_A\left(\sum_{i,j} \alpha_{i,j} \ketbra{b_i}{b_j} \right) \coloneqq \sum_{i,j} \alpha_{j,i} \ketbra{b_i}{b_j} \, .
\end{equation}
\begin{definition}{}{}
  For two quantum systems $A,B$, we define the set of positive partial transpose states $\ppt(A;B)$ to be all bipartite states $\rho_{AB}\in\dops(AB)$ which fulfill $(\idchan_A\otimes\transpose_B)(\rho_{AB}) \loewnergeq 0$.
\end{definition}
We leave it as an exercise to the reader to verify that this definition does not depend on the choice of orthonormal basis or the system on which the transpose is taken.
It is also evident to see that any separable operator has a positive partial transpose, i.e., $\sep(A;B)\subset\ppt(A;B)$.
This gives rise to the well-known Peres-Horodecki criterion: A quantum state that does not exhibit a positive partial transpose must be entangled~\cite{peres1996_separability}.
For small systems, this criterion is faithful.
\begin{lemma}{\cite{horodecki1996_separability}}{ppt_small_system}
  If $\dimension(A)\cdot\dimension(B)\leq 6$ then $\sep(A;B)=\ppt(A;B)$.
\end{lemma}
However, for larger systems, the inclusion $\sep(A;B)\subset\ppt(A;B)$ becomes strict.
The following state, due to \reference~\cite{bennett1999_unextendible}, provides an example of an entangled PPT state for $\dimension(A)=\dimension(B)=3$.
\begin{example}\label{ex:entangled_ppt_state}
  Consider the two systems $A,B$ with $\mathcal{H}_A=\mathcal{H}_B=\mathbb{C}^3$ and the five states
  \begin{align*}
    \ket{\psi_1} = \frac{1}{\sqrt{2}}\ket{0}\otimes(\ket{0}-\ket{1})
    & , &
    \ket{\psi_2} = \frac{1}{\sqrt{2}}(\ket{0}-\ket{1})\otimes\ket{2}
    , \\
    \ket{\psi_3} = \frac{1}{\sqrt{2}}\ket{2}\otimes(\ket{1}-\ket{2})
    & , &
    \ket{\psi_4} = \frac{1}{\sqrt{2}}(\ket{1}-\ket{2})\otimes\ket{0}
    , \\ 
    \multicolumn{3}{c}{$\ket{\psi_5} = \frac{1}{3}(\ket{0}+\ket{1}+\ket{2}) \otimes (\ket{0}+\ket{1}+\ket{2})$} \, .
  \end{align*}
  Any product state has non-zero overlap with at least one of these states.
  Therefore, any state in the orthogonal complement of these five states $\mathcal{H}^{\perp}$ must be entangled.
  Denote $\rho\coloneqq \frac{1}{4}\Pi$ where $\Pi$ is the projection onto $\mathcal{H}^{\perp}$.
  On one hand, $\rho$ must be entangled since $\Pi\rho\Pi=\rho$.
  On the other hand, $\rho$ is invariant under partial transposition, because $\rho\propto\id-\sum_{i=1}^5\proj{\psi_i}$ and each $\proj{\psi_i}$ is clearly invariant under partial transposition.
  Therefore, $\rho\in\ppt(A;B)$ but $\rho\notin\sep(A;B)$.
\end{example}

Next, we introduce four resource theories which are called LO, LOCC, SEPC and PPTC.
The QRT LOCC constitutes the cornerstone of the study of entanglement in quantum information theory.
It is operationally defined as the set of channels that can be realized by two\footnote{The four QRTs can straightforwardly be defined for arbitrarily many parties, but we restrict ourselves to the two-party case here.} parties which can perform arbitrary quantum operations on their respective local system, but are limited to classical communication between each other.
As such, it has a very rich structure and far-reaching applications for many entanglement manipulation tasks, such as the preparation, discrimination and transformation of entangled states.
In contrast, LO corresponds to the more restricted setting where the two parties cannot exchange classical communication with each other.
It doesn't have the same operational applications and is mathematically easy to handle, at least compared to LOCC.
For the purposes of QPS, LO captures the setting where the two sub-circuits do not have the capabilities to exchange classical information.

The set LOCC is notoriously difficult to work with, and as such, it is often necessary to resort to relaxations of it.
SEPC and PPTC are two supersets of LOCC that are mathematically easier to handle.
Compared to LO and LOCC, they have no operational meaning for the purposes of QPS and circuit knitting, and instead will serve purely as mathematical tools to bound the quasiprobability extent w.r.t. LOCC.

We now define the four resource theories in full detail.
\begin{description}
\item[Local operations (LO):]
The set of free transformations for this QRT is given by all product channels
\begin{equation}
  \lo(A;B\rightarrow A';B') \coloneqq \{\cE\otimes\cF | \cE\in\cptp(A\rightarrow A'), \cF\in\cptp(B\rightarrow B') \}
\end{equation}
and the associated free instruments are
\begin{equation}
  \instr{\lo}(A;B\rightarrow A';B') \coloneqq \left\{ (\cE_i\otimes\cF_j)_{i,j} \middle\vert 
  \begin{array}{l}
    (\cE_i)_i \in\qi(A\rightarrow A') \\
    (\cF_j)_j\in \qi(B\rightarrow B')
  \end{array}
  \right\} \, .
\end{equation}
For the purposes of QPS, we are interested in the induced decomposition set that accounts for classical side information (recall the discussion in~\Cref{sec:qps_intermediate_measurements})
\begin{equation}
  \lo^{\star}(A;B\rightarrow A';B') \coloneqq \qids[\instr{\lo}(A;B\rightarrow A';B')] \, .
\end{equation}
The free states in the QRT of LO are precisely the product states of the form $\rho_A \otimes \sigma _B$ where $\rho_A\in\dops(A)$ and $\sigma_B\in\dops(B)$.

Despite its very simple definition, LO unfortunately lacks certain structure compared to the other three QRTs (LOCC, SEPC, PPTC) that makes it a bit cumbersome to work with in the context of QPS.
For instance, it is not a convex resource theory, so~\Cref{lem:qrt_convexity,lem:two_element_qpd,lem:one_element_qpd} don't apply, and we generally need to consider QPDs with an unbounded number of terms.
Furthermore, $\instr{\lo}$ is neither coarse-grainable nor trivially fine-grainable\footnote{Note that the two instruments on each bipartition can be individually coarse-grained or fine-grained. However, the combined instrument can generally not, as that would require additional classical communication.}, and hence no simple characterization of $\lo^{\star}$ in the sense of~\Cref{lem:characterization_expanded_decomposition_set} is possible.
In the context of~\Cref{cor:free_classical_computation}, this can be attributed to the fact that general classical computation is not free under LO (only \emph{local} classical computation is).

Some previous work instead considered the decomposition set~\cite{piveteau2024_circuit}
\begin{equation}\label{eq:old_lostar_def}
  \{ \cE\otimes\cF | \cE\in\cptp^{\star}(A), \cF\in\cptp^{\star}(B) \}
\end{equation}
(recall~\Cref{eq:char_cptpstar} for the characterization of $\cptp^{\star}$).
While this set spans the same space as $\lo^{\star}$ (as we will see later), it is a strict subset of $\lo^{\star}(A;B\rightarrow A';B')$ since the weighting based on the measurement outcomes on $A$ and $B$ is done independently.

\item[Local operations with classical communication (LOCC):]
The formal mathematical construction of the set LOCC is quite involved and typically phrased in terms of the associated free instruments $\instr{\locc}$.
Since this formal definition will not be particularly insightful for our considerations, we refrain from reproducing it here and instead refer the interested reader to the excellent exposition in \reference~\cite{chitambar2014_locc}.
It suffices to note that the set of free instruments $\instr{\locc}(A;B\rightarrow A'B')\subset\qi(AB\rightarrow A'B')$ is
\begin{itemize}
  \item coarse-grainable by construction.
  \item trivially fine-grainable, as any protocol can be extended by having one of the two parties sample a random bit with probability depending on all previous measurement outcomes.
  \item closed under mixture, as the two parties can randomly choose between one of two protocols using one round of classical communication at the beginning.
\end{itemize}
By~\Cref{lem:qrt_convexity,lem:characterization_expanded_decomposition_set}, the induced decomposition set $\locc^{\star}(A;B\rightarrow A';B')\coloneqq\qids[\instr{\locc}(A;B\rightarrow A'B')]$ is absolutely convex and can be characterized as
\begin{equation}
  \locc^{\star}(A;B\rightarrow A';B') = \{\cE-\cF | (\cE,\cF)\in\instr{\locc}(A;B\rightarrow A';B') \} \, .
\end{equation}

Under the resource theory of LOCC, the free states on a bipartite system $AB$ are precisely given by $\sep(A;B)$.

\item[Separable channels (SEPC):]
We start by defining the set of separable quantum instruments $\instr{\sepc}(A;B\rightarrow A';B')$ to contain all instruments $(\cE_i)_i\in\qi(AB\rightarrow A'B')$ which have a separable Kraus representation, i.e., for all $i$ we can write
\begin{equation}
  \cE_i (\rho) = \sum_j (M_j\otimes N_j)\rho (M_j^{\dagger}\otimes N_j^{\dagger})
\end{equation}
for some $M_j\in\lino(A,A')$, $N_j\in\lino(B,B')$.
As an alternative characterization, the instruments $\cI=(\cE_i)_i\in\instr{\sepc}$ are precisely those for which the associated Choi representations are separable $(\choi{\cE_i} / \tr[\choi{\cE_i}])\in\sep(AA';BB')$~\cite{cirac2001_entangling}.
Correspondingly, the set of separable channels $\sepc(A;B\rightarrow A';B')$ is defined to be the subset of all quantum channels $\cE\in\cptp(AB\rightarrow A'B')$ that exhibit a separable Choi operator $\choi{\cE}\in\sep(AA';BB')$.
It can also be shown that $\sepc$ precisely constitutes the set of quantum channels which are \emph{completely resource non-generating}, i.e, which cannot generate entanglement even when applied to a subsystem~\cite[Proposition 17]{chitambar2020_entanglement}.

By construction, it can easily be seen that any LOCC instrument is separable, and hence $\instr{\locc}(A;B\rightarrow A';B')\subset\instr{\sepc}(A;B\rightarrow A';B')$.
In fact, it has been shown that $\sepc(A;B\rightarrow A';B')$ is in general strictly larger than $\locc(A;B\rightarrow A';B')$~\cite{bennett1999_quantum}.

It is easy to verify that $\instr{\sepc}$ is coarse-grainable, trivially fine-grainable and closed under mixture.
By~\Cref{lem:qrt_convexity,lem:characterization_expanded_decomposition_set}, the induced decomposition set $\sepc^{\star}(A;B\rightarrow A';B')\coloneqq\qids[\instr{\sepc}(A;B\rightarrow A'B')]$ is absolutely convex and can be characterized as
\begin{equation}\label{eq:char_sepc_star}
  \sepc^{\star}(A;B\rightarrow A';B') = \left\{\cE-\cF \middle\vert
  \begin{array}{l}
    (\cE,\cF)\in\qi(AB\rightarrow A'B'), \\
    \frac{\choi{\cE}}{\tr[\choi{\cE}]},\frac{\choi{\cF}}{\tr[\choi{\cF}]}\in\sep(AA';BB')
  \end{array}
  \right\} \, .
\end{equation}
By the Choi characterization, the free states of this QRT are precisely the separable states $\sep(A;B)$.

\item[Positive partial transpose channels (PPTC):]
SEPC provides a relaxation of the complicated LOCC constraint with a much simpler separability constraint on the Choi matrix.
However, even separability is not a semidefinite constraint, and as such, it is sometimes useful to further relax the problem using the PPT criterion.
We define the PPTC instrument as
\begin{equation}
  \instr{\pptc}(A;B\rightarrow A';B') \coloneqq \{ (\cE_i)_i \in \qi(AB\rightarrow A'B') | \frac{\choi{\cE_i}}{\tr[\choi{\cE_i}]} \in \ppt(AA';BB') \}
\end{equation}
and correspondingly, the set of PPT channels $\pptc(A;B\rightarrow A';B')$ is given by all maps $\cE\in\cptp(AB\rightarrow A'B')$ with PPT Choi representation $\choi{\cE}\in\ppt(AA';BB')$.
Again, it can be shown that PPTC precisely constitutes the set of quantum channels which are completely PPT-preserving, i.e., which cannot create non-PPT states from PPT states, even when only applied to a subsystem~\cite[Proposition 17]{chitambar2020_entanglement}.
The QRT PPTC is sometimes also called the \emph{resource theory of NPT-entanglement}.

It is easy to verify that $\instr{\pptc}$ is coarse-grainable, trivially fine-grainable and closed under mixture.
By~\Cref{lem:qrt_convexity,lem:characterization_expanded_decomposition_set}, the induced decomposition set $\pptc^{\star}(A;B\rightarrow A';B')\coloneqq\qids[\instr{\pptc}(A;B\rightarrow A'B')]$ is absolutely convex and can be characterized as
\begin{equation}\label{eq:char_pptc_star}
  \pptc^{\star}(A;B\rightarrow A';B') = \left\{ \cE-\cF \middle\vert
  \begin{array}{l}
    (\cE,\cF)\in\qi(AB\rightarrow A'B'), \\
    (\idchan_{AA'}\otimes\transpose_{BB'})(\choi{\cE})\loewnergeq 0, \\
    (\idchan_{AA'}\otimes\transpose_{BB'})(\choi{\cF})\loewnergeq 0
  \end{array}
  \right\} \, .
\end{equation}
By the Choi characterization, the free states of this QRT are precisely the PPT states $\ppt(A;B)$.
\end{description}
To simplify the notation, we will omit the output system whenever it is equal to the input system, e.g., $\pptc^{\star}(A;B)\coloneqq \pptc^{\star}(A;B\rightarrow A;B)$ or $\instr{\locc}(A;B)\coloneqq\instr{\locc}(A;B\rightarrow A;B)$.
It is easy to verify that LO, LOCC, SEPC and PPTC are all valid QRTs with tensor product structure.
Thus many properties in~\Cref{sec:qps_basic_properties,sec:gamma_qrt} directly apply.
For instance, the associated quasiprobability extents
\begin{itemize}
  \item fulfill the chaining property (\Cref{lem:gamma_chaining}),
  \item are sub-multiplicative (\Cref{cor:gamma_submult}),
  \item are invariant under local unitaries (\Cref{cor:gamma_invariant_reversible}),
  \item are Lipschitz continuous with $L=\dimension(AB)^{5/2}$ on qubit systems $A=A',B=B'$ (\Cref{ex:modifiedendobasis}).
\end{itemize}
The QRTs LOCC, SEP and PPTC are even convex, so the associated quasiprobability extents are also
\begin{itemize}
  \item directly related to a robustness measure (\Cref{lem:robustness1,lem:robustness2}),
  \item expressible in terms of single-element QPDs (\Cref{lem:one_element_qpd}).
\end{itemize}
Our four QRTs constitute a hierarchy in the sense that
{\small
\begin{equation}\label{eq:qrt_hierarchy}
  \lo(A;B\rightarrow A';B') \subset \locc(A;B\rightarrow A';B') \subset \sepc(A;B\rightarrow A';B') \subset \pptc(A;B\rightarrow A';B')
\end{equation}
}
and correspondingly
{\footnotesize
\begin{equation}\label{eq:qrt_star_hierarchy}
  \lo^{\star}(A;B\rightarrow A';B') \subset \locc^{\star}(A;B\rightarrow A';B') \subset \sepc^{\star}(A;B\rightarrow A';B') \subset \pptc^{\star}(A;B\rightarrow A';B') \, .
\end{equation}
}
In terms of the quasiprobability extent, \Cref{lem:gamma_ds_bound} thus directly implies
\begin{equation}\label{eq:channel_gamma_simple_bounds}
\begin{tikzpicture}
  \def \x{2}
  \def \y{1.3}
  \node at (0,0) {$\gamma_{\mathrm{LO}^{\star}}$}; 
  \node at (0,-\y) {$\gamma_{\mathrm{LO}}$}; 
  \node at (\x,0) {$\gamma_{\mathrm{LOCC}^{\star}}$}; 
  \node at (\x,-\y) {$\gamma_{\mathrm{LOCC}}$}; 
  \node at (2*\x,0) {$\gamma_{\mathrm{SEP}^{\star}}$}; 
  \node at (2*\x,-\y) {$\gamma_{\mathrm{SEP}}$}; 
  \node at (3*\x,0) {$\gamma_{\mathrm{PPT}^{\star}}$}; 
  \node at (3*\x,-\y) {$\gamma_{\mathrm{PPT}}$};  

  \node[rotate=90] at (0,-0.5*\y) {$\geq$};
  \node[rotate=90] at (\x,-0.5*\y) {$\geq$};
  \node[rotate=90] at (2*\x,-0.5*\y) {$\geq$};
  \node[rotate=90] at (3*\x,-0.5*\y) {$\geq$};
  \node at (0.5*\x,0) {$\geq$}; 
  \node at (1.5*\x,-\y) {$\geq$}; 
  \node at (0.5*\x,-\y) {$\geq$}; 
  \node at (1.5*\x,0) {$\geq$};
  \node at (2.5*\x,0) {$\geq$};  
  \node at (2.5*\x,-\y) {$\geq$};
\end{tikzpicture} \, .
\end{equation}
As we proceed further in this chapter, we will obtain a more advanced understanding of these inequalities and whether they are tight or not.
The culminated results will later be summarized in~\Cref{fig:bounds_overview}.
Let us highlight once again that only the two quantities $\gamma_{\lo^{\star}}$ and $\gamma_{\locc^{\star}}$ have practical relevance in the context of circuit knitting.
They precisely characterize the optimal sampling overhead for the two technologically relevant cases when the smaller sub-partitions can exchange classical communication or not.
The main effort is therefore geared towards characterizing these two quantities, and the other ones purely serve as mathematical tools to find bounds.

When specializing from channels to states (i.e. channels with trivial input space), the situation in~\Cref{eq:channel_gamma_simple_bounds} simplifies significantly.
First of all, since LOCC is a resource theory that fulfills the axiom of free instruments, \Cref{lem:negativity_state_qpd} implies that the quasiprobability extent is not impacted by the inclusion of intermediate measurements
\begin{equation}\label{eq:state_qpd_locc}
  \forall\omega\in\herm(AB): \gamma_{\locc^{\star}(\mathbb{C};\mathbb{C}\rightarrow A;B)}(\omega) = \gamma_{\locc(\mathbb{C};\mathbb{C}\rightarrow A;B)}(\omega) = \gamma_{\sep(A;B)}(\omega)
\end{equation}
and we can focus our study purely on the quasiprobability extent w.r.t. the set of separable channels.
In fact, the same also holds for LO and SEPC\footnote{We note that~\Cref{lem:negativity_state_qpd} cannot be invoked for LO, since it doesn't fulfill the axiom of free instruments. Therefore, this statement needs to be re-proven explicitly for LO.}, and an analogous statement holds for PPTC.
\begin{lemma}{}{gammas_overview_nonlocal_state}
  Consider $\omega\in\herm(AB)$.
  If $\ds$ is either $\lo,\lo^{\star},\locc,\locc^{\star},\sepc$ or $\sepc^{\star}$, then
  \begin{equation}
    \gamma_{\ds(\mathbb{C};\mathbb{C}\rightarrow A;B)}(\omega) = \gamma_{\sep(A;B)}(\omega) \, .
  \end{equation}
  If instead $\ds$ is either $\pptc$ or $\pptc^{\star}$, then
  \begin{equation}
    \gamma_{\ds(\mathbb{C};\mathbb{C}\rightarrow A;B)}(\omega) = \gamma_{\ppt(A;B)}(\omega) \, .
  \end{equation}
\end{lemma}
So to summarize, in the case of states \Cref{eq:channel_gamma_simple_bounds} collapses to
\begin{equation}
  \gamma_{\sep} \geq \gamma_{\ppt} \, .
\end{equation}
\begin{proof}
  By~\Cref{eq:channel_gamma_simple_bounds,eq:state_qpd_locc}, it remains to show $\gamma_{\lo}(\omega)\leq\gamma_{\sepc^{\star}}(\omega)$ to prove the first equation.
  For this purpose, consider a QPD $\omega=\sum_i a_i\sigma_i$ with $\sigma_i\in\sepc^{\star}$.
  This means that we can write $\sigma_i = \sum_j \beta_j^{(i)} \tau_j^{(i)}$ where the $\tau_j^{(i)}$ are sub-normalized separable states that add up to a normalized separable state.
  We write $\tau_j^{(i)}=\sum_k p_{ijk}\proj{f_{ijk}}\otimes\proj{g_{ijk}}$ where $\sum_{j,k}p_{ijk}=1$.
  We obtain a valid QPD of $\omega$ w.r.t. $\lo$ (i.e. w.r.t. to separable states)
  \begin{equation}
    \omega = \sum_{i,j,k} a_i\beta_j^{(i)}p_{ijk} \proj{f_{ijk}}\otimes\proj{g_{ijk}}
  \end{equation}
  with 1-norm $\sum_{ijk} \abs{a_i}\abs{\beta_j^{(i)}}p_{ijk} \leq \sum_{ijk} \abs{a_i}p_{ijk} = \norm{a}_1$.

  For the $\pptc$ part, a very similar proof applies.
  Alternatively, it also follows from the more general statement in~\Cref{prop:negtrick_sepc_pptc} that we will prove later.
\end{proof}

\subsection{Utility of classical side information}\label{sec:gate_cut_side_info}
This section will focus on understanding the utility of classical side information from intermediate measurements for the task of QPS of nonlocal operations.
More mathematically, this means that we will compare $\gamma_{\lo^{\star}}$ with $\gamma_{\lo}$, $\gamma_{\locc^{\star}}$ with $\gamma_{\locc}$, $\gamma_{\sepc^{\star}}$ with $\gamma_{\sepc}$ and also $\gamma_{\pptc^{\star}}$ with $\gamma_{\pptc}$.
Recall that in the special case of states, we already saw previously in~\Cref{lem:gammas_overview_nonlocal_state} that intermediate measurements don't impact the extent.
We now study the more general case of superoperators.

In the LO setting, it turns out that classical side information is absolutely crucial to realize space-like cuts.
In fact, the only operations that we can simulate without intermediate measurements are non-signalling channels, which are essentially useless for the purpose of circuit knitting.
\begin{definition}{}{}
  A bipartite superoperator $\cE\in\hp(AB\rightarrow A'B')$ is called non-signalling if $\tr_{B'}[\cE(\rho)] = \tr_{B'}[\cE(\sigma)]$ for all $\rho_{AB},\sigma_{AB}\in\dops(AB)$ with $\tr_B[\rho_{AB}]=\tr_{B}[\sigma_{AB}]$ and vice versa if we exchange of $A\leftrightarrow B$, $A'\leftrightarrow B'$.
  We denote the set of all non-signalling superoperators by $\nonsig(A;B\rightarrow A;B')$.
\end{definition}
The idea is that if $\rho$ and $\sigma$ are indistinguishable to the party $A$, then the application of $\cE$ cannot change the statistics of any measurement that $A$ performs on its side afterwards.
\begin{lemma}{}{span_lo}
  The set of superoperators that exhibit a QPD under $\lo$ and $\lo^{\star}$ are 
  \begin{align}
    \spn_{\mathbb{R}}(\lo(A;B\rightarrow A';B'))
    &= \mathbb{R}\hptp(A\rightarrow A')\otimes \mathbb{R}\hptp(B\rightarrow B') \label{eq:span_lo_eq1}\\
    &= \nonsig(A;B\rightarrow A';B') \label{eq:span_lo_eq2}
  \end{align}
  and
  \begin{equation}\label{eq:span_lo_eq3}
    \spn_{\mathbb{R}}(\lo^{\star}(A;B\rightarrow A';B')) = \hp(AB\rightarrow A'B') \, .
  \end{equation}
\end{lemma}
Recall the definition of $\mathbb{R}\hptp$ in~\Cref{lem:cptp_span}.
Note that the tensor product here should be understood as the \emph{real} tensor product of two subspaces of the \emph{real} vector space of Hermitian-preserving superoperators.
\begin{proof}
  To show \Cref{eq:span_lo_eq1}, notice on one hand that $\lo(A;B\rightarrow A'B')\subset \mathbb{R}\hptp(A\rightarrow A')\otimes \mathbb{R}\hptp(B\rightarrow B')$.
  Conversely, by~\Cref{lem:cptp_span} any pure tensor in $\mathbb{R}\hptp(A\rightarrow A')\otimes \mathbb{R}\hptp(B\rightarrow B')$ can be written as a linear combination of LO channels.
  For the proof of~\Cref{eq:span_lo_eq2} we refer to~\cite[Theorem 13]{gutoski2009_properties} or~\cite[Theorem 5.1]{cavalcanti2022_decomposing}.

  Next, we consider~\Cref{eq:span_lo_eq3}.
  Clearly, $\spn_{\mathbb{R}}(\lo^{\star})\subset\hp$, so it remains to show the converse.
  Because $\hp(AB\rightarrow A'B')=\hp(A\rightarrow A')\otimes\hp(B\rightarrow B')$, the desired statement follows from~\Cref{lem:cptp_span}. 
\end{proof}

In contrast, classical side information is not required in the setting of LOCC to simulate all CPTP maps.
\begin{lemma}{}{span_locc}
  The set of superoperators that exhibit a QPD w.r.t. $\locc$ and $\locc^{\star}$ are 
  \begin{equation}\label{eq:span_locc_eq1}
    \spn_{\mathbb{R}}(\locc(A;B\rightarrow A';B')) = \mathbb{R}\hptp(AB\rightarrow A'B')
  \end{equation}
  and
  \begin{equation}\label{eq:span_locc_eq2}
    \spn_{\mathbb{R}}(\locc^{\star}(A;B\rightarrow A';B')) = \hp(AA'\rightarrow BB') \, .
  \end{equation}
\end{lemma}
Notice that classical side information is strictly necessary to simulate non-trace-preserving maps.
This is unsurprising, and in fact already implied by the result in~\Cref{lem:cptp_span}.
In this context, one can interpret~\Cref{lem:span_locc} as showing that $\locc$ and $\locc^{\star}$ span the ``largest possible'' spaces achievable by decomposition sets without and with side information.
\begin{proof}
  By~\Cref{lem:cptp_span}, we know that the inclusion $\subset$ must hold for~\Cref{eq:span_locc_eq1,eq:span_locc_eq2}.
  The converse side for~\Cref{eq:span_locc_eq2} follows from~\Cref{eq:span_lo_eq3} and $\lo^{\star}\subset\locc^{\star}$.
  It thus remains to show $\supset$ in~\Cref{eq:span_locc_eq1}.

  By~\Cref{lem:cptp_span}, this can be proven by showing that any $\cE\in\cptp(AB\rightarrow A'B')$ lies in $\spn_{\mathbb{R}}(\locc)$.
  We assume without loss of generality that $A=A'$ and $B=B'$ are multi-qubit systems, since any quantum system can be embedded into a large enough multi-qubit system.
  By the Stinespring dilation theorem, we can thus write $\cE$ as a unitary evolution on some enlarged Hilbert space
  \begin{equation}
    \cE(\rho_{AB}) = \tr_E\left[ U_{ABE} (\rho_{AB}\otimes\proj{0}_E) U_{ABE}^{\dagger}\right]
  \end{equation}
  where $U_{ABE}\in\uni(ABE)$.
  It is well-known that any unitary can be decomposed into CNOT gates and single-qubit unitaries~\cite{barenco1995_elementary}.
  We have already seen that CNOT gates can be decomposed into LOCC operations in~\Cref{fig:cnot_teleportation}.
  By decomposing all CNOT gates in the circuit of $U_{ABE}$ that act nonlocally across $A$ and $BE$, we can write $\indsupo{U_{ABE}}$ as a QPD of channels in $\locc(A;BE\rightarrow A;BE)$.
  By tracing out $E$, these channels induce a QPD of $\cE$ into $\locc(A;B\rightarrow A;B)$.
\end{proof}

As a direct consequence of~\Cref{eq:qrt_hierarchy,eq:qrt_star_hierarchy} as well as~\Cref{lem:cptp_span,lem:span_locc}, the analogous result also holds in the SEPC and PPTC settings.
\begin{corollary}{}{}
  The set of superoperators that exhibit a QPD w.r.t. SEPC and PPTC are
  {\small
  \begin{equation}
    \spn_{\mathbb{R}}(\sepc(A;B\rightarrow A';B')) = \spn_{\mathbb{R}}(\pptc(A;B\rightarrow A';B')) = \mathbb{R}\hptp(AB\rightarrow A'B')
  \end{equation}
  }
  and
  {\small
  \begin{equation}
    \spn_{\mathbb{R}}(\sepc^{\star}(A;B\rightarrow A';B')) = \spn_{\mathbb{R}}(\pptc^{\star}(A;B\rightarrow A';B')) = \hp(AA'\rightarrow BB') \, .
  \end{equation}
  }
\end{corollary}

In fact, for this case there is a simple argument that the quasiprobability extent of any quantum channel is identical with and without classical side information.

\begin{proposition}{}{negtrick_sepc_pptc}
  For any $\cE\in\mathbb{R}\hptp(AB\rightarrow A'B')$ one has $\gamma_{\pptc}(\cE)=\gamma_{\pptc^{\star}}(\cE)$ and $\gamma_{\sepc}(\cE)=\gamma_{\sepc^{\star}}(\cE)$.
\end{proposition}
\begin{proof}
  In the following, $\ds$ may denote either $\sepc$ or $\pptc$.
  Clearly, $\gamma_{\ds}\geq\gamma_{\ds^{\star}}$ by~\Cref{lem:gamma_ds_bound}, so we only need to show the converse.
  Consider a general QPD of $\cE$ w.r.t. $\ds^{\star}$.
  Without loss of generality, we can assume it to only have one element by~\Cref{lem:one_element_qpd}, i.e., $\cE=\kappa\cF$, $\cF\in\ds^{\star}$.
  By the characterization of $\ds^{\star}$ (either~\Cref{eq:char_sepc_star} or~\Cref{eq:char_pptc_star}), we can write $\cF=\cF^+ - \cF^-$.
  If either $\cF^+=0$ or $\cF^-=0$, then we already have a QPD into $\ds$.
  Otherwise, observe that both $\choi{\cF^+}+\choi{\cF^-}$ and $\choi{\cF^+}-\choi{\cF^-}$ are in $\mathbb{R}\hptp$.
  Clearly, this is only possible if both $\choi{\cF^+}$ and $\choi{\cF^-}$ are themselves in $\mathbb{R}\hptp$.
  The following is therefore a valid QPD of $\cE$ w.r.t. $\ds$
  \begin{equation}
    \cE = \kappa\tr[\choi{\cF^+}] \frac{\cF^+}{\tr[\choi{\cF^+}]} - \kappa\tr[\choi{\cF^-}] \frac{\cF^-}{\tr[\choi{\cF^-}]}
  \end{equation}
  and the 1-norm of its coefficients is given by $\kappa(\tr[\choi{\cF^+}] + \tr[\choi{\cF^-}]) = \kappa$.

\end{proof}

In summary, we have already established some interesting observations from rather simple mathematical arguments.
While classical side information is absolutely crucial for space-like cuts in the LO setting, it is completely unnecessary for SEPC and PPTC (unless one wants to simulate non-trace-preserving maps).
The setting LOCC lies somewhere in the middle, and we know that intermediate measurements are not required in order to enable QPS of general CPTP maps.
However, it remains an open question whether there exists some CPTP map $\cE$ for which $\gamma_{\locc^{\star}}(\cE)$ is strictly smaller than $\gamma_{\locc}(\cE)$.

\subsection{Quasiprobability extent of states}\label{sec:gamma_nonlocal_states}
Before we delve into the most general setting by studying the quasiprobability extent of arbitrary channels, we first focus on the comparatively easier task of characterizing the quasiprobability extent $\gamma_{\sep(A;B)}(\rho_{AB})$ for bipartite states $\rho_{AB}\in\dops(AB)$ (recall from~\Cref{lem:gammas_overview_nonlocal_state} that $\gamma_{\locc^{\star}}(\rho)=\gamma_{\lo^{\star}}(\rho)=\gamma_{\sep}(\rho)$).
In many ways, SEP is a well-behaved set and many properties from~\Cref{chap:qpsim} are directly applicable to $\gamma_{\sep}$.
For instance, since SEP is convex and closed, we can write its quasiprobability extent as
\begin{equation}
  \gamma_{\sep}(\rho_{AB}) = \min\{a^++a^- | \rho=a^+\sigma^+-a^-\sigma^-, \sigma^{\pm}\in\sep(A;B),a^{\pm}\geq 0\}
\end{equation}
(see~\Cref{lem:two_element_qpd,lem:compact_ds}).
Furthermore, $\gamma_{\sep}$ is invariant under local unitaries (see~\Cref{cor:gamma_invariant_reversible}) and it can be expressed in terms of a robustness measure
\begin{equation}
  \gamma_{\sep}(\rho_{AB}) = 1 + 2R_{\sep}(\rho_{AB})
\end{equation}
(see~\Cref{lem:robustness1}) where
\begin{equation}
   R_{\sep}(\rho_{AB}) \coloneqq \min\{t\geq 0 | \frac{\rho_{AB}+t\sigma_{AB}}{1+t}\in\sep(A;B) \text{ for some } \sigma_{AB}\in\sep(A;B)\}
\end{equation}
The quantity $R_{\sep}$ is commonly referred to as the \emph{robustness of entanglement} and it was first introduced by Vidal and Tarrach in \reference~\cite{vidal1999_robustness}.
The original motivation for the robustness of entanglement was to capture the amount of (separable) noise that a certain entangled state can tolerate before it becomes itself separable.
Since then, it has found applications in various results on entanglement manipulation~\cite{liu2019_oneshot,brandao2010_generalization,brandao2015_reversible,lami2023_computable}.

Furthermore, we note that $\gamma_{\sep}$ is Lipschitz continuous.
\begin{proposition}{}{gamma_sep_lipschitz}
  Let $\omega,\tau\in\herm(AB)$ and $d_A\coloneqq\dimension(A)$, $d_B\coloneqq \dimension(B)$.
  Then
  \begin{equation}
    \forall \omega,\tau\in\herm(AB):
    \abs{\gamma_{\sep}(\omega)-\gamma_{\sep}(\tau)}
    \leq
    2^{\frac{3}{2}\left( \ceil{\log(d_A)} + \ceil{\log(d_B)} \right)} \norm{\omega-\tau}_2  \, .
  \end{equation}
\end{proposition}
\begin{proof}
  We first consider the case where $A$ and $B$ are $n$-qubit and $m$-qubit system.
  By~\Cref{lem:gamma_hulls} we can consider $\gamma_{\aconv(\sep)}$ instead of $\gamma_{\sep}$.
  Denote the $n+m$-qubit Pauli strings by $\{Q_0,\dots,Q_{4^{n+m}-1}\}$.
  Notice that the normalized Pauli matrices $\frac{1}{2^{n+m}}Q_i$ all lie in $\aconv(\sep(A;B))$, since they each have an eigenbasis of product states and eigenvalues $-1$ and $1$.
  Furthermore, they form an orthogonal basis of $\herm(AB)$, so we can apply~\Cref{lem:gamma_lipschitz}.
  The Gram matrix is diagonal and its smallest eigenvalues is $2^{-2(n+m)}$.

  We now turn our attention to the general case.
  Consider embedding $A$ and $B$ into large enough $n$-qubit and $m$-qubit systems $\bar{A}=A\oplus A^{\perp}$ and $\bar{B}=B\oplus B^{\perp}$.
  Define by $\mathrm{Emb}_{A\rightarrow\bar{A}}$ and $\mathrm{Emb}_{B\rightarrow\bar{B}}$ the CPTP maps that naturally embed a state into the larger space.
  Correspondingly, denote by $\mathrm{Proj}_{\bar{A}\rightarrow A}$ and $\mathrm{Proj}_{\bar{B}\rightarrow B}$ the CPTN maps $\rho\mapsto \Pi\rho\Pi$ where $\Pi$ is the projection onto $A$ or $B$ respectively.
  Define the extended state $\omega\in\herm(\bar{A}\bar{B})$ by
  \begin{equation}
    \bar{\omega}\coloneqq (\mathrm{Emb}_{A\rightarrow\bar{A}}\otimes\mathrm{Emb}_{B\rightarrow\bar{B}})(\omega)
  \end{equation}
  which clearly fulfills
  \begin{equation}
    \omega = (\mathrm{Proj}_{\bar{A}\rightarrow A}\otimes\mathrm{Emb}_{\bar{B}\rightarrow B})(\bar{\omega}) \, .
  \end{equation}
  Since both $(\mathrm{Emb}_{A\rightarrow\bar{A}}\otimes\mathrm{Emb}_{B\rightarrow\bar{B}})(\omega)$ and $(\mathrm{Proj}_{\bar{A}\rightarrow A}\otimes\mathrm{Emb}_{\bar{B}\rightarrow B})$ lie in $\locc^{\star}$, the chaining property in~\Cref{cor:gamma_state_from_channel} implies that $\gamma_{\sep}(\omega)=\gamma_{\sep}(\bar{\omega})$.
  Finally, notice that that by appropriately choosing a basis, the embedding does not change the non-zero matrix entries of a matrix. Hence, for two extended states one has
  \begin{equation}
    \norm{\bar{\omega}-\bar{\tau}}_2
    =
    \norm{\omega-\tau}_2 \, .
  \end{equation}
  Hence,
  \begin{align}
    \abs{\gamma_{\sep}(\omega)-\gamma_{\sep}(\tau)}
    &= \abs{\gamma_{\sep}(\bar{\omega})-\gamma_{\sep}(\bar{\tau})} \\
    &\leq 2^{\frac{3}{2}(n+m)}\norm{\bar{\omega}-\bar{\tau}}_2 \\
    &= 2^{\frac{3}{2}(n+m)} \norm{\omega-\tau}_2 \, .
  \end{align}
\end{proof}

While $\gamma_{\sep}$ is the quantity that we aim to characterize due to its operational importance, we will also study the quantity $\gamma_{\ppt}$.
It provides a lower bound to $\gamma_{\sep}$ and it is efficiently computable:
Since $\ppt$ is itself also convex and closed, we can analogously write
\begin{equation}
  \gamma_{\ppt}(\rho_{AB}) = \min\{a^++a^- | \rho=a^+\sigma^+-a^-\sigma^-, \sigma^{\pm}\in\ppt(A;B),a^{\pm}\geq 0\}
\end{equation}
which through a simple substitution can be expressed in the form of a SDP
\begin{equation}\label{eq:gamma_ppt_sdp}
  \gamma_{\ppt}(\rho) = \left \lbrace
  \begin{array}{r l}
    \min\limits_{\sigma^{\pm}\in\pos(AB)} & \tr[\sigma^+] + \tr[\sigma^-] \\
    \textnormal{s.t.} & \rho =\sigma^+ - \sigma^- \\
    & (\idchan_A\otimes\transpose_B)(\sigma^{\pm}) \loewnergeq 0
  \end{array} \right .
\end{equation}
This SDP characterization enables the efficient numerical evaluation of $\gamma_{\ppt}$ and also serves as a valuable analytical tool for proving the optimality of an QPD.
\begin{lemma}{}{gamma_ppt_sdp_dual}
  The following is a valid dual SDP formulation of~\Cref{eq:gamma_ppt_sdp}.
  \begin{equation}
    \gamma_{\ppt}(\rho) = \left \lbrace
    \begin{array}{r l}
      \max\limits_{X\in\herm(AB),Z^{\pm}\in\pos(AB)} & \tr[X\rho] \\
      \textnormal{s.t.} & -\id_{AB} + (\idchan_A\otimes\transpose_B)(Z^-)\loewnerleq X \\
      & X\loewnerleq \id_{AB} - (\idchan_A\otimes\transpose_B)(Z^+)
    \end{array} \right .
  \end{equation}
\end{lemma}
\begin{proof}
  In the following, we will denote the unnormalized Hilbert-Schmidt inner product by $\langle A,B\rangle\coloneqq\tr[A^{\dagger}B]$.
  The Lagrangian is given by
  \begin{align}
    L(\sigma^{\pm};X,Z^{\pm}) &= \tr[\sigma^+] + \tr[\sigma^-] - \langle X, \rho-\sigma^++\sigma^-\rangle \nonumber\\
    &- \langle Z^+,(\idchan_A\otimes\transpose_B)(\sigma^+)\rangle - \langle Z^-,(\idchan_A\otimes\transpose_B)(\sigma^-)\rangle \nonumber\\
    &= -\langle X,\rho\rangle
    + \langle \sigma^+, \id_{AB} + X - (\idchan_A\otimes\transpose_B)(Z^+)) \rangle \nonumber\\
    &+ \langle \sigma^-, \id_{AB} - X - (\idchan_A\otimes\transpose_B)(Z^-)) \rangle \, .
  \end{align}
  The dual function $g(X,Z^{\pm})\coloneqq\inf_{\sigma^{\pm}}L(\sigma^{\pm};X,Z^{\pm})$ can only take values greater than $-\infty$ if $\id_{AB}\pm X-(\idchan_A\otimes\transpose_B)(Z^\pm)\geq 0$.
  We thus get the dual SDP
  \begin{equation}
    \gamma_{\ppt}(\rho) = \left \lbrace
    \begin{array}{r l}
      \max\limits_{X,\in\herm(AB)Z^{\pm}\in\pos(AB)} & -\tr[X\rho] \\
      \textnormal{s.t.} & -\id_{AB} + (\idchan_A\otimes\transpose_B)(Z^+) \loewnerleq X \\
      & X \loewnerleq \id_{AB} - (\idchan_A\otimes\transpose_B)(Z^-) \, .
    \end{array} \right .
  \end{equation}
  A simple variable substitution $X\rightarrow -X$ yields the desired form of the dual SDP.
\end{proof}

The rest of this section is structured as follows.
First, we will consider the quasiprobability extent of pure entangled states.
We will see that in this case, it is possible to find a closed-form analytical expression of $\gamma_{\sep}$ in terms of the Schmidt coefficients.
Then, we will turn our attention to mixed states, which are generally more complicated to handle.
We will deal with the one-shot and asymptotic regimes separately form each other, as they will entail some different techniques.

\subsubsection{Pure states}\label{sec:gamma_nonlocal_pure}
Since the quasiprobability extent w.r.t. $\sep$ is invariant under local unitaries, $\gamma_{\sep}(\proj{\psi}_{AB})$ can clearly only depend on the Schmidt coefficients of a bipartite state $\ket{\psi}_{AB}$.
In fact, $\gamma_{\sep}$ takes a simple form and one can even find an expression for an achieving QPD.
\begin{theorem}{\cite{vidal1999_robustness}}{optimal_qpd_pure}
  Let $\ket{\psi}_{AB}$ be a bipartite pure state with Schmidt decomposition
  \vspace*{-1em}
  \begin{equation}
    \ket{\psi}_{AB} = \sum_{j=1}^r u_j \ket{f_j}_A\otimes \ket{g_j}_B \, .
  \end{equation}
  Then
  \vspace*{-1em}
  \begin{equation}
    \gamma_{\sep}(\proj{\psi}_{AB}) = 2\left(\sum_ju_j\right)^2 - 1
  \end{equation}
  and the QPD $\proj{\psi}_{AB} = a^+\sigma^+ - a^-\sigma^-$ achieves the optimal extent,
  where
  \begin{equation}
    a^- \coloneqq (\sum_ju_j)^2-1 \,,\quad a^+\coloneqq (\sum_ju_j)^2 \, ,
  \end{equation}
  \begin{equation}\label{eq:pure_state_qpd_sigma_plus}
    \sigma^-\coloneqq \frac{1}{(\sum_ju_j)^2-1}\sum_{i\neq j}u_iu_j\proj{f_i}_A\otimes\proj{g_j}_B \, ,
  \end{equation}
  \begin{equation}\label{eq:pure_state_qpd_sigma_min}
    \sigma^+\coloneqq \frac{1}{(\sum_ju_j)^2}\left(\proj{\psi}_{AB} + ((\sum_ju_j)^2-1)\sigma^-\right) \, .
  \end{equation}
\end{theorem}
The proof follows from two observations.
First, it can be shown that $\proj{\psi}=a^+\sigma^+-a^-\sigma^-$ is a valid QPD w.r.t. $\sep$.
Secondly, using the SDP formulation of $\gamma_{\ppt}$, one can verify that this QPD also achieves the quasiprobability extent w.r.t. PPT.
Combining these two insights gives us matching lower and upper bounds for $\gamma_{\sep}(\proj{\psi})$, hence finishing the proof.
In the following, we show both of these steps separately.

We briefly note that the proof of the converse statement differs from the original one introduced in \reference~\cite{vidal1999_robustness}.
Our approach based on convex programming simplifies the proof, and to our knowledge has not been proposed in literature before.
\begin{lemma}{}{pure_state_qpd_valid}
  The states $\sigma^{\pm}$ in~\Cref{eq:pure_state_qpd_sigma_plus,eq:pure_state_qpd_sigma_min} are separable.
\end{lemma}
\begin{proof}
  Clearly, $\sigma^-$ is separable, so we only need to consider $\sigma^+$.
  We will show that
  \begin{equation}
    \sigma^+ = \frac{1}{2^r-1}\sum\limits_{j=1}^{2^r-1} \proj{\phi_j}\otimes \proj{\tau_j}
  \end{equation}
  where 
  \begin{equation}
    \ket{\phi_j} \coloneqq \frac{1}{\sqrt{\sum_l u_l}} \sum\limits_{k=1}^{r}\sqrt{u_k}\exp\left( 2\pi ij \frac{2^{k-1}-1}{2^r-1}\right) \ket{f_k} \, ,
  \end{equation}
   \begin{equation}
    \ket{\tau_j} \coloneqq \frac{1}{\sqrt{\sum_l u_l}} \sum\limits_{k=1}^{r}\sqrt{u_k}\exp\left( -2\pi ij \frac{2^{k-1}-1}{2^r-1}\right) \ket{g_k} \, .
  \end{equation}
  To show this, we consider the matrix elements of $\sum\limits_{j=1}^{2^r-1} \proj{\phi_j}\otimes \proj{\tau_j}$ in the Schmidt basis
  \begin{align}
    & (\bra{f_a}\otimes\bra{g_b})\sum\limits_{j=1}^{2^r-1} \proj{\phi_j}\otimes \proj{\tau_j}(\ket{f_c}\otimes\ket{g_d}) \nonumber\\
    & \quad = \frac{\sqrt{u_au_bu_cu_d}}{(\sum_lu_l)^2} \sum_{j=1}^{2^r-1} \exp\left( \frac{2\pi i j}{2^r-1} (2^{a-1}+2^{d-1}-2^{b-1}-2^{c-1})\right) \, .
  \end{align}
  Notice that $\abs{2^{a-1}+2^{d-1}-2^{b-1}-2^{c-1}}$ is always smaller than $2^{r}-1$.
  Due to a standard geometric series argument, the matrix elements are zero unless $2^{a-1}+2^{d-1}-2^{b-1}-2^{c-1}=0$.
  This in turn can only occur when either $a=b$ and $c=d$ or otherwise when $a=c$ and $b=d$.
  In summary, the only non-vanishing matrix elements are
  \begin{align}
    &(\bra{f_a}\otimes\bra{g_a})\sum\limits_{j=1}^{2^r-1} \proj{\phi_j}\otimes \proj{\tau_j}(\ket{f_b}\otimes\ket{g_b}) \\
    &= (\bra{f_a}\otimes\bra{g_b})\sum\limits_{j=1}^{2^r-1} \proj{\phi_j}\otimes \proj{\tau_j}(\ket{f_a}\otimes\ket{g_b}) \\
    &= \frac{u_au_b}{(\sum_lu_l)^2}
  \end{align}
  which corresponds precisely to the matrix representation of $\sigma^+$.
\end{proof}
\begin{lemma}{}{optimal_qpd_pure_ppt}
  Let $\ket{\psi}_{AB}$ be a bipartite pure state with Schmidt coefficients $(u_j)_j$.
  Then
  \begin{equation}
    \gamma_{\ppt}(\proj{\psi}_{AB}) = 2\left(\sum_ju_j\right)^2 - 1
  \end{equation}
  and the QPD $\rho=a^+\sigma^+-a^-\sigma^-$ in~\Cref{thm:optimal_qpd_pure} achieves $\gamma_{\ppt}(\proj{\psi}_{AB})$.
\end{lemma}{}{}
\begin{proof}
  Since the QPD $\rho=a^+\sigma^+-a^-\sigma^-$ is a valid QPD w.r.t. $\sep$ (see~\Cref{lem:pure_state_qpd_valid}), it is automatically also a valid QPD w.r.t. $\ppt$.
  It remains to show its optimality.
  We accomplish this by finding a feasible solution of the dual SDP given in~\Cref{lem:gamma_ppt_sdp_dual} which achieves the same value of its objective.

  With a slight abuse of notation, we write the Schmidt decomposition of $\ket{\psi}_{AB}$ as
  \begin{equation}
    \ket{\psi}_{AB} = \sum_{j=1}^r u_j \ket{j}_A\otimes \ket{j}_B \, .
  \end{equation}
  where $\ket{i}_A$ and $\ket{i}_B$ are orthonormal basis states on $A$ and $B$ respectively.
  We now pick
  \begin{equation}
    X = 2\sum_{i,j}\ket{ii}\bra{jj}_{AB} - \id_{AB}
  \end{equation}
  \begin{equation}
    Z^- = 0
  \end{equation}
  and
  \begin{equation}
    Z^+ = \sum_{i<j}(\ket{ij}_{AB}-\ket{ji}_{AB})(\bra{ij}_{AB}-\bra{ji}_{AB}) \, .
  \end{equation}
  Clearly, $Z^-\loewnergeq$, $Z^+\loewnergeq 0$ and $X\loewnergeq -\id_{AB}$.
  The final constraint can be shown as follows
  \begin{align}
    \id_{AB} - (\idchan_A\otimes\transpose_B)(Z^+) &= \id_{AB} - \sum_{i<j}(\idchan_A\otimes\transpose_B)\left((\ket{ij}-\ket{ji})(\bra{ij}-\bra{ji})\right)\\
    &= \id_{AB} - \sum_{i<j}\left( \proj{ij} + \proj{ji} - \ket{ii}\bra{jj} - \ket{jj}\bra{ii} \right) \\
    &= \id_{AB} - 2\sum_{i\neq j}\left( \proj{ij} - \ket{ii}\bra{jj} \right) \\
    &= \id_{AB} - 2\id_{AB} + 2\sum_i \proj{ii} + 2\sum_{i\neq j} \ket{ii}\bra{jj} \\
    &= -\id_{AB} + 2\sum_{i,j}\ket{ii}\bra{jj} \\
    &= X
  \end{align}
  It remains to evaluate the objective function for this feasible point:
  \begin{align}
    \tr[X\rho] &= 2\sum_{i,j}\tr[\ket{ii}\bra{jj}_{AB} \rho] - \tr[\rho] \\
    & = 2\sum_{i,j}\bra{ii}\rho\ket{jj} - 1 \\
    & = 2\sum_{i,j,k,l}u_ku_l\bra{ii}\ket{kk}\bra{ll}\ket{jj} - 1 \\
    & = 2\left( \sum_iu_i \right)^2 - 1 \, .
  \end{align}
\end{proof}

An immediate consequence of~\Cref{thm:optimal_qpd_pure} is that we can also evaluate the regularized quasiprobability extent for pure states.
\begin{corollary}{}{gamma_reg_pure}
  Let $\ket{\psi}_{AB}$ be a bipartite pure state on $AB$ with Schmidt coefficients $(u_j)_{j=1,\dots,r}$.
  Then
  \begin{equation}
    \gammareg_{\sep}(\proj{\psi}_{AB}) = 1 + R_{\sep}(\proj{\psi}_{AB}) = \left(\sum_j u_j\right)^2 = 2^{H_{1/2}(\tr_B[\proj{\psi}_{AB}])}
  \end{equation}
  where $H_{1/2}$ is the quantum Rényi-$\nicefrac{1}{2}$ entropy.
\end{corollary}
Notice that this regularized quasiprobability extent has halved distance to $1$ compared to $\gamma_{\sep}(\proj{\psi})=1+2R_{\sep}(\proj{\psi})$.
The quantum Rényi entropy of order $\alpha\in (0,1)\cup(1,\infty)$ is defined as
\begin{equation}
  H_{\alpha}(\rho) \coloneqq \frac{1}{1-\alpha} \log \tr[\rho^{\alpha}]  \, .
\end{equation}
The specialization $H_{1/2}$ is commonly referred to as the \emph{max-Entropy} (also denoted $H_{\mathrm{max}}$), though older literature used to refer to $H_{0}$ by this name.
\begin{proof}
  The Schmidt coefficients of $\ket{\psi}^{\otimes n}$ are given by the collection $\left( \prod\limits_{i=1}^n u_{k_i} \right)_{k\in\mathbb{F}_r^n}$ where $\mathbb{F}_r\coloneqq \{1,\dots,r\}$.
  By~\Cref{thm:optimal_qpd_pure} the quasiprobability extent is thus
  \begin{align}
    \gamma_{\sep}(\proj{\psi}^{\otimes n}) &= 
    2\left(\sum\limits_{k\in\mathbb{F}_r^n} \prod\limits_{i=1}^n u_{k_i}\right)^2 - 1 \\
    &= 2\left(\sum\limits_{i=1}^ru_i \right)^{2n} - 1 \, .
  \end{align}
  Observing that $\sum_iu_i\geq \sqrt{\sum_iu_i^2} = 1$, we see that taking the limit $\lim\limits_{n\rightarrow\infty}(\cdot ^{1/n})$ on both sides precisely gives us the desired statement.
\end{proof}

\begin{example}\label{ex:maximally_entangled_state}
  Consider a two-qudit system $AB$, $\dimension(A)=\dimension(B)=d$ and the maximally entangled state
  \begin{equation}
    \ket{\Psi} = \frac{1}{\sqrt{d}}\sum_{i=1}^d \ket{i}_A\otimes \ket{i}_B
  \end{equation}
  for some choice of orthonormal bases on $A$ and $B$.
  Then
  \begin{equation}
    \gamma_{\sep}(\proj{\Psi}) = 2 \left(\sum_i \frac{1}{\sqrt{d}}\right)^2 - 1 = 2d-1
  \end{equation}
  and
  \begin{equation}
    \gammareg_{\sep}(\proj{\Psi}) = \lim\limits_{n\rightarrow\infty} (2d^n-1)^{1/n} = d \, .
  \end{equation}
  Since any state in $\dops(AB)$ can be reached by applying some $\locc$ protocol on $\proj{\Psi}$, $\gamma_{\sep}(\proj{\Psi})=2^d-1$ is the maximal value that the quasiprobability $\gamma_{\sep}$ takes over all states in $\dops(AB)$ (recall~\Cref{cor:gamma_state_from_channel}).

  For $d=2$, the following QPD achieves the quasiprobability extent $\gamma_{\sep}(\proj{\Psi})=3$
  \begin{align}
    \proj{\Psi} = \frac{1}{2}\big(
      &\proj{0}\otimes\proj{0}
      +\proj{1}\otimes\proj{1} \nonumber\\
      &-\proj{+}\otimes\proj{-}
      -\proj{-}\otimes\proj{+} \nonumber\\
      &+\proj{i+}\otimes\proj{i-}
      +\proj{i-}\otimes\proj{i+}
    \big)
  \end{align}
  where $\ket{\pm}\coloneqq(\ket{0}\pm\ket{1})/\sqrt{2}$ and $\ket{i\pm}\coloneqq(\ket{0}\pm i\ket{1})/\sqrt{2}$.
\end{example}

\subsubsection{Mixed states in the single-shot regime}\label{sec:gamma_sep_singleshot}
Previously, we have observed that the quasiprobability extent of an arbitrary pure state w.r.t. $\sep$ can be characterized by a simple formula involving the corresponding Schmidt coefficients.
Unfortunately, it is highly unlikely that an analogous result exists for general mixed states.
It is known that detecting entanglement is generally NP-hard~\cite{gurvits2003_classical,gharibian2010_strong}, so the computation of any faithful entanglement measure (such as the $\gamma_{\sep}$, recall~\Cref{lem:gamma_faithful2}) must be at least as difficult.
As such, we are fundamentally restricted to finding computationally inefficient characterizations of the quasiprobability extent or computationally efficient lower and upper bounds.

We have already seen that $\gamma_{\ppt}$ is one such lower bound which reaches the optimal value for any pure state.
As a direct consequence of~\Cref{lem:ppt_small_system}, it is also optimal for small systems.
\begin{corollary}{}{gamma_mixed_small_sys}
  Let $\dimension(A)\dimension(B)\leq 6$ and $\omega\in\herm(AB)$.
  Then
  \begin{equation}
    \gamma_{\sep}(\omega) = \gamma_{\ppt}(\omega) \, .
  \end{equation}
\end{corollary}
\begin{example}
  Consider the depolarized Bell state across two qubits $A$ and $B$
  \begin{equation}
    \rho_p \coloneqq (1-p) \proj{\Psi}_{AB} + p \frac{1}{4}\id_{AB} 
  \end{equation}
  where $\ket{\Psi}_{AB}=(\ket{00}+\ket{11})/\sqrt{2}$ and $p\in[0,1]$.
  Because of~\Cref{cor:gamma_mixed_small_sys}, we can numerically evaluate $\gamma_{\sep}(\rho_p)$.
  Up to numerical precision, we observe that
  \begin{equation}
    \gamma_{\sep}(\rho_p) = \max(1, 3(1-p)) \, .
  \end{equation}
\end{example}
However, there are states for which the lower bound does not achieve $\gamma_{\sep}$.
This is exemplified by any entangled PPT state $\rho$, such as the one introduced in~\Cref{ex:entangled_ppt_state}.
Clearly, one has $\gamma_{\ppt}(\rho)=1<\gamma_{\sep}(\rho)$ in this situation.

Conversely, to find upper bounds to $\gamma_{\sep}$, one approach is to simply guess a finite decomposition set $\ds\subset \sep(A;B)$ and use the linear programming approach from~\Cref{sec:qps_conv_programming} to evaluate $\gamma_{\ds}(\rho)$.
By~\Cref{lem:gamma_ds_bound}, $\gamma_{\ds}(\rho)\geq\gamma_{\sep}(\rho)$.
\begin{example}
  Consider the maximally entangled state $\ket{\Psi}_{AB}$ across two systems $A,B$ which are each two-qubit systems.
  We know by~\Cref{thm:optimal_qpd_pure} that $\gamma_{\sep}(\proj{\Psi})=7$.
  We now consider the depolarized version of $\proj{\Psi}$ with strength $p$
  \begin{equation}
    \rho_p \coloneqq (1-p) \proj{\Psi}_{AB} + p \frac{1}{16}\id_{AB} \, .
  \end{equation}
  To obtain an upper bound to the quasiprobability extent, we choose the decomposition basis $\ds$ which contains all product states $\sigma\otimes\tau$ where $\sigma$ and $\tau$ are either one of the four $2$-qubit maximally entangled states, or themselves product states of $\ket{0},\ket{1},\ket{+},\ket{-},\ket{i+}$ or $\ket{i-}$.
  The numerical evaluation of $\gamma_{\ds}(\rho_p)$ and $\gamma_{\ppt}(\rho_p)$ is depicted in~\Cref{fig:two_noisy_bell}.
  Up to numerical precision, we observe
  \begin{equation}
    \gamma_{\ppt}(\rho_p) = \max(1, 7-\frac{15}{2}p) \text{ and } \gamma_{\ds}(\rho_p) = \max(1, 7(1-p)) \, .
  \end{equation}
\end{example}
\begin{figure}
  \centering
  \includegraphics{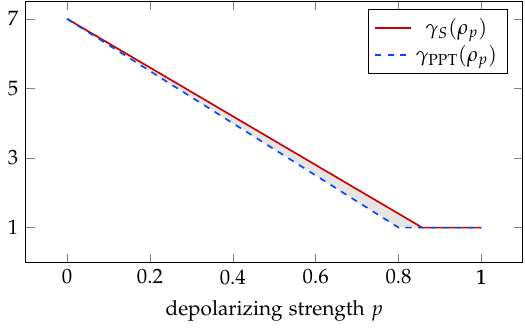}
  \caption{Numerical evaluation of $\gamma_{\ppt}$ and $\gamma_{\ds}$ on a depolarized maximally entangled $4$-qubit state $\rho_p$. The finite decomposition set $\ds$ is heuristically chosen. $\gamma_{\sep}$ must lie in the shaded region.}
  \label{fig:two_noisy_bell}
\end{figure}

In the following, we will introduce a more principled approach to characterize the quasiprobability extent $\gamma_{\sep}$ by expressing it in terms of a converging hierarchy of SDPs.
This approach is based on the well-known Doherty-Parrilo-Spedalieri (DPS) hierarchy of approximations for the set of separable states~\cite{doherty2004_complete}.
We briefly summarize its most important ideas below.

The main workhorse of the DPS hierarchy is the notion of extendibility.
Let $k\in\mathbb{N}$ and $\rho_{AB}\in\dops(AB)$ be a bipartite state.
Let $B_1,\dots,B_k$ be identical copies of the system $B$.
We say $\rho_{AB}$ is \emph{$k$-extendible} w.r.t. $B$ if there exists a state $\tilde{\rho}_{AB_1\dots B_k}\in\dops(AB_1\dots B_k)$
which fulfills following properties:
\begin{itemize}
  \item $\tilde{\rho}$ is invariant under permutations of the subsystems $B_i$, i.e.
  \begin{equation}
    \forall i\neq j: \swap_{B_i,B_j} \tilde{\rho} \swap_{B_i,B_j} = \tilde{\rho}
  \end{equation}
  where for any states $\ket{\psi},\ket{\phi}$ the operation $\swap_{B_i,B_j}$ acts as
  \begin{equation}
    \swap_{B_i,B_j} \ket{\psi}_{B_i}\otimes \ket{\phi}_{B_j} = \ket{\phi}_{B_i} \otimes \ket{\psi}_{B_j} \, .
  \end{equation}
  \item $\rho_{AB}=\tilde{\rho}_{AB_1}\coloneqq \tr_{B_2\dots B_k}[\tilde{\rho}_{AB_1\dots B_k}]$ where the equality should be understood by identifying $B$ and $B_1$.
\end{itemize}
Furthermore, we say $\rho_{AB}$ is \emph{$k$-PPT-extendible} w.r.t. $B$ if additionally 
\begin{itemize}
  \item $\tilde{\rho}$ remains positive under any partial transposition w.r.t. to its $n+1$ sub-systems.
  Because of the permutation invariance, this amounts to
  \begin{equation}
    \left(\idchan_A\otimes \transpose_{B_1} \otimes \idchan_{B_2}\otimes \idchan_{B_3}\otimes \dots \otimes \idchan_{B_k}\right)(\tilde{\rho}) \loewnergeq 0
  \end{equation}
  \begin{equation}
    \left(\idchan_A\otimes \transpose_{B_1} \otimes \transpose_{B_2}\otimes \idchan_{B_3}\otimes \dots \otimes \idchan_{B_k}\right)(\tilde{\rho}) \loewnergeq 0
  \end{equation}
  \begin{center}
    $\vdots$
  \end{center}
  \begin{equation}
    \left(\idchan_A\otimes \transpose_{B_1} \otimes \transpose_{B_2}\otimes \transpose_{B_3}\otimes \dots \otimes \transpose_{B_k}\right)(\tilde{\rho}) \loewnergeq 0
  \end{equation}
\end{itemize}
For $k\geq 1$ we define
\begin{equation}
   \sep_k(A;B) \coloneqq \{\rho\in\dops(AB) | \rho \text{ is $k$-PPT-extendible w.r.t. $B$ } \} \, .
\end{equation}
Evidently, we have $\sep_1(A;B)=\ppt(A;B)$.
It can also be verified that for any $k\geq 1$ one has $\sep_{k+1}\subset \sep_k$.
We leave the proof as an exercise for the reader and refer to \reference~\cite{doherty2004_complete} for more details.
Clearly, any separable state $\sum_i p_i\rho_i\otimes\sigma_i$ is $k$-PPT-extendible for any $k\in\mathbb{N}$ as we can pick the extension $\sum_ip_i\rho_i\otimes\sigma_i^{\otimes k}$.
Furthermore, it turns out that 
\begin{equation}
  \sep(A;B) = \bigcap\limits_{k=1}^{\infty} \sep_k(A;B) \, .
\end{equation}
This is a consequence of the completeness of the k-extendibility criterion for separability.
\begin{proposition}{\cite[Theorem 1]{doherty2004_complete}}{dps_completeness}
  A bipartite state $\rho\in\dops(AB)$ is separable if and only if it is $k$-extendible w.r.t. $B$ for all $k\in\mathbb{N}$.
\end{proposition}
The proof is based on techniques related to the quantum de Finetti theorem~\cite{caves2002_unknown}.
In fact, quantum de Finetti theorems can even provide an upper bound on the convergence speed of the hierarchy.
\begin{proposition}{\cite[Theorem II.7']{christandl2007_definetti}}{dps_convergence}
  Let $\rho\in\dops(AB)$ be a $k$-extendible state w.r.t. $B$.
  There exists a state $\sigma\in\sep(A;B)$ such that
  \begin{equation}
    \norm{\rho - \sigma}_1 \leq \frac{2\dimension(B)^2}{k}\, .
  \end{equation}
\end{proposition}
We note that the PPT constraint in our definition of $\sep_k$ is clearly not necessary for the convergence of the hierarchy.
The PPT constraint makes the hierarchy strictly tighter and as such is often included to get better results for small $k$.
We also note that there is some variation in literature w.r.t. how to the PPT constraint on the extension is defined.
For instance, in \reference~\cite{navascues2009_complete,navascues2009_power} the PPT criterion is not applied to every possible bipartition of $AB_1\dots B_k$, but instead only to one.
In this thesis, we chose to follow the convention from~\cite{doherty2004_complete}.

By definition, the inclusion in the set $\sep_k(A;B)$ exhibits a semidefinite representation.
As such, we can write $\gamma_{\sep_k}(\rho)$ of a state $\rho\in\dops(AB)$ as an SDP
\begin{equation}
  \gamma_{\sep_k}(\rho) = \left \lbrace
  \begin{array}{r l}
    \min\limits_{\sigma^{\pm}\in\pos(AB_1\dots B_k)} & \tr[\sigma^+] + \tr[\sigma^-] \\
    \textnormal{s.t.} & \rho = \tr_{B_2\dots B_k}[\sigma^+] - \tr_{B_2\dots B_k}[\sigma^-] , \\
    & \swap_{B_i,B_j} \sigma^{\pm} \swap_{B_i,B_j} = \sigma^{\pm} , \\
    & \left(\idchan_A\otimes \transpose_{B_1} \otimes \dots \otimes \idchan_{B_k}\right)(\sigma^{\pm}) \loewnergeq 0 , \\
    & \quad\quad \vdots \\
    & \left(\idchan_A\otimes \transpose_{B_1} \otimes \dots \otimes \transpose_{B_k}\right)(\sigma^{\pm}) \loewnergeq 0
  \end{array} \right . 
\end{equation}
where we used the convexity of $\sep_k$, \Cref{lem:two_element_qpd} and the same substitution as for~\Cref{eq:gamma_ppt_sdp}.
We highlight again that the first level of this SDP hierarchy precisely retrieves the $\ppt$ lower bound, i.e., $\gamma_{\sep_1}(\rho)=\gamma_{\ppt}(\rho)$.
As such, the SDP hierarchy can be considered a refinement of the $\ppt$ lower bound, that provably converges in the limit of large $k$.

\begin{theorem}{}{}
  For a bipartite operator $\rho\in\dops(AB)$ the quasiprobability extent $\gamma_{\sep_k}(\rho)$ monotonically approaches $\gamma_{\sep}(\rho)$ from below.
  Furthermore, the convergence speed is bounded by 
  \begin{equation}
    \gamma_{\sep}(\rho) - \gamma_{\sep_k}(\rho) \leq \frac{64\dimension(A)^{3/2}\dimension(B)^{7/2}\min(\dimension(A),\dimension(B))}{k} \, .
  \end{equation}
\end{theorem}
\begin{proof}
  The set $\sep_k$ is convex and compact (because it is the image of a compact set under the partial trace, which is continuous).
  Thus, by~\Cref{lem:two_element_qpd,lem:compact_ds}, the quasiprobability extent of $\rho$ w.r.t. $\sep_k(A;B)$ is achieved by some QPD of the form
  \begin{equation}
    \rho = (1+t)\rho^+ - t\rho^-
  \end{equation}
  for some $\rho^{\pm}\in\sep_k(A;B)$.
  By~\Cref{prop:dps_convergence}, there exist $\sigma^+,\sigma^-\in\sep(A;B)$ such that $\norm{\rho^{\pm}-\sigma^{\pm}}_1\leq 2\dimension(B)^2/k$.
  We can write
  \begin{equation}
    \rho - \Delta = (1+t)\sigma^+ - t\sigma^-
  \end{equation}
  where $\Delta\coloneqq (1+t)(\rho^+-\sigma^+) - t(\rho^--\sigma^-)$.
  Using the Lipschitz continuity of the quasiprobability extent (see~\Cref{prop:gamma_sep_lipschitz}), we have
  \begin{align}
    \gamma_{\sep}(\rho) &\leq \gamma_{\sep}(\rho-\Delta) + L\norm{\Delta}_2 \\
    &\leq (1+2t) + L\norm{\Delta}_2 \\
    &= \gamma_{\sep_k}(\rho) + L\norm{\Delta}_2
  \end{align}
  where $L\leq 2^{\frac{3}{2}(\log\dimension(A) + \log\dimension(B) + 2)} = 8(\dimension(A)\dimension(B))^{3/2}$.
  As such, the desired statement follows from
  \begin{align}
    \norm{\Delta}_2 
    &\leq \norm{\Delta}_1 \\
    &\leq (1+2t)\left(\norm{\rho^+-\sigma^+}_1 + \norm{\rho^--\sigma^-}_1\right) \\
    &\leq \gamma_{\sep_k}(\rho)\frac{4\dimension(B)^2}{k} \\
    &\leq \gamma_{\sep}(\rho)\frac{4\dimension(B)^2}{k} \\
    &\leq 2\min(\dimension(A),\dimension(B))\frac{4\dimension(B)^2}{k}
  \end{align}
  where we used that the maximal quasiprobability extent $\gamma_{\sep}$ is achieved by the maximally entangled state across $A$ and $B$ (recall~\Cref{ex:maximally_entangled_state}).
\end{proof}

We briefly illustrate an example where the SDP hierarchy provides an improvement over the PPT lower bound.
\begin{example}
  We have already seen previously that entangled PPT states provide an example where $\gamma_{\sep}$ is strictly larger than $\gamma_{\ppt}$.
  Taking $\rho$ to be the two-qutrit state from~\Cref{ex:entangled_ppt_state}, we numerically evaluate the SDP hierarchy $\gamma_{\sep_k}(\rho)$ for the first levels $k=1,2,3,4$.
  \begin{center}
  \begin{tabular}{r|cccc}
  & $k=1$ & $k=2$ & $k=3$ & $k=4$ \\
  \hline
  $\gamma_{\sep_k}(\rho)$ & $1.000000$ & $1.111037$ & $1.111037$ & $1.111037$
  \end{tabular}
  \end{center}
  As expected, the first level recovers the trivial bound $\gamma_{\ppt}(\rho)=1$.
  Interestingly, the hierarchy seemingly converges already at the second level.
\end{example}

\subsubsection{Mixed states in the asymptotic regime}\label{sec:gamma_sep_asymptotic}
All previously introduced techniques to find bounds for $\gamma_{\sep}$ rely on linear or semidefinite programming methods where the runtime scales polynomially with the dimension of the involved systems.
Clearly, these approaches are ineffective in computing the regularized and asymptotic quasiprobability extent $\gammareg_{\sep}, \gammasreg_{\sep}$ as the involved system size tends to infinity.
In the following, we explore some more analytical strategies to address this issue.

A natural starting point for such an endeavor are conversion bounds based on the so-called \emph{entanglement cost} and \emph{distillable entanglement}.
\begin{definition}{}{cost_and_distillable}
  Let $\epsilon \geq 0$ and $\rho\in\dops(AB)$. The \emph{single-shot entanglement cost} and \emph{single-shot distillable entanglement} of $\rho$ w.r.t. some set $\ds\subset\cptp$ are defined as
  \begin{equation}
    E_{C,\ds}^{(1),\epsilon}(\rho)\coloneqq \min\limits_{m\in\mathbb{N}} \{\log m | \exists\cE\in\ds : \frac{1}{2}\norm{\rho-\cE(\proj{\Psi_m})}_1 \leq \epsilon \}
  \end{equation}
  \begin{equation}
    E_{D,\ds}^{(1),\epsilon}(\rho)\coloneqq \max\limits_{m\in\mathbb{N}} \{\log m | \exists\cE\in\ds : \frac{1}{2}\norm{\proj{\Psi_m} - \cE(\rho)}_1 \leq \epsilon \}
  \end{equation}
  where $\ket{\Psi_m}$ denotes the maximally entangled state with local dimension $m$.
\end{definition}
Recall from~\Cref{cor:gamma_state_from_channel} that if two bipartite states $\rho,\sigma\in\dops(AB)$ are related by an LOCC map $\cE$ in the sense that $\rho=\cE(\sigma)$, then we know that $\gamma_{\sep}(\rho)\leq \gamma_{\sep}(\sigma)$.
The quantities $E_{C,\locc}^{\epsilon}(\rho)$ and $E_{D,\locc}^{\epsilon}(\rho)$ characterize (up to an error $\epsilon$) how many Bell pairs are required to create $\rho$, respectively, how many Bell pairs can be extracted from $\rho$.
Combined with the aforementioned state conversion bound, this yields following inequality.
\begin{lemma}{}{gamma_single_shot_bound}
  Let $\epsilon\geq 0$. For any bipartite state $\rho\in\dops(AB)$ one has
  \begin{equation}
    2^{E^{(1),\epsilon}_{D,\locc}(\rho)+1}-1 - 4L\epsilon\leq \gamma^{\epsilon}_{\sep}(\rho) \leq 2^{E^{(1),\epsilon}_{C,\locc}(\rho) + 1} - 1 \, .
  \end{equation}
  where $L$ is the Lipschitz constant of $\gamma_{\sep}$ from~\Cref{prop:gamma_sep_lipschitz}.
\end{lemma}
Recall the definition of the smoothed quasiprobability extent in~\Cref{def:smoothed_gamma}.
\begin{proof}
  Let us denote $c\coloneqq E_{C,\ds}^{(1),\epsilon}(\rho)$ and $d\coloneqq E_{D,\ds}^{(1),\epsilon}(\rho)$.
  There exist LOCC maps $\cE$ and $\cF$ as well as states $\rho',\rho''$ such that
  \begin{equation}
    \cE(\rho) = \rho' \quad \text{ and } \quad \rho'\approx_{\epsilon}\proj{\Psi_d} \, ,
  \end{equation}
  \begin{equation}
    \cF(\proj{\Psi_c}) = \rho'' \quad \text{ and } \quad \rho''\approx_{\epsilon}\rho \, .
  \end{equation}
  On one hand, we have
  \begin{equation}
    \gamma_{\sep}^{\epsilon}(\rho) \leq \gamma_{\sep}(\rho'') \leq \gamma_{\sep}(\proj{\Psi_c}) \, ,
  \end{equation}
  and conversely,
  \begin{align}
    \gamma_{\sep}^{\epsilon}(\rho)
    &\geq \gamma_{\sep}^{\epsilon}(\rho') \\
    &= \gamma_{\sep}(\tau) \\
    &\geq \gamma_{\sep}(\proj{\Psi_d}) - L\norm{\tau - \proj{\Psi_d}}_2 \\
    &\geq \gamma_{\sep}(\proj{\Psi_d}) - L(\norm{\tau-\rho'}_1 + \norm{\rho' - \proj{\Psi_d}}_1) \\
    &\geq \gamma_{\sep}(\proj{\Psi_d}) - 4L\epsilon \\
  \end{align}
  where $\tau$ is defined to be the achieving state of $\gamma_{\sep}^{\epsilon}(\rho')$.
  We used $\norm{\cE(\rho)-\cE(\sigma)}_1 \leq \norm{\rho-\sigma}_1$ and the Lipschitz continuity of $\gamma_{\sep}$.
  
\end{proof}

One of the most central quantities in entanglement theory are the optimal asymptotic conversion rates, which characterize the optimal number of Bell pairs that a state can be converted to/from in the asymptotic regime.
\begin{definition}{}{}
  The \emph{asymptotic entanglement cost} and \emph{asymptotic distillable entanglement} of a bipartite state $\rho\in\dops(AB)$ w.r.t. $\ds\subset\cptp$ are defined as
  \begin{equation}
    E_{C,\ds}(\rho) \coloneqq \lim\limits_{\epsilon\rightarrow 0^+} \liminf\limits_{n\rightarrow\infty}\frac{1}{n}E_{C,\ds}^{(1),\epsilon}(\rho^{\otimes n})
    \quad\text{and}\quad
    E_{D,\ds}(\rho) \coloneqq \lim\limits_{\epsilon\rightarrow 0^+} \limsup\limits_{n\rightarrow\infty}\frac{1}{n}E_{D,\ds}^{(1),\epsilon}(\rho^{\otimes n}) \, .
  \end{equation}
\end{definition}
Notice that in the above definitions, we allow for approximate state conversion as long as the approximation error can be made to vanish asymptotically.
The reasoning here is that exact conversion might not be possible in the first place.
If we do not allow for such an approximation error, we obtain a slightly different quantity, which is arguably less interesting from an operational point of view.
\begin{definition}{}{}
  The \emph{exact} asymptotic entanglement cost and \emph{exact} asymptotic distillable entanglement of a bipartite state $\rho\in\dops(AB)$ w.r.t. $\ds\subset\cptp$ are defined as
  \begin{equation}
    E_{C,\ds}^{\mathrm{exact}}(\rho) \coloneqq \lim\limits_{n\rightarrow\infty}\frac{1}{n}E_{C,\ds}^{(1),0}(\rho^{\otimes n})
    \quad\text{and}\quad
    E_{D,\ds}^{\mathrm{exact}}(\rho) \coloneqq \limsup\limits_{n\rightarrow\infty}\frac{1}{n}E_{D,\ds}^{(1),0}(\rho^{\otimes n}) \, .
  \end{equation}
\end{definition}

By applying the proper regularization on both of the inequalities in~\Cref{lem:gamma_single_shot_bound}, we obtain the bounds
\begin{equation}\label{eq:gammareg_sep_bound1}
  E^{\mathrm{exact}}_{D,\locc }(\rho)\leq \log\gammareg_{\sep}(\rho) \leq E^{\mathrm{exact}}_{C,\locc}(\rho)
\end{equation}
and
\begin{equation}\label{eq:gammareg_sep_bound2}
  \log\gammasreg_{\sep}(\rho) \leq E_{C,\locc}(\rho) \, .
\end{equation}
Note that the Lipschitz constant $L$ grows too quickly in the system size to allow for a lower bound of $\log\gammasreg_{\sep}(\rho)$ by the asymptotic distillable entanglement.
We will see later that this shortcoming can be fixed.

For arbitrary mixed states, the asymptotic entanglement cost can generally be extremely hard to evaluate.
It can be written as the regularized \emph{entanglement of formation} \smash{$E_{C,\locc}(\rho)=\lim\limits_{n\rightarrow\infty}\frac{1}{n}E_F(\rho^{\otimes n})$}~\cite{hayden2001_asymptotic}.
We omit the exact definition of $E_F$, but only comment that it can be readily be verified to be sub-additive, hence showing $E_{C,\locc}(\rho)\leq E_F(\rho)$.
If the support of a bipartite state is known to form a so-called \emph{entanglement breaking subspace}~\cite{vidal2002_entanglement,zhao2019_additivity}, then the entanglement of formation is even additive, implying that $E_{C,\locc}(\rho)=E_F(\rho)$.
For two-qubit states $\sigma_{AB}$, an explicit expression for $E_F$ is known~\cite{wooters1998_entanglement,wooters2001_entanglement}
\begin{equation}\label{eq:E_F_formula}
  E_F(\sigma_{AB}) = h\left(\frac{1 + \sqrt{1-C^2(\sigma_{AB})}}{2}\right)
\end{equation}
where $h(x)\coloneqq -x\log(x) - (1-x)\log(1-x)$ and
\begin{equation}
  C(\sigma_{AB}) \coloneqq \max\{0,\lambda_1-\lambda_2-\lambda_3-\lambda_4\}
\end{equation}
where $\lambda_1,\lambda_2,\lambda_3$ and $\lambda_4$ are the eigenvalues in decreasing order of the matrix $\abs{\sqrt{\sigma}\sqrt{(Y\otimes Y)\bar{\sigma}(Y\otimes Y)}}$ and $\bar{\rho}$ is the complex conjugate of $\rho$.

While the bounds provided by~\Cref{eq:gammareg_sep_bound1,eq:gammareg_sep_bound2} did give us some interesting insights, they ultimately fall a bit short, because they are generally not tight and the involved quantities are often quite difficult to compute.
In fact, the \emph{extremality theorem} states that the regularization of any ``sufficiently nice'' entanglement measure lies between the distillable entanglement and the entanglement cost~\cite{horodecki2000_limits,donald2002_uniqueness}.
Remarkably, as we will see below, the approach of relating the regularized quasiprobability extent to an entanglement cost is still extremely fruitful --- however we have to shift our attention to the entanglement cost w.r.t. to a different set of channels $\ds$ instead of LOCC.

\begin{definition}{}{}
  A quantum channel $\cE\in\cptp(AB\rightarrow A'B')$ is said to be \emph{separability preserving} if
  \begin{equation}
    \forall \rho\in\sep(A;B): \cE(\rho)\in \sep(A';B') \, .
  \end{equation}
  We denote by $\sepp(A;B\rightarrow A';B')$ the set of all separability preserving channels from $AB$ to $A'B'$.
\end{definition}
Note that $\sepp$ is a strictly larger set than $\sepc$, since the latter only contains \emph{completely} separability preserving channels, i.e., channels that cannot generate any entanglement even when applied to a subsystem.
The SWAP operation is an example of a channel that lies in $\sepp$ but not $\sepc$.
Thanks to its simpler structure, the finite-shot entanglement cost under $\sepp$ has a much simpler characterization.
For this purpose, consider a smoothed version of the robustness of entanglement of a state $\rho\in\dops(AB)$
\begin{equation}
  R_{\sep}^{\epsilon}(\rho) \coloneqq \min\limits_{\sigma\in B_{\epsilon}(\rho)}R_{\sep}(\sigma)
\end{equation}
where $\epsilon\geq 0$ and $B_{\epsilon}(\rho) \coloneqq \{\sigma\in\dops(AB) | \frac{1}{2}\norm{\rho-\sigma}_1 \leq \epsilon\}$.
It is clearly equivalent to the smoothed quasiprobability extent $\gamma^{\epsilon}_{\sep}(\rho_{AB}) = 1+2R^{\epsilon}_{\sep}({\rho_{AB}})$ by~\Cref{lem:robustness1}.
\begin{lemma}{\cite[Theorem 2]{brandao2011_oneshot}, see also~\cite[Equation (13.363)]{gour2024_resources}}{oneshot_sepp_cost}
  Let $\rho\in\dops(AB)$ be a bipartite state and $\epsilon \geq 0$.
  Then
  \begin{equation}
    \log(1 + R_{\sep}^{\epsilon}(\rho)) \leq E_{C,\sepp}^{(1),\epsilon}(\rho) \leq \log(1 + R_{\sep}^{\epsilon}(\rho)) + 1 \, .
  \end{equation}
\end{lemma}
\begin{proof}
  The following proof works independently of the distance measure that is utilized for the smoothing of the robustness and the error tolerance of the entanglement cost, as long as the same distance is used in both cases.
  To reflect this, we denote the distance measure by $\delta(\rho_1,\rho_2)$.

  We start by proving the lower bound of $E_{C,\sepp}^{(1),\epsilon}(\rho)$.
  Consider a SEPP channel $\cE$ that fulfills $\delta(\cE(\proj{\Psi_M}), \rho) \leq \epsilon$ where $\ket{\Psi_M}$ denotes the maximally entangled state with local dimension $M=2^{E_{C,\sepp}^{(1),\epsilon}(\rho)}$.
  Using the monotonicity of the robustness and the formula for the robustness of a maximally entangled state (see~\Cref{ex:maximally_entangled_state}), we obtain
  \begin{align}
    \log(1 + R^{\epsilon}_{\sep}(\rho)) &\leq \log(1 + R_{\sep}(\cE(\proj{\Psi_M}))) \\
    &\leq \log(1 + R_{\sep}(\proj{\Psi_M})) \\
    &= \log(1 + M - 1) \\
    &= E_{C,\sepp}^{(1),\epsilon}(\rho)
  \end{align}

  We now turn to the upper bound of $E_{C,\sepp}^{(1),\epsilon}(\rho)$.
  Let $\sigma$ be the achieving state in the minimization of $R_{\sep}^{\epsilon}(\rho)$.
  Then, let $\tau$ be the achieving state in the minimization of $R_{\sep}(\sigma)$ such that
  \begin{equation}
    \frac{\sigma + R(\sigma)\tau}{1+R(\sigma)}\in \sep(A;B) \, .
  \end{equation}
  We define $M\coloneqq 1+\ceil{R(\sigma)}$ and
  \begin{equation}
    \kappa\coloneqq \frac{\sigma + (M-1)\tau}{M} = \frac{\sigma + \ceil{R(\sigma)}\tau}{1 + \ceil{R(\sigma)}}
  \end{equation}
  which is also in $\sep(A;B)$, since $M\geq 1+R(\sigma)$.
  Finally, we define the quantum channel
  \begin{equation}
    \cE(X) \coloneqq \tr[\proj{\Psi_M}X]\sigma + (1-\tr[\proj{\Psi_M}X])\tau
  \end{equation}
  where again, $\ket{\Psi_M}$ denotes the maximally entangled state with local dimension $M$.
  First, notice that clearly $\cE(\proj{\Psi_M}) = \sigma$.
  Next, we show that $\cE\in\sepp$.
  For this purpose, let $X$ be a separable state, and define $q\coloneqq \tr[\proj{\Psi_M}X]\cdot M$.
  It holds that $q\in [0,1]$, because $\max_{X\in\sep}\tr[\proj{\Psi_M} X]$ is achieved by a pure product state and
  \begin{equation}
    \tr[\proj{\Psi_M}\cdot \proj{f}\otimes \proj{g}] = \frac{1}{M} \abs{\sum_i \braket{i}{f}\braket{i}{g}} \leq \frac{1}{M}
  \end{equation}
  where $\ket{i}$ is an orthonormal basis and we used the Cauchy-Schwartz inequality.
  With this in mind, we now observe that
  \begin{equation}
    \cE(X) = q\frac{1}{M}\sigma + (1-\frac{q}{M})\tau = q\kappa + (1-q) \tau
  \end{equation}
  is a convex mixture of two separable states and thus itself separable.
  Since $\cE$ is in $\sepp$, we now conclude that
  \begin{align}
    E_{C,\sepp}^{(1),\epsilon}(\rho) \leq \log M &= \log(1 + \ceil{R_{\sep}(\sigma)}) \\
    &\leq \log(1 + R_{\sep}(\sigma)) + 1 \\
    &= \log(1 + R_{\sep}^{\epsilon}(\rho)) + 1 \, .
  \end{align}
\end{proof}

This result allows us to directly relate the regularized quasiprobability extent with the asymptotic entanglement cost under $\sepp$.

\begin{theorem}{}{gamma_sep_is_sepp_cost}
  For any bipartite state $\rho\in\dops(AB)$, the regularized and asymptotic quasiprobability extent w.r.t. $\sep$ can be expressed as
  \begin{equation}
    \log \gammareg_{\sep}(\rho) = E_{C,\sepp}^{\mathrm{exact}}(\rho)
    \quad\text{and}\quad
    \log \gammasreg_{\sep}(\rho) = E_{C,\sepp}(\rho) \, .
  \end{equation}
\end{theorem}
\begin{proof}
  We can write the regularized quasiprobability extent as
  \begin{equation}
    \log\gammareg_{\sep}(\rho) = \lim\limits_{n\rightarrow\infty} \frac{1}{n}\log\gamma_{\sep}(\rho^{\otimes n})\, .
  \end{equation}
  Note that the sequence $\frac{1}{n}\log\gamma_{\sep}(\rho^{\otimes n})$ converges as $\log\gamma_{\sep}$ is sub-additive.
  The sequence $\frac{1}{n}\log \left(1+R_{\sep}(\rho^{\otimes n})\right)$ must hence also converge to $\log\gammareg_{\sep}(\rho)$ as
  {\small
  \begin{align}
    \abs{ \frac{1}{n}\log\left(1+R_{\sep}(\rho^{\otimes n})\right) - \log\gammareg_{\sep}(\rho) }
    &\leq \frac{\abs{ \log\left(1+R_{\sep}(\rho^{\otimes n})\right) - \log\left( 1+2R_{\sep}(\rho^{\otimes n}) \right) }}{n} \nonumber\\
    &\quad + \abs{ \frac{1}{n}\log\gamma_{\sep}(\rho^{\otimes n}) - \log\gammareg_{\sep}(\rho) } \\
    &\leq \frac{1}{n}\frac{R_{\sep}(\rho^{\otimes n})}{\ln(2)(1+R_{\sep}(\rho^{\otimes n}))} \nonumber\\
    &\quad + \abs{ \frac{1}{n}\log\gamma_{\sep}(\rho^{\otimes n}) - \log\gammareg_{\sep}(\rho) } \\
    &\xrightarrow{n\rightarrow\infty} 0 \, .
  \end{align}
  }
  The statement $\log \gammareg_{\sep}(\rho) = E_{C,\sepp}^{\mathrm{exact}}(\rho)$ follows by taking the appropriate limit (with $\epsilon=0$) on both sides of~\Cref{lem:oneshot_sepp_cost}.

  A similar argument can be used to show
  \begin{align}
    \liminf\limits_{n\rightarrow\infty} \frac{\log\left(1+2R^{\epsilon}_{\sep}(\rho^{\otimes n})\right)}{n}
    &= \liminf\limits_{n\rightarrow\infty} \frac{\log\left(1+R^{\epsilon}_{\sep}(\rho^{\otimes n})\right)}{n} \\
    &= \liminf\limits_{n\rightarrow\infty} \frac{\log\left(1+R^{\epsilon}_{\sep}(\rho^{\otimes n})\right) +1}{n}
  \end{align}
  which then implies $\log \gammasreg_{\sep}(\rho) = E_{C,\sepp}(\rho)$.
\end{proof}
One direct consequence of this result is that the upper bound $\log\gammasreg_{\sep}\leq E_{C,\locc}$ from~\Cref{eq:gammareg_sep_bound2} is actually tight for pure states.
This is a direct consequence of the chain of inequalities
\begin{equation}
  E_{D,\locc}(\rho)\leq E_{D,\sepp}(\rho) \leq E_{C,\sepp}(\rho) \leq E_{C,\locc}(\rho) \, .
\end{equation}
together with the fact that for pure states $E_{D,\locc}$ and $E_{C,\locc}$ coincide and are equal to the entanglement entropy~\cite{bennett1996_concentrating}.
\begin{corollary}{}{gamma_sep_asymptotic}
  For any pure state $\rho_{AB}\in\dops(AB)$, one has
  \begin{equation}
    \log\gammasreg_{\sep}(\rho_{AB}) = H(\rho_A)
  \end{equation}
  where $\rho_A\coloneqq \tr_B[\proj{\psi}_{AB}]$ and $H$ denotes the von Neumann entropy.
\end{corollary}
This is especially remarkable, as we previously showed in~\Cref{cor:gamma_reg_pure} that the (exact) regularized quasiprobability extent $\gammareg_{\sep}$ is given in terms of the Rényi-$\nicefrac{1}{2}$ entanglement entropy instead.
We have thus proven that $\gammareg_{\sep}$ and $\gammasreg_{\sep}$ must strictly differ for almost all pure states.
The difference between the two is depicted in~\Cref{fig:entropy_comparison} for general two-qubit pure states
\begin{figure}
  \centering
  \includegraphics{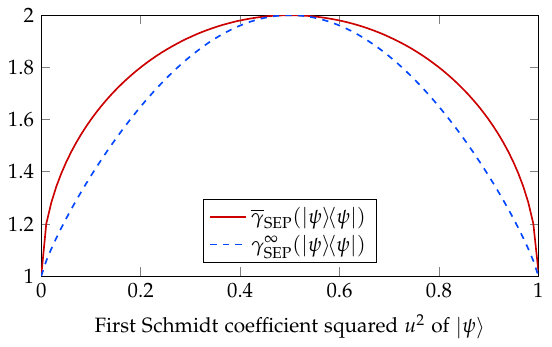}
  \caption{Comparison of $\gammareg_{\sep}(\proj{\psi})$ and $\gammasreg_{\sep}(\proj{\psi})$ for a pure two-qubit state $\ket{\psi}_{AB}$ with Schmidt coefficients $u$ and $\sqrt{1-u^2}$. The two quasiprobability extents are only identical for separable states ($u^2=0,1$) and the Bell state ($u^2=0.5$).}
  \label{fig:entropy_comparison}
\end{figure}

\Cref{thm:gamma_sep_is_sepp_cost} provides a direct connection between the quasiprobability extent and the $\sepp$ entanglement cost.
However, beyond pure states it is unclear whether $E_{C,\sepp}^{\mathrm{exact}}(\rho)$ or $E_{C,\sepp}(\rho)$ can be efficiently computed for interesting states $\rho$.
To address this, we will now relax the free channels even further and turn our attention to computing $\gammareg_{\ppt}$, which is a lower bound for $\gammareg_{\sep}$.
To achieve this goal, we first need to introduce another entanglement measure called $\kappa$-entanglement~\cite{wang2023_exact,wang2020_cost} which turns out to characterize the one-shot entanglement cost under PPTC.
\begin{definition}{}{}
  The $\kappa$-entanglement of a bipartite state $\rho\in\dops(AB)$ is defined as
  \begin{equation}
    E_{\kappa}(\rho) \coloneqq \log \inf_{S\in\pos(AB)} \{ \tr[S] | -(\idchan_A\otimes\transpose_B)(S) \leq (\idchan_A\otimes\transpose_B)(\rho) \leq (\idchan_A\otimes\transpose_B)(S) \} \, .
  \end{equation}
\end{definition}
The $\kappa$-entanglement is a special case of the so-called \emph{$\alpha$-logarithmic negativity} in the limit of $\alpha\rightarrow\infty$~\cite{wang2020_alpha}.
Note that by definition, the $\kappa$ entanglement can be expressed as a semidefinite program.
Analogous to other quantities, we defined the smoothed $\kappa$-entanglement as
\begin{equation}
  E_{\kappa}^{\epsilon}(\rho) \coloneqq \min\limits_{\sigma\in B_{\epsilon}(\rho)}E_{\kappa}(\sigma) \, .
\end{equation}
\begin{lemma}{}{gamma_ppt_kappaent_bound}
  For any bipartite state $\rho\in\dops(AB)$ one has
  \begin{equation}
    \gamma_{\ppt}^{\epsilon}(\rho) \geq 2^{E_{\kappa}^{\epsilon}(\rho)} - 1 \, .
  \end{equation}
\end{lemma}
\begin{proof}
  We first prove the non-smoothed case $\epsilon = 0$.
  Let $\rho=(1+t)\sigma^+-t\sigma^-$ be the QPD which achieves $\gamma_{\ppt}(\rho)$.
  Due to the positive partial transpose of $\sigma^{\pm}$, we have
  \begin{equation}
    -(1+t)(\idchan_A\otimes\transpose_B)(\sigma^+ + \sigma^-) \loewnerleq (\idchan_A\otimes\transpose_B)(\rho) \loewnerleq (1+t)(\idchan_A\otimes\transpose_B)(\sigma^++\sigma^-) \, .
  \end{equation}
  As such, $S\coloneqq (1+t)(\sigma^+ + \sigma^-)$ is a feasible point for $E_{\kappa}(\rho)$, which implies
  \begin{equation}
    2^{E_{\kappa}(\rho)} \leq (1+t)\tr[\sigma^++\sigma^-] = (1+2t) + 1 = \gamma_{\ppt}(\rho) + 1 \, .
  \end{equation}
  The smoothed case follows immediately
  \begin{align}
    \gamma_{\ppt}^{\epsilon}(\rho)
    &= \min\limits_{\sigma\in B_{\epsilon}(\rho)}\gamma_{\ppt}(\sigma) \\
    &\geq \min\limits_{\sigma\in B_{\epsilon}(\rho)}2^{E_{\kappa}(\sigma)} - 1 \\
    &= 2^{E_{\kappa}^{\epsilon}(\rho)} - 1 \, .
  \end{align}
\end{proof}
The next lemma relates the one-shot entanglement cost under $\pptc$ with the $\kappa$-entanglement, analogously to how the log-robustness is related to the entanglement cost under $\sepp$.
\begin{lemma}{\cite[Theorem 13.25]{gour2024_resources}, see also~\cite[Proposition 2]{wang2020_cost}}{oneshot_pptc_cost}
  For any bipartite state $\rho\in\dops(AB)$ we have
  \begin{equation}
    2^{E_{\kappa}^{\epsilon}(\rho)} - 1 \leq 2^{E_{C,\pptc}^{(1),\epsilon}(\rho)} \leq 2^{E_{\kappa}^{\epsilon}(\rho)} + 1 \, .
  \end{equation}
\end{lemma}
\begin{theorem}{}{regularized_gamma_ppt}
  The regularized and asymptotic quasiprobability extent of a bipartite state $\rho\in\dops(AB)$ w.r.t. $\ppt$ are given by
  \begin{equation}\label{eq:regularized_gamma_ppt_eq1}
    \log \gammareg_{\ppt}(\rho) = \lim\limits_{n\rightarrow\infty}\frac{1}{n}E_{\kappa}(\rho^{\otimes n}) = E_{C,\pptc}^{\mathrm{exact}}(\rho) 
  \end{equation}
  and
  \begin{equation}\label{eq:regularized_gamma_ppt_eq2}
    \log \gammasreg_{\ppt}(\rho) = \lim\limits_{\epsilon\rightarrow 0^+}\lim\limits_{n\rightarrow\infty}\frac{1}{n}E_{\kappa}^{\epsilon}(\rho^{\otimes n}) = E_{C,\pptc}(\rho) \, .
  \end{equation}
\end{theorem}
\begin{proof}
  We first note that the sequence $\frac{1}{n}E_{\kappa}^{\epsilon}(\rho^{\otimes n})$ converges for any $\epsilon\in[0,1)$, see e.g.~\cite[Theorem 13.26]{gour2024_resources}.
  As such, \Cref{lem:oneshot_pptc_cost} directly implies the second equalities in~\Cref{eq:regularized_gamma_ppt_eq1,eq:regularized_gamma_ppt_eq2}.
  The first inequalities follow by taking $\liminf\limits_{n\rightarrow\infty}\frac{1}{n}(\cdot)$ on each term of
  \begin{equation}
    \log(2^{E_{\kappa}^{\epsilon}(\rho)}-1) \leq \log\gamma_{\ppt}^{\epsilon}(\rho) \leq \log(2^{E_{C,\pptc}^{(1),\epsilon}(\rho)+1}-1) \, .
  \end{equation}
  where the first inequality follows from~\Cref{lem:gamma_ppt_kappaent_bound} and the second one can be derived analogously to the proof of~\Cref{lem:gamma_single_shot_bound} (recall that the quasiprobability extent of the maximally entangled state is equal under SEP and PPT, see~\Cref{lem:optimal_qpd_pure_ppt}).
\end{proof}
It was originally stated in \reference~\cite{wang2020_cost} that the $\kappa$-entanglement is additive, and thus directly equal to the entanglement cost.
However, this was later shown to be wrong~\cite{lami2024_computable,wang2024_errata}, as $E_{\kappa}$ is only sub-additive, and in some cases strictly sub-additive.
In our quest to find efficiently computable lower bounds of $\gammareg_{\sep}$, it thus remains to show that $E_{C,\pptc}^{\mathrm{exact}}(\rho)$ is an efficiently computable quantity.
For some classes of states, it has a closed-form expression.
\begin{proposition}{~\cite{eisert2003_entanglement}}{ecost_zero_bineg}
  Let $\rho\in\dops(AB)$ be a bipartite state with \emph{zero bi-negativity}, i.e.,
  \begin{equation}
    (\idchan_A\otimes\transpose_B)\left(\abs{(\idchan_A\otimes\transpose_B)(\rho)}\right) \loewnergeq 0 \, .
  \end{equation}
  Then $E_{C,\pptc}^{\mathrm{exact}}$ can be written as 
  \begin{equation}\label{eq:exact_expression_ECppt}
    E_{C,\pptc}^{\mathrm{exact}}(\rho) = \log\norm{(\idchan_A\otimes\transpose_B)(\rho)}_1 \, .
  \end{equation}
\end{proposition}
Note that the right-hand side of~\Cref{eq:exact_expression_ECppt} is the definition of the entanglement measure called \emph{entanglement negativity}~\cite{vidal2002_computable,plenio2005_logarithmic}.

For general states with non-zero bi-negativity, the situation is more complicated.
Here, the entanglement negativity only provides a lower bound to $E_{C,\pptc}^{\mathrm{exact}}$~\cite{eisert2003_entanglement}.
Recent work has shown that $E_{C,\pptc}^{\mathrm{exact}}$ can still be computed through the evaluation of two SDP hierarchies that converge from above and below respectively.
We refrain from delving into the details here and simply summarize the main result.
\begin{proposition}{\cite[Theorem 3]{lami2024_computable}}{pptc_cost_efficient_algo}
  Let $\rho\in\dops(AB)$ be a bipartite state.
  There exists a classical algorithm that computes $E_{C,\pptc}^{\mathrm{exact}}(\rho)$ up to additive error $\epsilon$ with runtime
  \begin{equation}
    \mathcal{O}\left((dD)^{6+o(1)}\cdot \mathrm{polylog}(1/\epsilon) \right)
  \end{equation}
  where $D\coloneqq \dimension(A)\dimension(B)$ and $d\coloneqq\min\{\dimension(A),\dimension(B)\}$.
\end{proposition}

\begin{example}\label{ex:noisy_bell_state}
  Consider the two noisy $2$-qubit states $\rho_p$ and $\sigma_p$ which are defined to be the depolarized version of the Bell state $\ket{\Psi}$
  \begin{equation}
    \rho_p \coloneqq (1-p) \proj{\Psi}_{AB} + p \frac{1}{4}\id_{AB} 
  \end{equation}
  and of the pure entangled state $\ket{\phi}\coloneqq \frac{1}{3}\ket{00} + \frac{2\sqrt{2}}{3}\ket{11}$
  \begin{equation}
    \sigma_p \coloneqq (1-p) \proj{\phi}_{AB} + p \frac{1}{4}\id_{AB} \, .
  \end{equation}
  In~\Cref{fig:noisy_regularized_states} we depict for $\rho_p$ and $\sigma_p$
  \begin{itemize}
    \item the quasiprobability extent $\gamma_{\sep}$ which is numerically computed through the SDP for $\gamma_{\ppt}$ (recall~\Cref{cor:gamma_mixed_small_sys}).
    \item the exact PPTC entanglement cost $E_{C,\pptc}^{\mathrm{exact}}$, which is numerically computed through~\Cref{prop:ecost_zero_bineg}.
    We note that both $\rho_p$ and $\sigma_p$ have zero bi-negativity for all values of $p\in [0,1]$.
    \item the entanglement of formation $E_F$ which is computed through~\Cref{eq:E_F_formula}.
    \item the regularized and asymptotic quasiprobability extents $\gammareg_{\sep}$ and $\gammasreg_{\sep}$ at $p=0$, which are analytically known through~\Cref{cor:gamma_reg_pure,cor:gamma_sep_asymptotic}.
  \end{itemize}
  The exact PPTC entanglement cost provides a lower bound for $\gammareg_{\sep}$.
  This bound coincides with the analytic expression of $\gammareg_{\sep}$ at $p=0$ and with $\gamma_{\sep}$ at $p=p^*$, where $p^*$ is the smallest probability $p$ at which $\gamma_{\sep}$ reaches the value $1$ ($p=2/3$ for $\rho_p$ and $p\approx 0.56$ for $\sigma_p$).
  Furthermore, we observe (up to numerical precision) that  $E_{C,\pptc}^{\mathrm{exact}}$ linearly interpolates between these two points.
  Remarkably this indicates that $\gammareg_{\sep}=E_{C,\pptc}^{\mathrm{exact}}$ due to the convexity of the quasiprobability extent.

  Finally, we note that the entanglement of formation provides an upper bound to $\gammasreg_{\sep}$ as
  \begin{equation}
    \gammasreg_{\sep} = E_{C,\sepp} \leq E_{C,\locc} = \lim\limits_{n\rightarrow\infty}\frac{1}{n}E_F(\cdot^{\otimes n}) \leq E_F
  \end{equation}
  due to~\Cref{thm:gamma_sep_is_sepp_cost}.
  Interestingly, for $p=0$, this upper bound coincides with the analytical expression of $\gammasreg_{\sep}$ up to numerical precision (note that for the Bell state $\ket{\Psi}$ it also coincides with  $\gammareg_{\sep}$ since the von Neumann and Rényi-$\nicefrac{1}{2}$ entropies take the same value).
  For $p>0$, $\gammasreg_{\sep}$ must lie somewhere in the shaded region.
\end{example}
\begin{figure}
  \centering
  \includegraphics{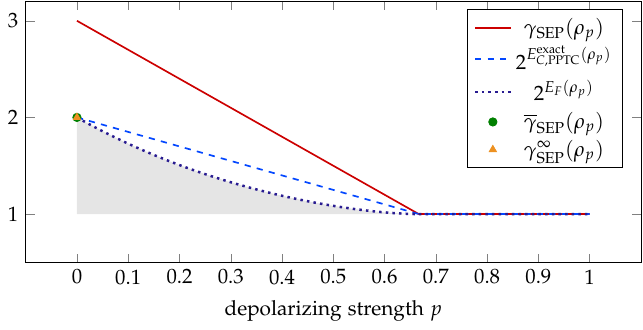}
  \includegraphics{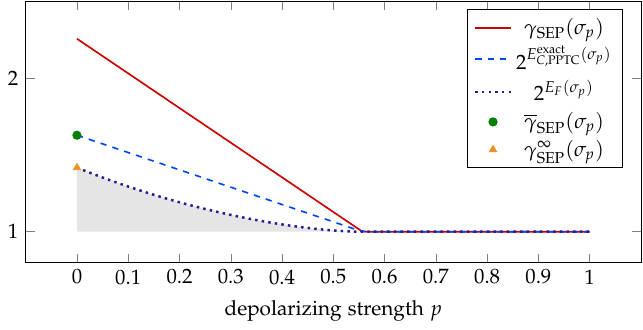}
  \caption{Numerical evaluation of various entanglement quantities on $\rho_p$ (top) and $\sigma_p$ (bottom) which characterize $\gamma_{\sep}$, $\gammareg_{\sep}$ and $\gammasreg_{\sep}$. See main text for explanation. We argue that $\gammareg_{\sep}$ must be equal to $2^{E_{C,\pptc}^{\mathrm{exact}}}$ and that $\gammasreg_{\sep}$ must lie in the shaded region.}
  \label{fig:noisy_regularized_states}
\end{figure}

\subsection{Quasiprobability extent of channels}\label{sec:gamma_nonlocal_channels}
In this section, we study the quasiprobability extent $\gamma_{\lo^\star}(\cE)$ and $\gamma_{\locc^\star}(\cE)$ of a quantum channel $\cE\in\cptp(AB\rightarrow A'B')$.
This is a generalization of the quasiprobability extent of states, which can be considered to be the specialization with trivial input space $A=B=\mathbb{C}$.
At first, one might hope that the problem could be reduced back to evaluating $\gamma_{\sep}$ through the state-channel duality of the Choi isomorphism.
However, the complicated structure of $\locc^{\star}$ and $\lo^{\star}$ Choi states as well as the partial trace constraint make this endeavor impossible.
These difficulties persist even for unitary channels $\cU$: As illustrated by~\Cref{ex:toffoli} further below, $\gamma_{\sep}(\choi{\cU})$ can generally be strictly smaller than $\gamma_{\sepc}(\cU),\gamma_{\locc^{\star}}(\cU)$ and $\gamma_{\lo^{\star}}(\cU)$.

As a result of these difficulties, we don't know an explicit expression for $\gamma_{\lo^\star}(\cU)$ and $\gamma_{\locc^\star}(\cU)$ for all unitary channels $\cU$, which is somewhat disappointing considering we have a complete understanding of $\gamma_{\sep}$ for pure states.
Still, many ideas and techniques that we utilized in the state setting can be (at least partly) generalized to the channel setting.
In this section will explore these generalizations.
We will start off with two specific sub-classes of unitaries for which we do have a characterization of the quasiprobability extent (\Cref{sec:gamma_clifford,sec:gamma_kaklike}).
Then, we will turn our attention to more general techniques to find lower and upper bounds for both unitary and non-unitary channels.

\subsubsection{Clifford gates}\label{sec:gamma_clifford}
In the introduction of this chapter, we have already discussed that a CNOT gate can be realized with an LOCC protocol that consumes a pre-shared Bell pair.
Conversely, a Bell pair can clearly be realized from a product state using a CNOT gate.
In this sense, a Bell pair and a CNOT gate are equally powerful resources under the QRT of LOCC.
This insight can be straightforwardly generalized to arbitrary Clifford gates through the generalized gate teleportation protocol.
\begin{proposition}{Gate teleportation}{gate_teleportation}
  Let $A,A',A''$ be $n$-qubit systems and $B,B',B''$ $m$-qubit systems.
  For any bipartite Clifford unitary $U\in\uni(AB)$ with induced channel $\cU\in\cptp(AB)$, there exists an LOCC protocol $\cE\in\locc(AA'A'';BB'B''\rightarrow A;B)$ such that
  \begin{equation}
    \forall \rho_{AB}\in\dops(AB): \cU(\rho_{AB}) = \cE(\rho_{AB}\otimes (\choi{\cU})_{A'B'A''B''}) \, .
  \end{equation}
\end{proposition}
Let us denote the $n$-qubit Pauli operators by $\{Q_0^n,Q_1^n,\dots,Q_{4^n-1}^n\}$.
The generalized $n$-qubit Bell basis is defined as
\begin{equation}
  \ket{\Psi^n_i} \coloneqq (Q_i^n\otimes\id) \frac{1}{\sqrt{2^n}}\sum_{i=0}^{2^n-1}\ket{i}\otimes \ket{i} \, .
\end{equation}
It is an orthonormal basis because $\braopket{\Psi^n_0}{O\otimes\id}{\Psi^n_0} = 2^{-n}\tr[O]$ for any $n$-qubit operator $O$.
The gate teleportation protocol $\cE$ that achieves~\Cref{prop:gate_teleportation} is depicted in~\Cref{fig:gate_teleportation}.
Both parties perform a measurement of the generalized Bell basis on the input systems $A,B$ and the corresponding environment systems $A'',B''$ of the Choi state.
They then apply some correction operation that depends on these two measurement outcomes.

This protocol is essentially an extension of the well-known quantum teleportation protocol (which we will also encounter again later, see~\Cref{fig:state_teleportation}) with the correction operation being ``pushed through'' the unitary $U$.
Whenever $U$ is Clifford, the resulting correction $U(Q_i^n\otimes Q_j^m)U^{\dagger}$ is Pauli and hence locally implementable.
Therefore, $\cE$ is in LOCC.

\begin{figure}
  \centering
  \includegraphics{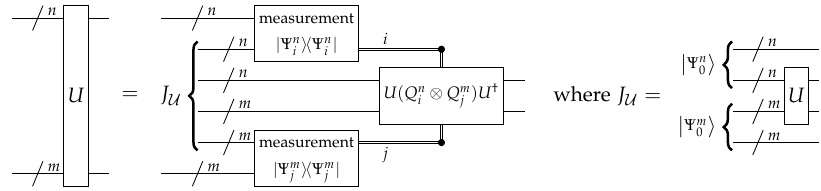}
  \caption{Gate teleportation for a general unitary $U$ with associated channel $\cU$. If $U$ is a Clifford gate, then the correction operation $U(Q_i^n\otimes Q_j^m)U^{\dagger}$ is a Pauli gate and hence locally implementable.}
  \label{fig:gate_teleportation}
\end{figure}

\begin{corollary}{}{gamma_locc_clifford}
  Let $A$ and $B$ be $n$-qubit and $m$-qubit systems.
  For any bipartite Clifford unitary $U\in\uni(AB)$ with induced channel $\cU\in\cptp(AB)$, one has
  \begin{equation}
    \gamma_{\locc^\star}(\cU) = \gamma_{\locc}(\cU) = \gamma_{\sep(AA';BB')}(\choi{\cU}) \, .
  \end{equation}
\end{corollary}
\begin{proof}
  A priori, we know by~\Cref{cor:gamma_state_from_channel} that
  \begin{equation}
    \gamma_{\sep}(\choi{\cU}) \leq \gamma_{\locc^\star}(\cU) \leq \gamma_{\locc}(\cU)
  \end{equation}
  since the Choi state can be prepared from a product state and an invocation of $\cU$.
  The converse inequality $\gamma_{\sep}(\choi{\cU})\geq \gamma_{\locc}(\cU)$ follows from~\Cref{prop:gate_teleportation}:
  Any QPD of $\choi{\cU}$ into separable states automatically induces a QPD of $\cU$ into LOCC protocols.
\end{proof}
Using the closed-form formula for $\gamma_{\sep}(\choi{\cU})$ in~\Cref{thm:optimal_qpd_pure}, we can thus efficiently compute the quasiprobability extent for various gates.
Similarly, $\gammareg_{\locc^\star}(\cU) = \gammareg_{\locc}(\cU) = \gammareg_{\sep}(\choi{\cU})$ can be evaluated through~\Cref{cor:gamma_reg_pure}.
\begin{example}\label{ex:clifford_gates}
  Every two-qubit Clifford gate is equivalent to either the CNOT, SWAP or iSWAP gate up to local Clifford gates.
  The Schmidt coefficients of $\choi{\cnot}$ are $\smash{(\frac{1}{\sqrt{2}},\frac{1}{\sqrt{2}})}$ and those of $\choi{\swap}$ and $\choi{\iswap}$ are $\smash{(\frac{1}{2},\frac{1}{2},\frac{1}{2},\frac{1}{2})}$.
  Hence, we get
  \begin{equation}
    \gamma_{\locc^{\star}}(\cnot) = 3 \text{ and } \gammareg_{\locc^{\star}}(\cnot) = 2
  \end{equation}
  and similarly
  \begin{equation}
    \gamma_{\locc^{\star}}(\swap) = \gamma_{\locc^{\star}}(\iswap) = 7 \text{ and } \gammareg_{\locc^{\star}}(\swap) = \gammareg_{\locc^{\star}}(\iswap) = 4 \, .
  \end{equation}
\end{example}
\begin{example}
  Consider the general permutation gate $\mathrm{Perm}_{\pi}$ across $n+m$ qubits for some $\pi\in S_{n+m}$ where $S_{k}$ denotes the permutation group.
  Define the number of ``crossings'' $\mathrm{cross}(\pi)$ to be the number of elements $i\in\{1,\dots,n\}$ such that $\sigma(i)>n$. The Choi state $\choi{\mathrm{Perm}_{\pi}}$ is equivalent to $2\mathrm{cross}(\pi)$ Bell pairs, so we get by~\Cref{ex:maximally_entangled_state}
  \begin{equation}
    \gamma_{\locc^{\star}}(\mathrm{Perm}_{\pi}) = \gamma_{\locc}(\mathrm{Perm}_{\pi}) = 2^{2\mathrm{cross}(\sigma)+1}-1 \, .
  \end{equation}
\end{example}

\subsubsection{Gates with unitary operator Schmidt decomposition}\label{sec:gamma_kaklike}
While the result in~\Cref{cor:gamma_locc_clifford} is elegant, the approach cannot be applied to channels without an equivalence to a state, and more generally, it also cannot make any statement about $\gamma_{\lo^\star}$.
Here, we provide another approach that is applicable to a class of unitaries which exhibit a so-called unitary operator Schmidt decomposition.
Any operator $X\in\lino(AB)$ can be written in terms of its operator Schmidt decomposition~\cite{nielsen1998_thesis}
\begin{equation}
  X = \sum_j u_j A_j\otimes B_j
\end{equation}
where $u_j>0$ and $A_j\in\lino(A),B_j\in\lino(B)$ are orthonormal operators w.r.t. the Hilbert-Schmidt norm, i.e., $\langle A_j,A_k\rangle = \delta_{jk}$ and $\langle B_j,B_k\rangle = \delta_{jk}$.

\begin{theorem}{}{gamma_kaklike_unitary}
  Let $U\in\lino(AB)$ be a bipartite unitary with an operator Schmidt decomposition
  \begin{equation}
    U_{AB} = \sum_j u_j (L_j)_A \otimes (R_j)_B
  \end{equation}
  such that the $L_i$ and $R_i$ are unitary operators on $A$ and $B$ respectively.
  Then, the quasiprobability extent of the induced channel $\cU\in\cptp(AB)$ is
  \begin{equation}
    \gamma_{\locc^\star}(\cU)
    = \gamma_{\lo^\star}(\cU)
    = 2\left(\sum_j u_j\right)^2 - 1 \, .
  \end{equation}
\end{theorem}
This result is applicable to any two-qubit unitary $U$, as Cartan's KAK decomposition implies that it can always be written in the form
\begin{equation}
  U = (V_1\otimes V_2)\exp\left(i\theta_X X\otimes X + i\theta_Y Y\otimes Y + i\theta_Z Z\otimes Z \right)(V_3\otimes V_4)
\end{equation}
for some single-qubit unitaries $V_1,V_2,V_3,V_4$ and angles $\theta_X,\theta_Y,\theta_Z\in\mathbb{R}$ fulfilling $\abs{\theta_Z}\leq\theta_Y\leq\theta_X\leq \pi/4$~\cite{khaneja2000_cartan,tucci2005_introduction}.
Evaluating the exponential function, this directly provides us (up to phases) with a unitary operator Schmidt decomposition
\begin{equation}
  U = (V_1\otimes V_2)\left(u_0\id\otimes\id + u_1X\otimes X + u_2Y\otimes Y + u_3 Z\otimes Z \right)(V_3\otimes V_4)
\end{equation}
for some $u_j\in\mathbb{C}$ fulfilling $\sum_{j=0}^3\abs{u_j}^2 = 1$.
By extension, any tensor product of two-qubit unitaries has a unitary operator Schmidt decomposition allowing us to apply~\Cref{thm:gamma_kaklike_unitary}.
This insight enables us to evaluate the regularized quasiprobability extent.
\begin{corollary}{}{}
  Let $U\in\lino(AB)$ be a bipartite unitary as in~\Cref{thm:gamma_kaklike_unitary}.
  Then,
  \begin{equation}
    \gammareg_{\locc^\star}(\cU) = \gammareg_{\lo^\star}(\cU) = (\sum_j u_j)^2 \, .
  \end{equation}
\end{corollary}
The proof is completely analogous to the one of~\Cref{cor:gamma_reg_pure}.
Note that the regularization halves the distance to $1$ as
\begin{equation}
  \gamma_{\locc^\star}(\cU) = \gamma_{\lo^\star}(\cU) = 1 + 2 R_{\locc^\star}(\cU) 
\end{equation}
and
\begin{equation}
  \gammareg_{\locc^\star}(\cU) = \gammareg_{\lo^\star}(\cU) = 1 + R_{\locc^\star}(\cU) \, ,
\end{equation}
where $R_{\locc^\star}(\cU) = (\sum_j u_j)^2-1$ is the robustness measure as in~\Cref{lem:robustness2}.

We now focus our attention on the proof of~\Cref{thm:gamma_kaklike_unitary}.
The general strategy consists of finding matching upper and lower bounds for the quasiprobability extents.
The upper bound is provided by explicitly constructing an achieving QPD.
This construction is applicable in a slightly more general setting, so we state it as a separate lemma
\begin{lemma}{}{gamma_unitary_upperbound}
  Let $U\in\lino(AB)$ be a bipartite unitary of the form
  \begin{equation}
    U_{AB} = \sum_j u_j (L_j)_A \otimes (R_j)_B
  \end{equation}
  with induced channel $\cU$, where $u_j>0$, $L_j\in\uni(A)$ and $R_j\in\uni(B)$.
  Then,
  \begin{equation}
    \gamma_{\lo^\star}(\cU) \leq 2(\sum_j u_j)^2 - \sum_j u_j^2 \, .
  \end{equation}
\end{lemma}
We refer to~\Cref{app:optimal_circuits} for an explicit circuit implementation of the achieving QPD that only requires one ancilla qubit on each bipartition.
\begin{proof}
  In the following, we construct an explicit $\lo^{\star}$ QPD of $\cU$ that achieves a 1-norm of $2(\sum_j u_j)^2 - \sum_j u_j^2$.
  The action of the bipartite unitary can be written as
  \begin{align}\label{eq:proto_qpd_kaklike}
    \cU(\rho) &= \sum_j u_j^2 (L_j\otimes R_j)\rho(L_j^{\dagger}\otimes R_j^{\dagger}) \nonumber\\
    &+ \sum_{j<k} u_ju_k \left[(L_j\otimes R_j)\rho(L_k^{\dagger}\otimes R_k^{\dagger}) + (L_k\otimes R_k)\rho(L_j^{\dagger}\otimes R_j^{\dagger})\right] \, .
  \end{align}
  The first term is already a mixture of local channels, so we only need to concentrate our efforts to decompose the second term.
  For this purpose, we introduce the two linear superoperators
  \begin{equation}
    \mathcal{A}_{j,k}^{\phi}(\rho) \coloneqq \frac{1}{2}\left( e^{i\phi} L_j\rho L_k^{\dagger} + e^{-i\phi} L_k\rho L_j^{\dagger} \right)
  \end{equation}
  and 
  \begin{equation}
    \mathcal{B}_{j,k}^{\phi}(\rho) \coloneqq \frac{1}{2}\left( e^{i\phi} R_j\rho R_k^{\dagger} + e^{-i\phi} R_k\rho R_j^{\dagger} \right)
  \end{equation}
  where $\phi\in[0,2\pi)$ is some angle.
  A simple calculation shows that
  \begin{align}
    & 2\left[\mathcal{A}_{j,k}^0\otimes\mathcal{B}_{j,k}^0 - \mathcal{A}_{j,k}^{\pi /2}\otimes\mathcal{B}_{j,k}^{\pi /2} \right](\rho) \nonumber\\
    &= \frac{1}{2}\Big[ (L_j\otimes R_j)\rho(L_k^{\dagger}\otimes R_k^{\dagger}) + (L_j\otimes R_k)\rho(L_k^{\dagger}\otimes R_j^{\dagger}) \nonumber\\
    &\quad\quad + (L_k\otimes R_j)\rho(L_j^{\dagger}\otimes R_k^{\dagger}) + (L_k\otimes R_k)\rho(L_j^{\dagger}\otimes R_j^{\dagger}) \Big] \nonumber\\
    &- \frac{1}{2}\Big[ -(L_j\otimes R_j)\rho(L_k^{\dagger}\otimes R_k^{\dagger}) + (L_j\otimes R_k)\rho(L_k^{\dagger}\otimes R_j^{\dagger}) \nonumber\\
    &\quad\quad + (L_k\otimes R_j)\rho(L_j^{\dagger}\otimes R_k^{\dagger}) - (L_k\otimes R_k)\rho(L_j^{\dagger}\otimes R_j^{\dagger}) \Big] \\
    &= (L_j\otimes R_j)\rho(L_k^{\dagger}\otimes R_k^{\dagger}) + (L_k\otimes R_k)\rho(L_j^{\dagger}\otimes R_j^{\dagger}) \, .
  \end{align}
  As such, we have
  \begin{equation}
    \cU = \sum_j u_j^2 \mathcal{L}_j\otimes\mathcal{R}_j + \sum_{j<k} 2u_ju_k \left(\mathcal{A}_{j,k}^0\otimes\mathcal{B}_{j,k}^0 - \mathcal{A}_{j,k}^{\pi /2}\otimes\mathcal{B}_{j,k}^{\pi /2} \right)
  \end{equation}
  where $\mathcal{L}_j(\rho) = L_j\rho L_j^{\dagger}$ and $\mathcal{R}_j(\rho) = R_j\rho R_j^{\dagger}$.
  To show that this is indeed a valid QPD, it remains to show that $\mathcal{A}_{j,k}^{\phi}\otimes\mathcal{B}_{j,k}^{\phi}\in\lo^{\star}(A;B)$.
  For this, it suffices to show $\mathcal{A}_{j,k}^{\phi}\in\cptp^{\star}(A)$ and $\mathcal{B}_{j,k}^{\phi}\in\cptp^{\star}(B)$ as $\cptp^{\star}(A)\otimes\cptp^{\star}(B)\subset\lo^{\star}(A;B)$ (recall the discussion of~\Cref{eq:old_lostar_def}).
  Using the characterization of $\cptp^{\star}$ in~\Cref{eq:char_cptpstar}, the desired statement follows from the observation that
  \begin{equation}
    \mathcal{A}_{j,k}^{\phi}(\rho) = \mathcal{A}_{j,k}^{+,\phi}(\rho) - \mathcal{A}_{j,k}^{-,\phi}(\rho)
  \end{equation}
  \begin{equation}
    \mathcal{B}_{j,k}^{\phi}(\rho) = \mathcal{B}_{j,k}^{+,\phi}(\rho) - \mathcal{B}_{j,k}^{-,\phi}(\rho)
  \end{equation}
  where 
  \begin{equation}
    \mathcal{A}_{j,k}^{\pm,\phi}(\rho) \coloneqq \frac{1}{4}(L_j\pm e^{-i\phi}L_k)\rho (L_j^{\dagger}\pm e^{i\phi}L_k^{\dagger})
  \end{equation}
  \begin{equation}
    \mathcal{B}_{j,k}^{\pm,\phi}(\rho) \coloneqq \frac{1}{4}(R_j\pm e^{-i\phi}R_k)\rho (R_j^{\dagger}\pm e^{i\phi}R_k^{\dagger}) \, .
  \end{equation}
  Clearly, $\mathcal{A}_{j,k}^{\pm,\phi}$ and $\mathcal{B}_{j,k}^{\pm,\phi}$ are completely positive maps.
  To see that $\mathcal{A}_{j,k}^{+,\phi}+\mathcal{A}_{j,k}^{-,\phi}$ is trace-preserving, observe that
  \begin{align}
    & \frac{1}{4}(L_j^{\dagger}+ e^{i\phi}L_k^{\dagger})(L_j + e^{-i\phi}L_k) + \frac{1}{4}(L_j^{\dagger} - e^{i\phi}L_k^{\dagger})(L_j - e^{-i\phi}L_k) \nonumber\\
    &= \frac{1}{4}( 2L_j^{\dagger}L_j + 2L_k^{\dagger}L_k ) = \id \, .
  \end{align}
  Analogously, $\mathcal{B}_{j,k}^{+,\phi}+\mathcal{B}_{j,k}^{-,\phi}$ is trace-preserving.
  Since we have a valid QPD, we conclude that
  \begin{equation}
    \gamma_{\lo^{\star}}(\cU) \leq \sum_j u_j^2 + 4\sum_{j<k}u_ju_k = 2(\sum_j u_j)^2 - \sum_j u_j^2 \, .
  \end{equation}
\end{proof}
\begin{proof}[Proof of~\Cref{thm:gamma_kaklike_unitary}]
  Notice that
  \begin{align}
    \dimension(AB) &=
    \tr[UU^{\dagger}] \\
    &= \sum_{j,k}u_ju_k\tr[(L_jL_k^{\dagger})\otimes (R_jR_k^{\dagger})] \\
    &= \sum_{j,k}u_ju_k\tr[L_jL_k^{\dagger}]\tr[R_jR_k^{\dagger}] \\
    &= \dimension(AB)\sum_{j,k}u_ju_k\delta_{jk}\delta_{jk} \\
    &= \dimension(AB)\sum_{j}u_j^2
  \end{align}
  and hence $\sum_i u_i^2 = 1$.
  Thanks to~\Cref{lem:gamma_unitary_upperbound} and $\gamma_{\locc^{\star}}\leq \gamma_{\lo^{\star}}$, it remains to show that $\gamma_{\locc^{\star}}\geq 2(\sum_ju_j)^2 - 1$.

  Recall that the Choi state of the unitary channel $\cU$ is defined as
  \begin{equation}
    \choi{\cU} = (\cU\otimes\idchan_{A'B'}) (\proj{\Psi}_{AA'}\otimes\proj{\Phi}_{BB'})
  \end{equation}
  where $A',B'$ are copies of $A,B$ and $\ket{\Psi}_{AA'},\ket{\Phi}_{BB'}$ are maximally entangled states.
  By~\Cref{cor:gamma_state_from_channel}, this implies that
  \begin{equation}
    \gamma_{\locc^{\star}(A;B)}(\cU) \geq \gamma_{\sep(AA';BB')}(\choi{\cU}) \, . 
  \end{equation}
  Since the Choi state is pure $\choi{\cU}=\proj{\choi{\cU}}$, we can express its quasiprobability extent in terms of its Schmidt coefficients $v_j$ (see~\Cref{thm:optimal_qpd_pure})
  \begin{equation}
    \gamma_{\sep(AA';BB')}(\choi{\cU}) = 2(\sum_j v_j)^2 - 1 \, . 
  \end{equation}
  It remains to explicitly evaluate the Schmidt coefficients for the Choi state.
  For some fixed orthonormal bases $\{\ket{i}\}_i$ and $\{\ket{j}\}_j$ of $A$ and $B$ we get
  \begin{align}
    \ket{\choi{U}}_{A'B'AB} &= (U\otimes\idchan_{A'B'})\frac{1}{\sqrt{\dimension(A)\dimension(B)}} \sum_{i=1}^{\dimension(A)} \sum_{j=1}^{\dimension(B)} \ket{i}_{A}\ket{i}_{A'}\ket{j}_B\ket{j}_{B'} \\
    &= \sum_k u_k \frac{1}{\sqrt{\dimension(A)\dimension(B)}} \sum_{i=1}^{\dimension(A)} \sum_{j=1}^{\dimension(B)} (L_k\ket{i}_{A})\ket{i}_{A'}(R_k\ket{j}_B)\ket{j}_{B'} \\
    &= \sum_k u_k \ket{f_k}_{AA'}\ket{g_k}_{BB'}
  \end{align} 
  where
  \begin{align}
    \ket{f_k}_{A'A} &\coloneqq \frac{1}{\sqrt{\dimension(A)}}\sum_{i=1}^{\dimension(A)}(L_k\ket{i}_A)\ket{i}_{A'} \\
    \ket{g_k}_{B'B} &\coloneqq \frac{1}{\sqrt{\dimension(B)}}\sum_{j=1}^{\dimension(B)}(R_k\ket{j}_B)\ket{j}_{B'} \, .
  \end{align}
  It remains to verify that the set of states $\{\ket{f_k}\}_k$ and $\{\ket{g_k}\}_k$ are orthogonal, which then implies that we have found a valid Schmidt decomposition of $\ket{\choi{U}}$ with Schmidt coefficients $u_j$.
  This follows directly from the orthogonality of the unitaries
  \begin{equation}
    \braket{f_j}{f_k} = \frac{1}{\dimension(A)}\sum_{i,i'} \braket{i}{i'}\cdot \bra{i}L_j^{\dagger}L_k\ket{i'} = \frac{1}{\dimension(A)}\tr[L_j^{\dagger}L_k] = \delta_{jk}
  \end{equation}
  and analogously $\braket{g_j}{g_k} = \delta_{jk}$.
\end{proof}

\begin{example}\label{ex:gamma_2qb_gates}
  As an application of~\Cref{thm:gamma_kaklike_unitary}, we list a few important two-qubit gates and their associated quasiprobability extent.
  \begin{description}
    \item[Controlled rotation gate]
    The controlled rotation gate w.r.t. the Bloch sphere vector $\vec{n}\in\mathbb{R}^3$, $\norm{\vec{n}}_2 = 1$ and angle $\theta\in[0, 2\pi)$ is defined as
    \begin{equation}
      \mathrm{CR}_{\vec{n}}(\theta) \coloneqq e^{-i\frac{\theta}{2} \proj{1}\otimes \left(n_1X + n_2Y + n_3Z\right)} \, .
    \end{equation}
    Its Schmidt coefficients are given by $u=(\cos\frac{\theta}{4}, \sin\frac{\theta}{4}, 0, 0)$ and as such one has
    \begin{equation}
      \gamma_{\lo^{\star}}(\mathrm{CR}_{\vec{n}}(\theta)) = \gamma_{\locc^{\star}}(\mathrm{CR}_{\vec{n}}(\theta)) = 1 + 2\sin(\theta /2)
    \end{equation}
    \begin{equation}
      \gammareg_{\lo^{\star}}(\mathrm{CR}_{\vec{n}}(\theta)) = \gammareg_{\locc^{\star}}(\mathrm{CR}_{\vec{n}}(\theta)) = 1 + \sin(\theta /2) \, .
    \end{equation}

    \item[CNOT gate]
    We recover the result from~\Cref{ex:clifford_gates} as a special case of the rotation gate with $\theta=\pi$.

    \item[Two-qubit Pauli rotation gate]
    For $P\in \{X,Y,Z\}$ and an angle $\theta\in[0,2\pi)$ we define
    \begin{equation}
      \mathrm{R}_{PP}(\theta) \coloneqq e^{-i\frac{\theta}{2}P\otimes P} \, .
    \end{equation}
    The associated Schmidt coefficients are $u=(\abs{\cos\frac{\theta}{2}}, \sin\frac{\theta}{2}, 0, 0)$ and as such one has
    \begin{equation}
      \gamma_{\lo^{\star}}(\mathrm{R}_{PP}(\theta)) = \gamma_{\locc^{\star}}(\mathrm{R}_{PP}(\theta)) = 1 + 2\abs{\sin\theta} 
    \end{equation}
    \begin{equation}
      \gammareg_{\lo^{\star}}(\mathrm{R}_{PP}(\theta)) = \gammareg_{\locc^{\star}}(\mathrm{R}_{PP}(\theta)) = 1 + \abs{\sin\theta} \, .
    \end{equation}

    \item[SWAP/iSWAP gate]
    We recover the result from~\Cref{ex:clifford_gates} by using that both the SWAP and iSWAP gates have the Schmidt coefficients $u=(\frac{1}{2},\frac{1}{2},\frac{1}{2},\frac{1}{2})$.
    Note that they have the maximal quasiprobability extent of any two-qubit unitary as the associated Schmidt vector $u\in\mathbb{R}^4$ must always fulfill
    \begin{equation}
      2(\sum_j u_j)^2 - 1 = 2\norm{u}_1^2 - 1 \leq 2 (\sqrt{4}\norm{u}_2)^2 - 1 = 7
    \end{equation}
    where we used the Cauchy-Schwarz inequality.
  \end{description}
\end{example}
\begin{example}\label{ex:cnot_lostar_qpd}
  We briefly illustrate the $\lo^{\star}$ QPD of the CNOT gate that achieves $\gamma_{\lo^{\star}}(\cnot)=3$ throught the construction of~\Cref{thm:gamma_kaklike_unitary}.
  Up to local unitaries, we can equate the CNOT with the unitary
  \begin{equation}
    V\coloneqq \frac{1}{\sqrt{2}} \id\otimes \id + \frac{1}{\sqrt{2}} Z \otimes (-i Z)
  \end{equation}
  as
  \begin{equation}
    \cnot = \left( S^{\dagger}\otimes (HS^{\dagger}) \right) V  \left(\id\otimes H\right)  \, .
  \end{equation}
  We can thus focus our considerations on finding the optimal $\lo^{\star}$ QPD of $V$.
  Following the construction in~\Cref{lem:gamma_unitary_upperbound}, the optimal QPD reads
  \begin{equation}
    \indsupo{V} = \frac{1}{2}\idchan\otimes\idchan + \frac{1}{2}\indsupo{Z}\otimes\indsupo{Z} + \mathcal{A}_{1,2}^{0}\otimes\mathcal{B}_{1,2}^{0} - \mathcal{A}_{1,2}^{\pi / 2}\otimes\mathcal{B}_{1,2}^{\pi /2}
  \end{equation}
  where
  \begin{align}
    \mathcal{A}_{1,2}^{0} &= \indsupo{\proj{0}} - \indsupo{\proj{1}} \, ,
    &
    \mathcal{A}_{1,2}^{\pi / 2} &= \frac{1}{2}\indsupo{S^{\dagger}} - \frac{1}{2}\indsupo{S} \, , \\
    \mathcal{B}_{1,2}^{0} &= \frac{1}{2}\indsupo{S} - \frac{1}{2}\indsupo{S^{\dagger}} \, ,
    &
    \mathcal{B}_{1,2}^{\pi / 2} &= \indsupo{\proj{1}} - \indsupo{\proj{0}} \, .
  \end{align}
\end{example}
We note that~\Cref{lem:gamma_unitary_upperbound} provides us a technique to find upper bounds for the quasiprobability extent of large unitaries that are composed of two-qubit nonlocal gates interleaved with local unitaries.
Consider for example the unitary $U_{\mathrm{tot}}$ depicted in~\Cref{fig:largeunitary_upperbound}.
\begin{figure}
  \centering
  \scalebox{1.2}{
  \includegraphics{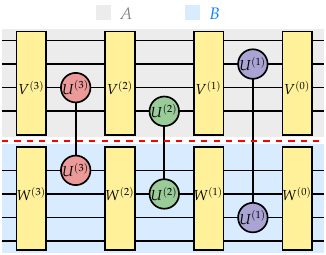}
  }
  \caption{Cut of a large unitary that can be written as the composition of multiple two-qubit unitaries (here denoted by $U^{(1)},U^{(2)},U^{(3)}$) interleaved with local unitaries (denoted by $V^{(i)},W^{(i)}$).}
  \label{fig:largeunitary_upperbound}
\end{figure}
From the chaining property (see~\Cref{lem:gamma_chaining,lem:gamma_tp_idchan}), one can obtain the trivial upper bound
\begin{equation}
  \gamma_{\lo^{\star}}(\cU_{\mathrm{tot}}) \leq \gamma_{\lo^{\star}}(\cU^{(1)}) \gamma_{\lo^{\star}}(\cU^{(2)}) \gamma_{\lo^{\star}}(\cU^{(3)})
\end{equation}
However, following corollary shows that the cost of cutting $U_{\mathrm{tot}}$ is at most the cost of jointly cutting the $U^{(i)}$, instead of cutting them separately
\begin{equation}
  \gamma_{\lo^{\star}}(\cU_{\mathrm{tot}}) \leq \gamma_{\lo^{\star}}(\cU^{(1)} \otimes \cU^{(2)} \otimes \cU^{(3)}) \, .
\end{equation}
which is a better bound due to the strict sub-multiplicativity of $\gamma_{\lo^{\star}}$.
\begin{corollary}{}{}
  Let $U_{\mathrm{tot}}\in\uni(AB)$ be a unitary of the form 
  \begin{equation}
    U_{\mathrm{tot}} = (V_A^{(0)}\otimes W_B^{(0)}) \prod_{i=1}^n \left(U_{AB}^{(i)} (V_A^{(i)}\otimes W_B^{(i)}) \right) 
  \end{equation}
  where $V^{(i)}_A\in\uni(A)$, $W^{(i)}_B\in\uni(B)$ and the $U_{AB}^{(i)}\in\uni(AB)$ act as a two-qubit unitaries across $A$ and $B$.
  We then have
  \begin{equation}\label{eq_parallel-as-upper}
    \gamma_{\lo^{\star}}(\cU_{\mathrm{tot}}) \leq \gamma_{\lo^{\star}}\left(\bigotimes_{i=1}^n \cU^{(i)}\right)
  \end{equation}
  where $\cU_{\mathrm{tot}}$ and $\cU^{(i)}$ denote the induced channels of $U_{\mathrm{tot}}$ and $U^{(i)}$ respectively.
\end{corollary}
We note that a better upper bound is not possible without more knowledge about the specific structure of the unitary, as the bound is clearly achieved when the $V^{(i)},W^{(i)}$ are the identity and the $U^{(i)}$ act on different qubits.
Conversely, a meaningful lower bound for $\gamma_{\lo^{\star}}(\cU_{\mathrm{tot}})$ in terms of only the $\gamma_{\lo^{\star}}(\cU^{(i)})$ is also impossible, as generally the entangling gates $U^{(i)}$ could undo the effects of each other leading to $U$ being the identity.
\begin{proof}
  Combining the KAK decompositions of the individual two-qubit unitaries $U^{(i)}$, we immediately obtain a decomposition of $U_{\mathrm{tot}}$ that is of the form as in~\Cref{lem:gamma_unitary_upperbound}.
\end{proof}

\subsubsection{State conversion bounds}\label{sec:sepc_conversion_bound}
Consider a quantum channel $\cE\in\cptp(AB\rightarrow A'B')$.
As a consequence of the chaining property (recall~\Cref{cor:gamma_state_from_channel,lem:gamma_tp_idchan}), we have
\begin{equation}
  \gamma_{\sepc^{\star}}(\cE) \geq \gamma_{\sep(A'\bar{A};B'\bar{B})}\big((\idchan_{\bar{A}\bar{B}}\otimes \cE)(\rho)\big) 
\end{equation}
for any systems $\bar{A},\bar{B}$ and for any separable state $\rho\in\sep(A\bar{A};B\bar{B})$.
This in turn provides us with a lower bound for $\gamma_{\locc^{\star}}(\cE)$ and $\gamma_{\lo^{\star}}(\cE)$.
The power of this approach lies in the reduction to the state setting, as estimating $\gamma_{\sep}$ is often much easier.
The choice of $\rho$ is critical: Ideally, $(\idchan\otimes\cE)(\rho)$ should be as strongly entangled as possible.

We already used this lower bound in the proof of~\Cref{cor:gamma_locc_clifford,thm:gamma_kaklike_unitary}, where we chose $\rho$ to be the product of two maximally entangled states on $A\bar{A}$ and $B\bar{B}$.
In that case, $(\idchan\otimes\cE)(\rho)$ is precisely the Choi operator $\choi{\cE}$.
It is natural to ask whether this choice of $\rho$ is always the best.
The example below shows that this is not the case, even for unitary channels $\cE$.

\begin{example}\label{ex:toffoli}
  Consider the three-qubit Toffoli gate acting on the bipartite system $AB$ where $A$ consists of the two control qubits and $B$ of the target qubit.
  The Toffoli gate is equivalent to a CNOT gate across $A$ and $B$ up to some local operations, i.e., a Toffoli gate can be realized with one invocation of a CNOT gate and vice versa.
  The corresponding circuit identities are depicted in~\Cref{fig:toffoli_cnot_equiv}.
  \begin{figure}
    \centering
    \includegraphics{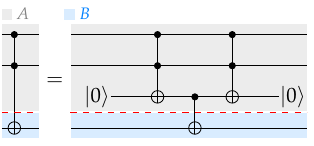}
    \quad \quad \quad
    \raisebox{0.3cm}{\includegraphics{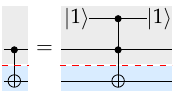}}
    \caption{A $\toffoli$ gate and a $\cnot$ gate are equally powerful in the sense that each of the two can be realized through the other using only local operations. The red dotted line denotes the partition into two subsystems. The ancilla qubits are uncomputed and end up in a well-defined single-qubit state.}
     \label{fig:toffoli_cnot_equiv}
  \end{figure}
  As such, it follows from the chaining property (recall~\Cref{lem:gamma_chaining}) that
  \begin{equation}
    \gamma_{\locc^{\star}}(\toffoli) = \gamma_{\lo^{\star}}(\toffoli) = \gamma_{\locc^{\star}}(\cnot) = 3 \, .
  \end{equation}
  However, the Choi-based lower bound does not achieve this value.
  A simple calculation reveals that the Schmidt coefficients of $\choi{\toffoli}$ are given by $\frac{1}{2}$ and $\frac{\sqrt{3}}{2}$, hence showing that
  \begin{equation}
    \gamma_{\sep}(\choi{\toffoli}) = 1 + \sqrt{3} \approx 2.732
  \end{equation}
  using~\Cref{thm:optimal_qpd_pure}.
  Yet, clearly we can pick a simple state like $\ket{1}\otimes\ket{+}\otimes\ket{0}$ to achieve the quasiprobability extent using the state conversion bound.
  Since the state
  \begin{equation}
    \toffoli \ket{1}\ket{+}\ket{0} = \frac{1}{\sqrt{2}} \left( \ket{10}\otimes\ket{0} + \ket{11}\otimes\ket{1} \right) 
  \end{equation}
  has Schmidt coefficients $(\frac{1}{\sqrt{2}},\frac{1}{\sqrt{2}})$, it achieves a quasiprobability extent of $3$.
\end{example}
Coming back to the general case, the strongest version of the state-based lower bound takes the form
\begin{equation}\label{eq:state_conversion_bound}
  \gamma_{\sepc^{\star}}(\cE) \geq \sup_{\bar{A},\bar{B},\rho\in\sep(A\bar{A};B\bar{B})} \gamma_{\sep}((\idchan\otimes\cE)(\rho)) \, .
\end{equation}
We can restrict the optimization to pure states and systems with equal dimension $\dimension(\bar{A})=\dimension(A)$, $\dimension(B)=\dimension(\bar{B})$, since
\begin{align}
  \sup_{\bar{A},\bar{B},\rho\in\sep(A\bar{A};B\bar{B})} \gamma_{\sep}((\idchan\otimes\cE)(\rho))
  &= \sup_{\bar{A},\bar{B}} \max_{\rho\in\sep(A\bar{A};B\bar{B})} \gamma_{\sep}((\idchan\otimes\cE)(\rho)) \\
  &= \sup_{\bar{A},\bar{B}} \max_{\substack{\ket{\psi}\in\mathcal{H}_{A\bar{A}}\\\ket{\phi}\in\mathcal{H}_{B\bar{B}}}} \gamma_{\sep}((\idchan\otimes\cE)(\proj{\psi}\otimes\proj{\phi})) 
\end{align}
where we used the convexity of $\gamma_{\sep}$.
For any given $\ket{\psi},\ket{\phi}$, we can restrict $\bar{A},\bar{B}$ to be the spaces that are spanned by the corresponding Schmidt basis.

The quantity on the right-hand side of~\Cref{eq:state_conversion_bound} is essentially the \emph{unassisted entangling capacity} introduced in \reference~\cite{campbell2010_optimal}.
The authors also showed that the lower bound cannot be improved by picking an already entangled state $\rho$ and using $\gamma_{\sepc^{\star}}(\cE) \geq \frac{\gamma_{\sep}(\cE(\rho))}{\gamma_{\sep}(\rho)}$, see~\cite[Corollary 1]{campbell2010_optimal}.


\subsubsection{\texorpdfstring{$\pptc$}{PPTC} lower bound}\label{sec:pptc_lower_bound}
One useful strategy to find lower bounds of $\gamma_{\lo^{\star}}$ and $\gamma_{\locc^{\star}}$ is to relax the decomposition set to $\pptc^{\star}$.
In full analogy to the SDP of $\gamma_{\ppt}$ in~\Cref{eq:gamma_ppt_sdp}, we can express $\gamma_{\pptc^{\star}}(\cE)$ for a channel $\cE\in\cptp(AB\rightarrow A'B')$ as an SDP
\begin{equation}\label{eq:gamma_pptc_sdp}
  \gamma_{\pptc^{\star}}(\cE) = \gamma_{\pptc}(\cE) = \left \lbrace
  \begin{array}{r l}
    \min\limits_{\Lambda^{\pm}\in\pos(ABA'B')} & \tr[\Lambda^+] + \tr[\Lambda^-] \\
    \textnormal{s.t.} & \choi{\cE} =\Lambda^+ - \Lambda^-, \\
    & \tr_{A'B'}[\Lambda^{\pm}] = \tr[\Lambda^{\pm}] \frac{1}{\dimension(AB)}\id_{AB}, \\
    & (\idchan_{AA'}\otimes\transpose_{BB'})(\Lambda^{\pm}) \loewnergeq 0
  \end{array} \right . \, .
\end{equation}
where we used~\Cref{prop:negtrick_sepc_pptc} and~\Cref{lem:two_element_qpd}.

\begin{lemma}{}{gamma_pptc_sdp_dual}
  The following is a valid dual SDP of~\Cref{eq:gamma_pptc_sdp}
  \begin{equation}
    \begin{array}{r l}
      \max\limits_{\substack{X\in\herm(ABA'B')\\ Y^{\pm}\in\herm(AB)\\ Z^{\pm}\in\pos(ABA'B')}} & \tr[X\choi{\cE}] \\
      \textnormal{s.t.} 
      & \id_{A'B'}\otimes Y^- + (\idchan_{AA'}\otimes\transpose_{BB'})(Z^-) - \id_{AA'BB'}(1+\frac{\tr[Y^-]}{\dimension(AB)}) \loewnerleq X, \\
      & X \loewnerleq - \id_{A'B'}\otimes Y^+ - (\idchan_{AA'}\otimes\transpose_{BB'})(Z^+) + \id_{AA'BB'}(1+\frac{\tr[Y^+]}{\dimension(AB)})
    \end{array} 
  \end{equation}
\end{lemma}
The proof is very similar to the one of~\Cref{lem:gamma_ppt_sdp_dual}.
\begin{proof}
  In the following, we will denote the unnormalized Hilbert-Schmidt inner product by $\langle A,B\rangle\coloneqq\tr[A^{\dagger}B]$.
  The Lagrangian is given by
  {\small
  \begin{align}
    L(\Lambda^{\pm};X,Y^{\pm},Z^{\pm}) &= \tr[\Lambda^+] + \tr[\Lambda^-] - \langle X, \choi{\cE} -\Lambda^++\Lambda^-\rangle \nonumber\\
    &- \langle Y^+,\tr_{A'B'}[\Lambda^+]-\frac{\tr[\Lambda^+]}{\dimension(AB)}\id_{AB}\rangle \nonumber\\
    &- \langle Y^-,\tr_{A'B'}[\Lambda^-]-\frac{\tr[\Lambda^-]}{\dimension(AB)}\id_{AB}\rangle \nonumber\\
    &- \langle Z^+,(\idchan_{AA'}\otimes\transpose_{BB'})(\Lambda^+)\rangle - \langle Z^-,(\idchan_{AA'}\otimes\transpose_{BB'})(\Lambda^-)\rangle \\
    &= -\langle X,\choi{\cE}\rangle \nonumber\\ 
    &+ \langle\Lambda^+,\id_{ABA'B'}(1+\frac{\tr[Y^+]}{\dimension(AB)}) + X - \id_{A'B'}\otimes Y^+ - (\idchan_{AA'}\otimes\transpose_{BB'})(Z^+)\rangle \nonumber\\
    &+ \langle\Lambda^-,\id_{ABA'B'}(1+\frac{\tr[Y^-]}{\dimension(AB)}) - X - \id_{A'B'}\otimes Y^- - (\idchan_{AA'}\otimes\transpose_{BB'})(Z^-)\rangle \, .
  \end{align}
  }
  where we used the relation $\langle R_{A},\tr_B[S_{AB}]\rangle = \langle R_{A}\otimes\mathds{1}_B,S_{AB}\rangle$.
  The constraints of the dual SDP thus read
  \begin{equation}
    \id_{ABA'B'}(1+\frac{\tr[Y^{\pm}]}{\dimension(AB)}) \pm X - \id_{A'B'}\otimes Y^{\pm} - (\idchan_{AA'}\otimes\transpose_{BB'})(Z^{\pm}) \loewnergeq 0
  \end{equation}
  and the desired form of the dual SDP is obtained with the substitution $X\rightarrow -X$.
\end{proof}

For a unitary channel $\cU\in\cptp(AB)$, the $\pptc$ lower bound is provably at least as good as the state conversion bound discussed in~\Cref{sec:sepc_conversion_bound}
\begin{align}
  \gamma_{\pptc}(\cU)
  &\geq \max\limits_{\ket{\psi}\text{ product state}} \gamma_{\ppt}((\cU\otimes\idchan_{\bar{A}\bar{B}})(\proj{\psi})) \\
  &= \max\limits_{\ket{\psi}\text{ product state}} \gamma_{\sep}((\cU\otimes\idchan_{\bar{A}\bar{B}})(\proj{\psi}))
\end{align}
where we used that $\gamma_{\sep}$ and $\gamma_{\ppt}$ coincide on pure states (recall~\Cref{lem:optimal_qpd_pure_ppt}).
In the context of~\Cref{ex:toffoli}, this implies that $\gamma_{\pptc}(\toffoli)=3$.
This can readily be verified using a numerical solver.

\begin{example}
  We consider the noisy CNOT channel
  \begin{equation}
    \cE_p(\rho) \coloneqq (1-p)\cnot(\rho) + p\mathcal{D}(\rho)
  \end{equation}
  where $\cD$ is the fully depolarizing two-qubit channel.
  In~\Cref{fig:noisy_cnot}, we depict the numerical evaluation of $\gamma_{\pptc^{\star}}(\cE_p)$ as well as $\gamma_{\ds}(\cE_p)$ where $\ds$ is a heuristically chosen finite decomposition set containing
  \begin{itemize}
    \item the elements $\cF\in\ds_{\cnot}$ where $\ds_{\cnot}$ is the decomposition set achieving the optimal $\lo^{\star}$ QPD of the CNOT gate discussed in~\Cref{ex:cnot_lostar_qpd},
    \item the two-qubit Pauli channels,
    \item the superoperators of the form $\cG\circ\cF$ where $\cF\in\ds_{\cnot}$ and $\cG$ is a two-qubit Pauli channel.
  \end{itemize}
  Since $\ds\subset\lo^{\star}$, we know that $\gamma_{\lo^{\star}}(\cE_p)$ and $\gamma_{\locc^{\star}}(\cE_p)$ must lie in the shaded region.
  We observe that $\gamma_{\ds}$ does not perform better than simply taking the trivial bounds implied by the convexity of the quasiprobability extent.
  Interestingly, $\gamma_{\lo^{\star}}(\cE_p)$ and $\gamma_{\locc^{\star}}(\cE_p)$ are strictly larger than the corresponding extent of the noisy Bell state $\gamma_{\sep}(\rho_p)$ that we considered in~\Cref{ex:noisy_bell_state}.
\end{example}
\begin{figure}
  \centering
  \includegraphics{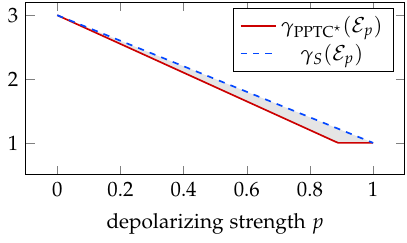}
  \caption{Numerical evaluation of $\gamma_{\pptc^{\star}}(\cE_p)$ and $\gamma_{\ds}(\cE_p)$ for the depolarized CNOT gate $\cE_p$. Here, $\ds$ is a heuristically chosen decomposition set that is contained in $\lo^{\star}$.}
  \label{fig:noisy_cnot}
\end{figure}

\subsubsection{SDP hierarchy for \texorpdfstring{$\gamma_{\sepc}$}{gamma SEPC}}\label{sec:sdp_hierarchy_channels}
In the setting of states, we have previously seen in~\Cref{sec:gamma_sep_singleshot} that $\gamma_{\sep}$ can be expressed in terms of an efficiently converging SDP hierarchy.
Unfortunately, this approach cannot be directly applied to $\gamma_{\locc^{\star}}$.
Instead, we have to content ourselves with using the DPS hierarchy to evaluate the quasiprobability extent w.r.t. to the set of channels with separable Choi operators $\sepc$.
This in turn at least provides a lower bound $\gamma_{\lo^{\star}}\geq \gamma_{\locc^{\star}}\geq \gamma_{\sepc^{\star}}=\gamma_{\sepc}$ thanks to~\Cref{prop:negtrick_sepc_pptc}.
The techniques will mirror closely the approach we previously introduced in~\Cref{sec:gamma_sep_singleshot}.
As such, we will not delve into as many details and rather focus on the differences and difficulties that specifically occur in this setting.

We define an outer approximation of the set $\sepc(A;B\rightarrow A';B')$ as
\begin{equation}
   \sepc_k(A;B\rightarrow A';B') \coloneqq \{\cE\in\cptp(AB) | \choi{\cE}\in \sep_k(AA';BB') \} \, .
\end{equation}
As for $\sep_k$, we can optimize over elements of $\sepc_k$ by considering symmetric extensions in the larger space $AB_1\dots B_kA'B_1'\dots B_k'$ where the $B_i$ and $B_j'$ are copies of $B$ or $B'$ respectively.
Following the same steps in the derivation of $\gamma_{\sep_k}$, we obtain
\begin{equation}
  \gamma_{\sepc_k}(\cE) = \left \lbrace
  \begin{array}{r l}
    \min\limits_{\Lambda^{\pm}\in\pos(AB_1\dots B_kA'B_1'\dots B_k')} & \tr[\Lambda^+] + \tr[\Lambda^-] \\
    \textnormal{s.t.} & \choi{\cE} = \tr_{B_2B_2'\dots B_kB_k'}[\Lambda^+] - \tr_{B_2B_2'\dots B_kB_k'}[\Lambda^-] , \\
    & \tr_{A'B_2\dots B_kB_1'\dots B_k'}[\Lambda^{\pm}] = \frac{1}{\dimension(AB)}\id_{AB_1} , \\
    & \swap_{B_iB_i',B_jB_j'} \Lambda^{\pm} \swap_{B_iB_i',B_jB_j'} = \Lambda^{\pm} , \\
    & \left(\idchan_{AA'}\otimes \transpose_{B_1B_1'} \otimes \dots \otimes \idchan_{B_kB_k'}\right)(\Lambda^{\pm}) \loewnergeq 0 , \\
    & \quad\quad \vdots \\
    & \left(\idchan_{AA'}\otimes \transpose_{B_1B_1'} \otimes \dots \otimes \transpose_{B_kB_k'}\right)(\Lambda^{\pm}) \loewnergeq 0
  \end{array} \right . \, .
\end{equation}
The first level of the hierarchy precisely corresponds to $\sepc_1(A;B\rightarrow A';B')=\pptc(A;B\rightarrow A';B')$.
The completeness of the DPS hierarchy (see~\Cref{prop:dps_completeness}) implies that this outer approximation will also converge, i.e.,
\begin{equation}
  \sepc(A;B\rightarrow A';B') = \bigcap\limits_{k=1}^{\infty} \sepc_k(A;B\rightarrow A';B') \, .
\end{equation}
Unfortunately, the statement about convergence speed in~\Cref{prop:dps_convergence} cannot be directly translated to this setting.
Given the optimal QPD $\cE=a^+\cE^+ - a^-\cE^-$, the de Finetti theorem only guarantees that $\choi{\cE^+}$ and $\choi{\cE^-}$ are close to separable operators, but not necessarily to separable operators that fulfill the partial trace constraint of a trace-preserving channel.
There exist modifications of the de Finetti theorem that allow for imposing linear constraints on the approximating state, such as the result by Berta \etal~in \reference~\cite[Theorem 2.3]{berta2021_semidefinite}.
However, this modified de Finetti theorem can only incorporate linear constraints that separately act on $AA'$ and $BB'$, and is therefore not applicable to our case.
A slightly different modification would be required, which to our knowledge has not been established in previous literature.


\subsubsection{Overview of one-shot bounds}\label{sec:overview_channel_bounds}
Here, we come back to the chain of inequalities we previously introduced in~\Cref{eq:channel_gamma_simple_bounds} and try to complete that picture with a revised understanding of which of the inequalities are strict and which ones are not.
We summarize our results in~\Cref{fig:bounds_overview}, which also includes the Choi lower bounds.
The relations $\gneqq$ and $\lneqq$ indicate that one quantity is generally at least/at most as large as the other and that there exists a specific counter example of a CPTP map $\cE$ for which the two quantities do not coincide.
\begin{figure}
  \centering
  \includegraphics{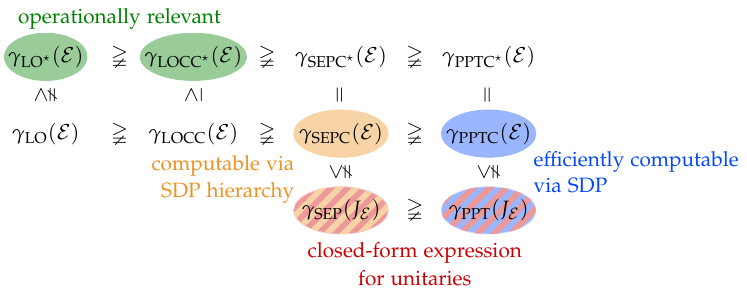}
  \caption{Relations between various quasiprobability extents for some CPTP map $\cE$. The relations $\gneqq$ and $\lneqq$ indicate that one quantity is generally at least/at most as large as the other, and there exists a specific counter example of a quantum channel $\cE$ for which the two quantities do not coincide.}
  \label{fig:bounds_overview}
\end{figure}
The following relations are directly implied by previously shown results.
\begin{itemize}
  \item $\gamma_{\lo}\gneqq \gamma_{\lo^{\star}}$ and $\gamma_{\lo}\gneqq \gamma_{\locc}$ due to~\Cref{lem:span_lo,lem:span_locc}.
  \item $\gamma_{\sepc^{\star}}=\gamma_{\sepc}$ and $\gamma_{\pptc^{\star}}=\gamma_{\pptc}$ due to~\Cref{prop:negtrick_sepc_pptc}.
  \item $\gamma_{\sepc}(\cE)\gneqq\gamma_{\sep}(\choi{\cE})$ and $\gamma_{\pptc}(\cE)\gneqq\gamma_{\ppt}(\choi{\cE})$ by the example of the Toffoli gate.
  Recall from~\Cref{ex:toffoli} that $\gamma_{\sep}(\choi{\toffoli})=\gamma_{\ppt}(\choi{\toffoli})\approx 2.732$.
  At the same time, we also saw that $\gamma_{\pptc}(\toffoli)=\gamma_{\sepc}(\toffoli)=3$.
\end{itemize}
For two decomposition sets $\ds\subsetneqq\ds'\subset\cptp$ such that $\ds$ is convex and compact, we have $1=\gamma_{\ds'}(\cF) < \gamma_{\ds}(\cF)$ for any $\cF\in\ds'\setminus\ds$ due to~\Cref{lem:gamma_faithful2}.
This implies following insights.
\begin{itemize}
  \item $\gamma_{\sepc}\gneqq\gamma_{\pptc}$ because $\sepc$ is convex and compact and $\pptc$ is strictly larger than $\sepc$, due to the existence of PPT entangled states (recall that we view states as a special case of channels).
  \item For the same reason, $\gamma_{\sep}(\choi{\cE})\gneqq\gamma_{\ppt}(\choi{\cE})$.
  \item To see that $\gamma_{\locc}\gneqq\gamma_{\sepc}$, notice that $\locc\subset\overline{\locc}\subset\sepc$ where $\overline{\locc}$ is the closure of $\locc$ and we used that $\sepc$ is closed.
  Since there exist channels in $\sepc$ which do not lie in $\overline{\locc}$ (this has been demonstrated in the context of state discrimination~\cite{bennett1999_quantum,koashi2007_quantum,childs2013_framework} and random distillation~\cite{cui2011_randomly,chitambar2012_increasing}), one thus finds that
  \begin{equation}
    \gamma_{\locc} \geq \gamma_{\overline{\locc}} \gneqq \gamma_{\sepc} \, .
  \end{equation}
\end{itemize}
The inequality $\gamma_{\lo^{\star}}\gneqq \gamma_{\locc^{\star}}$ is a bit more involved and operationally very important, since it implies that classical communication is advantageous for the simulation of certain channels.
\begin{lemma}{}{}
  There exist quantum channels $\cE\in\cptp(AB\rightarrow A'B')$ such that $\gamma_{\lo^{\star}}(\cE) > \gamma_{\locc^{\star}}(\cE)=1$.
\end{lemma}
We briefly note that we will later even encounter an explicit quantum channel $\cE$ for which $1<\gamma_{\locc^{\star}}(\cE)<\gamma_{\lo^{\star}}(\cE)$ in~\Cref{sec:id_wire_cut}.
\begin{proof}
  Let us consider some $\cE$ that lies in $\locc(A;B\rightarrow A';B')$ but not in $\conv(\lo(A;B\rightarrow A';B'))$.
  As, such, we clearly have $\gamma_{\locc^{\star}}(\cE) = 1$.
  Assume that it would also hold that $\gamma_{\lo^{\star}}(\cE) = 1$.
  This means that for any $\epsilon > 0$ there exists a QPD $\cE=\sum_ia_i\cF_i$ with $\cF_i\in\lo^{\star}$ such that $\norm{a}_1\leq 1+\epsilon$.
  For every $\cF_i$, let $(\cG^{(i)}_j)_j\in\instr{\lo}(A;B\rightarrow A';B')$ be the associated instrument s.t. $\cF_i=\sum_j\beta_j\cG^{(i)}_j$.
  Define the superoperator $\tilde{\cE}\coloneqq \sum_{i,j}\abs{a_i}\cG_{j}^{(i)}$.
  We clearly have
  \begin{equation}
    \choi{\cE} = \sum_{i,j}a_i\beta_j^{(i)}\choi{\cG_j^{(i)}} \loewnerleq \sum_{i,j}\abs{a_i}\choi{\cG_j^{(i)}} = \choi{\tilde{\cE}}
  \end{equation}
  and hence
  \begin{align}
    \norm{\choi{\cE}-\choi{\tilde{\cE}}}_1 
    &= \tr\left[\choi{\tilde{\cE}}-\choi{\cE}\right] \\
    &= \sum_{i,j}\abs{a_i}\tr[\choi{\cG_{j}^{(i)}}] - 1 \\
    &= \norm{a}_1 - 1 \\
    &\leq \epsilon \, .
  \end{align}
  The channel $\frac{1}{\norm{a}_1}\tilde{\cE}$ lies in $\conv(\lo)$ and approximates $\cE$ up to error
  \begin{equation}
    \norm{\choi{\cE}-\frac{1}{\norm{a}_1}\choi{\tilde{\cE}}}_1
    \leq \norm{\choi{\cE} - \choi{\tilde{\cE}}}_1 + \norm{\choi{\tilde{\cE}} - \frac{1}{\norm{a}_1}\choi{\tilde{\cE}}}_1
    \leq \epsilon + \frac{\epsilon}{1+\epsilon} \leq 2\epsilon \, .
  \end{equation}
  In summary, we have shown that any such $\cE$ can be approximated by an element of $\conv(\lo)$ to arbitrary small error.
  But since $\conv(\lo)$ is compact, this would also imply that $\cE\in\conv(\lo)$, which is a contradiction.
\end{proof}
The argument in the proof can be recycled to show a separation between $\gamma_{\locc^{\star}}$ and $\gamma_{\sepc^{\star}}$ since there exist separable channels which can't be arbitrary well approximated by $\locc$ channels (recall the separation between $\overline{\locc}$ and $\sepc$).
More precisely, for any $\cE\in\sepc(AB\rightarrow A'B')\setminus\overline{\locc}(AB\rightarrow A'B')$ we have
\begin{equation}
  \gamma_{\locc^{\star}}(\cE) > \gamma_{\sepc^{\star}}(\cE) = 1 \, .
\end{equation}
It remains unclear, whether there exist quantum channels $\cE$ for which $\gamma_{\locc}(\cE)\neq \gamma_{\locc^{\star}}(\cE)$.
This question amounts to whether classical side information is useful for circuit knitting with classical communication.

\subsubsection{Asymptotic characterization with entanglement cost}\label{sec:asymptotic_channel_extent}
The main achievement of~\Cref{sec:gamma_sep_asymptotic} was to establish a connection between the quasiprobability extent of states and various entanglement cost measures.
In this section, we will replicate this result as much as possible in the setting of channels.
Unfortunately, its unknown if this can be achieved for the decomposition sets $\lo^{\star}$ and $\locc^{\star}$.
As such, we will restrict ourselves to studying $\gammareg_{\sepc}$, $\gammasreg_{\sepc}$, $\gammareg_{\pptc}$ and $\gammasreg_{\pptc}$, which at least provide interesting lower bounds.

Since the results in this section are completely analogous to the previously discussed case for states, we will introduce and summarize them in more brevity.
Furthermore, we note that the dynamic resource theory of entanglement is much less developed for channels compared to states, so not all aspects are equally well understood.

We require some notion of the entanglement cost of a quantum channel, i.e., a characterization of how many Bell pairs are required to realize it.
This is achieved by adapting~\Cref{def:cost_and_distillable} by replacing states with channels, channels with superchannels and the trace distance with the diamond  distance~\cite{gour2020_dynamical}.
\begin{definition}{}{}
  Let $\epsilon \geq 0$ and $\cE\in\cptp(AB\rightarrow A'B')$.
  The \emph{single-shot entanglement cost} of $\cE$ w.r.t. some set of superchannels $\ds$ is defined as
  \begin{equation}
    E_{C,\ds}^{(1),\epsilon}(\cE)\coloneqq \min_{m\in\mathbb{N}} \{\log m | \exists \cF\in\ds : \frac{1}{2}\dnorm{\cE-\cF(\proj{\Psi_m})} \leq \epsilon \}
  \end{equation}
  where $\proj{\Psi_m}$ denotes the maximally entangled state of Schmidt rank $m$ and should be understood as a quantum channel with trivial input space.
\end{definition}
The asymptotic entanglement cost and exact asymptotic entanglement cost are then given by
\begin{equation}
  E_{C,\ds}(\cE) \coloneqq \lim\limits_{\epsilon\rightarrow 0^+} \liminf\limits_{n\rightarrow\infty}\frac{1}{n}E_{C,\ds}^{(1),\epsilon}(\cE^{\otimes n})
  \quad\text{and}\quad
  E_{C,\ds}^{\mathrm{exact}}(\cE) \coloneqq \lim\limits_{n\rightarrow\infty}\frac{1}{n}E_{C,\ds}^{(1),0}(\cE^{\otimes n}) \, .
\end{equation}
We will study the entanglement cost under following two sets of superchannels.
The \emph{separability-preserving superchannels} $\seppsc(AB,A'B'\rightarrow CD,C'D')$ contain all superchannels $\Theta\in\schan(AB,A'B'\rightarrow CD,C'D')$ such that
\begin{equation}
  \forall \cE\in \sepc(A;B\rightarrow A';B') : \Theta(\cE)\in \sepc(C;D\rightarrow C';D') \, .
\end{equation}
Furthermore, the set of \emph{PPT superchannels} $\pptsc(AB,A'B'\rightarrow CD,C'D')$ contains all $\Theta\in\schan(AB,A'B'\rightarrow CD,C'D')$ which have a Choi matrix $\choi{\Theta}$ with positive partial transpose $(\idchan_{A,A',C,C'}\otimes\transpose_{B,B',D,D'})(\choi{\Theta})\loewnergeq 0$.
It can be shown that $\pptsc$ precisely constitutes the set of superchannels which completely preserve PPTC~\cite[Theorem 5]{gour2021_entanglement}.

The one-shot entanglement cost of $\cE\in\cptp(AB\rightarrow A'B')$ w.r.t. $\seppsc$ can be characterized by the smooth robustness
\begin{equation}
  R_{\sepc}^{\epsilon}(\cE) \coloneqq \min\limits_{\cF\in B_{\epsilon}(\cE)}R_{\sepc}(\cF)
\end{equation}
where $\epsilon\geq 0$ and $B_{\epsilon}(\cE) \coloneqq \{\cF\in\cptp(AB\rightarrow A'B') | \frac{1}{2}\dnorm{\cE-\cF} \leq \epsilon\}$.
\begin{lemma}{\cite[Theorem III.1]{kim2021_oneshot}}{oneshot_seppsc_cost}
  Let $\cE\in\cptp(AB\rightarrow A'B')$ and $\epsilon\geq 0$.
  Then
  \begin{equation}
    \log(1+R_{\sepc}^{\epsilon}(\cE)) \leq E_{C,\seppsc}^{(1),\epsilon}(\cE) \leq \log(1+R_{\sepc}^{\epsilon}(\cE)) + 2 \, .
  \end{equation}
\end{lemma}
As a consequence, we get the statement corresponding to~\Cref{thm:gamma_sep_is_sepp_cost}.
\begin{theorem}{}{}
  For any bipartite channel $\cE\in\cptp(AB\rightarrow A'B')$, the regularized and asymptotic quasiprobability extent w.r.t. $\sepc$ can be expressed as
  \begin{equation}
    \log \gammareg_{\sepc}(\cE) = E_{C,\seppsc}^{\mathrm{exact}}(\cE)
    \quad\text{and}\quad
    \log \gammasreg_{\sepc}(\cE) = E_{C,\seppsc}(\cE) \, .
  \end{equation}
\end{theorem}
We omit the proof because it is completely analogous.
To characterize the one-shot entanglement cost of a channel w.r.t. $\pptsc$, we need to define a dynamical resource measure called \emph{max-logarithmic negativity} which was introduced in \reference~\cite{gour2021_entanglement}.
\begin{definition}{}{}
  The max-logarithmic negativity of $\cE\in\cptp(AB\rightarrow A'B')$ is defined as $\maxlogneg(\cE)\coloneqq \max\{ \maxlogneg^{(0)}(\cE), \maxlogneg^{(1)}(\cE) \}$ where
 \begin{equation}
  2^{\maxlogneg^{(0)}(\cE)} \coloneqq \left \lbrace
  \begin{array}{r l}
    \inf\limits_{\Lambda\in\pos(ABA'B')} & \opnorm{\tr_{A'B'}[\Lambda]} \dimension(AB) \\
    \textnormal{s.t.} & - (\idchan_{AA'}\otimes\transpose_{BB'})(\Lambda) \loewnerleq (\idchan_{AA'}\otimes\transpose_{BB'})(\choi{\cE}) \\
    & (\idchan_{AA'}\otimes\transpose_{BB'})(\choi{\cE}) \loewnerleq (\idchan_{AA'}\otimes\transpose_{BB'})(\Lambda) \}
  \end{array} \right . \, ,
  \end{equation}
 \begin{equation}
  2^{\maxlogneg^{(1)}(\cE)} \coloneqq \left \lbrace
  \begin{array}{r l}
    \inf\limits_{\Lambda\in\pos(ABA'B')} & \opnorm{(\idchan_{A}\otimes\transpose_{B})(\tr_{A'B'}[\Lambda])} \dimension(AB) \\
    \textnormal{s.t.} & - (\idchan_{AA'}\otimes\transpose_{BB'})(\Lambda) \loewnerleq (\idchan_{AA'}\otimes\transpose_{BB'})(\choi{\cE}) \\
    & (\idchan_{AA'}\otimes\transpose_{BB'})(\choi{\cE}) \loewnerleq (\idchan_{AA'}\otimes\transpose_{BB'})(\Lambda) \}
  \end{array} \right . \, .
  \end{equation}
\end{definition}
Note that $\maxlogneg^{(0)}$ and $\maxlogneg^{(1)}$ are both defined in terms of an SDP and can hence be efficiently computed.
It can be shown that $\maxlogneg(\cE)$ reduces to the $\kappa$-entanglement when $\dimension(A)=\dimension(B)=1$~\cite{gour2021_entanglement}.

\begin{lemma}{\cite[Lemma 9]{gour2021_entanglement}}{oneshot_pptsc_cost}
  Let $\cE\in\cptp(AB\rightarrow A'B')$ and $\epsilon\geq 0$.
  Then
  \begin{equation}
    2^{\maxlogneg(\cE)} - 1 \leq 2^{E_{C,\pptsc}^{(1),0}(\cE)} \leq 2^{\maxlogneg(\cE)} + 2  \, .
  \end{equation}
\end{lemma}
\begin{lemma}{~\cite[Lemma 5]{jing2024_circuit}}{}
  For any $\cE\in\cptp(AB\rightarrow A'B')$ one has
  \begin{equation}
    \gamma_{\pptc}(\cE) \geq 2^{\maxlogneg(\cE)} - 1
  \end{equation}
\end{lemma}
\begin{proof}
  Consider the achieving QPD $\cE=(1+t)\cF^+-t\cF^-$ for $\cF^{\pm}\in\pptc$.
  Since $\choi{\cF^{\pm}}\in\ppt$ we have
  \begin{align}
    -(1+t)(\id_{AA'}\otimes\transpose_{BB'})(\choi{\cF^+}+\choi{\cF^-})
    &\leq (\id_{AA'}\otimes\transpose_{BB'})(\choi{\cE}) \\
    &\leq (1+t)(\id_{AA'}\otimes\transpose_{BB'})(\choi{\cF^+}+\choi{\cF^-}) \, .
  \end{align}
  Hence, $\tilde{\Lambda}\coloneqq (1+t)(\choi{\cF^+}+\choi{\cF^-})$ is a feasible point for $\maxlogneg(\cE)$.
  Notice that
  \begin{equation}
    \frac{2(1+t)}{\dimension(AB)}\id_{AB} = \tr_{A'B'}[\tilde{\Lambda}] = (\idchan_{A}\otimes\transpose_{B})(\tr_{A'B'}[\tilde{\Lambda}])
  \end{equation}
  which implies
  \begin{align}
    2^{\maxlogneg(\cE)} 
    &\leq \max\left(\opnorm{\tr_{A'B'}[\tilde{\Lambda}]}, \opnorm{(\idchan_{A}\otimes\transpose_{B})(\tr_{A'B'}[\tilde{\Lambda}]} \right) \dimension(AB) \\
    &\leq \max\left( \frac{2(1+t)}{\dimension(AB)}, \frac{2(1+t)}{\dimension(AB)} \right) \dimension(AB) \\
    &= (1+2t)+1 \\
    &= \gamma_{\pptc}(\cE) + 1 \, .
  \end{align}
\end{proof}

\begin{theorem}{}{}
  The regularized quasiprobability extent of $\cE\in\cptp(AB\rightarrow A'B')$ w.r.t. $\pptc$ is given by
  \begin{equation}
    \log \gammareg_{\pptc}(\cE) = \lim\limits_{n\rightarrow\infty}\frac{1}{n}\maxlogneg(\cE^{\otimes n}) = E_{C,\pptsc}^{\mathrm{exact}}(\cE) \, .
  \end{equation}
\end{theorem}
\begin{proof}
  The second equality follows directly from regularizing the statement in \Cref{lem:oneshot_pptsc_cost}.
  The first equality follows by regularizing all terms
  \begin{equation}
    \log(2^{\maxlogneg(\cE)}-1) \leq \log \gamma_{\pptc}(\cE) \leq \log( 2^{E_{C,\pptc}^{(1),0}(\cE)+1}-1 )
  \end{equation}
  where the second inequality follows from the monotonicity of $\gamma_{\pptc}$ in the sense of a dynamical resource monotone (recall~\Cref{prop:log_gamma_dynamic_monotone}).
  More precisely, if we have a PPT superchannel $\Theta$ that realizes $\cE$ from the maximally entangled state $\proj{\Psi_d}$, then $\gamma_{\pptc}(\cE)$ can be at most equal to $\gamma_{\pptc}(\proj{\Psi_d})=\gamma_{\ppt}(\proj{\Psi_d}) = 2d - 1$ where we used~\Cref{lem:optimal_qpd_pure_ppt}.
\end{proof}

Previous work has claimed that $\maxlogneg$ is an additive quantity~\cite{gour2021_entanglement}, which would imply that $\lim\limits_{n\rightarrow\infty}\frac{1}{n}\maxlogneg(\cE^{\otimes n})=\maxlogneg(\cE)$ is efficiently computable.
However, this is wrong by the same counter-example that the $\kappa$-entanglement is not additive.
To our knowledge, there is no known method at this time to efficiently evaluate $\log \gammareg_{\pptc}(\cE)$.
It remains an open question whether the techniques from~\Cref{prop:pptc_cost_efficient_algo} can be adapted from the state setting to the channel setting.

\begin{example}
  We again consider the noisy CNOT channel
  \begin{equation}
    \cE_p(\rho) \coloneqq (1-p)\cnot(\rho) + p\mathcal{D}(\rho)
  \end{equation}
  where $\cD$ is the fully depolarizing two-qubit channel.
  In~\Cref{fig:noisy_cnot_regularized} we depict the numerical evaluation of $\gamma_{\pptc}(\cE_p)$ and $2^{\maxlogneg(\cE_p)}$.
  The latter provides an upper bound for the regularized quasiprobability extent $\gammareg_{\pptc}(\cE_p)$ which must hence lie in the gray region.
  Remarkably, the max-logarithmic negativity coincides with the exact result for $\gammareg_{\pptc}(\cE_p)$ known at $p=0$ (from~\Cref{ex:clifford_gates}), hinting at the possibility that it might be additive for $\cE_p$.
\end{example}
\begin{figure}
  \centering
  \includegraphics{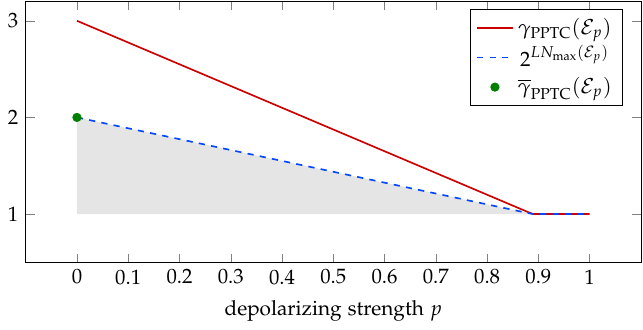}
  \caption{Numerical evaluation of $\gamma_{\pptc}(\cE_p)$ and $2^{\maxlogneg(\cE_p)}$ for the depolarized CNOT gates $\cE_p$}
  \label{fig:noisy_cnot_regularized}
\end{figure}

\subsection{Exploiting strict sub-multiplicativity in the black box setting}\label{sec:submult}
The simplest approach for circuit knitting is to separately find space-like cuts for each individual nonlocal operation in a circuit, and then combine these into a QPD of the total circuit.
However, treating all the nonlocal components individually is generally not optimal, and in some instances it is possible to find better QPDs of the circuit.
For instance, we observed that the quasiprobability extents $\gamma_{\sep},\gamma_{\lo^{\star}}$ and $\gamma_{\locc^{\star}}$ exhibit strictly sub-multiplicative behavior.
Therefore, cutting multiple \emph{parallel} instances of some operation (i.e. operations that occur in the same time slice of the circuit) is strictly cheaper than cutting them individually.
Unfortunately, in practical situations, the nonlocal gates/states may not occur in the same time slice, but rather at possibly distant locations in the circuit.
In this section, we will explore more refined QPS protocols that operate in this more realistic regime and how they can sometimes still achieve a sampling overhead that is the same as if the operations were all performed in parallel.

From now on, we will call this setup with non-parallel gates the \emph{black box cut setting} to differentiate it from the \emph{parallel cut setting}.
The two settings, exemplified on three two-qubit gates, are depicted in~\Cref{fig:gate_cutting_settings}.
The name black box cut is chosen to reflect that the circuit knitting protocol should work independently of what the operations $\cB_1$ and $\cB_2$ do, since we assume that we generally do not have any further information about them.
\begin{figure}
    \centering
    \begin{subfigure}[b]{0.49\textwidth}
        \centering
        \includegraphics{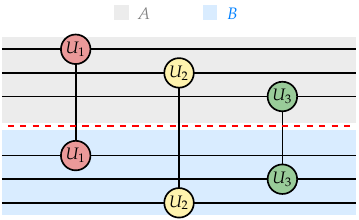}
        \caption{Parallel cut}
        \label{fig:parallelCut}
    \end{subfigure}        
    \begin{subfigure}[b]{0.49\textwidth}
        \centering
        \includegraphics{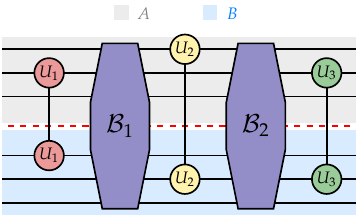}
        \caption{Black box cut}
        \label{fig:blackBoxCut}
    \end{subfigure}
\caption{Graphical depiction of the two cutting scenarios for multiple space-like cuts. The dashed line represents the cut. The unitaries $U_i$ are two-qubit gates acting on $A$ and $B$. In the parallel cut setting, the unitaries act on distinct qubits and can be considered to occur in the same time slice (not depicted here for visual clarity). In the black box setting, $\cB_1$ and $\cB_2$ denote arbitrary unknown CPTP maps, which represent all operations in the circuit that occur between $U_1,U_2$ and between $U_2,U_3$ respectively.}
\label{fig:gate_cutting_settings}
\end{figure}

We now mathematically formalize the notion of the black box cut.
Consider a quantum channel $\cG\in\cptp(AB)$ that is a composition of $n$ nonlocal gates $\cE_i\in\cptp(A_iB_i)$, where $A_i,B_i$ are subsystems of $A,B$, interleaved by some black box operations $\cB_j\in\cptp(AB)$
\begin{equation}
  \cG\coloneqq \cE_n \circ \cB_{n-1} \circ \cE_{n-1} \circ \cB_{n-2} \circ \cdots \circ \cE_{2} \circ \cB_{1} \circ \cE_{1} \, .
\end{equation}
We are interested in finding QPDs of $\cG$ of the form
\begin{equation}\label{eq:bb_cut_qpd}
  \cG(\rho) = \sum_{i=1}^m a_i \tr_{E_AE_B} \left[ \cF_n^{(i)}\circ \cB_{n-1} \circ \cF_{n-1}^{(i)} \circ \cB_{n-2} \circ \dots \circ \cB_{1} \circ \cF_1^{(i)} (\rho\otimes\proj{0}_{E_AE_B}) \right]
\end{equation}
where the $\smash{\cF_j^{(i)}}$ are either in $\lo^{\star}(A_iE_A;B_iE_B)$ or in $\locc^{\star}(A_iE_A;B_iE_B)$\footnote{We briefly note that choosing $\cF_j^{(i)}$ to be either in $\lo^{\star}$ or $\locc^{\star}$ effectively means that the weighting of the classical side information is done independently for each gate. A priori, this is not the most general scheme imaginable, as the weighting could be done in a correlated fashion across the individual gates. However, this does not affect the lowest possible overhead, as such a correlated weighting scheme could also be effectively achieved by passing classical information between the operations $\cF_j^{(i)}$ through the ancilla systems $E_A,E_B$.}.
Here, the systems $E_A$ and $E_B$ capture the notion that one might use ancilla systems on both partitions to store information between the nonlocal gates $\cE_i$.
$E_A$ and $E_B$ are initialized in some arbitrary product state $\proj{0}_{E_AE_B}$.
Graphically, this can be depicted as follows for $n=3$
\begin{equation}
  \includegraphics{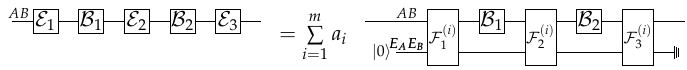}
  \, .
\end{equation}
Crucially, we are interested in QPDs such that this equality holds for any choice of the black box channels $\cB_j$.
In summary, the optimal sampling overhead in the black box setting is captured by following quantity.
\begin{definition}{}{gamma_bb}
  Let $n\in\mathbb{N}$, $\cE_i\in\hp(A_iB_i)$ for $i=1,\dots,n$ and let $\ds$ be either $\lo^{\star}$ or $\locc^{\star}$.
  We define the \emph{black box quasiprobability extent} $\gammabb_{\ds}(\cE_1,\cE_2,\dots,\cE_n)$ to be the infimum of the 1-norm $\norm{a}_1$ over all valid black box QPDs with $\cF_j^{(i)}\in\ds(A_iE_A;B_iE_B)$ for which~\Cref{eq:bb_cut_qpd} holds for every system $AB$, for every embedding of $A_iB_i$ into $AB$ and for every choice of $\cB_1,\dots,\cB_{n-1}\in\cptp(AB)$.
\end{definition}
We briefly note two trivial bounds for the black box quasiprobability extent.
\begin{lemma}{}{gamma_bb_bounds}
  Let $n\in\mathbb{N}$, $\cE_i\in\hp(A_iB_i)$ for $i=1,\dots,n$ and let $\ds$ be either $\lo^{\star}$ or $\locc^{\star}$.
  Then,
  \begin{equation}
    \gamma_{\ds}\left( \bigotimes\limits_{i=1}^n\cE_i \right)
    \leq 
    \gammabb_{\ds}(\cE_1,\cE_2,\dots,\cE_n) 
    \leq
    \prod\limits_{i=1}^n\gamma_{\ds}(\cE_i) \, .
  \end{equation}
\end{lemma}
\begin{proof}
  Clearly, optimally cutting the individual gates separately of each other induces a valid black box cut.
  This implies the upper bound.
  For the lower bound, notice that by appropriately choosing the embedding and fixing the black box channels to be the identity, we can obtain $\cG = \otimes_{i=1}^n\cE_i$.
  Therefore, any black box cut directly induces a QPD for $\otimes_{i=1}^n\cU_i$ with identical 1-norm.
\end{proof}
To illustrate the simplest instance of how circuit knitting in the black box setting can operate, we start by considering that the nonlocal gates that we want to cut are all $\cnot$ gates.
Furthermore, we will for the moment consider the $\locc$ setting.
In this case, the $\cnot$ gates can be realized through the gate teleportation circuit (recall~\Cref{fig:cnot_teleportation}) and the task reduces to the preparation of a shared Bell state, which can itself be simulated using QPS.
In a circuit with multiple CNOT gates, these bell states could all be prepared at the same time at the very beginning of the circuit, as depicted in~\Cref{fig:joint_ebit_prep}.
This in turns allows us to harvest the strictly sub-multiplicative behavior of $\gamma_{\sep}$ to jointly prepare the Bell states more efficiently.
This is fully analogous to the Clifford+T setting discussed in~\Cref{ex:magic_state_submult}, where we used a magic state injection circuit and utilized the sub-multiplicative behavior of the quasiprobability extent to simulate the joint preparation of multiple magic states more efficiently.

\begin{figure}
  \centering
  \includegraphics{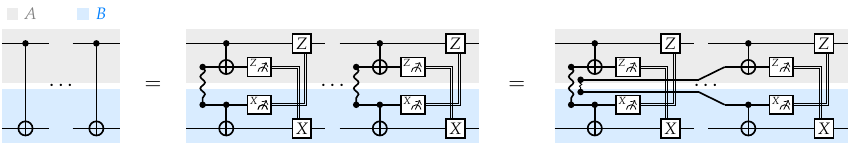}
  \caption{Two distant $\cnot$ gates can be realized through gate teleportation under LOCC by each consuming a Bell pair (depicted as wavy line). By generating the two Bell pairs simultaneously (instead generating a single Bell pair twice), we can utilize the sub-multiplicativity of the quasiprobability extent under the tensor product to reduce the total sampling overhead.}
  \label{fig:joint_ebit_prep}
\end{figure}

Clearly, the strategy outlined above can be straightforwardly extended to any gate that can be teleported from a state, such as all Clifford gates.
Together with~\Cref{cor:gamma_locc_clifford} and~\Cref{lem:gamma_bb_bounds}, we obtain following result.
\begin{proposition}{}{}
  The black-box quasiprobability extent of $n$ bipartite Clifford gates $\cU_1,\dots,\cU_n$ is given by
  \begin{equation}\label{eq:bb_optimal_gamma}
    \gammabb_{\locc^{\star}}(\cU_1,\cU_2,\dots,\cU_n) = \gamma_{\locc^{\star}}\left(\bigotimes\limits_{i=1}^n\cU_i\right) \, .
  \end{equation}
\end{proposition}
\begin{remark}{}{}
  The definition of the black box quasiprobability extent does not put any limit on the size of the ancilla systems $E_A$ and $E_B$.
  This can become problematic for practical applications of circuit knitting, as the small number of qubits is the main motivation for circuit knitting in the first place.
  Indeed, we run into this issue with the black box cutting strategy for CNOT gates discussed above.
  The systems $E_A$ and $E_B$ contain the ancilla qubits that store the simulated Bell pairs, so their size grows linearly in the number of nonlocal CNOT gates.

  In practice, this problem can be addressed by picking a fixed number of ancilla qubits $K$ on each side of the partition that are used to jointly generate $K$ shared Bell pairs.
  Once all $K$ Bell pairs are consumed by gate teleportations, the same qubits are re-used in another round of (simulated) preparation of $K$ Bell pairs.
  The size $K$ of such an ``entanglement factory'' allows for a trade-off: Increasing $K$ means that more ancilla qubits are required, but also lowers the effective quasiprobability extent per $\cnot$ to $(2^{K+1}-1)^{1/K}$.
  We note that a very small $K$ is already sufficient to get very close to an effective extent of $2$.
\end{remark}

It is natural to ask at this point, whether optimal black box cutting (as given by~\Cref{eq:bb_optimal_gamma}) is also achievable for non-Clifford unitaries and without the usage of classical communication.
At first, this might seem unlikely, given the crucial role of gate teleportation in the previous scheme.
Somewhat surprisingly, the answer is positive for both questions, at least for general two-qubit unitaries.
\begin{theorem}{}{gamma_bb_twoqubit}
  Let $\cU_1,\cU_2,\dots,\cU_n$ be two-qubit unitary channels.
  Then
  \begin{equation}
    \gammabb_{\lo^{\star}}(\cU_1,\cU_2,\dots,\cU_n) = 
    \gammabb_{\locc^{\star}}(\cU_1,\cU_2,\dots,\cU_n) = 
    \gamma_{\lo^{\star}}(\cU_1\otimes \cU_2\otimes \dots \cU_n ) \, .
  \end{equation}
\end{theorem}
We note that the proof is constructive, i.e., we specify an explicit $\lo^{\star}$ protocol that achieves optimal black box cutting.
The essential ingredient is the utilization of a gate teleportation-like protocol that works for non-Clifford gates and doesn't require any sort of classical communication.
Instead, our protocol relies purely on the classical side information generated by mid-circuit measurements.
\begin{proof}
  By~\Cref{thm:gamma_kaklike_unitary} and~\Cref{lem:gamma_bb_bounds}, we directly have
  \begin{equation}
    \gamma_{\lo^{\star}}(\cU_1\otimes \cU_2\otimes \dots \cU_n )
    \leq \gammabb_{\locc^{\star}}(\cU_1,\cU_2,\dots,\cU_n)
    \leq \gammabb_{\lo^{\star}}(\cU_1,\cU_2,\dots,\cU_n) \, .
  \end{equation}
  As such, it remains to show $\gammabb_{\lo^{\star}}(\cU_1,\cU_2,\dots,\cU_n)\leq \gamma_{\lo^{\star}}(\cU_1\otimes \cU_2\otimes \dots \cU_n )$, which we will do by explicitly constructing an $\lo^{\star}$ black box cut.

  Let us denote the single-qubit Pauli strings by $Q_0=\id,Q_1=X,Q_2=Y,Q_3=Z$ and $n$-qubit Pauli strings by $Q_{i_1,i_2,\dots,i_n}\coloneqq Q_{i_1}\otimes Q_{i_2}\otimes\dots\otimes Q_{i_n}$.
  The unitary channel induced by a Pauli will be denoted $\mathcal{Q}_i(\rho)\coloneqq Q_i\rho Q_i$.
  The two-qubit Bell basis will be denoted by $\ket{\Psi^{i}}\coloneqq (\sigma_i\otimes\id)\ket{\Psi^0}$ where $\ket{\Psi^0}\coloneqq (\ket{00}+\ket{11})/\sqrt{2}$.
  Let us furthermore define the two-qubit superoperator $\cD_{\alpha,\beta}$ as follows
  \begin{equation}
    \cD_{\alpha,\beta}(\rho)\coloneqq Q_{\alpha,\alpha}\rho Q_{\beta,\beta} \, .
  \end{equation}
  Recall the gate teleportation protocol that we discussed in~\Cref{fig:gate_teleportation}.
  This protocol allows us to realize any unitary channel from its Choi state, and in case the unitary is Clifford, then the protocol is LOCC.
  In the even more specialized case where the gate is a Pauli channel $\mathcal{Q}_{\alpha}\otimes\mathcal{Q}_{\beta}$, then the protocol is even in LO.
  We will extend this result by showing that $\cD_{\alpha,\beta}$ can also be realized from its Choi state using an $\lo^{\star}$ protocol.

  For this purpose, it is useful to first recapitulate the teleportation protocol for a Pauli channel $\mathcal{Q}_{\alpha}\otimes\mathcal{Q}_{\beta}$.
  We can write
  \begin{equation}\label{eq:gate_teleport_qq}
    (\mathcal{Q}_{\alpha}\otimes\mathcal{Q}_{\beta})(\rho_{\tilde{A}\tilde{B}}) = \sum_{i,j} T_{i,j}\left( \rho_{\tilde{A}\tilde{B}} \otimes (\choi{\mathcal{Q}_{\alpha}\otimes\mathcal{Q}_{\beta}})_{ABA'B'} \right)
  \end{equation}
  where $(T_{i,j})_{i,j}\in\instr{\lo}(AA'\tilde{A};BB'\tilde{B} \rightarrow A;B)$ is the quantum instrument describing the teleportation protocol which consists of
  \begin{enumerate}
    \item measuring $\tilde{A}A'$ in the Bell basis, obtaining outcome $i\in\{0,1,2,3\}$,
    \item measuring $\tilde{B}B'$ in the Bell basis, obtaining outcome $j\in\{0,1,2,3\}$,
    \item applying the operation $Q_{i,j}$ on $AB$.
  \end{enumerate}

  Now, we turn our attention to the gate teleportation of $\cD_{\alpha,\beta}$.
  Consider applying the same three steps as above, but now on the initial state $\rho_{\tilde{A}\tilde{B}}\otimes(\choi{\cD_{\alpha,\beta}})_{ABA'B'}$.
  By a standard state teleportation argument, the post-measurement state after obtaining the outcome $i,j$ is given by $\cD_{\alpha,\beta}( Q_{i,j}\rho Q_{i,j} )$.
  After applying the correction operation $Q_{i,j}$ we hence have
  \begin{align}
    Q_{i,j}\cD_{\alpha,\beta}( Q_{i,j}\rho Q_{i,j} ) Q_{i,j}
    &= Q_{i,j}Q_{\alpha,\alpha}Q_{i,j}\rho Q_{i,j}Q_{\beta,\beta}Q_{i,j} \\
    &= (-1)^{\symplip{Q_i}{Q_\alpha} + \symplip{Q_j}{Q_\alpha} + \symplip{Q_i}{Q_\beta} + \symplip{Q_j}{Q_\beta} } \cD_{\alpha,\beta}(\rho) \, .
  \end{align}
  The factor $f(i,j,\alpha,\beta)\coloneqq (-1)^{\symplip{Q_i}{Q_\alpha} + \symplip{Q_j}{Q_\alpha} + \symplip{Q_i}{Q_\beta} + \symplip{Q_j}{Q_\beta} }$ can be compensated by appropriately weighting with $\pm 1$ depending on $i$ and $j$ (recall that we are allowing $\lo^{\star}$ operations, not just LO).
  In summary, we have shown that
  \begin{equation}\label{eq:gate_teleport_D}
    \cD_{\alpha,\beta}(\rho_{\tilde{A}\tilde{B}}) = \sum_{i,j}  f(i,j,\alpha,\beta) T_{i,j}\left( \rho_{\tilde{A}\tilde{B}} \otimes (\choi{\cD_{\alpha,\beta}})_{ABA'B'} \right) \, .
  \end{equation}
  Notice that $f(i,j,\alpha,\beta)$ is symmetric under exchange of $\alpha,\beta$.

  We now turn to proving the statement of the theorem.
  The KAK decomposition allows us to write the two-qubit unitaries $U_m$ as
  \begin{equation}
    U_m = \sum_{k=0}^3 v_k^{(m)} Q_k\otimes Q_k
  \end{equation}
  for coefficients $v^{(m)}\in\mathbb{C}^4$.
  We separate the absolute value and the phase $v^{(m)}_k = u^{(m)}_k e^{i\phi^{(m)}_k}$ for $u^{(m)}_k\geq 0$, $\phi^{(m)}_k\in[0,2\pi)$.
  The induced channels are therefore
  \begin{equation}
    \mathcal{U}_m(\rho) = \sum_{k,l} u^{(m)}_ku^{(m)}_l e^{i(\phi^{(m)}_k-\phi^{(m)}_l)} \cD_{k,l} \, .
  \end{equation}
  Using the gate teleportation protocol in~\Cref{eq:gate_teleport_D}, we can rewrite this as
  \begin{align}
    \mathcal{U}_m(\rho) &= \sum_{i,j,k,l} u^{(m)}_ku^{(m)}_le^{i(\phi^{(m)}_k-\phi^{(m)}_l)} f(i,j,k,l)  T_{i,j}\left(\rho \otimes \choi{\cD_{k,l}} \right) \\
    &= \sum_{k,l} u^{(m)}_ku^{(m)}_le^{i(\phi^{(m)}_k-\phi^{(m)}_l)} \mathcal{T}^{(k,l)}\left(\rho \otimes \choi{\cD_{k,l}} \right)  \, .
  \end{align}
  where $\mathcal{T}^{(k,l)}\coloneqq \sum_{i,j} f(i,j,k,l) T_{i,j} \in \lo^{\star}(AA'\tilde{A};BB'\tilde{B}\rightarrow A;B)$.

  The main idea for our black box cutting strategy is to apply the gate teleportation to every nonlocal gate $\cU_m$ in the black box circuit
  \begin{equation}
    \cG = \cU_n \circ \cB_{n-1} \circ \cU_{n-1} \circ \cB_{n-2} \circ \cdots \circ \cU_{2} \circ \cB_{1} \circ \cU_{1} 
  \end{equation}
  and then use the ancilla system to prepare all the Choi states of the form $\choi{\cD_{k,l}}$ together at the beginning of the circuit.
This allows us to write
\begin{equation}
  \cG = 
  \sum_{\bm{k},\bm{l}\in\mathbb{F}_4^n} \bm{u}_{\bm{k}}\bm{u}_{\bm{l}} e^{i\theta(\bm{k},\bm{l})}
  \Big( \mathcal{T}^{k_n,l_n}\circ\cB_{n-1}\circ\mathcal{T}^{k_{n-1},l_{n-1}}\circ\dots\circ\cB_1\circ\mathcal{T}^{k_1,l_1} \Big)
  (\rho \otimes \choi{\bm{D}_{\bm{k},\bm{l}}})
\end{equation}
where $\mathbb{F}_4\coloneqq\{0,1,2,3\}$ and
\begin{equation}
  \bm{u}_{\bm{k}}\coloneqq u^{(1)}_{k_1}u^{(2)}_{k_2}\cdots u^{(n)}_{k_n} \, ,
\end{equation}
\begin{equation}
  \theta(\bm{k},\bm{l}) \coloneqq \left(\phi^{(1)}_{k_1}-\phi^{(1)}_{l_1}\right) + \left(\phi^{(2)}_{k_2}-\phi^{(2)}_{l_2}\right) + \dots + \left(\phi^{(n)}_{k_n}-\phi^{(n)}_{l_n}\right) \, ,
\end{equation}
\begin{equation}
  \bm{D}_{\bm{k},\bm{l}} \coloneqq D_{k_1,l_1}\otimes D_{k_2,l_2}\otimes \dots \otimes D_{k_n,l_n} \, .
\end{equation}
We re-arrange the sum to obtain
\begin{align}\label{eq:bbcut_proof}
  \cG = 
  & \sum_{\bm{k}\in\mathbb{F}_4^n} \bm{u}_{\bm{k}}^2 
  \Big( \mathcal{T}^{k_n,k_n}\circ\cB_{n-1}\circ\mathcal{T}^{k_{n-1},k_{n-1}}\circ\dots\circ\cB_1\circ\mathcal{T}^{k_1,k_1} \Big)
  (\rho \otimes \choi{\bm{D}_{\bm{k},\bm{k}}}) \nonumber\\
  & + \sum_{\bm{k}<\bm{l}} 4\bm{u}_{\bm{k}}\bm{u}_{\bm{l}}
  \Big( \mathcal{T}^{k_n,l_n}\circ\cB_{n-1}\circ\mathcal{T}^{k_{n-1},l_{n-1}}\circ\dots\circ\cB_1\circ\mathcal{T}^{k_1,l_1} \Big)
  (\rho \otimes \choi{\bm{\tilde{D}}_{\bm{k},\bm{l}}})
\end{align}
where we used that $\mathcal{T}^{k,l}=\mathcal{T}^{l,k}$, $\bm{k}<\bm{l}$ should be understood in lexicographic order and
\begin{equation}
  \bm{\tilde{D}}_{\bm{k},\bm{l}} \coloneqq \frac{1}{4}\left( e^{i\theta(\bm{k},\bm{l})}\bm{D}_{\bm{k},\bm{l}} + e^{-i\theta(\bm{k},\bm{l})}\bm{D}_{\bm{l},\bm{k}} \right) \, .
\end{equation}
Clearly, $\bm{D}_{\bm{k},\bm{k}}\in\lo$, and further below, we will also show that $\bm{\tilde{D}}_{\bm{k},\bm{l}}\in\lo^{\star}$.
This implies that~\Cref{eq:bbcut_proof} provides a valid $\lo^{\star}$ black box cut with a 1-norm of
\begin{equation}
  \sum_{\bm{k}} \bm{u}_{\bm{k}}^2 + \sum_{\bm{k}<\bm{l}} 4\bm{u}_{\bm{k}}\bm{u}_{\bm{l}}
  = 2\left(\sum_{\bm{k}}\bm{u}_{\bm{k}}\right)^2  - \sum_{\bm{k}} \bm{u}_{\bm{k}}^2
  = \gamma_{\lo^{\star}}(\cU_1\otimes \cU_2\otimes \dots \cU_n ) 
\end{equation}
where we used~\Cref{thm:gamma_kaklike_unitary} and the fact that the Schmidt coefficients of $\otimes_{i=1}^n\cU_i$ are precisely the $\bm{u}_{\bm{k}}$.

It now remains to show that $\bm{\tilde{D}}_{\bm{k},\bm{l}}\in\lo^{\star}$.
Notice that we can write the operator Schmidt decomposition of $\otimes_{i=1}^nU_i$ as
\begin{equation}
  \bigotimes_{i=1}^nU_i = \sum_{\bm{k}} \bm{u}_{\bm{k}} \prod_{j=1}^n e^{i\phi_{k_j}^{(j)}} Q_{\bm{k}}\otimes Q_{\bm{k}} \, .
\end{equation}
Following the construction in the proof of~\Cref{lem:gamma_unitary_upperbound}, the optimal $\lo^{\star}$ QPD of $\otimes_{i=1}^n\cU_i$ is precisely
\begin{equation}
  \bigotimes_{i=1}^n\cU_i
   = \sum_{\bm{k}} \bm{u}_{\bm{k}}^2 \mathcal{Q}_{\bm{k}}\otimes \mathcal{Q}_{\bm{k}}
   + \sum_{\bm{k}<\bm{l}} 4\bm{u}_{\bm{k}}\bm{u}_{\bm{l}} \tilde{D}_{\bm{k},\bm{l}} \, .
\end{equation}
Therefore, $\bm{\tilde{D}}_{\bm{k},\bm{l}}\in\lo^{\star}$ is a direct consequence from the proof of~\Cref{lem:gamma_unitary_upperbound}.
\end{proof}

\subsection{Applications}\label{sec:gate_cut_app}
While our main focus has been the fundamental characterization of the overhead of circuit knitting, we will briefly touch upon some potential applications of space-like cuts in this section.
It is not easy to find computational tasks that are well-suited for circuit knitting, as the involved quantum circuit needs to be separable into weakly interacting partitions.
Asymptotically speaking, the number of space-like cuts needs to be logarithmic in the problem size, at least when the entangling strength of the cut gates is not vanishing.

Another consideration when trying to demonstrate some form of quantum advantage with circuit knitting is that the individual partitions of the circuit should be out of reach of classical simulation algorithms.
Otherwise, the original computation can be classically simulated by simply simulating the circuit knitting procedure.

First, we will study the quantum Fourier transform, which is a central building block for some of the most important quantum algorithms such as phase estimation~\cite{kitaev1995_quantum} and Shor's algorithms for prime factorization and the discrete logarithm~\cite{shor1997_polynomial}.
Its action on an $n$-qubit computational basis state $\ket{x}$ for $x\in\{0,\dots,2^n-1\}$ is defined as
\begin{equation}
  \ket{x} \mapsto \frac{1}{\sqrt{2^n}} \sum_{y=0}^{2^n-1} \omega_{2^n}^{xy} \ket{y}
\end{equation}
where $\omega_N\coloneqq e^{\frac{2\pi i}{N}}$.
It can be efficiently realized by the circuit in~\Cref{fig:qft} followed by layer of SWAP gates $\swap_{Q_1Q_n}\swap_{Q_2Q_{n-1}}\swap_{Q_3Q_{n-3}}\dots$ that inverts the order of the qubits.
\begin{figure}
  \centering
  \includegraphics{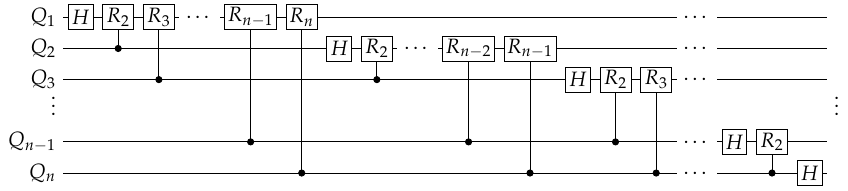}
  \caption{Circuit implementing quantum Fourier transform without qubit order inversion at the end. Here, $R_i$ denotes the rotation gate $R_Z(\frac{2\pi}{2^i})$.}
  \label{fig:qft}
\end{figure}
While the QFT is known to be very strongly entangling~\footnote{The QFT is maximally entangling in the sense that its operator Schmidt decomposition has uniform singular values~\cite{tyson2003_operatorschmidt,nielsen2003_quantum}.}, it is entirely due to the qubit order reversal at the end~\cite{chen2023_quantum}.
In fact, if we take the circuit from~\Cref{fig:qft} and individually cut the various controlled rotation gates, we get a quasiprobability extent that does not diverge as a function of $n$.

\begin{proposition}{}{gamma_qft}
  The quasiprobability extent of the $n$-qubit quantum Fourier transform without qubit order reversal $\mathrm{QFT}_n$ is asymptotically constant in $n$
  \begin{equation}
    \gamma_{\lo^{\star}(A;B)}(\mathrm{QFT}_n) = \mathcal{O}(1)
  \end{equation}
  w.r.t to the bipartition $A=Q_1\dots Q_i$ and $B=Q_{i+1},\dots,Q_n$ for any $1\leq i\leq n-1$.
\end{proposition}
\begin{proof}
  We prove this statement by taking the circuit in~\Cref{fig:qft} and separately cutting every controlled rotation gate which acts across the bipartition.
  This does most likely not produce the optimal QPD of $\mathrm{QFT}_n$, but it provides us an upper bound to the quasiprobability extent.

  A simple counting argument shows that no matter the value of $n$ and the choice of the partition, the number of controlled $R_{k+1}$ gates that we cut is at most $k$ where $k\geq 1$.
  Recalling the result in~\Cref{ex:gamma_2qb_gates}, we get
  \begin{align}
    \gamma_{\lo^{\star}}(\mathrm{QFT}_n)
    &\leq \prod_{k=1}^{\infty}\gamma_{\lo^{\star}}(CR_{k+1})^k \\
    &= \prod_{k=1}^{\infty}\left( 1 + 2\sin\frac{2\pi}{2^{k+1}} \right)^k \\
    &\leq \prod_{k=1}^{\infty}\left( 1 + 2\frac{2\pi}{2^{k+1}} \right)^k \\
    &= \prod_{k=1}^{\infty}\left( \frac{2^k + 2\pi}{2^k} \right)^k \\
    &= \exp_2\left[ \sum_{k=1}^{\infty} k \left( \log(2^k + 2\pi) - \log(2^k) \right) \right]
  \end{align}
  where $\exp_2(x)$ denotes $2^x$.
  Since $\log$ is concave, we can upper bound it with a first-order Taylor approximation at any point, so
  \begin{align}
    \gamma_{\lo^{\star}}(\mathrm{QFT}_n)
    &\leq \exp_2\left[ \sum_{k=1}^{\infty} k2\pi\frac{1}{\ln(2)2^k} \right] \\
    &= \exp_2\left[ \frac{4\pi}{\ln{2}} \right] \\
    &= e^{4\pi} \\
    &\approx 286751 
  \end{align}
  where we used $\sum\limits_{k=1}^{\infty}\frac{k}{2^k} = 2$.
\end{proof}
Beside the upper bound in the proof, we also numerically evaluated the quantity $\prod_{k=1}^{\infty}\gamma_{\lo^{\star}}(CR_{k+1})^k$ and obtained a value of roughly $\approx 26498$, which is more than an order of magnitude smaller than $e^{4\pi}$.

It is an interesting open question whether better QPDs of the quantum Fourier transform (or at least better estimates of its quasiprobability extent) can be obtained.
Related results in literature have shown that the singular values of the operator Schmidt decomposition of the quantum Fourier transformation (without qubit reversal) are exponentially decaying~\cite{chen2023_quantum}.
Unfortunately, this does not directly provide us with some means of upper bounding the quasiprobability extent, as the operator Schmidt decomposition is not known to be of the unitary form required for~\Cref{thm:gamma_kaklike_unitary}.

Next, we briefly consider Hamiltonian simulation as a second application of space-like cuts, which has been thoroughly studied by Harrow and Lowe.
We informally summarize their result as follows.
\begin{theorem}{\cite[Theorem 3.12]{harrow2024_optimal} (informal, simplified)}{hamiltonian_sim}
  Consider a bipartite Hamiltonian $H_{\mathrm{tot}}\in\herm(AB)$ that can be separated into two local parts and a weak interaction term
  \begin{equation}
    H_{\mathrm{tot}} = H_A\otimes\id_B + \id_A\otimes H_B + H_{\mathrm{int}}
  \end{equation}
  where $H_A,H_B$ and $H_{\mathrm{int}}$ are sums of Pauli terms acting on $\mathcal{O}(1)$ qubits.
  Circuit knitting can simulate the time evolution operator $\mathcal{T}(\rho)\coloneqq e^{-iH_{\mathrm{tot}}t}\rho e^{iH_{\mathrm{tot}}t}$ up to error $\epsilon$ with efficient circuits and a sampling overhead $\mathcal{O}(e^{8\eta t}/\epsilon^2)$ where $\eta\leq \norm{H_{\mathrm{int}}}$.
\end{theorem}
For the full formal statement, we refer the reader to \reference~\cite{harrow2024_optimal}.
The main idea is to use the standard Trotterization approach to approximate $\mathcal{T}$ with interleaved local and interaction operations.
Trotterized Hamiltonian simulation has 2 freely choosable parameters: The simulation time $t$ and the number of steps $N$.
Remarkably, the quasiprobability extent of the Trotterized circuit only grows with $t$, but not with $N$.
The reason for this is that, while the number of nonlocal gates increases linearly with $N$, their interaction strength also decreases with $N$.

As a simple model to illustrate this effect, consider the two-qubit Pauli operator $Z\otimes Z$ as our interaction term.
The time evolution induced by this interaction is of the form of a rotation gate $\mathrm{R}_{ZZ}(\theta) \coloneqq e^{-i\frac{\theta}{2}Z\otimes Z}$.
We previously showed in~\Cref{ex:gamma_2qb_gates} that $\gamma_{\lo^{\star}}(\mathrm{R}_{ZZ}(\theta)) = 1 + 2\abs{\sin\theta}$.
If we sub-divide the rotation into $N$ smaller rotations which we cut individually, then the resulting sampling overhead converges to some fixed value as $N$ grows to infinity
\begin{align}
  \gamma_{\lo^{\star}}(\mathrm{R}_{ZZ}(\theta / N))^N
  &= \left(1 + 2\sin\frac{\abs{\theta}}{N} \right)^N \\
  &\approx \left(1 + \frac{2\theta}{N} \right)^N \\
  & \xrightarrow{N\rightarrow\infty} \exp(2\theta) .
\end{align}

In this section, we have focused on applications of circuit knitting where we partitioned a large circuit into two or more smaller circuits.
There is also a much more mundane, but perhaps more widely applicable usage of the technique.
Many experimental quantum computing platforms exhibit strong locality constraints.
For instance, superconducting quantum chips typically only allow for two-qubit interactions between neighboring qubits.
As such, two-qubit gates between distant qubits require expensive swap chains that scale in depth with the distance.
This difficulty remains in the fault-tolerant settings, at least when the logical qubits are encoded in a topological code like the surface code or color code.
As such, it could be useful to locally simulate a few of these nonlocal operations, foregoing the need for swap chains at a cost of an increased sampling overhead.
If quantum resources are more limited than the available time (i.e. number of shots), then this trade-off can potentially be advantageous.

\section{Time-like cuts}\label{sec:timelike_cuts}
We have previously seen that space-like cuts are deeply connected to the resource theory of entanglement.
Remarkably, there is a similar connection for time-like cuts, with the lesser known \emph{resource theory of quantum memories}.
This is a dynamical resource theory where the non-free channels are the ones with the capability of preserving entanglement.
Correspondingly, the free operations are called \emph{entanglement breaking}.
\begin{definition}{}{}
  A map $\cE\in\cptp(A\rightarrow A')$ is called entanglement breaking if $(\cE\otimes\idchan_B)(\rho_{AB})\in\sep(A';B)$ for any system $B$ and density operator $\rho\in\dops(AB)$.
  We denote the set of such entanglement breaking channels by $\ebc(A\rightarrow A')$.
\end{definition}
In some ways, the resource theory of quantum memories can be considered the counterpart to the resource theory of entanglement as it studies quantum correlations in time rather than spatial correlations.
It has found applications for the certification and benchmarking of quantum memories~\cite{rosset2018_resource,yuan2021_universal}.

In the context of circuit knitting, we will study decompositions of the identity channel
\begin{equation}
  \idchan = \sum_i a_i \cF_i \, .
\end{equation}
A major question is what decomposition set we will allow for the channels $\cF_i$ to lie in.
Clearly, the $\cF_i$ must be entanglement breaking, as otherwise they could generate entanglement across the circuit partitions.
In fact, the set $\ebc$ is a very natural decomposition set, as it precisely contains the channels that can be realized with a measurement, followed by the preparation of a state that depends on the measurement outcome.
\begin{proposition}{\cite[Theorem 2]{horodecki2003_entanglement}}{}
  A channel $\cE\in\cptp(A\rightarrow A')$ is entanglement breaking if and only if it can be written in the form
  \begin{equation}
    \cE(\rho) = \sum_k \tr[M_k\rho] \sigma_k
  \end{equation}
  for some $(M_k)_k\in\povm(A)$ and $\sigma_k\in\dops(A')$.
\end{proposition}
As previously discussed in~\Cref{fig:example_wire_cut}, decomposing the identity into such channels allows us to realize circuit knitting.
The choice of $\ebc$ as our decomposition set necessitates that our circuit partitions can exchange classical communication during the circuit execution.
This is a requirement that we might not want to impose.
Similarly to the study of the decomposition sets $\lo$ and $\locc$ for space-like cuts, we also introduce a second decomposition set that reflects the constraint that the two involved circuit partitions cannot exchange classical communication.
It consists of measurement and prepare channels $\cE\in\cptp(A\rightarrow A')$, which do not allow for the prepared state to depend on the measurement result
\begin{equation}
  \cE(\rho) = \sum_k \tr[M_k\rho] \sigma
\end{equation}
for some $(M_k)_k\in\povm(A)$ and $\sigma\in\dops(A')$.
We denote the set of such channels by $\mpc(A\rightarrow A')$.
At first glance, these channels might seem completely useless, as any information from their input is completely lost, as $\cE(\rho)=\tr[\rho]\sigma$.
We will see in the next section that these channels can become useful if we allow for classical side information that is generated from the measurement.
So while $\mpc(A\rightarrow A')$ lacks the expressibility required for circuit knitting, the set $\mpc^{\star}(A\rightarrow A')$ does not.

The main goal of this section is to identify the quasiprobability extent of the identity channel $\gamma_{\ebc^{\star}}(\idchan)$ and $\gamma_{\mpc^{\star}}(\idchan)$, which in turn characterizes the sampling overhead of circuit knitting with time-like cuts.
We will also briefly touch upon the quasiprobability extent of other channels, but it will be a lesser focus since they are operationally not as important.

We note that time-like cuts can be considered a special case of space-like cuts by choosing the output system on one partition and the input system on the other partition to be trivial.
As such, one can readily verify that $\ebc(A\rightarrow B')=\locc(A;B\rightarrow A'B'),\mpc(A\rightarrow B')=\lo(A;B\rightarrow A';B')$, and correspondingly $\ebc^{\star}(A\rightarrow B')=\locc^{\star}(A;B\rightarrow A';B'),\mpc^{\star}(A\rightarrow B')=\lo^{\star}(A;B\rightarrow A';B')$, in the case that $\dimension(B)=\dimension(A')=1$.
This will however prove to be of limited use to us, as our results on space-like cuts are not powerful enough to cover the time-like cuts that we will study in this section.

Astute readers might also have noticed that $\mpc$ and $\ebc$ do not actually constitute valid QRTs as per~\Cref{def:qrt}, as the identity channel is not free.
This issue can be fixed by identifying the channels in $\mpc$ and $\ebc$ with their associated states through the Choi isomorphism.
Under this new perspective, channels get promoted (or rather \emph{demoted}) to states, and superchannels to channels.
As such, the identity channel (the demoted identity superchannel) is now free, providing us with a valid quantum resource theory.
This subtlety will be of no further importance in this section, besides the fact that we need to be a bit careful when applying results from~\Cref{sec:gamma_qrt}.

\subsection{Decomposition sets}\label{sec:timelike_ds}
Here, we introduce the decomposition sets of interest with full mathematical rigor.
To capture the notion of classical side information through intermediate measurements, we start by defining the associated sets of quantum instruments.
Note that in principle, both the sender and receiver have the capability of generating classical side information.
In the case without classical communication, the most general setup consists of the sender doing a POVM measurement on the incoming state, and the receiver randomly preparing some state $\sigma_j$ with probability $p_j$.
The index of the chosen state $j$ can serve as classical side information as well.
We thus define $\instr{\mpc}(A\rightarrow A')\subset\qi(A\rightarrow A')$ as the set of all instruments $(\cE_{i,j})_{i,j}$ of the form\footnote{More axiomatically, the sets $\instr{\mpc}(A\rightarrow B')$ and $\instr{\ebc}(A\rightarrow B')$ could equivalently also be defined as $\instr{\lo}(A;B\rightarrow A';B')$ and $\instr{\locc}(A;B\rightarrow A';B')$ where $\dimension(A')=\dimension(B)=1$.}
\begin{equation}
  \cE_{i,j}(\rho)=\tr[M_i\rho]p_j\sigma_j
\end{equation}
where $(M_i)_i\in\povm(A)$, $(p_j)_j$ is a probability vector and $\sigma_j\in\dops(A')$.
In the setup with classical communication, the randomly prepared state of the receiver may additionally depend on the measurement outcome of the POVM.
As such, we define $\instr{\ebc}(A\rightarrow A')\subset\qi(A\rightarrow A')$ as all instruments $(\cE_i)_i$ of the form
\begin{equation}
  \cE_i(\rho)=\tr[M_i\rho]\sigma_i
\end{equation}
as well as all coarse-grainings of such instruments, where $(M_i)_i\in\povm(A)$ and $\sigma_i\in\dops(A')$.
Note that without loss of generality, we could assume that the receiver prepares a deterministic state $\sigma_i$ given the POVM outcome $i$, as the case of a randomly chosen state can be captured by fine-graining the POVM elements $M_i$.

We can directly recover $\ebc(A\rightarrow A')=\qichan[\instr{\ebc}(A\rightarrow A')]$ and $\mpc(A\rightarrow A')=\qichan[\instr{\mpc}(A\rightarrow A')]$.
For practical purposes, we are more interested in the decomposition sets
\begin{align}\label{eq:mpc_star}
  \mpc^{\star}(A\rightarrow A') &\coloneqq \qids[\instr{\mpc}(A\rightarrow A')] \\
  &= \left\{ \rho\mapsto \sum_{i,j}\beta_{i,j} \tr[M_i\rho]p_j\sigma_j \middle\vert
  \begin{array}{l}
    \abs{\beta_{i,j}}\leq 1, (M_i)_i\in\povm(A), \\
    p_j \text{ probabilities }, \sigma_j\in\dops(A')
  \end{array}
  \right\}
\end{align}
and
\begin{align}
  \ebc^{\star}(A\rightarrow A') &\coloneqq\qids[\instr{\ebc}(A\rightarrow A')] \\
  &= \{\rho\mapsto \sum_i\beta_i \tr[M_i\rho]\sigma_i | \abs{\beta_i}\leq 1, (M_i)_i\in\povm(A), \sigma_i\in\dops(A')\}
\end{align}
which capture the possibility of classical side information.
Since $\instr{\ebc}(A\rightarrow A')$ is closed under mixture (the sender can randomly choose between two POVMs to perform and broadcast that choice to the receiver), $\ebc^{\star}$ is absolutely convex by~\Cref{lem:decompset_absolutely_convex}.

We note that our definition of $\mpc^{\star}$ is slightly more general than the decomposition basis
\begin{equation}
  \underline{\mpc}(A\rightarrow A')\coloneqq \{\rho\mapsto \tr[O\rho]\sigma | O\in\herm(A),\sigma\in\herm(A'),\opnorm{O}\leq 1,\norm{\sigma}_1\leq 1\}
\end{equation}
introduced in \reference~\cite{brenner2023_optimal}.
Indeed, our construction of $\mpc^{\star}$ utilizes the full generality of classical side information by allowing the weighting to be jointly dependent on the classical side information of both the sender and receiver, whereas
\begin{equation}
  \underline{\mpc}(A\rightarrow A') = \{\rho \mapsto \sum_{i,j}  \tr[\alpha_i\rho] \beta_jp_j\sigma_j | \abs{\alpha_i}\leq 1, \abs{\beta_j}\leq 1\}
\end{equation}
doesn't allow for correlated weighting.
This discrepancy is essentially the same as the one between $\lo^{\star}$ and~\Cref{eq:old_lostar_def} that we already mentioned previously.
Clearly, $\underline{\mpc}(A\rightarrow A')\subset \mpc^{\star}(A\rightarrow A')$, but the two sets are generally not equal.
Fortunately, as we will see, the discrepancy is of no further consequence for the study of the quasiprobability extent.
In fact, it will at times be more convenient to work with $\underline{\mpc}(A\rightarrow A')$ for technical reasons.

Working directly with the channel representations of these various sets can become somewhat unwieldy.
It can be useful to consider the associated Choi representations instead.

\begin{proposition}{}{wirecut_ds_choi_characterization}
  Let $\cE\in\hp(A\rightarrow A')$.
  Then
  \begin{itemize}
    \item $\cE\in\mpc(A\rightarrow A') \Longleftrightarrow \choi{\cE}=\frac{\id_A}{\dimension(A)}\otimes \sigma$ for some $\sigma\in\dops(A')$.
    \item $\cE\in\underline{\mpc}(A\rightarrow A') \Longleftrightarrow \choi{\cE} = \frac{O}{\dimension(A)}\otimes\sigma$ \\
    for some $O\in\herm(A),\opnorm{O}\leq 1$ and $\sigma\in\herm(A'),\norm{\sigma}_1\leq 1.$
    \item
    $\begin{aligned}[t]
        \cE\in\mpc^{\star}(A\rightarrow A') \Longleftrightarrow \choi{\cE}
        &= \frac{1}{\dimension(A)} \sum_i \left( M_i \otimes \left(\sum_j\beta_{i,j}p_j\sigma_j\right)\right) \\
        &= \frac{1}{\dimension(A)} \sum_j p_j \left(\sum_i \beta_{i,j}M_i\right) \otimes \sigma_j
    \end{aligned}$ \\
    for some $(M_i)_i\in\povm(A)$, probabilities $(p_j)_j$, $\abs{\beta_{i,j}}\leq 1$ and $\sigma_j\in\dops(A')$.
    \item $\cE\in\ebc(A\rightarrow A') \Longleftrightarrow \choi{\cE}\in\sep(A;A')$ and $\tr_{A'}[\choi{\cE}]=\frac{\id_A}{\dimension(A)}.$
    \item $\cE\in\ebc^{\star}(A\rightarrow A') \Longleftrightarrow \choi{\cE}=\rho^+ - \rho^-$ \\
    where $\rho^{\pm}\in[0,1]\cdot\sep(A;A')$ are sub-normalized separable states and $\tr_{A'}[\rho^++\rho^-]=\frac{\id_A}{\dimension(A)}$.
  \end{itemize}
\end{proposition}
The statement about the Choi characterization of $\ebc$ is originally due to Horodecki, Shor and Ruskai~\Cite[Theorem 4]{horodecki2003_entanglement}.
\begin{proof}
  Consider a quantum instrument $(\cE_{i,j})_{i,j}\in \instr{\mpc}(A\rightarrow A')$ of the form $\cE_{i,j}(\rho)=\tr[M_i\rho]p_j\sigma_j$.
  Its Choi representation is as follows
  \begin{align}
    (\idchan\otimes\cE_{i,j})\left(\frac{1}{\dimension(A)}\sum_{k,l}\ketbra{kk}{ll}\right)
    &= \frac{1}{\dimension(A)}\sum_{k,l} \ketbra{k}{l} \otimes \left( \tr[M_i\ketbra{k}{l}]p_j\sigma_j\right) \\
    &= \frac{p_j}{\dimension(A)}\left( \sum_{k,l}\tr[M_i\ketbra{k}{l}] \ketbra{k}{l} \right) \otimes \sigma_j \\
    &= \frac{p_j}{\dimension(A)} M_i^{T} \otimes \sigma_j \, .
  \end{align}
  As such, the Choi representations of a map $\cE\in\mpc(A\rightarrow A')$ is given by
  \begin{equation}
    \sum_{i,j} \frac{p_j}{\dimension(A)} M_i^{T} \otimes \sigma_j = \frac{\id_A}{\dimension(A)}\otimes \left(\sum_jp_j\sigma_j\right) \, .
  \end{equation}
  Conversely, every state $\frac{\id_A}{\dimension(A)}\otimes\sigma$ can be written in this form by appropriately choosing $p_j,\sigma_j$.
  Analogously, the Choi representation of a map $\cE\in\underline{\mpc}(A\rightarrow A')$ is given by
  \begin{equation}
    \sum_{i,j} \alpha_i\beta_j \frac{p_j}{\dimension(A)} M_i^T \otimes \sigma_j =\frac{1}{\dimension(A)}\left(\sum_i \alpha_i M_i^T\right)\otimes \left(\sum_jp_j\beta_j\sigma_j\right) \, .
  \end{equation}
  This can be made equal to $\frac{1}{\dimension(A)}O\otimes\sigma$ for any $O,\sigma$ such that $\opnorm{O},\norm{\sigma}_1\leq 1$ by appropriately choosing the $M_i,\alpha_i,p_j,\beta_j$ and $\sigma_j$.
  Finally, the Choi representation of a map $\cE\in\mpc^{\star}(A\rightarrow A')$ is precisely
  \begin{equation}
    \sum_{i,j} \beta_{i,j} \frac{p_j}{\dimension(A)} M_i^T \otimes \sigma_j = \sum_i \left(\frac{M_i^T}{\dimension(A)} \otimes \left(\sum_j\beta_{i,j}p_j\sigma_j\right)\right) \, .
  \end{equation}

  Next, we look at a quantum instrument $(\cE_{i})_{i}\in \instr{\ebc}(A\rightarrow A')$ of the form $\cE(\rho)=\tr[M_i\rho]\sigma_i$.
  By a similar calculation as above, its Choi representation is
  \begin{align}
    (\idchan\otimes\cE_i)\left(\frac{1}{\dimension(A)}\sum_{k,l}\ketbra{kk}{ll}\right)
    &= \frac{1}{\dimension(A)}\sum_{k,l} \ketbra{k}{l} \otimes \left( \tr[M_i\ketbra{k}{l}]\sigma_i\right) \\
    &= \frac{M_i^T}{\dimension(A)} \otimes \sigma_i \, .
  \end{align}
  As such the maps $\cE\in\ebc(A\rightarrow A')$ are of the form
  \begin{equation}
    \sum_i \frac{M_i^T}{\dimension(A)} \otimes \sigma_i = \sum_i \frac{\tr[M_i]}{\dimension(A)} \frac{M_i^T}{\tr[M_i]}\otimes \sigma_i \, .
  \end{equation}
  These maps are clearly separable and have partial trace $\sum_i \frac{M_i^T}{\dimension(A)} = \frac{\id_A}{\dimension(A)}$.
  Conversely, a general separable state in $\sep(A;A')$ can be written as $\sum_i p_i \proj{f_i}\otimes\proj{g_i}$.
  If the state additionally fulfills the partial trace constraint, then $\sum_i p_i\proj{f_i} = \frac{\id_A}{\dimension(A)}$.
  The above calculation shows that this state is the Choi representation of the map
  \begin{equation}
    \cE'(\rho)\coloneqq \sum_i\tr[\dimension(A)p_i\proj{f_i}\rho]\proj{g_i} \, .
  \end{equation}
  This map $\cE'$ lies in $\ebc(A\rightarrow A')$ since $\sum_i \dimension(A)p_i\proj{f_i}=\id_A$.

  Finally, the characterization of $\ebc^{\star}(A\rightarrow A')$ is a direct consequence of \Cref{lem:characterization_expanded_decomposition_set}: Any $\cE\in\ebc^{\star}(A\rightarrow A')$ can be written as $\cE^+-\cE^-$ for some instrument $(\cE^+,\cE^-)\in\instr{\ebc}$.
  That is, there exists a POVM $(M_i)_{i=1,\dots,m}$ and states $(\sigma_i)_{i=1,\dots,m}$ and some $k\in\{0,\dots,m-1\}$ such that
  \begin{equation}
    \cE^+(\rho) = \sum_{i=1}^k \tr[M_i\rho]\sigma_i \quad \text{ and } \quad \cE^-(\rho) = \sum_{i=k+1}^m \tr[M_i\rho]\sigma_i \, .
  \end{equation}
  By the same argument as for the Choi characterization of $\ebc(A\rightarrow A')$, $\choi{\cE^+}$ and $\choi{\cE^-}$ are sub-normalized separable states and $\choi{\cE^+ + \cE^-}\in\sep(A;A')$ fulfills the partial trace constraint.
  For the converse, consider some general sub-normalized states $\rho^{\pm}$ which are separable
  \begin{equation}
    \rho^+ = \sum_i p_i^+ \proj{f_i^+}\otimes\proj{g_i^-} \quad \text{ and } \quad \rho^- = \sum_j p_j^- \proj{f_j^-}\otimes\proj{g_j^-}
  \end{equation}
  and fulfill the partial trace constraint
  \begin{equation}
    \sum_i p_i^+\proj{f_i^+} + \sum_j p_j^-\proj{f_i^-} = \frac{\id_A}{\dimension(A)} \, .
  \end{equation}
  Then $\rho^+-\rho^-$ is the Choi representation of the channel
  \begin{align}
    \cE'(\rho) &\coloneqq \sum_i \tr[\dimension(A)p_i^+\proj{f_i^+}\rho]\proj{g_i^+} \\
    &\quad + \sum_j (-1) \tr[\dimension(A)p_j^-\proj{f_j^-}\rho]\proj{g_j^-}
  \end{align}
  which lies in $\ebc^{\star}$ since 
  \begin{equation}
    \sum_i\dimension(A)p_i^+\proj{f_i^+} + \sum_j\dimension(A)p_j^-\proj{f_j^-} = \id_A \, .
  \end{equation}
\end{proof}

The Choi characterization of $\mpc^{\star}$ is more complicated than the other sets.
Luckily, $\mpc^{\star}$ and $\underline{\mpc}$ are equivalent in terms of their convex hulls.
\begin{corollary}{}{equivalence_mpc_def}
  For any systems $A,A'$ one has
  \begin{equation}
    \conv(\underline{\mpc}(A\rightarrow A')) = \conv(\mpc^{\star}(A\rightarrow A')) \, .
  \end{equation}
\end{corollary}
Since the quasiprobability extent only depends on the convex hull of the decomposition set (recall~\Cref{lem:gamma_hulls}), this corollary implies that it doesn't matter whether we work with $\gamma_{\underline{\mpc}}$ or $\gamma_{\mpc^{\star}}$.
\begin{proof}
  Since $\underline{\mpc}(A\rightarrow A')\subset \mpc^{\star}(A\rightarrow A')$, it suffices to show that $\mpc^{\star}(A\rightarrow A')\subset \conv(\underline{\mpc}(A\rightarrow A'))$.
  This is evident from the Choi characterization in~\Cref{prop:wirecut_ds_choi_characterization}.
\end{proof}

\subsection{Utility of classical side information}\label{sec:wire_cut_side_info}
In this section, we will study the importance of intermediate measurements and classical side information for the achievability of time-like cuts.
This discussion will mirror the one about space-like cuts in~\Cref{sec:gate_cut_side_info}.
As alluded to previously, the channels $\mpc(A\rightarrow A')$ without classical side information are useless for the task for circuit knitting.
\begin{lemma}{}{span_mpc_ebc}
  For any systems $A,A'$ one has
  \begin{equation}\label{eq:span_mpc}
    \spn_{\mathbb{R}}(\mpc(A\rightarrow A')) = \{\rho\mapsto \tr[\rho]\sigma | \sigma\in\herm(A') \}
  \end{equation}
  \begin{equation}
    \spn_{\mathbb{R}}(\ebc(A\rightarrow A')) = \mathbb{R}\hptp(A\rightarrow A')
  \end{equation}
  and
  \begin{equation}
    \spn_{\mathbb{R}}(\mpc^{\star}(A\rightarrow A')) = \spn_{\mathbb{R}}(\ebc^{\star}(A\rightarrow A'))= \hp(A\rightarrow A') \, .
  \end{equation}
\end{lemma}
The proof of this statement follows quite directly from the Choi characterizations in~\Cref{prop:wirecut_ds_choi_characterization}.
We refrain from stating the full proof, as it essentially follows the identical arguments as~\Cref{lem:span_lo} and~\Cref{lem:span_locc}.
In fact,~\Cref{lem:span_mpc_ebc} can be considered a special instance of these two results.

To summarize, we have shown that classical side information through intermediate measurements is absolutely crucial for time-like cuts without classical communication.
In contrast, when classical communication is available, circuit knitting remains possible even without classical side information.
In fact, one can even show that in this case intermediate measurements provide no advantage in reducing the overhead.
\begin{proposition}{}{negtrick_ebc}
  Let $\cE\in\mathbb{R}\hptp(A\rightarrow A')$.
  Then $\gamma_{\ebc}(\cE)=\gamma_{\ebc^{\star}}(\cE)$.
\end{proposition}
The proof follows the same chain of arguments as the one of~\Cref{prop:negtrick_sepc_pptc}.
\begin{proof}
  By~\Cref{lem:gamma_ds_bound} we know $\gamma_{\ebc}(\cE)\geq\gamma_{\ebc^{\star}}(\cE)$, so it remains to show the converse.
  By the absolute convexity of $\ebc^{\star}$ and~\Cref{lem:one_element_qpd}, we can without loss of generality assume all QPDs w.r.t. $\ebc^{\star}$ to only have one element.
  Let $\cE=\kappa\cF$, $\cF\in\ebc^{\star}$ be such a QPD.
  By the Choi characterization in~\Cref{prop:wirecut_ds_choi_characterization} we can write $\choi{\cF}=\sigma^+-\sigma^-$ for two sub-normalized separable states $\sigma^+,\sigma^-$ fulfilling
  \begin{equation}
    \tr_{A'}[\sigma^++\sigma^-] = \frac{\id_A}{\dimension(A)} \, .
  \end{equation}
  At the same time,
  \begin{equation}
    \tr_{A'}[\sigma^+-\sigma^-] = \frac{1}{\kappa}\tr_{A'}[\choi{\cE}] = \frac{\tr[\choi{\cE}]}{\kappa}\frac{\id_A}{\dimension(A)}.
  \end{equation}
  This implies that both $\tr_{A'}[\sigma^+]$ and $\tr_{A'}[\sigma^-]$ are proportional to $\frac{\id_A}{\dimension(A)}$.
  As such, $\frac{\sigma^+}{\tr[\sigma^+]}$ and $\frac{\sigma^-}{\tr[\sigma^-]}$ are the Choi representations of two channels in $\ebc(A\rightarrow A')$ (in the case $\tr[\sigma^{\pm}]=0$, we define $\frac{\sigma^+}{\tr[\sigma^+]}$ to be zero) and we have a valid QPD w.r.t $\ebc$
  \begin{equation}
    \choi{\cE} = \kappa\tr[\sigma^+]\frac{\sigma^+}{\tr[\sigma^+]} - \kappa\tr[\sigma^-]\frac{\sigma^-}{\tr[\sigma^-]}
  \end{equation}
  with 1-norm of its coefficients $\kappa$.
\end{proof}

\subsection{Quasiprobability extent of the identity channel}\label{sec:id_wire_cut}
In this section, we study the quasiprobability extent of the most important channel in the context of circuit knitting with time-like cuts: the identity channel.
We will study this quasiprobability extent both w.r.t. $\ebc^{\star}$ and w.r.t. $\mpc^{\star}$.
\begin{theorem}{}{gamma_mpc_id}
  Let $A$ and $A'$ be $n$-qubit systems.
  Then
  \begin{equation}
    \gamma_{\mpc^{\star}(A\rightarrow A')}(\idchan) = 4^n \, .
  \end{equation}
\end{theorem}
\begin{proof}
  Due to~\Cref{cor:equivalence_mpc_def} and~\Cref{lem:gamma_hulls}, we can restrict ourselves to showing that $\gamma_{\underline{\mpc}(A\rightarrow A')}(\idchan) = 4^n$.
  We denote by $P_n$ the set of $n$-qubit Pauli strings $P_n\coloneqq \{\id,X,Y,Z\}^{\otimes n}$.
  We enumerate the elements of $P_n$ as $P_n=\{Q_1,Q_2,\dots,Q_{4^n}\}$.
  The set $P_n$ forms an orthonormal basis of $\herm(A)$ under the Hilbert-Schmidt inner product.

  We first show $\gamma_{\underline{\mpc}(A\rightarrow A')}(\idchan) \leq 4^n$ by finding an explicit QPD that fulfills the desired upper bound.
  We have
  \begin{align} 
    \idchan(\rho) &= \rho \\
    &= \sum_{Q\in P_n} \langle Q,\rho\rangle Q \\
    &= \sum_{Q\in P_n} \tr[Q\rho]\frac{Q}{2^n} \, .
  \end{align}
  Since for each $Q\in P_n$ one has $\norm{\frac{Q}{2^n}}_1=1$ and $\opnorm{Q}=1$, the maps
  \begin{equation} 
    \rho \mapsto \tr[Q\rho]\frac{Q}{2^n}
  \end{equation}
  are elements of $\underline{\mpc}(A\rightarrow A')$.
  As such, we have found a valid QPD with 1-norm $\abs{P_n}=4^n$.

  It remains to show the optimality of this decomposition.
  For this purpose, consider an arbitrary decomposition $\idchan=\sum_i a_i\cF_i$ where $\cF_i\in\underline{\mpc}(A\rightarrow A')$.
  Taking the trace on both sides, we get
  \begin{equation}
    4^n = \tr[\idchan] = \tr[\sum_i a_i\cF_i] = \sum_i a_i \tr[\cF_i] \,
  \end{equation}
  The operation $\cF_i$ is of the form $\rho\mapsto \tr[O_i\rho]\tau_i$ with $\opnorm{O_i}\leq 1$ and $\norm{\tau_i}_1\leq 1$.
  The Pauli transfer matrix representation (i.e. the representation in the orthonormal basis of Pauli strings) $M^{(i)}$ of $\cF_i$ is given by
  \begin{align}
    M^{(i)}_{\alpha,\beta} &= \langle Q_{\alpha}, \cF_i(Q_{\beta})\rangle \\
    &= \frac{1}{2^n}\tr[Q_{\alpha} \tr[O_iQ_{\beta}]\tau_i] \\
    &= \frac{1}{2^n}\tr[Q_{\alpha}\tau_i]\tr[O_iQ_{\beta}] \, .
  \end{align}
  Since $\tr[\cF_i]=\tr[M^{(i)}]$, we get
  \begin{align}
    4^n &= \sum_i a_i \tr[\cF_i] \\
    &= \sum_{i,\alpha} \frac{a_i}{2^n} \tr[Q_{\alpha}\tau_i] \tr [O_iQ_{\alpha}] \\
    &= \sum_{i,\alpha} \frac{a_i}{2^n} \tr[(Q_{\alpha}\otimes Q_{\alpha})(O_i\otimes\tau_i)] \, .
  \end{align}
  Let us next introduce the $\swap$ operator $\swap\coloneqq \frac{1}{2^n}\sum_{Q\in P_n}Q\otimes Q$.
  A well-known technique called the \emph{SWAP trick} states that for any two operators $A_1,A_2$, we can write $\tr[A_1A_2]=\tr[(A_1\otimes A_2)\swap]$~\cite{ando1979_concavity,carlen2009_trace}.
  As such, we obtain
  \begin{align}
    4^n &= \sum_{i} a_i \tr[O_i\tau_i] \\
    &\leq \sum_{i} \abs{a_i} \abs{\tr[O_i\tau_i]} \\
    &\leq \sum_{i} \abs{a_i} \norm{O_i\tau_i}_1 \\
    &\leq \sum_{i} \abs{a_i} \opnorm{O_i}\norm{\tau_i}_1 \\
    &\leq \sum_{i} \abs{a_i} 
  \end{align}
  where we used Hölder's inequality for Schatten norms.
  As such, no QPD can have a 1-norm smaller than $4^n$.
\end{proof}
Due to its practical importance, we briefly take the time to explicitly write down the achieving QPD in the $n=1$ case.
\begin{example}\label{ex:wirecut_no_cc}
  Let $A,A'$ be single-qubit systems.
  By~\Cref{thm:gamma_mpc_id}, we know that $\gamma_{\mpc^{\star}(A\rightarrow A')}(\idchan)=4$ which is achieved by the QPD $\idchan =\sum_ia_i\cF_i$ where
  \begin{equation}
    \cF_i(\rho) = \tr[O_i\rho] \sigma_i
  \end{equation}
  where the $a_i,O_i$ and $\sigma_i$ are given in~\Cref{tab:identity_cut}.
  This decomposition can be attributed to Peng \etal~\cite{peng2020_simulating}.
  \begin{table} 
  \centering
  \begin{tabular}{c|c|c}
     $O_i$ & $\sigma_i$  & $a_i$ \\ \hline
     $\id$ & $\proj{0}$ & $+1/2$ \\ \hline
     $\id$ & $\proj{1}$ & $+1/2$ \\ \hline
     $X$ & $\proj{+}$ & $+1/2$ \\ \hline
     $X$ & $\proj{-}$ & $-1/2$ \\ \hline
     $Y$ & $\proj{i}$ & $+1/2$ \\ \hline
     $Y$ & $\ket{-i}\bra{-i}$ & $-1/2$ \\ \hline
     $Z$ & $\proj{0}$ & $+1/2$ \\ \hline
     $Z$ & $\proj{1}$ & $-1/2$ 
  \end{tabular}
  \caption{Operators, states and weights required for a single time-like cut in a circuit.}
  \label{tab:identity_cut}
  \end{table}
\end{example}

Next, we consider the situation where classical communication is allowed for.
\begin{theorem}{\cite[Proposition 4.2]{brenner2023_optimal} and \cite[Lemma 1 (supplementary materials)]{yuan2021_universal}}{gamma_ebc_id}
  Let $A,A'$ be $d$-dimensional systems.
  Then
  \begin{equation}
    \gamma_{\ebc^{\star}(A\rightarrow A')}(\idchan) = 2d-1 \, .
  \end{equation}
\end{theorem}
\begin{proof}
  Thanks to~\Cref{prop:negtrick_ebc}, we can restrict our attention to $\gamma_{\ebc}$ instead of $\gamma_{\ebc^{\star}}$.
  The Choi representation of the identity channel is simply given by maximally entangled state $\ket{\Psi}_{AA'}$.
  We know from~\Cref{thm:optimal_qpd_pure} that the optimal QPD of $\proj{\Psi}$ w.r.t. $\sep(A;A')$ has 1-norm $2d-1$ and is given by $\proj{\Psi} = d\sigma^+ - (d-1)\sigma^-$ where
  \begin{equation}
    \sigma^- = \frac{1}{d(d-1)}\sum_{i\neq j} \proj{i}\otimes\proj{j} \in \sep(A;A')
  \end{equation}
  and
  \begin{equation}
    \sigma^+ = \frac{1}{d}\left( \proj{\Psi} + (d-1) \sigma^- \right) \in \sep(A;A').
  \end{equation}
  Fortunately, these two states happen to also fulfill the partial trace constraints
  \begin{equation}
    \tr_{A'}[\sigma^-]
    = \frac{1}{d(d-1)}\sum_{i\neq j} \proj{i}
    = \frac{1}{d}\sum_{i} \proj{i}
    = \frac{\id_A}{d}
  \end{equation}
  and
  \begin{equation}
    \tr_{A'}[\sigma^+] = \frac{1}{d}\left( \tr_{A'}[\proj{\Psi}] + (d-1) \tr_{A'}[\sigma^-] \right) = \frac{\id_A}{d} \, .
  \end{equation}
  Because of~\Cref{prop:wirecut_ds_choi_characterization}, $\sigma^{\pm}$ are the Choi representations of two channels in $\ebc(A\rightarrow A')$, hence giving us a valid QPD of the identity channel.
  Clearly, $\gamma_{\ebc}(\idchan)$ cannot be any lower than $2d-1$, as that would contradict $\gamma_{\sep}(\proj{\Psi})=2d-1$.
\end{proof}
The above proof hints at a deeper connection between time-like cuts of the identity channel and space-like cuts of maximally entangled states $\ket{\Psi}$.
While the proof is quite succinct, it gives little insight into this connection.
We explore this by providing a second argument that takes a much more operational approach to proving~\Cref{thm:gamma_ebc_id}.
For this, we will restrict ourselves to $A$ and $A'$ being $n$-qubit systems.

The argument for $\gamma_{\ebc^{\star}}(\idchan)\leq\gamma_{\sep}(\proj{\Psi})$ in the case of $n=1$ is graphically summarized in~\Cref{fig:state_teleportation}.
The first equality depicts the famous state teleportation protocol~\cite{bennett1993_teleporting}, which allows for the transfer of a qubit (i.e. the realization of the identity channel) from one party to another using a pre-shared Bell pair $\ket{\Psi}$ and one round of classical communication.
By taking taking a QPD of $\proj{\Psi}$ into product states $\proj{\Psi}=\sum_i a_i \rho_i\otimes \sigma_i$, we can write the protocol as a mixture of entanglement breaking channels, which lie in $\ebc^{\star}(A\rightarrow A')$.
The quasiprobability extent $\gamma_{\sep}(\proj{\Psi})$ precisely characterizes the lowest achievable 1-norm $\sum_i\abs{a_i}$ through this means.

\begin{figure}
  \centering
  \includegraphics{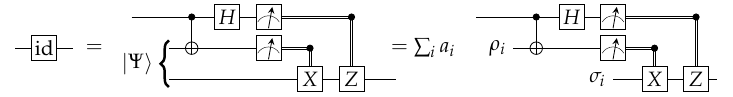}
  \caption{The state teleportation protocol (center) allows the transmission of one qubit through a round of classical communication by consuming one Bell pair $\ket{\Psi}$ of pre-shared entanglement. By using a QPD of the Bell state into product states $\proj{\Psi}=\sum_i a_i\rho_i\otimes\sigma_i$, we can write this protocol as a quasiprobabilistic mixture of entanglement breaking channels.}
   \label{fig:state_teleportation}
\end{figure}

On the converse side, a Bell-pair shared between two parties can also be prepared using one instance of an identity channel from one side to the other.
As such, any QPD of the identity channel into $\ebc^{\star}$ channels automatically induces a QPD of $\proj{\Psi}$ into separable states.

The above arguments for the $n=1$ case straightforwardly generalize to arbitrary $n\geq 1$ by using $n$ parallel instances of gate teleportation, or identity channels respectively.
Hence, $\gamma_{\ebc^{\star}}(\idchan^{\otimes n}) = \gamma_{\sep}(\proj{\Psi}^{\otimes n}) = 2^{n+1}-1$.
The equivalence between nonlocal identity channels and Bell pairs that we have found is strongly reminiscent of the equivalence of Clifford gates and their Choi states that we previously explored in~\Cref{sec:gamma_clifford}.

We also note that the sub-multiplicativity of $\gamma_{\ebc^{\star}}(\idchan)$ can be practically harnessed for time-like cuts that do not occur in the same time slice of the circuit.
We refrain from an in-depth discussion of this approach, since it exactly mirrors the technique utilized for black box cuts of Clifford gates in~\Cref{sec:submult}.
The main idea is to realize all space-like cuts using the state teleportation protocol, and to jointly simulate the preparation of all Bell pairs to harness the sub-multiplicative behavior of $\gamma_{\sep}$.
To avoid having to use a large number of ancilla qubits, an \emph{entanglement factory} can be utilized, which re-uses ancillas for Bell pair simulation.
In fact, the same entanglement factory can be shared between space-like and time-like cuts.

In a black box cut setting, it seems unlikely that one could entirely avoid such ancilla qubits.
However, in the parallel cut setting, there exist constructions that do not require any ancilla qubits in their implementation, which allows for a significantly simpler practical realization.
We refer the reader to \reference~\cite{harada2023_doublyoptimal,pednault2023_alternative} for details on these constructions.

As a somewhat surprising conclusion, this chapter demonstrated that classical communication can significantly reduce the overhead of circuit knitting with time-like cuts.
Indeed, to cut one single wire in a quantum circuit, the associated extent is $\gamma_{\mpc^{\star}}(\idchan)=4$ and $\gamma_{\ebc^{\star}}(\idchan)=3$ respectively.
When we cut many wires in parallel, the gap widens even more, as the regularized gamma factors are $\gammareg_{\mpc^{\star}}(\idchan)=4$ and 
\begin{equation}
  \gammareg_{\ebc^{\star}}(\idchan) = \lim\limits_{n\rightarrow\infty}(2^{n+1}-1)^{1/n} = 2 \, .
\end{equation}
This is especially remarkable, considering that for space-like cuts, we were unable to find an advantage provided by classical communication for a large variety of unitary gates.

\subsection{Going beyond the identity channel}\label{sec:general_timelike_cuts}
In this section, we briefly study the quasiprobability extents $\gamma_{\mpc^{\star}}(\cE)$ and $\gamma_{\ebc^{\star}}(\cE)$ for channels $\cE\in\cptp(A\rightarrow A')$ beyond the identity.
A first insight is that the identity channel is in some sense the maximally difficult channel to cut.
\begin{corollary}{}{wirecut_id_hardest}
  Let $\cE\in\cptp(A\rightarrow A')$ and either $\ds=\mpc^{\star}$ or $\ds=\ebc^{\star}$.
  Then
  \begin{equation}
    \gamma_{\ds(A\rightarrow A')}(\cE)\leq \min\{ \gamma_{\ds(A,A)}(\idchan_A), \gamma_{\ds(A',A')}(\idchan_{A'}) \} \,. 
  \end{equation}
\end{corollary}
This follows from a relatively simple insight, namely that a cut of $\cE$ can be straightforwardly realized by first cutting the identity channel and then running $\cE$ on the receiver's system, or analogously, running $\cE$ on the sender's system and then cutting the identity.
This insight is more generally captured by a stronger version of the chaining property.
\begin{lemma}{}{gamma_timelike_composition}
  Let $\cE\in\cptp(A\rightarrow A')$, $\cF\in\cptp(A'\rightarrow A'')$ and either $\ds=\mpc^{\star}$ or $\ds=\ebc^{\star}$.
  Then
  \begin{equation}
    \gamma_{\ds(A\rightarrow A'')}(\cF\circ\cE)\leq \min\{ \gamma_{\ds(A\rightarrow A')}(\cE), \gamma_{\ds(A',A'')}(\cF) \} \,. 
  \end{equation}
\end{lemma}
\begin{proof}
  It suffices to show that any QPD of $\cE$ w.r.t. $\ds(A\rightarrow A')$ and any QPD of $\cF$ w.r.t. $\ds(A'\rightarrow A'')$ both can induce a QPD of $\cF\circ\cE$ w.r.t. $\ds(A\rightarrow A'')$ with the same 1-norm of coefficients.
  This follows immediately from the fact that if you prepend or append a CPTP map $\cG$ to any map in $\ds(A\rightarrow A')$, one again obtains an element in $\ds(A\rightarrow A')$.
  For $\ds=\ebc^{\star}$ this is quite evident as
  \begin{equation}
    \cG\left(\sum_i \beta_i \tr[M_i\rho]\sigma_i \right) = \sum_i \beta_i \tr[M_i\rho] \cG(\sigma_i)
  \end{equation}
  and
  \begin{equation}
    \sum_i \beta_i \tr[M_i\cG(\rho)]\sigma_i = \sum_i \beta_i \tr[\cG^{\dagger}(M_i)\rho]\sigma_i 
  \end{equation}
  where $\cG^{\dagger}$ is the adjoint channel of $\cG$ and $\cG^{\dagger}(M_i)\geq 0,\sum_i\cG(M_i)=\id$ since $\cG$ is completely positive and unital (see e.g.~\cite[Proposition 2.18,Theorem 2.26]{watrous2018_theory}.
  An analogous argument applies for $\ds=\mpc^{\star}$.
\end{proof}
Another consequence of~\Cref{lem:gamma_timelike_composition} is that unitaries are exactly as difficult to cut as the identity channel.
\begin{corollary}{}{}
 Let $\cU\in\cptp(A)$ be a unitary channel and either $\ds=\mpc^{\star}$ or $\ds=\ebc^{\star}$.
 Then
 \begin{equation}
   \gamma_{\ds}(\cU) = \gamma_{\ds}(\idchan) \, .
 \end{equation}
\end{corollary}
\begin{proof}
  By~\Cref{cor:wirecut_id_hardest}, we know that $\gamma_{\ds}(\cU) \leq \gamma_{\ds}(\idchan)$.
  For the converse, we use~\Cref{lem:gamma_timelike_composition}
  \begin{equation}
    \gamma_{\ds}(\idchan) = \gamma_{\ds}(\cU\cU^{-1}) \leq \gamma_{\ds}(\cU) \, .
  \end{equation}
\end{proof}

So if we want to find some interesting physics beyond the identity channel, we necessarily need to consider non-unitary channels.
As for space-like cuts, non-unitarity significantly increases the complexity of evaluating the quasiprobability extent, so in the following we content ourselves with summarizing various techniques for finding upper and lower bounds.

As a first result, we note that the proof technique of~\Cref{thm:gamma_mpc_id} can be straightforwardly generalized to arbitrary channels, leading to following lower and upper bounds to $\gamma_{\mpc^{\star}}$.
\begin{proposition}{}{gamma_mpc_generalization}
  Let $A$ and $A'$ be $n$-qubit systems and $\cE\in\hp(A\rightarrow A')$.
  Then
  \begin{equation}
    \tr[M] \leq \gamma_{\mpc^{\star}(A\rightarrow A')}(\cE) \leq \sum_{Q\in P_n} \opnorm{\cE^{\dagger}(Q)}
  \end{equation}
  where $M$ is the Pauli transfer matrix representation of $\cE$, $P_n$ denotes the set of $n$-qubit Pauli strings and $\cE^{\dagger}$ is the adjoint map of $\cE$ (in the sense that $\langle \cE(X),Y)\rangle = \langle X,\cE^{\dagger}(Y)\rangle$).
\end{proposition}
\begin{proof}
  The lower bound is achieved by the exact same steps as in the proof of~\Cref{thm:gamma_mpc_id}.
  For the upper bound, we use
  \begin{align} 
    \cE(\rho)
    &= \sum_{Q\in P_n} \langle Q,\cE(\rho)\rangle Q \\
    &= \sum_{Q\in P_n} \langle \cE^{\dagger}(Q),\rho\rangle Q \\
    &= \sum_{Q\in P_n} \opnorm{\cE^{\dagger}(Q)}\tr\left[\frac{\cE^{\dagger}(Q)}{\opnorm{\cE^{\dagger}(Q)}}\rho\right]\frac{Q}{2^n}
  \end{align}
  which provides us with a QPD of $\cE$ w.r.t. the decomposition set $\underline{\mpc}(A\rightarrow A')$ with 1-norm $\sum_{Q\in P_n}\opnorm{\cE^{\dagger}(Q)}$.
\end{proof}
\begin{example}
  Consider the $1$-qubit depolarizing channel
  \begin{equation}
    \cD_p(\rho) = (1-p)\rho + \tr[\rho]\frac{1}{2}\id \, .
  \end{equation}
  It is easy to verify that it is self-adjoint $\cD_p^{\dagger}=\cD_p$ and that its Pauli transfer matrix representation is $M=\diag(1,1-p,1-p,1-p)$ since 
  \begin{align}
    \cD_p(\id) &= \id \, , & \cD_p(X) &= (1-p)X \, , \\
    \cD_p(Y) &= (1-p)Y \, , & \cD_p(Z) &= (1-p)Z  \, .
  \end{align}
  As such, we get matching lower and upper bounds from~\Cref{prop:gamma_mpc_generalization} and
  \begin{equation}
    \gamma_{\mpc^{\star}}(\cD_p) = 4 - 3p \, .
  \end{equation}
\end{example}

Next, we focus purely on the quasiprobability extent w.r.t. $\ebc$, which is the same as the one w.r.t. $\ebc^{\star}$ as shown in~\Cref{prop:negtrick_ebc}.
Through the Choi characterization of~\Cref{prop:wirecut_ds_choi_characterization}, $\gamma_{\ebc(A\rightarrow A')}(\cE)$ is almost the same quantity as $\gamma_{\sep(A;A')}(\choi{\cE})$, except that we impose an additional partial trace constraint on the decomposition set.
As such, many techniques utilized for space-like cutting can be directly applied to bound $\gamma_{\ebc}$.
In fact, we can even consider $\ebc(A\rightarrow B')$ to be a special case of the decomposition set $\sepc(A;B\rightarrow A';B')$ where $A'$ and $B$ are chosen to be of dimension $1$.
Therefore, the techniques from~\Cref{sec:gamma_nonlocal_channels} to estimate $\gamma_{\sepc}$ are directly applicable here.
We briefly summarize them below.

\paragraph{State conversion bounds}
For any bipartite state $\rho_{AB}\in\dops(AB)$ and any $\cE\in\cptp(A\rightarrow A')$ one has
\begin{equation}
  \gamma_{\ebc^{\star}}(\cE) \geq \gamma_{\sep}(\cE\otimes\idchan_B(\rho)) \, .
\end{equation}
This can be verified by checking that any QPD of $\cE$ induces one of $\cE\otimes\idchan_B(\rho)$ with equal 1-norm.
In the special case where $\rho_{AB}$ is chosen as the maximally entangled state, we get
\begin{equation}
  \gamma_{\ebc^{\star}}(\cE) \geq \gamma_{\sep}(\choi{\cE}) \, .
\end{equation}

\paragraph{PPT relaxation bound}
By using a PPT relaxation on the separability constraint in the Choi characterization of $\ebc$, we obtain a lower bound for the quasiprobability extent of $\cE\in\cptp(A\rightarrow A')$
\begin{equation}
  \gamma_{\ebc^{\star}(A\rightarrow A')}(\cE) \geq \left \lbrace
  \begin{array}{r l}
    \min\limits_{\Lambda^{\pm}\in\pos(AA')} & \tr[\Lambda^+] + \tr[\Lambda^-] \\
    \textnormal{s.t.} & \choi{\cE} = \Lambda^+ - \Lambda^- , \\
    & \tr_{A'}[\Lambda^{\pm}] = \tr[\Lambda^{\pm}]\frac{1}{\dimension(A)}\id_{A} , \\
    & (\idchan_{A}\otimes\transpose_{A'})(\Lambda^{\pm}) \loewnergeq 0
  \end{array} \right .
\end{equation}
which is an SDP that can be efficiently evaluated.

\paragraph{SDP hierarchy}
As a refinement of the above PPT relaxation, we can lower bound $\gamma_{\ebc(A\rightarrow A')}(\cE)$ through the SDP $\gamma_{\ebc_k(A\rightarrow A')}(\cE)$ where
\begin{equation}
  \ebc_k(A\rightarrow A') \coloneqq \{ \cE\in\cptp(A\rightarrow A') | \choi{\cE}\in\sep_k(A;A') \}
\end{equation}
is a converging outer approximation $\ebc(A\rightarrow A')=\bigcap\limits_{k=1}^{\infty}{\ebc_k(A\rightarrow A')}$.

\paragraph{Asymptotic characterization}
Let us define the set of entanglement breaking-preserving superchannels $\ebpsc(A,A'\rightarrow B,B')$ to be the superchannels which map elements from $\ebc(A\rightarrow A')$ to $\ebc(B\rightarrow B')$.
The single-shot dilution cost $E_{C,\ebpsc}^{(1)\epsilon}(\rho)$ is characterized by the smooth log-robustness
\begin{equation}
  \log(1+R_{\ebc}^{\epsilon}(\cE)) \leq E_{C,\ebpsc}^{(1),\epsilon}(\cE) \leq \log(1+R_{\ebc}^{\epsilon}(\cE)) + 2 \, .
\end{equation}
This was independently proven in \reference~\cite[Theorem 1 (supplementary materials)]{yuan2021_universal}, though it can also be seen as a direct consequence of~\Cref{lem:oneshot_seppsc_cost}.
As such, the regularized and asymptotic quasiprobability extent w.r.t. $\ebc^{\star}$ can be expressed as
\begin{equation}
  \log \gammareg_{\ebc^{\star}}(\cE) = E_{C,\ebpsc}^{\mathrm{exact}}(\cE)
\end{equation}
\begin{equation}
  \log \gammasreg_{\ebc^{\star}}(\cE) = E_{C,\ebpsc}(\cE) \, .
\end{equation}

\section{Further reading}
The idea of extending the size of a quantum computer was first proposed by Bravyi \etal~\cite{bravyi2016_trading}, specifically in the context of sparse quantum circuits and Pauli-based computation.
The idea is to simulate an $n+k$-qubit quantum circuit with an $n$-qubit quantum computer by using a hybrid classical-quantum scheme that simulates $k$ ``virtual qubits''.
The overhead of the technique is exponential in $k$, which is generally less favorable compared to circuit knitting.

This shortcoming was addressed in the first circuit knitting scheme by Peng \etal~\cite{peng2020_simulating}.
Their approach uses the time-like QPD of the identity channel in~\Cref{ex:wirecut_no_cc} to decompose the total circuit into sub-circuits that fit on a small quantum device, though their approach is phrased in the language of tensor networks.
The individual circuit clusters are either combined through QPS or via a tensor network contraction.

The first circuit knitting scheme with space-like cuts was introduced by Mitarai and Fujii in \reference~\cite{mitarai2021_constructing,mitarai2021_overhead}.
The authors propose explicit QPDs of arbitrary two-qubit gates in terms of a decomposition set that contains single-qubit rotations and projective measurements (including classical side information derived from those).
This decomposition set is a strict subset of $\lo^{\star}$, and the authors explicitly define the quasiprobability extent of a channel w.r.t. that set (which they call \emph{channel robustness of nonlocality}), however without being able to evaluate it.

In a later paper by Piveteau and Sutter~\cite{piveteau2024_circuit}, the decomposition sets $\lo^{\star}$ and $\locc^{\star}$ were introduced for the first time, though their definition of $\lo^{\star}$ did not encompass the full generality of classical side information (an inconsequential technicality which we discussed around~\Cref{eq:old_lostar_def}).
The QPD from~\cite{mitarai2021_overhead} was also proven to be optimal w.r.t. $\lo^{\star}$ and $\locc^{\star}$ for two-qubit gates with at most two non-zero KAK parameters.
In that work, the connection to the robustness of entanglement was first established.
The authors also considered for the first time the possibility to utilize real-time classical communication between the two parties.
They proposed teleportation-based circuit knitting in the black box setting and proved the optimal quasiprobability extent of Clifford gates w.r.t. $\locc^{\star}$.

Two subsequent papers, one by Schmitt, Piveteau and Sutter~\cite{schmitt2024_cutting} and the other by Harrow and Lowe~\cite{harrow2024_optimal}, simultaneously and independently established the quasiprobability extent for unitaries with unitary operator Schmidt decompositions (and as such arbitrary two-qubit gates).
The former paper additionally extended the black-box results to arbitrary two-qubit unitaries.
The latter instead applied the techniques to clustered Hamiltonian simulation and also provided an improved wire cut-based circuit kniting scheme that can further reduce the sampling overhead given a certain sparsity assumption on the observable.
A third paper by Ufrecht \etal~\cite{ufrecht2024_optimal} also appeared simultaneously which describes an optimal sub-multiplicative black box cutting method for the more restricted set of two-qubit rotation gates.

The $\pptc$-based lower bound for the quasiprobability extent $\gamma_{\locc}$ (without intermediate measurements) of channels was first proposed by Jing \etal~\cite{jing2024_circuit}.
They also established the lower bounds for the regularized quasiprobability extent $\gammareg_{\pptc}$ and $\gammareg_{\sepc}$ in terms of the entanglement cost for the first time.

We briefly note that the quasiprobability extent $\gamma_{\sepc}(\cU)$ of unitary channels $\cU$ with unitary operator Schmidt decompositions was introduced and characterized long before the advent of circuit knitting by Harrow and Nielsen in \reference~\cite{harrow2003_robustness}.


Focusing on wire cutting, the quasiprobability extent $\gamma_{\ebc}$ (without classical side information) of the identity channel was first characterized in \reference~\cite{yuan2021_universal}.
This work focuses more on resource theoretic aspects instead of circuit knitting and the authors were seemingly not aware of the previous work by Peng \etal~\cite{peng2020_simulating}.
Similarly, many subsequent circuit knitting papers were not aware of~\cite{yuan2021_universal} and did not refer to it.

Lowe \etal~\cite{lowe2023_fast} proposed the first circuit knitting scheme based on time-like cuts that utilizes  classical communication and exhibits strictly sub-multiplicative scaling.
While their approach does achieve the optimal asymptotic scaling, its finite-size extent is strictly sub-optimal.

Later work by Brenner, Piveteau and Sutter~\cite{brenner2023_optimal} established the notions of $\ebc^{\star}$ and $\underline{\mpc}$ and they provided and exact characterization of $\gamma_{\ebc^{\star}}(\idchan)$ and $\gamma_{\underline{\mpc}}(\idchan)$ with explicit QPDs achieving the two extents.
They also introduced teleportation-based time-like cutting which allowed for its application in the black-box setting.

The optimal $\ebc^{\star}$ QPD in \reference~\cite{brenner2023_optimal} requires the utilization of ancilla qubits for its practical realization.
Two subsequent works by Pednault~\cite{pednault2023_alternative} and by Harada, Wada and Yamamoto\cite{harada2023_doublyoptimal} provide alternative optimal QPDs with explicit circuit implementations that do not require these ancilla qubits.


To date, multiple experimental demonstrations of QPS-based circuit knitting have been realized~\cite{yamato2023_errorsuppresion,singh2024_experimental,vazquez2024_scaling}.

\section{Contributions}
This chapter almost exclusively contains novel contributions by the author.
Some of them were previously presented in \reference~\cite{piveteau2024_circuit,brenner2023_optimal,schmitt2024_cutting} while many others are new and previously unpublished.
We briefly summarize the contributions section-by-section.
\begin{description}
  \item[\Cref{sec:spacelike_ds}]
  The formal definition of the quasiprobability extent for circuit cutting was introduced by the author in \reference~\cite{piveteau2024_circuit} for LO and LOCC.
  In this section, we establish connections with the novel formalism from~\Cref{chap:qpsim} and consider the extension to the relaxations SEPC and PPTC.
  
  \item[\Cref{sec:gate_cut_side_info}]
  This section only contains previously unpublished results.
  
  \item[\Cref{sec:gamma_nonlocal_states}]
  The characterization of $\gamma_{\sep}$ for pure states is originally due to Vidal and Tarrach~\cite{vidal1999_robustness}.
  Though we note that the convex programming-based proof presented in~\Cref{sec:gamma_nonlocal_pure} is new, arguably simpler and as a side result also provides a characterization of $\gamma_{\ppt}$ for pure states.

  All results in~\Cref{sec:gamma_sep_singleshot,sec:gamma_sep_asymptotic} are new and previously unpublished.
  
  \item[\Cref{sec:gamma_nonlocal_channels}]
  The~\Cref{sec:gamma_clifford,sec:gamma_kaklike} summarize the authors' results published in~\reference~\cite{piveteau2024_circuit,schmitt2024_cutting}.
  During the preparation of the second publication, we became aware of efforts by Lowe and Harrow who had independently discovered a proof of~\Cref{thm:gamma_kaklike_unitary}~\cite{harrow2024_optimal}.
  While their work was published later than ours, we acknowledge an equal claim to the discovery of the result.

  The work in~\Cref{sec:sepc_conversion_bound,sec:sdp_hierarchy_channels,sec:overview_channel_bounds,sec:asymptotic_channel_extent} is new and previously unpublished.

  We note that the idea of using the PPTC relaxation for lower bounds, as presented in~\Cref{sec:pptc_lower_bound}, is due to Jing, Zhu and Wang~\cite{jing2024_circuit}.
  Furthermore, the characterization of the regularized quasiprobability extent in~\Cref{sec:asymptotic_channel_extent} can be seen as a refinement of the work in \reference~\cite{jing2024_circuit}, where the authors showed that $\gammareg_{\locc}$ can be lower bounded by an entanglement cost.

  \item[\Cref{sec:submult}]
  This section summarizes results by the author from \reference~\cite{piveteau2024_circuit,schmitt2024_cutting}.

  \item[\Cref{sec:gate_cut_app}]
  To our knowledge, the analysis of the quantum Fourier transform is new and previously unpublished.
  The results on clustered Hamiltonian simulation are due to Harrow and Lowe~\cite{harrow2024_optimal}.

  \item[\Cref{sec:timelike_ds}]
  The decomposition sets $\underline{\mpc}$ and $\ebc^{\star}$ were initially introduced by the author in \reference~\cite{brenner2023_optimal}.
  This section extends this work by formally introducing $\mpc^{\star}$ and determining the Choi characterization of the various decomposition sets.

  \item[\Cref{sec:wire_cut_side_info}]
  This section only contains previously unpublished results.
 
  \item[\Cref{sec:id_wire_cut}]
  The two main results that characterize $\gamma_{\mpc^{\star}}(\idchan)$ and $\gamma_{\ebc^{\star}}(\idchan)$ (\Cref{thm:gamma_mpc_id,thm:gamma_ebc_id}) were both introduced by the author in \reference~\cite{brenner2023_optimal}.
  After publishing this paper, we were made aware of a previous work that had independently obtained a characterization of $\gamma_{\ebc}(\idchan)$, though with a different proof technique~\cite[Lemma 1 (supplementary materials)]{yuan2021_universal}.
  Together with our insight from~\Cref{prop:negtrick_ebc}, this also leads to a characterization of $\gamma_{\ebc^{\star}}(\idchan)$.

  \item[\Cref{sec:general_timelike_cuts}]
  This section only contains previously unpublished results, except for the PPT lower bound, which was proposed in \reference~\cite{yuan2021_universal}.

\end{description}

\chapter{Simulating magic computation}\label{chap:magic}
In this chapter, we discuss the simulation of non-Clifford circuits using a Clifford quantum computer.
This application of QPS was already briefly covered as a motivational example in~\Cref{sec:motivational_example} and the discussion here provides a more in-depth treatment.

The famous Gottesman-Knill theorem states that any quantum circuit can be efficiently simulated classically if it only consists of Clifford gates as well as preparations and measurements of qubits in the computational basis~\cite{gottesman1998_heisenberg,aaronson2004_improved}.
Clifford gates also play an important role in quantum error correction, as the naturally fault-tolerant (transversal) gates often form a subgroup of the Clifford group.
The Clifford gates are not universal, and in order to reach universality, they need to be complemented by states or gates which are in some sense ``non-classical''.
The most common such choice is the T gate, or equivalently the magic state $\ket{H}\coloneqq \left(\ket{0} + e^{i\pi / 4} \ket{1} \right)/\sqrt{2}$ which can realize a T gate using the injection circuit in~\Cref{fig:magic_state_injection}.

The resource theory of magic aims to quantify the amount of non-classicality (called ``magic'') that non-Clifford channels or non-stabilizer states contain.
This chapter will formally introduce this resource theory and study the induced quasi\-pro\-ba\-bi\-li\-ty extent.
We will restrict ourselves to qubits (as opposed to higher-dimensional qudits) since that represents the computationally most relevant setting.
Intriguingly, the set of free states is finite (respectively convex with finitely many extremal points).
At first, this might suggest that this resource theory could be simpler to work with compared to others like the QRT of entanglement.
Alas, the opposite is the case: While it is indeed possible to numerically evaluate the quasiprobability extent of any state for small systems through the linear program in~\Cref{sec:qps_conv_programming}, it is very hard to find analytic closed-form expressions, and we know very little about the asymptotic behavior of the quasiprobability extent, even for the simplest cases.
Very little is also known about the utility of classical side information through intermediate measurements in this theory - this topic has not previously been studied to our knowledge.

The obvious practical application of QPS w.r.t. the QRT of magic is classical simulation of non-Clifford circuits.
However, it is important to note that there exist other approaches that can simulate near-Clifford circuits more efficiently.
The state-of-the-art approach is instead based on the stabilizer rank simulation algorithm proposed in \reference~\cite{bravyi2016_improved,bravyi2019_simulation} and generally seems to exhibit roughly a square root speedup over the QPS approach.
The main difference of the stabilizer rank simulation approach is that it is based on quasiprobability decompositions of magic states and gates in the pure state quantum mechanics formalism instead of the mixed state formalism. 
That is, one decomposes non-Clifford unitaries (or non-stabilizer states) into Clifford unitaries (or stabilizer states) instead of working with channels (or density matrices).
The runtime of the algorithm is then related to the number of elements in this decomposition, which interestingly can also be related to the 1-norm of the coefficients~\cite[Theorem 1]{bravyi2019_simulation}.
However, due to the nature of the algorithm, stabilizer-rank simulators cannot incorporate mixed states and non-unitary channels the same way as QPS does, at least not without modifications.

Another simulation algorithm that outperforms QPS w.r.t. the QRT of magic is the \emph{dyadic frame simulator}~\cite{seddon2021_quantifying}.
This technique can actually also be considered an instance of QPS itself, but with an extended decomposition set.

In summary, there exist better simulation algorithms for the task of simulating near-Clifford circuits with a classical computer.
However, these improved techniques explicitly rely on using a \emph{classical} computer that allows for strong simulation.
We framed QPS as a method for simulating quantum computers using a \emph{quantum} computer itself (a Clifford quantum computer in this case), so these improved techniques are not on the same footing.
Still, the practical implications of this chapter are less clear compared to the other ones.
For this reason, we will go into less detail and content ourselves with sketching some basic properties of the decomposition sets and the quasiprobability extent.

\section{Decomposition sets}
The core idea behind the Gottesman-Knill theorem is that any intermediate state of a Clifford circuit necessarily belongs to the class of \emph{stabilizer states}, a finite subset of all $n$-qubit quantum states which can be represented with only $\mathrm{poly}(n)$ memory.
Conversely, these states precisely represent the set of reachable states from a computational basis state using a Clifford unitary.
\begin{definition}{}{}
  Let $S\subset \mathrm{P}_n$ be an Abelian subgroup of the Pauli group which does not contain $-(\id^{\otimes n})$.
  The $n$-qubit \emph{stabilizer code} $\mathcal{C}(S)$ associated to $S$ is defined to be the simultaneous $+1$ eigenspace of the operators in $S$.
\end{definition}
It can easily be shown that the dimension of a stabilizer code is precisely $\dimension(\mathcal{C}(S))=2^k$ where $n-k$ is the minimal number of generators of $S$.
We say that $\mathcal{C}(S)$ is a $[n,k]$ stabilizer code.
A $[n,0]$ stabilizer code only represents a single state (up to a complex scaling factor).
States that can be represented as a $[n,0]$ stabilizer code are called \emph{stabilizer states}, and we denote the set of stabilizer states of an $n$-qubit system $A$ by $\stab(A)$.

Having defined the free states of our QRT, we next turn to defining the free channels.
There are at least two reasonable choices that can be made here.
The first choice is very operational in nature, as it precisely entails the set of operations that can be efficiently simulated through the Gottesman-Knill theorem.
\begin{definition}{}{}
  Let $A$ and $A'$ be multi-qubit systems.
  We say that a quantum channel from $A$ to $A'$ is a \emph{stabilizer operation} if it is a composition of
  \begin{itemize}[noitemsep]
    \item preparations of stabilizer states,
    \item Clifford gates,
    \item Pauli measurements,
    \item discarding of qubits,
  \end{itemize}
  which may include classical control (based on previous measurement outcomes) that can involve randomness.
  We denote the set of such stabilizer operations by $\stabops(A\rightarrow A')$.
\end{definition}
Note that we do not impose any sort of computational restrictions on the classical control logic - it could in principle entail very long and expensive classical calculations.
This is a simplifying assumption akin to the one we made for circuit knitting with classical communication, where we allowed LOCC protocols with an unbounded number of rounds.

Another possible choice for the free operations is the set of completely resource non-generating channels.
\begin{definition}{}{}
  Let $A$ and $A'$ be multi-qubit systems.
  The set of completely stabilizer-preserving channels $\csp(A\rightarrow A')\subset\cptp(A\rightarrow A')$ is defined to contain all channels $\cE$ such that for any ancillary system $E$
  \begin{equation}
    \forall \rho\in\stab(AE): (\cE\otimes \idchan_E)(\rho) \in \stab(A'E) \, .
  \end{equation}
\end{definition}
While $\csp(A\rightarrow A')$ is a more axiomatic choice of free operations, the operationally meaningful decomposition set for the purposes of QPS is clearly given by $\stabops(A\rightarrow A')$.
The relation between $\stabops(A\rightarrow A')$ and $\csp(A\rightarrow A')$ was recently studied in detail by Heimendahl~\etal~\cite{heimendahl2022_axiomatic}.
Interestingly, they observed that $\stabops(A\rightarrow A')$ is a strict subset, i.e., there exist completely stabilizer-preserving operations which lie outside of $\stabops(A\rightarrow A')$.
It can still be useful to consider $\csp(A\rightarrow A')$ as a well-behaved relaxation of $\stabops(A\rightarrow A')$, since it exhibits a rather simple characterization in terms of the Choi representation.
\begin{proposition}{\cite[Theorem 3.1]{seddon2019_quantifying}}{csp_characterization}
  Let $A$ and $A'$ be multi-qubit systems and $\cE\in\supo(A\rightarrow A')$.
  Then
  \begin{equation}
    \cE\in\csp(A\rightarrow A') \Longleftrightarrow \choi{\cE} \in \conv(\stab(AA')) \text{ and } \tr_{A'}[\choi{\cE}]=\frac{1}{\dimension(A)}\id_A \, .
  \end{equation}
\end{proposition}
Similarly, there exists a ``normal form'' of stabilizer operations, though it is not quite as simple.
\begin{proposition}{\cite[Theorem 4]{heimendahl2022_axiomatic}}{stabops_standard_form}
  Let $A$ and $A'$ be $n$-qubit and $m$-qubit systems.
  Every $\cE\in\stabops(A\rightarrow A')$ can be written as a convex mixture of stabilizer operations of the type
  \begin{equation}\label{eq:stabops_standard_form}
    \cE(\rho) = \tr_{m+1,\dots,n+r} \left[ \sum_i U_i\left( \Pi_i\rho\Pi_i^{\dagger}\otimes \proj{0}^{\otimes r} \right) U_i^{\dagger} \right]
  \end{equation}
  for some $r\geq 0$, $n+r$-qubit Clifford unitaries $U_i$ and a projective measurement $\{\Pi_i\}_i$ of mutually orthogonal stabilizer code projectors $\Pi_i$.
\end{proposition}
In fact, the proof of~\Cref{prop:stabops_standard_form} in \reference~\cite{heimendahl2022_axiomatic} implies even a slightly stronger statement.
The free quantum instrument $\instr{\stabops}(A\rightarrow A')$ from $A$ to $A'$ (recall~\Cref{def:free_qi}) are precisely given by coarse-grainings of convex mixtures of instruments $(\cE_i)_i\in\qi(A\rightarrow A')$ of the form
\begin{equation}\label{eq:stabops_qi}
  \cE_i(\rho) =
  \tr_{m+1,\dots,n+r} \left[ U_i \left( \Pi_i\rho\Pi_i^{\dagger}\otimes \proj{0}^{\otimes r} \right) U_i^{\dagger} \right]
\end{equation}
for some $r\geq 0$, $n+r$-qubit Clifford unitaries $U_i$ and a projective measurement $\{\Pi_i\}_i$ of mutually orthogonal stabilizer code projectors $\Pi_i$.
Therefore, we can write the decomposition set $\stabops^{\star}(A\rightarrow A') \coloneqq \qids[\instr{\stabops}(A\rightarrow A')]$, which captures the notion of intermediate measurements, as the absolute convex hull of all superoperators of the form
\begin{equation}
    \rho \mapsto \tr_E \left[ \sum_i \beta_i U_i\left( \Pi_i\rho\Pi_i^{\dagger}\otimes \proj{0}^{\otimes r} \right) U_i^{\dagger} \right]
\end{equation}
for some $\beta_i\in[-1,1]$.

We note that $\stabops$ and $\csp$ are convex sets and $\stabops^{\star}$ is even absolutely convex (recall~\Cref{lem:qrt_convexity}).
As such, many properties from~\Cref{sec:qps_basic_properties} directly apply to the associated quasiprobability extents, such as the relation to a robustness measure and the reduction to QPDs with at most two (respectively one) elements.
Note also that any $\cptp$ map can be simulated using the decomposition sets $\stabops$ and $\csp$.
\begin{lemma}{}{}
  Let $A$ and $A'$ be multi-qubit systems.
  Then
  \begin{equation}
    \spn_{\mathbb{R}}(\stabops(A\rightarrow A')) = \spn_{\mathbb{R}}(\csp(A\rightarrow A')) = \mathbb{R}\hptp(A\rightarrow A') \, .
  \end{equation}
\end{lemma}
Recall that $\mathbb{R}\hptp(A\rightarrow B) \coloneqq \{r\cdot \cE | r\in\mathbb{R}, \cE\in\hptp(A\rightarrow B)\}$.
\begin{proof}
  Since $\stabops(A\rightarrow A')\subset\csp(A\rightarrow A')\subset\hptp(A\rightarrow A')$ it suffices to show $\mathbb{R}\hptp(A\rightarrow A')\subset\spn_{\mathbb{R}}(\stabops(A\rightarrow A'))$.
  For this purpose, consider an arbitrary basis $\{\cE_1,\cE_2,\dots,\cE_M\}\subset\cptp(A\rightarrow A')$ of $\hptp(A\rightarrow A'))$ (recall from~\Cref{lem:cptp_span} that $\cptp$ spans $\hptp$).
  Denote by $D\neq 0$ the determinant of the matrix obtained by stacking the vectorized representations of the $\cE_i$.
  Because of the Stinespring dilation and Solovay-Kitaev theorems and the fact that the T gate can be decomposed in terms of Clifford gates (see~\Cref{eq:optimal_T_qpd}), all the channels $\cE_i$ can be approximated arbitrarily well by linear superpositions of elements in $\stabops(A\rightarrow A')$.
  By choosing the approximation error to be small enough, we can ensure that the change in determinant is smaller than $\abs{D}$ such that the approximated channels still remain linearly independent (since the determinant is continuous).
\end{proof}
Similarly, one can verify that 
\begin{equation}
  \spn_{\mathbb{R}}(\stabops^{\star}(A)) = \hp(A)
\end{equation}
which immediately follows from the fact that the basis of $\hp(A)$ presented in~\Cref{ex:modifiedendobasis} lies in $\stabops^{\star}(A)$.

\section{Quasiprobability extent of states}
Due to unfortunate historical circumstances, the quasiprobability extent $\gamma_{\stab}$ is commonly called the ``robustness of magic'' in literature.
This clashes with the more commonly used notion of robustness $R_{\stab}$ (as introduced in~\Cref{def:robustness}) which relates to the extent by $\gamma_{\stab}(\rho)=1+2R_{\stab}(\rho)$ for $\rho\in\dops(A)$ by~\Cref{lem:robustness1}.
In this manuscript, we will stick to the nomenclature used in other resource theories and call $\gamma_{\stab}$ the ``quasiprobability extent'' and $R_{\stab}$ the ``robustness''.

Since the decomposition set is finite, we can express the quasiprobability extent $\gamma_{\stab(A)}(\rho)$ in terms of a linear program (recall~\Cref{sec:qps_conv_programming}).
As such, we can numerically evaluate it for small systems with a runtime of $\mathrm{poly}(\abs{\stab(A)})$.
Unfortunately, in contrast to the resource theory of entanglement, we do not know an explicit analytical expression for the quasiprobability extent of product states $\gamma_{\stab(A)}(\rho^{\otimes m})$ for arbitrary $m$, even for simple states like the magic state $\ket{H}$.
This makes the estimation of the regularized extent $\gammareg_{\stab}$ very difficult.
To make matters worse, $\smash{\abs{\stab(A)}=2^n\prod_{k=1}^{n}(2^{k}+1)}\geq 2^{n(n+3)/2}$ grows super-exponentially in the number of qubits $n$~\cite{aaronson2004_improved,singal2023_counting}, making the cost of numerically evaluating $\gamma_{\stab}$ grow out of control extremely quickly.
By exploiting the symmetry structure of the stabilizer polytope~\cite{heinrich2019_robustness} as well as algorithmic improvements in calculating stabilizer state overlaps~\cite{hamaguchi2024_handbook}, state-of-the-art techniques are able to compute the quasiprobability extent of arbitrary states up to $8$-qubits.
By considering states with certain symmetries or by allowing for approximations, even larger systems can be handled - see~\cite[Table 1]{hamaguchi2024_handbook} for a detailed overview.
We especially note an interesting approach in \reference~\cite{heinrich2019_robustness} which proposes a hierarchy of inner approximations of the stabilizer polytope by restricting the decomposition set to stabilizer states which are at most $k$-partite entangled for $1\leq k\leq n$.
This is somewhat reminiscent of the SDP hierarchy we proposed in~\Cref{sec:gamma_sep_singleshot}, and allows for efficiently computable upper bounds to the quasiprobability extent together with explicit QPDs that achieve these upper bounds.

A characteristic that $\gamma_{\stab}$ does share with its entanglement counterpart $\gamma_{\sep}$ is that it exhibits strict sub-multiplicativity.
\begin{example}
  The magic state $\ket{H}\coloneqq \left(\ket{0} + e^{i\pi / 4} \ket{1} \right)/\sqrt{2}$ has a quasiprobability extent $\gamma_{\stab}(\proj{H})=\sqrt{2}$ which is achieved by the QPD
  \begin{equation}
    \proj{H} = \frac{1}{2}\proj{+} + \frac{\sqrt{2}}{2}\proj{i+} - \frac{\sqrt{2}-1}{2}\proj{-}
  \end{equation}
  where $\ket{\pm}\coloneqq (\ket{0}\pm\ket{1})/\sqrt{2}$ and $\ket{i\pm}\coloneqq (\ket{0}\pm i\ket{1})/\sqrt{2}$.
  The normalized quasiprobability extent of $m$ copies of the magic state $\gamma_{\stab}(\proj{H}^{\otimes m})^{1/m}$ was evaluated in \reference~\cite{heinrich2019_robustness} and is depicted in~\Cref{fig:gamma_magic_state}.
\end{example}
\begin{figure}
  \centering
  \includegraphics{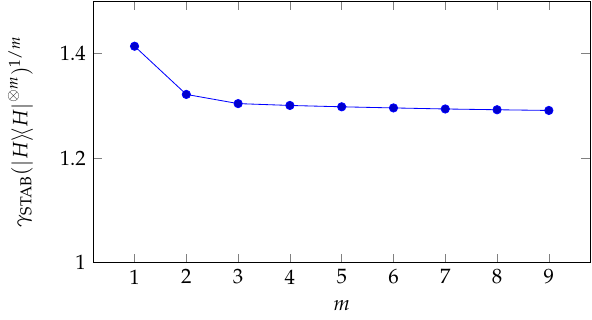}
  \caption{Normalized quasiprobability extent $\gamma_{\stab}(\proj{H}^{\otimes m})^{1/m}$ of $m$ magic states. Data taken from \reference~\cite{heinrich2019_robustness}.}
  \label{fig:gamma_magic_state}
\end{figure}

The sub-multiplicative behavior is very important as the asymptotic simulation overhead of a Clifford+T circuit with $t$ T gates can be reduced from $\mathcal{O}(\gamma_{\stab}(\proj{H})^{2t})$ to $\mathcal{O}(\gammareg_{\stab}(\proj{H})^{2t})$ where the regularized quasiprobability extent is estimated to be around $\gammareg_{\stab}(\proj{H})\approx 1.283 \pm 0.002$ in \reference~\cite{heinrich2019_robustness}.

\section{Quasiprobability extent of channels}
Generally, finding the quasiprobability extent $\gamma_{\stabops^{\star}}(\cE)$ of a channel $\cE\in\cptp(A\rightarrow A')$ is very difficult.
In the case where there is a one-to-one correspondence between $\cE$ and a state, we can relate $\gamma_{\stabops^{\star}}(\cE)$ to the extent of the state.
\begin{example}\label{ex:T_magic_correspondence}
  A T gate and a magic state $\ket{H}$ are equivalent under stabilizer operations, as one can be transformed into the other (recall the magic state injection in~\Cref{fig:magic_state_injection} and $\ket{H}=T\ket{+}$).
  As such, we have for all $m\in\mathbb{N}$
  \begin{equation}
    \gamma_{\stabops^{\star}}(T^{\otimes m}) = \gamma_{\stabops}(T^{\otimes m}) = \gamma_{\stab}(\proj{H}^{\otimes m})
  \end{equation}
  by~\Cref{lem:gamma_chaining,lem:negativity_state_qpd}
\end{example}
Unfortunately, such a state-channel correspondence is generally not available.
Due to the difficulty of estimating the quasiprobability extent w.r.t. $\stabops^{\star}$, we will from now on focus on estimating $\gamma_{\stabops}(\cE)$.
Below, we summarize a few strategies for finding lower and upper bounds, the latter through possibly non-optimal QPDs.

To find explicit QPDs for upper bounds, a general approach consists of finding inner approximations of $\stabops$ and use those as decomposition sets.
For instance, one could pick the finite set of all Clifford unitary channels (respectively its convex hull).
Interestingly those are sufficient for achieving an optimal QPD for the T gate as
\begin{equation}
  \indsupo{T} = \frac{1}{2}\idchan + \frac{1}{\sqrt{2}}\indsupo{S} - \frac{\sqrt{2}-1}{2}\indsupo{Z} 
\end{equation}
achieves the 1-norm of $\sqrt{2}$.
A more refined choice would be all stabilizer operators that require only a bounded number of ancillas $r$ in~\Cref{eq:stabops_standard_form}.
This is also a finite decomposition set, and in the limit of $r\rightarrow\infty$ the quasiprobability extent will achieve $\gamma_{\stabops}$.

To find lower bounds, we instead turn to outer approximations of $\stabops$.
A natural candidate in this regard is $\csp(A\rightarrow A')$, which is precisely the intersection of the polytope $\stab(AA')$ with the hyperplane describing the partial trace constraint.
It is itself a convex polytope and hence has finitely many extremal points.
In fact, in some cases this partial trace constraint can even be dropped, which reduces the problem to finding a QPD of the Choi state w.r.t. $\stab$.
\begin{proposition}{\cite[Theorem 3.2]{seddon2019_quantifying}}{}
  Let $A,A'$ be multi-qubit systems.
  If $\cE$ is a unitary channel corresponding to the third level of the Clifford hierarchy, then $\gamma_{\csp}(\cE)=\gamma_{\stab}(\choi{\cE})$.
\end{proposition}
It is a simple exercise to show that the Choi state of the T gate is equivalent to $\ket{0}\otimes\ket{H}$ up to Clifford gates.
From~\Cref{ex:T_magic_correspondence}, we can see that the $\csp$ and Choi-based lower bounds are therefore tight for the T gate
\begin{equation}
  \gamma_{\stabops}(T^{\otimes m}) \geq \gamma_{\csp}(T^{\otimes m}) = \gamma_{\stab}(\choi{T}^{\otimes m}) = \gamma_{\stab}(\proj{H}^{\otimes m}) = \gamma_{\stabops}(T^{\otimes m}) \, .
\end{equation}

\section{Further reading}
Historically, he resource theory of magic has was first developed for odd-dimensional quantum systems, since they exhibit well-behaved finite-dimensional analogs of the Wigner function~\cite{gross2006_hudson}.
For instance, Veitch~\etal~\cite{veitch2014_resource} introduced a magic monotone called \emph{mana}, which has a striking resemblance to the quasiprobability extent, as it is defined to be the logarithm of the 1-norm of the Wigner function.

The first instance of QPS for magic simulation (or as a matter of fact of QPS in general) dates back to Pashayan~\etal~\cite{pashayan2015_estimating}.
They frame their method in terms of quasiprobability representations w.r.t. some frame and its dual frame.
As such, this formulation lends itself very naturally to odd-dimensional qudit systems.
Later work by Howard~\etal~\cite{howard2017_application} extended QPS to qubit systems and studied QPDs of qubit magic states to be used in injection circuits.
They introduced the quasiprobability extent of states and characterized some of its properties, like the strict sub-multiplicativity.
Furthermore, they also used the quasiprobability extent to find bounds on optimal rates of synthesis for certain gates.
It should be noted that their definition of robustness slightly differs from the canonical notion introduced by Vidal and Tarrach~\cite{vidal1999_robustness}.
They denote by ``robustness of magic'' what we call the quasiprobability extent.

The first notion of QPS of magic channels was proposed shortly thereafter in \reference~\cite{bennink2017_unbiased}.
The authors apply their technique to the simulation of noisy quantum circuits, which allows them to simulate quantum error correction systems with non-Clifford errors --- a regime that is out of reach for the more conventional pure-state stabilizer rank simulators.
For the decomposition set, they consider the set of stabilizer-preserving operations.
The quasiprobability extent of channels was further developed in \reference~\cite{seddon2019_quantifying}.
The authors prove the Choi characterization of the set of completely stabilizer-preserving operation as well as various properties of the induced quasiprobability extent.

The difference between the completely stabilizer-preserving operations and and the more operational stabilizer operations was investigated in \reference~\cite{heimendahl2022_axiomatic}.
Interestingly, the authors prove an explicit separation between these two sets, and they provide a normal form of general stabilizer operations.

Due to the strictly sub-multiplicative behavior and the lack of analytic description for $\gamma_{\stab}$, a significant amount of research has gone into improving the efficiency of its numerical evaluation.
Notably, Heinrich \etal~demonstrated how the complexity of computing the quasiprobability extent can be reduced super-polynomially by exploiting symmetries of the stabilizer polytope and the state in question~\cite{heinrich2019_robustness}.
As such, they can exactly estimate $\gamma_{\stab}(\proj{H}^{\otimes m})$ for up to $m=9$.
Using approximations, they can estimate almost tight upper bounds of the quasiprobability extent for even larger $m$.
Further algorithmic improvements by Hamaguchi \etal~\cite{hamaguchi2024_handbook} allow the evaluation of the quasiprobability extent of arbitrary states for up to $8$ qubits.
We refer to~\cite[Table 1]{hamaguchi2024_handbook} for a nice succinct overview of the different approaches.


\section{Contributions}
This chapter mostly surveys recent developments in QPS for magic simulation from other authors in literature.
It is included it in this thesis for the sake of completeness.
The author's contributions are mainly focused on contextualizing these results within the QPS framework developed in previous chapters.
We also note that to our knowledge, classical side information from intermediate measurements has not been explored in the context of the QRT of magic before.

\chapter{Simulating noise-free computation}\label{chap:noisefree}
Experimental realizations of quantum computers suffer from a large amount of inherent noise, which makes their practical utilization difficult. 
In this chapter, we will study a method to remove the effect of noise by using QPS to simulate ideal noise-free quantum gates with a noisy quantum computer.
More concretely, we consider a decomposition set $\ds$ that contains a selection of noisy quantum operations that our computer can execute, such as noisy gates, noisy measurements and combinations thereof.
This noise suppression is relatively straightforward to implement in practical experiments, at least compared to full-blown quantum error correction (QEC).
It avoids the need for encoding the quantum information into a larger system as well as all the other difficulties that come with syndrome extraction and decoding.
However, it also entails a sampling overhead that scales exponentially in the circuit size.
This is perhaps not too surprising, given that the coherence time of the qubits is not physically prolonged.

As such, QPS of noise-free computation is an instance of quantum error mitigation (QEM).
In fact, it is one of the prototypical QEM methods and it commonly goes by the name \emph{probabilistic error cancellation}.
For near-term devices, probabilistic error cancellation holds the promise to systematically remove the effect of noise in a quantum computation without the need for any additional quantum resources.
It has been at the heart of some of the most impressive experimental quantum computing demonstrations that have come out in recent times~\cite{kim2023_evidence}.
Even as quantum hardware is poised to improve in the future, the overhead of error correction in early fault-tolerant computers is likely to make any practical utilization very difficult.
As such, some recent research has focused on finding ways to seamlessly combine QEC and QEM to split up the burden of suppressing noise and therefore interpolate between the advantages and disadvantages of the two approaches.
We will touch upon some of these ideas in this chapter.

Contrary to other applications of QPS that we explored in previous chapters, probabilistic error cancellation does unfortunately not exhibit a comparatively rich mathematical formalism, as the decomposition set is defined in a more ad-hoc fashion and a priori lacks any sort of structure.
Many of the challenges in probabilistic error cancellation are of more practical nature.
The goal of this chapter will be to briefly outline some of these difficulties as well as a powerful technique based on randomized compiling, which makes it possible to address these challenges while effectively enforcing some strict structure on the decomposition set.

\begin{remark}{}{}
  Since the decomposition set in probabilistic error cancellation is not derived from a QRT, it is a priori not clear that QPDs of smaller circuit elements can be combined into a QPD of the total circuit (recall the discussion in~\Cref{rem:combining_qpds}).
  Indeed, if an imperfect quantum computer can implement the noisy operations $\cE$ and $\cF$, it might not be able to implement $\cE\otimes\idchan$ or $\cE\circ\cF$ due to cross-talk and temporally correlated noise.
  This is an additional difficulty that practical QEM techniques need to take into account.
\end{remark}

\section{Basic setup of probabilistic error cancellation}
In the following, we describe a simple procedure how to simulate a unitary quantum channel $\cU$ using probabilistic error cancellation.
\begin{enumerate}
  \item Pick some finite set of (noisy) operations that the noisy quantum computer can natively execute.
  \item Characterize these noisy operations, for instance with a quantum process tomography protocol. The resulting quantum channels constitute the decomposition set $\ds$.
  \item Find the optimal QPD of $\cU$ into the decomposition set $\ds$ using the linear program introduced in~\Cref{sec:qps_conv_programming}.
  \item Perform QPS of $\cU$ w.r.t. the decomposition set $\ds$.
\end{enumerate}
There are some potential difficulties with this type of approach.
For instance, the choice of the operations for the decomposition set is generally unclear, as the hardware typically supports the execution of a continuum of gates and the choice of the decomposition set can strongly impact the resulting quasiprobability extent of $\cU$.
To achieve a small quasiprobability extent, $\ds$ should typically include the noisy realization $\cU_{\mathrm{noisy}}$ of the gate $\cU$. 
As such, we would need to pick a different decomposition set for every different type of gate.
This is not only cumbersome from a theoretical perspective, but more importantly also impractical, as the number of noisy operations that need to be characterized is very large.

As a solution to this problem, many protocols instead tackle the problem from a noise inversion perspective\footnote{This is essentially the same noise inversion procedure we encountered as an application of nonphysical computation in~\Cref{sec:nonphysical_applications}, but with a noisy decomposition set.}.
A noisy gate $\cU_{\mathrm{noisy}}$ can always be seen as a composition of the ideal gate with a noise channel $\cN\circ \cU$ where $\cN\coloneqq \cU_{\mathrm{noisy}}\circ\cU^{-1}$.
So instead of quasiprobabilistically simulating $\cU$, one can instead physically execute $\cU_{\mathrm{noisy}}$ and then immediately thereafter quasiprobabilistically simulate \footnote{Note that we assume the noise channel to be invertible here. The inverse is not required to be a physical CPTP map itself.} $\cN^{-1}$.
Since (at least for weak noise) $\cN^{-1}$ is close to the identity no matter the choice of gate $\cU$, it is much more reasonable to use the same decomposition set $\ds$ for every type of gate.
This approach is sometimes called the \emph{inverse method} to distinguish it from the \emph{compensation} method that we introduced initially.

\begin{example}
  As a simple toy model, consider a quantum computer that can only execute noisy single-qubit quantum gates, which consist of the ideal gate followed by the depolarizing channel
  \begin{equation}
    \cN_{\epsilon}(\rho) \coloneqq (1-{\epsilon})\rho + {\epsilon}\tr[\rho]\frac{\id}{2}
  \end{equation}
  for some noise rate ${\epsilon}$.
  For our decomposition set, we simply pick the four noisy single-qubit Pauli channels
  \begin{equation}
    \ds = \{\cN_{\epsilon}, \cN_{\epsilon}\circ\mathcal{X}, \cN_{\epsilon}\circ\mathcal{Y}, \cN_{\epsilon}\circ\mathcal{Z} \}
  \end{equation}
  where $\mathcal{X},\mathcal{Y},\mathcal{Z}$ are the unitary channels induced by $X,Y,Z$.
  Then the sampling overhead associated with inverting the noise of any quantum gate is given by
  \begin{equation}
    \gamma_{\ds}(\cN_{\epsilon}^{-1}) = \frac{1+\epsilon-\epsilon^2}{1 - 2\epsilon + 2\epsilon^2}
  \end{equation}
  which is achieved by the following QPD\footnote{One way to verify this is by multiplying both sides by $\mathcal{N}_{\epsilon}$ from the left and then applying~\Cref{ex:noise_inv_depol}.}
  \begin{equation}
    \cN_{\epsilon}^{-1} = \frac{4-\delta}{4(1-\delta)}\cN_{\epsilon} - \frac{\delta}{4(1-\delta)}\left(  \cN_{\epsilon}\circ\mathcal{X} + \cN_{\epsilon}\circ\mathcal{Y} + \cN_{\epsilon}\circ\mathcal{Z} \right)
  \end{equation}
  where $\delta\coloneqq 2\epsilon(1-\epsilon)$.

  This example nicely demonstrates an important property of probabilistic error cancellation: as the physical noise rate of the hardware improves, the sampling overhead decreases.
  Moreover, the quasiprobability extent converges to $1$ as the physical error rate vanishes.
  This is nicely illustrated by the first-order Taylor expansion for small $\epsilon$
\begin{equation}
  \gamma_{\ds}(\cN_{\epsilon}^{-1}) \approx 1 + 3\epsilon \, .
\end{equation}

  In order to achieve a crude theoretical handle on what sort of quasiprobability extents are achievable through the inverse method, we can look at the maximal possible extension of the decomposition set
  \begin{equation}
  \gamma_{\ds}(\cN^{-1}) \geq \gamma_{\cptp^{\star}}(\cN^{-1}) = \dnorm{\cN^{-1}} = \frac{1+\nicefrac{\epsilon}{2}}{1-\epsilon} \approx 1 + \frac{3}{2}\epsilon
  \end{equation}
  where the inequality follows from~\Cref{lem:gamma_ds_bound} and the equalities were shown in~\Cref{ex:noise_inv_depol}.
  As such, the first-order factor in the overhead differs by a factor of two from the theoretically optimal value.
\end{example}

\section{Tailoring noise with randomized compilation}\label{sec:randomized_compiling_pec}
There are a few technical difficulties which need to be addressed for an experimental realization of probabilistic error cancellation.
The most pressing one is that one requires exact knowledge of the underlying noise of the hardware.
This is in stark contrast to QEC, which is typically quite robust w.r.t. the underlying noise model.
For a simple back-of-the-envelope calculation, assume that we erroneously characterize the noisy operations $\cE_1,\dots,\cE_n$ and instead use the wrong decomposition set $\ds=\{\tilde{\cE}_1,\dots,\tilde{\cE}_n\}$ with  an error of about $\norm{\cE_i-\tilde{\cE}_i}\approx \epsilon$.
If we choose a QPD w.r.t. the erroneous basis $\cU=\sum_ia_i\tilde{\cE}_i$ and try to simulate $\cU$, we effectively cause an error
\begin{equation}
  \norm{\cU - \sum_i a_i\cE_i} = \norm{\sum_i a_i (\tilde{\cE}_i - \cE_i)} \lessapprox \norm{a}_1\epsilon \, .
\end{equation}
In summary, the quality of the simulation is at most as good as the characterization of the involved noisy operations.

The standard method for characterizing quantum channels is through the quantum process tomography protocol.
However, in practice it is inherently unreliable due to state preparation and measurement (SPAM) errors.
There exist modifications to process tomography that allow for the elimination of systematic errors caused by SPAM noise (see e.g. \emph{gate-set tomography}~\cite{greenbaum2015_introduction} and its utilization in QEM discussed by Endo~\etal~in \reference~\cite{endo2018_practical}), but they are very expensive and typically impractical.
To make matters worse, typical quantum devices suffer from cross-talk and correlated noise.
Modelling noisy gates by a simple quantum channel acting \emph{only} on the involved qubits is generally a bad approximation, and one needs to also account for the effect on surrounding qubits.
This further exacerbates the difficulties of tomography.

In this section, we briefly present an elegant solution to these problems due to van den Berg~\etal~\cite{vandenberg2023_pec}.
The foundational idea is based on the technique of \emph{randomized compilation}~\cite{wallman2016_noisetailoring} which enforces the noise in the circuit to effectively be Pauli channels.
We only present the main principles and refer the reader to \reference~\cite{wallman2016_noisetailoring,vandenberg2023_pec} for more details.

The main idea is to separate the gates into two categories: ``easy gates'' and ``hard gates'', which can correspondingly be implemented with more or less noise.
For the purposes of QEM, it is natural to choose the easy gates to be single-qubit gates and the hard gates to be $\cnot$ gates.
The circuit is modelled as alternating layers of easy and hard gates.
Randomized compiling effectively twirls the noise to enforce it to be of Pauli type.
This is achieved by pushing the Pauli twirl gates through the CNOT layer (which is Clifford) and integrating them into the previous and subsequent easy gate layers.
This noise tailoring is exact as long as the noise of the easy gates can be assumed to be gate-independent~\cite[Theorem 1]{wallman2016_noisetailoring}.

Having effectively tailored the noise to be purely of Pauli type, it remains to characterize the resulting Pauli channel.
Luckily, Pauli channels are much easier to estimate than general channels~\cite{flammia2020_efficient} and expensive error-prone process tomography can be avoided.
The main idea is to repeat the $\cnot$ gate layer multiple times and measure the corresponding Pauli expectation values at every depth.
The noise parameters can be fully extracted from the exponential decay rate, making them robust against any type of SPAM errors\footnote{This procedure is reminiscent of randomized benchmarking~\cite{emerson2005_scalable}.}.
Still, the number of parameters to learn is generally exponentially large in the number of qubits.
This can in term be avoided by assuming some sparsity structure in the noise.
Sparse Pauli-Lindblad models have proven very effective in the context of superconducting quantum processors~\cite{vandenberg2023_pec}.

Since the resulting noise channels $\cN$ are Pauli, $\cN^{-1}$ can itself be decomposed into a QPD of Pauli channels.
In fact, this QPD is even optimal in the sense $\gamma_{\mathcal{P}}(\cN^{-1}) = \gamma_{\cptp^{\star}}(\cN^{-1})$ where $\mathcal{P}$ denotes the set of Pauli channels (recall~\Cref{prop:gamma_cptp_paulisupo}).

QPS based on randomized compilation and sparse Pauli-Lindblad models has been at the heart of some of the most impressive experimental demonstrations of quantum computing in recent years~\cite{kim2023_evidence}.
Though it should be noted that in \reference~\cite{kim2023_evidence}, QPS is not used to remove the effect of noise, but rather to extrapolate its strength in order to combine the method with another well-known QEM technique called \emph{zero-noise noise extrapolation}~\cite{temme2017_errormitigation}.

\section{Combining QEM and QEC}\label{sec:qemqec}
Early implementations of QEC will likely be unable to implement any useful quantum algorithms due the restricted number of logical qubits and limited ability to sufficiently suppress the logical error rates.
Recent research is thus increasingly focused on finding meaningful ways to combine QEC and QEM to split the burden of suppressing the noise.
This prospect is very enticing, as one could interpolate between the drawbacks of QEC and QEM.
When the QEC capabilities of the hardware are more limited, QEM would do most of the heavy lifting, and once the hardware improves, QEC would become more dominant and thus reduce the sampling overhead.

The most straightforward way to achieve this goal is by simply performing error mitigation on the logical level, i.e., simply treat the logical qubits as if they were physical qubits and then run any QEM technique on top of that.
This is unlikely to be the most resource efficient approach, and more refined techniques should more specifically take into account and address the specific aspects of QEC that are especially difficult to realize in early fault-tolerant systems.
In this section, we will present precisely such a technique.

Logical Clifford gates are typically easier to realize fault-tolerantly than non-Clifford gates.
In fact, for some CSS codes they can be implemented transversally, i.e., by only coupling equivalent qubits across code blocks.
Transversal gates have a manageable overhead, because they do not significantly spread errors and are hence fault-tolerant by definition.
The Eastin-Knill theorem proves that no error correcting code can transversally implement a universal set of gates~\cite{eastin2009_restrictions}. 
To achieve universal computation, an additional fault-tolerant non-Clifford gate is thus needed.
The T gate is usually considered for this purpose.
A common technique to implement a T gate is through through a magic state injection, as previously discussed in~\Cref{fig:magic_state_injection}.
This approach reduces the problem of implementing a fault-tolerant T gate to the task of a fault-tolerant preparation of an encoded magic state.
The latter task can be achieved by magic state distillation~\cite{bravyi2005_universal}, where several noisy magic states are repeatedly transformed into fewer magic states of better fidelity.
While magic state distillation is a very elegant approach for achieving universal fault-tolerant quantum computing, the distillation process leads to a considerable overhead in practice.
It is meaningful to investigate the possibility of using error mitigation techniques to reduce this overhead.

In this section, we show two possible techniques how to perform probabilistic error cancellation on the logical level to simulate the execution of noise-free logical T gates using a noisy T gate, thus circumventing the need for expensive magic state distillation.
In both cases, the quasiprobability extent of the logical T gate will be $\gamma_{\epsilon}=1+2\kappa\epsilon+\mathcal{O}(\epsilon^2)$ where $\epsilon$ is the physical error rate and $\kappa$ is some constant.
The first method achieves a value of $\kappa\approx 2/5$, however it requires an additional logical ancilla qubit.
In contrast, the second method is easier to implement but incurs a larger overhead $\kappa\approx 30$.
When simulating a Clifford+T circuit with this technique, the sampling overhead scales as $\gamma_{\epsilon}^{2t}$ where $t$ is the number of T gates.
This overhead is much smaller than the best currently known classical simulation algorithm, which scales as $1.3161^t$~\cite{qassim2021_improved}, even for modestly low physical error rates $\epsilon$.

In the following, we present the two methods in detail.

\paragraph{Error-mitigated logical T gates via noisy magic states}
We consider an implementation of a T gate through magic state injection, but with a noisy magic state instead of a distilled noise-free one.
This will lead to a faulty T gate.
In the second step we use probabilistic error cancellation to suppress the error.
Let $\rho$ be the noisy encoded magic state.
We define the logical error rate as
\begin{align} \label{eq:eps_bar}
\bar\epsilon:=1-\bra{H}\rho \ket{H}   \, . 
\end{align}
Recall that $\ket{H}\coloneqq \left(\ket{0} + e^{i\pi / 4} \ket{1} \right)/\sqrt{2}$.
For simplicity, assume that the chosen error correction code enables a noiseless implementation of logical gates S, $X$, $\cnot$, preparation of logical $\ket{0}$ and $\ket{+}$ states, and measurements of a logical qubit in the Z and X basis~\footnote{The noise can be exponentially suppressed by choosing a large enough distance. These gates are Clifford gates that feature a transversal and hence fault-tolerant implementation for self-dual CSS codes.}.
We can twirl the noisy magic state $\rho$ with respect to the group $\{\id,A\}$, where $A$ is the Clifford gate defined as
\begin{align}
    A:=e^{-i\pi/4}SX = \proj{H}-\proj{\omega} \, ,
\end{align}
an $\ket{\omega}:=Z\ket{H}$ is a state orthogonal to $\ket{H}$.
The twirled state thus becomes
\begin{align} \label{eq:twirled_state}
    \tau:=\frac{1}{2}(\rho + A\rho A^\dagger) = (1-\bar\epsilon)\proj{H} + \bar\epsilon \proj{\omega} \, ,
\end{align}
where the second equality follows from the identity $A\ketbra{H}{\omega}A^{\dagger} + \ketbra{H}{\omega}=0$. 
Therefore, we can prepare the twirled state $\tau$ by first preparing $\rho$ and then applying the logical $A$ gate with probability $1/2$. The twirled state $\tau$ can be considered as the ideal magic state $\ket{H}$ that suffers from a random Pauli $Z$ error applied with probability $\bar \epsilon$, since $\ket{\omega}= Z\ket{H}$.
Note that a Pauli $Z$ error on the magic state $\ket{H}$ has no effect on the measurement in~\Cref{fig:magic_state_injection}.
Using the twirled state $\tau$ instead of the ideal magic state $\ket{H}$ inside the T gate gadget therefore leads to the noisy T gate
\begin{align}\label{eq:noisyT}
    \mathcal{T}_{\bar \epsilon} = (1-\bar \epsilon) \mathcal{T} + \bar \epsilon \, \mathcal{Z} \circ \mathcal{T} \, ,
\end{align}
where $\mathcal{T}$ and $\mathcal{Z}$ denote the quantum channels implementing the ideal T and Z gate, respectively, i.e., $\mathcal{T}(X)=TXT^\dagger$ and $\mathcal{Z}(X)=ZXZ^\dagger$.
To use error mitigation to simulate a perfect T gate using $\mathcal{T}_{\bar \epsilon}$, we need a good estimate for the logical error rate $\bar \epsilon$. This could be done via standard tomography techniques, however we present a more robust method further below.

It is crucial to understand how the logical error rate $\bar \epsilon$ defined in~\Cref{eq:eps_bar} depends on the physical error rate $\epsilon$.
It has been shown that there exist noisy state preparation circuits for the surface code that achieve logical error rates that are asymptotically constant w.r.t. to the code distance~\cite{horodecki2015_simple}.
In fact, by choosing the right protocol, $\bar\epsilon$ can be in the same order as the physical error rate $\epsilon$
\begin{align} \label{eq:logical_er}
    \bar \epsilon = \kappa \epsilon + \mathcal{O}(\epsilon^2) \, ,
\end{align}
where $\kappa$ is a conversion constant that depends on the details of the noise model.
For example, in the case of perfect single-qubit gates and no initialization errors, one can achieve $\kappa=2/5$~\cite{li2015_magicstate}.
The second order term in~\Cref{eq:logical_er} has numerically been shown to be negligible for realistic error rates using the protocol from \reference~\cite{li2015_magicstate}.



We now discuss how to experimentally determine the logical error rate $\bar\epsilon$. Denote the phase flip channel by $\cN_{\bar \epsilon}:= (1-\bar \epsilon) \idchan + \bar \epsilon \mathcal{Z}$.
Since $\mathcal{T}$ and $\mathcal{Z}$ commute, we have for any $p \in \mathbb{N}$
\begin{equation} \label{eq_commute}
    \mathcal{T}_{\bar \epsilon}^p = (\cN_{\bar \epsilon} \circ \mathcal{T})^p = \cN_{\bar \epsilon}^p \circ \mathcal{T}^p\, .
\end{equation}
For $p\equiv 0 \,\, (\!\!\!\!\mod 8)$ we have $T^p=\mathds{1}$ and $\mathcal{T}^p=\idchan$.
Hence, \Cref{eq_commute} ensures that $\mathcal{T}_{\bar \epsilon}^p=\cN_{\bar \epsilon}^p$.
We can now learn $\bar \epsilon$ by preparing a logical state $\ket{+}$ fault-tolerantly, applying $\mathcal{T}^p_{\bar \epsilon}$ and measuring the output in the $\{\ket{+},\ket{-}\}$ basis, which gives us an outcome $1$ or $0$.
Hence we can estimate the expectation value of the outcome 
\begin{align} \label{eq_tomography}
    f(p)
    := \bra{+} \mathcal{T}^p_{\bar \epsilon}(\proj{+}) \ket{+}
    =\frac{1}{2}\Big(1+(1-2\bar\epsilon)^{p}\Big) \, ,
\end{align}
with an approximation error that decreases in the number of circuit shots~\footnote{The final step in~\Cref{eq_tomography} can be checked with a simple recursive argument.}.
By measuring $f(p)$ for $p=8k$ with $k \in \mathbb{N}$, we can obtain a good estimate for $\bar \epsilon$ using exponential fitting.

The following QPD optimally decomposes the ideal T gate into its noisy variant with and without a subsequent Z gate
\begin{align}
    \mathcal{T} = \left( \frac{1-\bar \epsilon}{1-2\bar \epsilon} \right) \mathcal{T}_{\bar \epsilon} - \left( \frac{\bar \epsilon}{1-2\bar \epsilon} \right) \mathcal{Z} \circ \mathcal{T}_{\bar \epsilon} \, ,
\end{align}
with a 1-norm of $\gamma_{\epsilon} = 1/(1-2\bar \epsilon)$.
Expanding this term around $\bar \epsilon = 0$ gives 
\begin{align} \label{eq_gamma_approx}
    \gamma_{\epsilon}= 1+2\bar \epsilon + \mathcal{O}(\bar \epsilon^2) = 1 + 2\kappa \epsilon + \mathcal{O}(\epsilon^2) \, .
\end{align}

\paragraph{Error-mitigated logical T gates via code switching}
The implementation of a T gate via magic states has the drawback of using two encoded qubits.
It also includes a logical CNOT gate which may require more logical qubits for certain implementations~\cite{horsman2012_surface}.
One may ask if a logical T gate can be implemented using only one logical qubit.
In this section we show how to accomplish this task by combining probabilistic error cancellation with the code switching method pioneered by Paetznick and Reichardt~\cite{paetznick2013_universal}.

Suppose $\mathcal{S}_1$ and $\mathcal{S}_2$ are CSS-type~\cite{calderbank1996_good,steane1996_multiple} quantum codes with one logical qubit.
Our goal is to implement a logical T gate on a qubit encoded by $\mathcal{S}_1$.
We assume that $\mathcal{S}_1$ has a large distance for both $X$ and $Z$ errors and enables a fault-tolerant noise-free implementation of the S gate.
We assume that $\mathcal{S}_2$ is an asymmetric code with a large distance for $X$ errors and a distance $1$ for $Z$ errors such that a logical-$Z$ operator of $\mathcal{S}_2$ can be chosen as a single-qubit Pauli $Z$ on some physical qubit.
We shall denote this local logical-$Z$ operator as $\overline{Z}_{\mathrm{loc}}$.
This allows us to perform a logical T gate on a qubit encoded by $\mathcal{S}_2$ simply by applying a physical T gate on the qubit acted upon by $\overline{Z}_{\mathrm{loc}}$.
Indeed, since the T gate is a linear combination of the identity and the Pauli $Z$, the physical and the logical T gates become equivalent.
The code switching method enables a conversion between the codes $\mathcal{S}_1$ and $\mathcal{S}_2$ by measuring stabilizers of $\mathcal{S}_2$ on a logical state encoded by $\mathcal{S}_1$ or vice versa.
In certain cases this conversion can be performed fault-tolerantly~\cite{paetznick2013_universal}.
This is achieved by identifying stabilizers present in both codes and using the measured syndromes of such stabilizers to diagnose and correct errors. 
The error suppression  typically scales exponentially with  the code distance.
In our case the code switching protects the encoded qubit exponentially well only from logical $X$ errors.
This ensures that the effective noise channel acting on the logical qubit is dominated by $Z$-type (possibly coherent) errors.
Let $\mathcal{T}_n=\mathcal{N}\circ\mathcal{T}$ be the noisy logical T gate.
We can tailor the noise $\mathcal{N}$ to become stochastic Pauli-Z noise using by appropriate twirling
\begin{align}
  \frac{1}{2}(\mathcal{N} + \mathcal{X}\circ\mathcal{N}\circ\mathcal{X}) = (1-\bar\epsilon)\idchan + \bar\epsilon \mathcal{Z} =:  \mathcal{N}_{\bar \epsilon} \, ,
\end{align}
for some logical error rate $\bar \epsilon$ where $\mathcal{X}(B)=XBX^\dagger$.
Using the identity $XT=TA$ we get
\begin{align}\label{eq:deformation_twirling}
    \frac{1}{2}(\mathcal{T}_n + \mathcal{X}\circ\mathcal{T}_n\circ\cA) =
    \mathcal{N}_{\bar \epsilon}\circ\mathcal{T} =\mathcal{T}_{\bar \epsilon} \, ,
\end{align}
where $\cA(B)=ABA^\dagger$ and $\mathcal{T}_{\bar \epsilon}$ is defined as in~\Cref{eq:noisyT}.
By assumption, the code $\mathcal{S}_1$ enables a fault-tolerant implementation of the Pauli $X$ and the $S$ gate on the encoded qubit.
Since $A=SX$ up to an overall phase, the twirling in~\Cref{eq:deformation_twirling} can be implemented fault-tolerantly on a qubit encoded by $\mathcal{S}_1$ by applying either $\mathcal{T}_n$ or $\mathcal{X}\circ\mathcal{T}_n\circ\cA$ with the probability $1/2$ each. 
Since at this point we have implemented a noisy T gate with Pauli-Z noise, we can follow the same procedure as for the first method in order to determine the parameter $\bar \epsilon$ efficiently and perform probabilistic error cancellation to mitigate the noise.

In~\Cref{app:codeswitch} we specialize the code switching method to Kitaev's surface code.
It is shown that the logical error rate scales as $\bar \epsilon \approx 30 \epsilon$ for realistic error models. 
The code switching procedure is depicted in~\Cref{fig:kitaev_codeswitch}.

\begin{figure}
  \centerline{
  \includegraphics[width=0.90\textwidth]{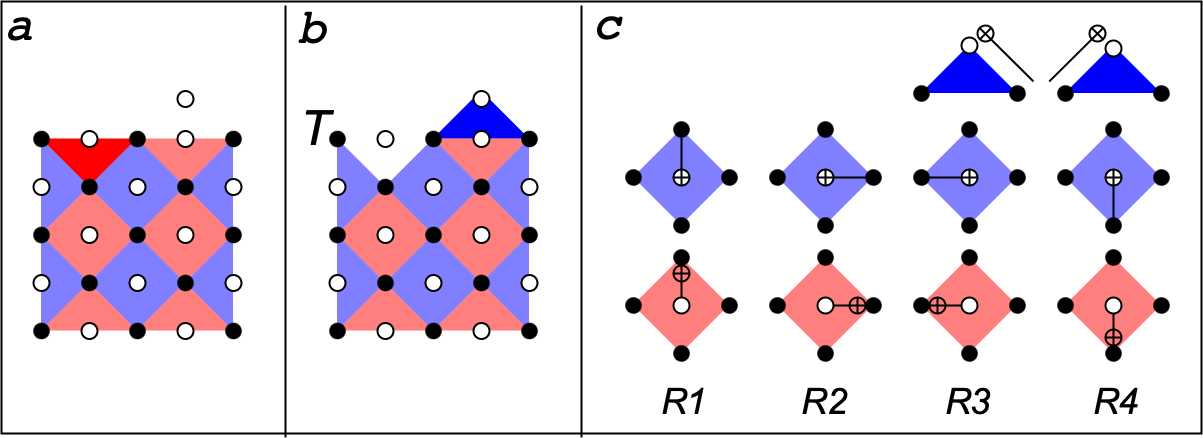}}
  \caption{Logical $T$-gate by code switching.
  (a) Surface code $\mathcal{S}_1$ with the distance $d=3$.
  Black and white circles indicate data and syndrome
  qubits. Red and blue faces indicate $X$ and $Z$ stabilizers.
  (b) Asymmetric surface code $\mathcal{S}_2$ exhibiting a single-qubit
  logical $T$-gate at the north-west corner of the lattice.
  Code switching $\mathcal{S}_1\to \mathcal{S}_2$ is performed by  
  turning off $X$ stabilizer $F_1$ (dark red triangle)
  and turning on $Z$ stabilizer $G_1$ (dark blue triangle).
  (c) The syndrome extraction cycle consists of four rounds of CNOTs R1, R2, R3, R4. To avoid clutter we only show a local schedule of CNOTs for each stabilizer.
  Schedules for weight-3 stabilizers are properly
  truncated. 
  The schedule extends to the full lattice
  in a translation invariant fashion.}
  \label{fig:kitaev_codeswitch}
\end{figure}

\section{Further reading}
Probabilistic error cancellation was originally introduced by Temme, Bravyi and Gambetta~\cite{temme2017_errormitigation} together with another important QEM technique, \emph{zero noise extrapolation}.
Their approach is formulated in terms of the compensation method, and they also formulated the linear program in~\Cref{eq:gamma_lp} expressing the optimal quasiprobability extent w.r.t. of a finite decomposition set.
Shortly thereafter, \reference~\cite{endo2018_practical} proposed several improvements to the technique, including the inverse method, an explicit decomposition set of single-qubit Clifford operations that spans the space of Hermitian-preserving superoperators as well as the usage of gate-set tomography to address the SPAM problem.

To minimize the sampling overhead of quantum error mitigation, it is of high interest to choose the best possible decomposition basis as to minimize the quasiprobability extent.
This is generally very difficult, unless one uses the randomized compiling approach from~\Cref{sec:randomized_compiling_pec}.
In \reference~\cite{piveteau2022_quasiprobability}, a heuristic method is proposed to iteratively construct the decomposition set in a greedy fashion.
Importantly, no assumptions on the hardware noise is required, since it is considered in each iteration through an instance of gate tomography.

A different approach in \reference~\cite{takagi2021_optimal} entirely avoids the choice of a finite decomposition set by modelling all gates to have identical noise $\cN$.
As such, the authors model their decomposition set as the union of all noisy gates $\{\cN\circ\indsupo{U} | U\in\uni(A) \}$ together with (noise-free) state preparation channels.
While this model is likely a strong oversimplification, it makes the computation of the quasiprobability extent more tractable and allowed the authors to find bounds or even closed-form formulas in some instances.

PEC based on randomized compilation was first proposed in \reference~\cite{vandenberg2023_pec}.
To characterize the twirled Pauli noise, the authors utilize a sparse Pauli-Lindblad model that can capture correlated noise across neighboring qubits on a superconducting quantum processor.
With an extensive amount of theoretical work, the authors show how this characterization can be done efficiently and reliably.
The methods are based on a broader line of research that shows that characterizing sparse Pauli noise can be done efficiently~\cite{flammia2020_efficient}.


The idea of integrating QEM with QEC was first formulated in \reference~\cite{suzuki2022_qem} and in \reference~\cite{piveteau2021_errormitigation,lostanglio2021_error}.
The former considers the possibility of using probabilistic error cancellation on the logical level to correct errors that either stem from erroneous decoding or from approximation errors in the Solovay-Kitaev decomposition.
The latter independently developed similar ideas, which broadly correspond to the discussion in~\Cref{sec:qemqec} on mitigating noisy non-transversal gates.
Many subsequent works were dedicated to further study the interplay between QEC and QEM~\cite{wahl2023_zne,gonzales2024_ftqem,dutkiewicz2024_error,tsubouchi2024_symmetric}.

Finally, we briefly mention a few other QEM techniques for the interested reader.
\emph{Zero-noise extrapolation} was introduced by \reference~\cite{temme2017_errormitigation,li2017_efficient} and tries to remove the bias from expectation values by extrapolating the circuit result while varying the noise strength.
A class of techniques called \emph{symmetry verification} uses post-selection to verify that the quantum state fulfills certain symmetries that should be conserved in the absence of noise~\cite{bonetmonroig2018_lowcost,mcardle2019_errormitigated}.
Another approach called \emph{subspace expansion} can mitigate noise in the task of eigenvalue estimation by classically solving the eigenvalue problem on a small subspace that is spanned by noisy states prepared by the hardware~\cite{mcclean2017_hybrid}.
We refer to the review in \reference~\cite{cai2023_qem} for an extensive overview of the state of quantum error mitigation.

\section{Contributions}
Most of this section serves as an overview of probabilistic error cancellation.
The method presented in~\Cref{sec:qemqec} is due to the author and was originally published in \reference~\cite{piveteau2021_errormitigation}.
We briefly note that the author also developed a heuristic algorithm to find near-optimal decomposition sets for probabilistic error cancellation using the compensation method, which was published in \reference~\cite{piveteau2022_quasiprobability}.
We chose not to include this work in this thesis.

\chapter{Conclusion}\label{chap:conclusion}
In this thesis, we have explored the application of QPS to four different settings, each corresponding to a different choice of decomposition set.
Looking back at these results, we revisit some of the important goals and questions that we originally stated in the introduction.

\paragraph{Utility of classical side information}
Since any linear combination of CPTP maps is itself proportional to an HPTP map, it is absolutely necessary to use classical side information for the simulation of any non-TP map.
However, we saw only very few practical examples where such non-TP maps actually occur, as even non-physical maps of interest are usually HPTP.
One exception is the imaginary time evolution discussed in~\Cref{sec:nonphysical_applications}.

For the simulation of CPTP maps, the utility of classical side information becomes a bit more murky and strongly depends on the underlying QRT.
For space-like and time-like cuts without classical communication, we saw that classical side information is absolutely crucial, as nontrivial channels cannot even be represented as a QPD without it (see~\Cref{lem:span_lo,lem:span_mpc_ebc}).
However, for many other decomposition sets we showed that side information is useless as the quasiprobability extent of a CPTP map $\cE$ remains unchanged without it, i.e., $\gamma_{\cptp}(\cE)=\gamma_{\cptp^{\star}}(\cE)$ (\Cref{prop:utility_cptp_interm_mmts}), $\gamma_{\sepc}(\cE)=\gamma_{\sepc^{\star}}(\cE)$, $\gamma_{\pptc}(\cE)=\gamma_{\pptc^{\star}}(\cE)$ (\Cref{prop:negtrick_sepc_pptc}) and $\gamma_{\ebc}(\cE)=\gamma_{\ebc^{\star}}(\cE)$ (\Cref{prop:negtrick_ebc}).
For space-like cuts with classical communication and for Clifford operations, it remains unknown whether $\gamma_{\locc^{\star}}(\cE)$ and $\gamma_{\stabops^{\star}}(\cE)$ can be strictly smaller than $\gamma_{\locc}(\cE)$ and $\gamma_{\stabops}(\cE)$.
We highlight that currently known constructions for QPDs that achieve $\gamma_{\locc^{\star}}(\cU)$ for non-Clifford unitary channels $\cU$ heavily utilize classical side information, hinting at a possible separation $\gamma_{\locc}(\cU)\neq \gamma_{\locc^{\star}}(\cU)$.

Remarkably, we do not know of any example where side information can reduce the quasiprobability extent, but is not absolutely required to find a valid QPD in the first place.

\paragraph{Submultiplicative behavior of quasiprobability extent}
For the resource theories of entanglement and magic, we saw that the associated quasiprobability extents exhibit a strictly sub-multiplicative behavior for many states and channels of interest.
Remarkably, in the QRT of entanglement we were able to evaluate $\gammareg_{\sep}$ and $\gammasreg_{\sep}$ for all pure states and show that they strictly differ in almost all instances.
A similar feat seems out of reach for the QRT of magic, as even for the simplest magic state $\ket{H}$ there is no known closed-form expression of the regularized quasiprobability extent.

In stark contrast, the quasiprobability extent $\gamma_{\cptp^{\star}}$ is multiplicative and hence $\gamma_{\cptp^{\star}}=\gammareg_{\cptp^{\star}}$.
This result also translates to some instances of probabilistic error cancellation, e.g., when using randomized compilation for noise tailoring as discussed in~\Cref{sec:randomized_compiling_pec}.

\paragraph{State of circuit knitting theory}
With the results presented in this thesis, we arguably have a mostly complete understanding of the practically relevant aspects of QPS-based circuit knitting.
Concretely, we know the optimal overheads for cutting bipartite states, two-qubit unitaries and wire cuts, as well as explicit protocols that achieve these overheads.
There is little room left to improve the sampling overhead in practical scenarios, unless a different circuit knitting paradigm emerges that operates in a different manner.

\paragraph{Practical utility of QPS}
Both circuit knitting and error mitigation have a clear practical appeal as they provide a potential avenue to overcome some of the stringent limitations of early quantum devices.
Correspondingly, they have received a lot of attention in the last few years and have been experimentally demonstrated multiple times.
However, we stress that it remains unclear if they can be used to demonstrate a useful quantum advantage on near-term devices, since the exponential sampling overhead is very steep and kicks in rather quickly.
Especially for circuit knitting, it remains a crucial open problem to identify suitable quantum computations which are ``nearly local'' with only a small amount of nonlocality across the partitions.
The clustered Hamiltonian simulation scheme by Harrow and Lowe~\cite{harrow2024_optimal}, briefly discussed in~\Cref{sec:gate_cut_app}, likely presents the most promising candidate in this regard.
At least for QEM, there exists some experimental evidence that it enables current devices to start solving non-trivial problems at useful scales~\cite{kim2023_evidence}, though a convincing quantum advantage has not been demonstrated to date.

Due to the steep exponential sampling overhead, it seems plausible that error mitigation and circuit knitting could find practical utility in combination with more traditional approaches that scale better asymptotically.
For example, we explored in \Cref{sec:qemqec} the possibility to combine error correction with QPS to simulate ideal logical T gates with noisy logical T gate.
This avoids costly magic state distillation and comes with a sampling overhead $\mathcal{O}(\gamma_{\epsilon}^{2t})$ exponential in the number of $t$ gates where the quasiprobability extent $\gamma_{\epsilon}$ converges to $1$ as the physical error rate vanishes.
To further reduce the sampling overhead, one could perform a small number of rounds of magic state distillation (insufficient for the desired logical accuracy) and then use QEM to mitigate the remaining noise of the intermediate magic state\footnote{This combination of QEM with conventional magic state distillation only works with the noisy magic state approach, but not with the code switching method.}.

Similar approaches could also be envisaged for circuit knitting: Instead of simulating non-local computation purely with local operations, one could instead simulate it with local operations and noisy nonlocal resources.
As the fidelity of this nonlocal resource improves, the resulting sampling overhead would decrease.
Some recent work has explored ideas in this direction~\cite{bechtold2024_jointwirecutting}.

The idea of using QPS to simulate nonphysical operations has only been sparsely explored to date.
In \Cref{sec:nonphysical_applications}, we briefly discussed some basic applications like noise inversion and entanglement detection, but further research is necessary.
We highlight the simulation of imaginary time evolution as a particularly interesting idea in this direction.

Finally, we briefly re-iterate that QPS for magic simulation only has limited practical utility, as there exist significantly more efficient classical algorithms to simulate near-Clifford circuits.
One advantage of QPS-based classical simulators is that they naturally allow for the inclusion of non-unitary channels, such as noise.
This can be particularly useful for the simulation of quantum error correction protocols under weak non-Pauli noise.

\paragraph{Utility of classical communication for circuit knitting}
Classical communication seems to have only limited practical utility for space-like cuts.
First, we saw that for states, $\gamma_{\locc^{\star}}$ and $\gamma_{\lo^{\star}}$ both reduced to $\gamma_{\sep}$ (\Cref{lem:gammas_overview_nonlocal_state}).
Furthermore, we saw that for two-qubit gates, $\gamma_{\locc^{\star}}$ and $\gamma_{\lo^{\star}}$ were also identical (\Cref{thm:gamma_kaklike_unitary}).
This equality even remains when considering tensor products of multiple two-qubit unitaries and also in the black box setting (\Cref{thm:gamma_bb_twoqubit}).
It remains an open question if a gap between the two arises for general unitaries.

Quite surprisingly, this picture is completely different for time-like cuts.
Cutting a single wire is already cheaper with classical communication as $\gamma_{\ebc^{\star}}(\idchan)=3$ (\Cref{thm:gamma_ebc_id}) and $\gamma_{\mpc^{\star}}(\idchan)=4$ (\Cref{thm:gamma_mpc_id}).
Asymptotically, this gap widens even more as $\gammareg_{\ebc^{\star}}(\idchan)=2$, but $\gammareg_{\mpc^{\star}}(\idchan)=4$.

\paragraph{Similar techniques across different QRTs}
We briefly review various broad techniques that have found applications in multiple different instances of QPS.

Generally, expressing quasiprobability extents as convex programs proved to be very fruitful, both for numerical evaluation and also to analytically prove certain properties.
In some cases, the desired quasiprobability extent can directly be expressed as a linear or semidefinite program (e.g. $\gamma_{\cptp^{\star}}$ for nonphysical simulation, $\gamma_{\stab}$ for magic states and also in the setting of probabilistic cancellation).
Whenever this is not possible, we instead reverted to finding outer approximations of the decomposition set, for which the inclusion does exhibit a semidefinite expression (e.g. $\ppt$ for $\sep$, $\pptc$ for $\locc$ and $\csp$ for $\stabops$).

Generally speaking, evaluating the quasiprobability extent tends to be easier for states than for channels.
Interestingly, we observed that the completely resource non-generating channels of our QRTs were characterized as having free Choi states, e.g., SEPC channels have separable Choi states (\Cref{sec:spacelike_ds}) and CSP channels have stabilizer Choi states (\Cref{prop:csp_characterization}).
Hence, the relaxation to this larger set (e.g. $\sepc$ instead of $\locc$ or $\csp$ instead of $\stabops$) simplifies the problem of finding the quasiprobability extent of a channel as it is now almost the same task as finding the extent of a state.

Another useful insight is that for any channel $\cE$ with a diagonal Pauli transfer matrix, the quasiprobability extent w.r.t. to the ``largest'' decomposition set $\gamma_{\cptp^{\star}}$ can also be achieved by restricting the decomposition set to the Pauli channels (recall \Cref{prop:gamma_cptp_paulisupo}).
Since the Pauli channels are included in virtually every decomposition set we considered (they are local Clifford operations and typically also available in probabilistic error cancellation), the associated quasiprobability extent of $\cE$ is hence also directly known.

\section{Open questions}
We summarize various open questions that we have left unanswered and that could be explored in future work.

\paragraph{General}
\begin{itemize}[noitemsep]
  \item
  We explored QPS applied to the QRTs of entanglement, magic and the ``maximal'' QRT.
  What other resource theories can QPS be applied to in a practically meaningful manner?
  A natural candidate would for example be the resource theory of non-Gaussianity for continuous-variable quantum systems~\cite{weedbrook2012_gaussian}.
  This QRT mirrors the resource theory of magic, as an analog of the Gottesman-Knill theorem exists for Gaussian information processing~\cite{bartlett2002_efficient}.
  It would be particularly interesting to investigate the implications of switching from a finite-dimensional to an infinite-dimensional Hilbert space.

  \item
  It would be interesting to further investigate the properties of the smooth quasiprobability extent.
  Can any meaningful statements be made for a finite decomposition set (respectively a polytope)?
  How do the precise details of the smoothing matter?
\end{itemize}

\paragraph{Nonphysical simulation}
\begin{itemize}[noitemsep]
  \item What interesting algorithmic applications (such as the imaginary time evolution in~\Cref{sec:nonphysical_applications}) can be achieved using QPS of nonphysical operations?
  \item How well can QPS emulate a universal quantum cloner? Does it provide an advantage over approximate CPTP universal quantum cloners?
  \item How does the overhead of entanglement detection with QPS compare to other approaches?
\end{itemize}

\paragraph{Nonlocal simulation}
\begin{itemize}[noitemsep]
  \item What are suitable computational tasks for which circuit knitting can be used?

  \item Does there exist a bipartite CPTP map $\cE\in\cptp(A;B\rightarrow A';B')$ for which $\gamma_{\locc^{\star}}$ is strictly smaller than $\gamma_{\locc}$?

  \item 
  Is there a simple closed-form expression for the quasiprobability extent of arbitrary unitary channels $\cU\in\cptp(AB)$?
  Does there exist a bipartite unitary channel $\cU$ for which classical communication reduces the overhead of space-like cutting, i.e., $\gamma_{\locc^{\star}}(\cU)\lneqq \gamma_{\lo^{\star}}(\cU)$?
  This question is difficult to address with current techniques, as essentially all of our proposed lower bounds are valid for both $\gamma_{\lo^{\star}}$ and $\gamma_{\locc^{\star}}$, making them useless to demonstrate a separation between the two.

  \item Is there an adapted de Finetti theorem which can prove the convergence speed of the SDP hierarchy $\gamma_{\sepc_k}$ introduced in~\Cref{sec:sdp_hierarchy_channels}?

  \item
  Our SDP hierarchies for $\gamma_{\sep}$ and $\gamma_{\sepc}$ discussed in \Cref{sec:gamma_sep_singleshot,sec:sdp_hierarchy_channels} are based on the original version of the Doherty-Parrilo-Spedalieri (DPS) hierarchy introduced in \reference~\cite{doherty2004_complete}.
  There exist improved versions of the DPS hierarchy that exhibit better convergence behavior, see for example \reference~\cite{harrow2017_improved}.
  It would be interesting to adapt our hierarchies for the quasiprobability extent with these improvements.

  \item Can a separation between the regularized and asymptotic quasiprobability extent be proven for arbitrary two-qubit gates, in analogy to the result for states in~\Cref{cor:gamma_reg_pure,cor:gamma_sep_asymptotic}?

  \item Can the exact PPT entanglement cost of a channel $E_{C,\pptsc}^{\mathrm{exact}}(\cE)$ be efficiently evaluated analogously to the state counterpart?

  \item Can the finite-shot smoothed quasiprobability extent of a channel w.r.t. PPTC be characterized by the smoothed max-logarithmic negativity, to lift~\Cref{lem:oneshot_pptsc_cost} to a statement analogous to~\Cref{lem:oneshot_pptc_cost}?
  This would allow us to relate the asymptotic extent $\gammasreg_{\pptc}(\cE)$ with the PPTSC entanglement cost $E_{C,\pptsc}(\cE)$.

  \item
  Can we find bounds on the quasiprobability extent of the quantum Fourier transform that improve upon the proof of \Cref{prop:gamma_qft}?

  \item Space-like cuts can be straightforwardly generalized to the $n$-partite settings for $n>2$.
  Can the quasiprobability extent of pure states and unitary channels be evaluated in this more general setting?

\end{itemize}

\paragraph{Magic simulation}
\begin{itemize}[noitemsep]
  \item Does approximate quasiprobability simulation with a vanishing error have an asymptotically lower overhead?
  More concretely, is $\gammasreg_{\stab}(\proj{H})$ strictly smaller than $\gammareg_{\stab}(\proj{H})$?
  \item Is there any utility in including classical side information in the QRT of magic, i.e., can $\gamma_{\stabops^{\star}}$ be strictly smaller than $\gamma_{\stabops}$ for certain channels?
\end{itemize}

\paragraph{Noise-free simulation}
\begin{itemize}[noitemsep]
  \item What is the most resource efficient way to combine QEC with QEM?
  \item Can classical side information be meaningfully used in some instance of probabilistic error cancellation?
\end{itemize}

\cleardoublepage
\phantomsection
\addcontentsline{toc}{chapter}{Bibliography}
\printbibliography

\appendix
\chapter{Hoeffding's inequality}\label{app:hoeffding}
\begin{theorem}{Hoeffding's inequality}{hoeffding}
  Let $X_1,\dots X_n$ be i.i.d random variables with $|X_i|\leq K$.
  Define $S_n\coloneqq \sum_{i=1}^n X_i$.
  Then for any $t>0$
  \begin{equation}
    \mathbb{P}\left[\abs{S_n-\mathbb{E}[S_n]}\geq t\right] \leq 2 \exp\left(-\frac{t^2}{2nK^2}\right) \, .
  \end{equation}
\end{theorem}

It will be informative for us to look at a simple example in which Hoeffding's bound is almost achieved, up to some constant factors.
\begin{lemma}{}{hoeffding_tightness}
  Consider i.i.d. random variables $X_1,\dots,X_n$ that take the values $\pm K$ with probability $1/2$ each and let $S_n\coloneqq \sum_{i=1}^n X_i$.
  Then for any $t\in [0,nK/4]$
  \begin{equation}
    \mathbb{P}\left[\abs{S_n} \geq t\right] \geq \frac{2}{15}\exp\left(-\frac{4t^2}{nK^2}\right) \, .
  \end{equation}
\end{lemma}
\begin{proof}
  We can write $X_i = 2K\cdot Y_i - K$ where $Y_i$ are i.i.d. random variables taking values $0$ and $1$.
  Clearly, $\sum_i Y_i$ follows a Binomial distribution on which we can use standard tail bounds, such as~\cite[Proposition 7.3.2]{matouek2009_the}
  \begin{equation}
    \mathbb{P}[\sum_i Y_i\geq \frac{n}{2}+t'] \geq \frac{1}{15}e^{-16t'^2 / n} .
  \end{equation}
  for $t'\in [0,n/8]$.
  The desired statement follows directly from
  \begin{align}
    \mathbb{P}\left[\abs{\sum_i X_i} \geq t\right] 
    &=\mathbb{P}\left[\sum_i X_i \geq t\right] + \mathbb{P}\left[\sum_i X_i < -t\right] \\
    &=2\mathbb{P}\left[\sum_i X_i \geq t\right] \\
    &=2\mathbb{P}\left[\sum_i Y_i - \frac{n}{2} \geq \frac{t}{2K}\right]
  \end{align}
  where we used the symmetry of the distribution $\sum_i X_i$.
\end{proof}

\chapter{Circuits for QPDs of unitaries}\label{app:optimal_circuits}
In this appendix, we provide an explicit circuit implementation for the QPD of a gate with unitary operator Schmidt decomposition that was constructed in~\Cref{lem:gamma_unitary_upperbound}.
Let $U\in\uni(AB)$ be a bipartite unitary with operator Schmidt decomposition
\begin{equation}
  U = \sum_j u_j L_j \otimes R_j
\end{equation}
for some unitaries $L_j\in\uni(A)$ and $R_j\in\uni(B)$.
Adopting the notation in the proof of~\Cref{lem:gamma_unitary_upperbound}, the QPS of $U$ involves three types of operations.
\begin{itemize}
  \item $\mathcal{L}_j\otimes\mathcal{R}_j$ with quasiprobability coefficient $u_j^2$.
  \item $\mathcal{A}_{j,k}^0\otimes\mathcal{B}_{j,k}^{0}$ with quasiprobability coefficient $2u_ju_k$.
  \item $\mathcal{A}_{j,k}^{\pi/2}\otimes\mathcal{B}_{j,k}^{\pi/2}$ with quasiprobability coefficient $-2u_ju_k$.
\end{itemize}
The circuit implementation of the first one is trivial.
In~\Cref{fig:nonlocal_kaklike_decomp_circuit}, we depict a circuit implementation of $\mathcal{A}_{j,k}^{\phi}\otimes\mathcal{B}_{j,k}^{\phi}$ which covers both the second and third types by appropriately choosing $\phi$.
Most importantly, this circuit implementation only requires one additional ancilla on each bipartition.

\begin{figure}
  \centering
  \includegraphics{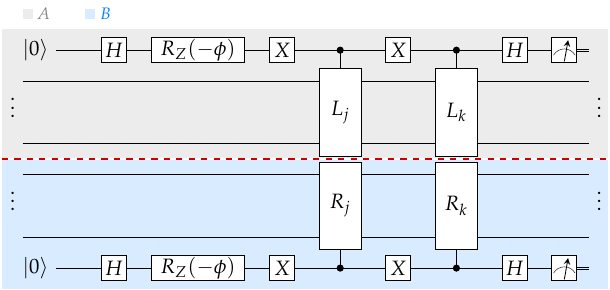}
  \caption{Circuit implementation for the simulation the operation $\mathcal{A}_{j,k}^{\phi}\otimes\mathcal{B}_{j,k}^{\phi}$ used in the QPD constructed in~\Cref{lem:gamma_unitary_upperbound}.}
  \label{fig:nonlocal_kaklike_decomp_circuit}
\end{figure}

The two ancilla qubits are each is prepared in the $\left(\ket{0} + e^{-i\phi}\ket{1}\right)/\sqrt{2}$ state.
Since the two-circuit halves do not interact, it suffices to show that the upper half acts correctly --- the argument for the lower half is analogous.
Consider an input state $\ket{\psi}$ on the bipartition $A$.
Controlled on the ancilla qubit, we either apply the gate $L_j$ or $L_k$, thus resulting into the state
\begin{equation}
  \frac{1}{\sqrt{2}}\left( \ket{0}\otimes L_j\ket{\psi} + e^{-i\phi}\ket{1}\otimes L_k\ket{\psi} \right) \, .
\end{equation}
Next, we apply a Hadamard gate on the ancilla qubit, yielding
\begin{equation}
  \frac{1}{2}\left( \ket{0}\otimes (L_j+e^{-i\phi}L_k)\ket{\psi} + \ket{1}\otimes (L_j-e^{-i\phi}L_k)\ket{\psi} \right) \, .
\end{equation}
If we next measure the ancilla qubit in the computation basis, the total process is thus described by the quantum instrument $(\mathcal{A}_{j,k}^{+,\phi},\mathcal{A}_{j,k}^{-,\phi})$.
In the QPS simulation procedure, we will thus need to weight the final measurement outcome by $+1$ or $-1$ depending on the outcome of the ancilla measurement.

\chapter{Code switching on Kitaev's surface code}\label{app:codeswitch}
In~\Cref{sec:qemqec} we described how to realize error-mitigated logical T-gates via code switching.
In this appendix, we present a detailed construction for Kitaev's surface code depicted in~\Cref{fig:kitaev_codeswitch}.
The content of this appendix is taken and adapted from the author's previous work in \reference~\cite{piveteau2021_errormitigation}.

A distance-$d$ surface code is defined on a square lattice of linear size $2d-1$.
Code qubits live at sites indicated by black circles in~\Cref{fig:kitaev_codeswitch}.
Let $\mathcal{S}_1$ be the stabilizer group of the surface code generated by $X$ and $Z$ stabilizers located on red ($X$) and blue ($Z$) faces of the lattice.
We choose logical Pauli operators $\overline{X}$ and $\overline{Z}$ as products of single-qubit $X$ and $Z$ along  the left ($X$) and the top ($Z$) boundary.

Next, let us define an asymmetric version of the surface code with a small distance for $Z$ errors.
Its stabilizer group $\mathcal{S}_2$ is obtained from $\mathcal{S}_1$ by introducing  some new $Z$ stabilizers and removing some $X$  stabilizers. This modification only affects stabilizers located near the top boundary, as shown on Figure~\ref{fig:kitaev_codeswitch}.
To describe this formally, let $d=2t+1$ be the code distance. 
Label code qubits located at the top boundary by integers $1,2,\ldots,d$ in the order from the left to the right such that $\overline{Z}=Z_1Z_2\cdots Z_d$.
Define Pauli operators 
\begin{align*}
 G_i := Z_{2i} Z_{2i+1}.
\end{align*}
The group $\mathcal{S}_2$ is obtained from $\mathcal{S}_1$ by adding $t$ new stabilizers $G_1,\ldots,G_t$ and removing every second $X$ stabilizer adjacent to the top boundary.
An $X$ stabilizer of $\mathcal{S}_1$ is removed if it anti-commutes with some operator $G_i$.
Let $q_{\mathrm{loc}}:= 1$ be the data qubit located at the north-west corner of the lattice.
The code $\mathcal{S}_2$ exhibits a local logical-$Z$ operator acting only on the qubit $q_{\mathrm{loc}}$, namely
\begin{equation}
\label{deform0}
\overline{Z}_{\mathrm{loc}} =
\overline{Z}\prod_{i=1}^t  G_i = Z_{q_{\mathrm{loc}}}.
\end{equation}
Note that $\overline{Z}_{\mathrm{loc}}$ and $\overline{Z}$ differ by stabilizers of $\mathcal{S}_2$ and thus have the same action on any logical state of $\mathcal{S}_2$.
Thus, the logical $T$-gate on a qubit encoded by $\mathcal{S}_2$ can be implemented by either of the following operators:
\begin{equation}
\label{deform1}
\overline{T} = e^{-i \frac{\pi}{8} \overline{Z}}
\quad \mbox{and}
\quad \overline{T}_{\mathrm{loc}} = e^{-i \frac{\pi}{8} \overline{Z}_{\mathrm{loc}}} \, .
\end{equation}
Here we ignore the overall phase.
Note that $\overline{T}$ also implements a logical $T$-gate for the code $\mathcal{S}_1$ since both codes have the same logical Pauli operators.
Let $\mathcal{G}=\langle G_1,\ldots,G_t\rangle$ be the group generated by $G_1,\ldots,G_t$.
We claim that a logical $T$-gate on a qubit encoded by the surface code $\mathcal{S}_1$ can be implemented as follows. 
\begin{algorithm}[H]
	\caption{Logical $T$-gate by code switching}
	\label{algo_1}
	\begin{algorithmic}[1]
    \State{Initialize code qubits in a logical state of the surface code $\mathcal{S}_1$}
    \State{Measure the syndrome of the asymmetric surface code $\mathcal{S}_2$}
    \State{Let $\sigma_i \in \{-1,1 \}$ be the measured syndrome of $G_i$}
    \State{Compute $\sigma=\prod_{i=1}^t \sigma_i$}
    \State{Apply a single-qubit operator $(\overline{T}_{\mathrm{loc}})^\sigma$
    to the post-measurement state}
    \State{Measure the syndrome of the surface code $\mathcal{S}_1$}
    \State{Compute a Pauli operator $R\in \mathcal{G}$ consistent with the measured syndrome of $\mathcal{S}_1$}
    \State{Apply $R$ to the post-measurement state}
    \end{algorithmic}
\end{algorithm}
Indeed, let $\ket{\psi_j}$ be a state obtained after executing the first $j$ steps of this algorithm. 
We need to show that $\ket{\psi_8} = \overline{T} \ket{\psi_1}$.
Let $\Pi$ be a projector describing the measurement at step~2 such that $\ket{\psi_2}=\Pi \ket{\psi_1}$, ignoring the normalization.
From $G_i\ket{\psi_2}=\sigma_i\ket{\psi_2}$ and~\eqref{deform0},\eqref{deform1} one infers $(\overline{T}_{\mathrm{loc}})^\sigma \ket{\psi_2}=
\overline{T} \ket{\psi_2}$.
Since $\overline{Z}$ is a logical operator for both codes $\mathcal{S}_i$, the logical gate $\overline{T}$ commutes with the syndrome measurement of $\mathcal{S}_2$.
Thus ,
\begin{align*}
\ket{\psi_5}=
(\overline{T}_{\mathrm{loc}})^\sigma \ket{\psi_2} =
\overline{T} \ket{\psi_2}=\Pi \overline{T} \ket{\psi_1} \, .
\end{align*}
Let $\Lambda$ be a projector describing the measurement at step~6.
In the absence of errors the measured syndrome is non-trivial only for stabilizers of $\mathcal{S}_1$ that anti-commute with some element of $\mathcal{S}_2$.
One can check that there are exactly $t$ such stabilizers which we denote $F_1,\ldots,F_t$.
The stabilizer $F_i$ acts on data qubits $2i-1,2i$ located at the top boundary as well as the data qubit located at the second row between $2i-1$ and $2i$, see Figure~\ref{fig:kitaev_codeswitch}.
Let $\lambda_i\in \{1,-1\}$ be the syndrome of $F_i$ measured at step~6. The operator $R$ used at step~7 is defined as
\begin{align*}
R=\prod_{i\, : \, \lambda_i=-1} 
\; \prod_{a=i}^t G_a \, . 
\end{align*}
Using the commutation rules between the operators $F_i$ and $G_a$ one can easily check that $R$ has the syndrome $\lambda$, that is, $RF_i=\lambda_i F_i R$ for all $i=1,\ldots,t$.
Write $\Lambda = R \Gamma R$, where $\Gamma$ is a projector onto the logical subspace of the surface code $\mathcal{S}_1$.
The above shows that 
\begin{align*}
\ket{\psi_8} = R\Lambda \ket{\psi_5}=
(R\Lambda R)R \ket{\psi_5}=\Gamma R \ket{\psi_5}=
\Gamma R \Pi \overline{T} \ket{\psi_1}=
\Gamma (R\Pi) \Gamma \overline{T} \ket{\psi_1} \, .
\end{align*}
To obtain the last equality we noted that $\Gamma \ket{\psi_1}=\ket{\psi_1}$ and that $\overline{T}$ commutes with $\Gamma$.
The operator $R\Pi$ is a linear combination of Pauli operators from the group $\mathcal{S}_1\cdot \mathcal{G}$.
One can easily check that this group contains only stabilizers of the surface code and detectable errors.
Thus, $\Gamma (R\Pi) \Gamma$ is proportional to $\Gamma$.
The latter has trivial action on the logical state $\overline{T} \ket{\psi_1}$.
This proves that $\ket{\psi_8}=\overline{T} \ket{\psi_1}$, as claimed.

We make Algorithm~\ref{algo_1} partially fault-tolerant by repeating syndrome measurements at step~2 and step~6 sufficiently many times.
The observed syndromes are used to compute an error-corrected version of the phase $\sigma$ at step~4 and to perform the final error correction after step~8. 
Accordingly, the algorithm can fail in two distinct ways.
First, the final state could be $(\overline{T})^{-1} \ket{\psi_1}$ instead of $\overline{T} \ket{\psi_1}$ because some syndrome $\sigma_i$
has been flipped due to an error.
Such error results in a wrong phase $\sigma$ computed at step~4.
We shall refer to such events as Pauli frame errors since they can be viewed as applying the $T$-gate in a wrong Pauli frame.
Likewise, an undetected $X$ error that occurs before step~5 on the qubit $q_{\mathrm{loc}}$ would result in a Pauli frame error, as can be seen from the identity $\overline{T}_{\mathrm{loc}}X_1 =X_1(\overline{T}_{\mathrm{loc}})^{-1}$.
Secondly, the algorithm may fail if the final error correction performed after step~8 resulted in a logical error.
We would like to achieve an exponential suppression for logical $X$ errors and, at the same time, ensure that Pauli frame and logical $Z$ errors occur with probability at most $\kappa \varepsilon$, where $\varepsilon$ is the physical error rate and $\kappa$ is a constant prefactor.

To measure the syndrome of the surface code $\mathcal{S}_1$ we use a well-known quantum circuit proposed in \reference~\cite{dennis2002_topological,fowler2009_high}.
It requires one ancillary syndrome qubit per each stabilizer, see Figure~\ref{fig:kitaev_codeswitch}.
The circuit applies four rounds of CNOTs that compute the syndrome of each stabilizer into the respective
syndrome qubit.
A syndrome measurement cycle begins by resetting each syndrome qubit to $\ket{0}$ or $\ket{+}$ state for $Z$ and $X$ stabilizers respectively. 
The cycle ends by measuring each syndrome qubit in the $Z$- or $X$-basis.
A circuit measuring the syndrome of $\mathcal{S}_2$ requires only a few minor modifications, as described in the figure caption.

For numerical simulations we chose the depolarizing noise model~\cite{fowler2009_high}. 
It depends on a single error rate parameter $\varepsilon$ such that each operation in the circuit (a gate, a measurement, or a qubit reset) becomes faulty with the probability $\varepsilon$.
Faults on different operations are independent. 
The model can be described by stochastic Pauli errors enabling efficient simulation using the stabilizer formalism~\cite{aaronson2004_improved}.
A faulty CNOT gate is modelled as the ideal CNOT followed by a Pauli error drawn uniformly at random from the two-qubit Pauli group.
A faulty measurement is modelled as the ideal measurement whose outcome is flipped.
A faulty qubit reset prepares a basis state orthogonal to the ideal one. 
A faulty idle qubit suffers from a Pauli error X, Y, or Z.

\begin{figure}
\centerline{
\includegraphics[width=0.65\textwidth]{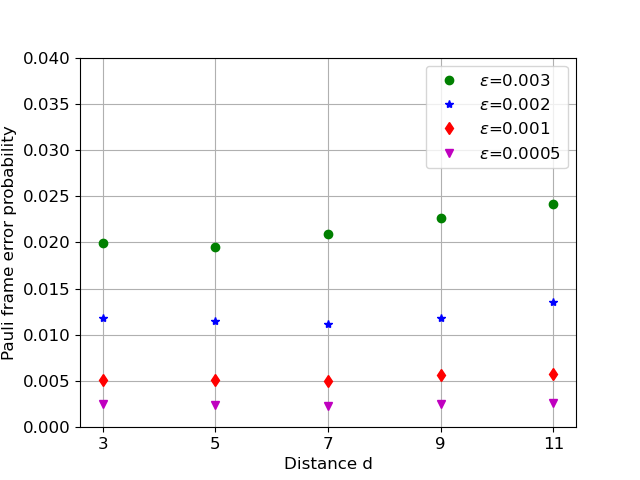}}
\caption{\textbf{Probability of a Pauli frame error} $P_F$ for the
code switching $\mathcal{S}_1\to \mathcal{S}_2$ 
with $L=3$ syndrome measurement cycles for the asymmetric
surface code $\mathcal{S}_2$. The initial state is chosen
as an ideal logical state of the surface code $\mathcal{S}_1$.
A Pauli frame error in the code switching
results in the implementation
of the logical gate $(\overline{T})^{-1}$ instead of $\overline{T}$. 
Our data suggest a scaling
$P_F\approx 6\varepsilon + O(d \varepsilon^2)$. 
\label{fig:plot1}
}
\end{figure}

\begin{figure}
\centerline{
\includegraphics[width=0.65\textwidth]{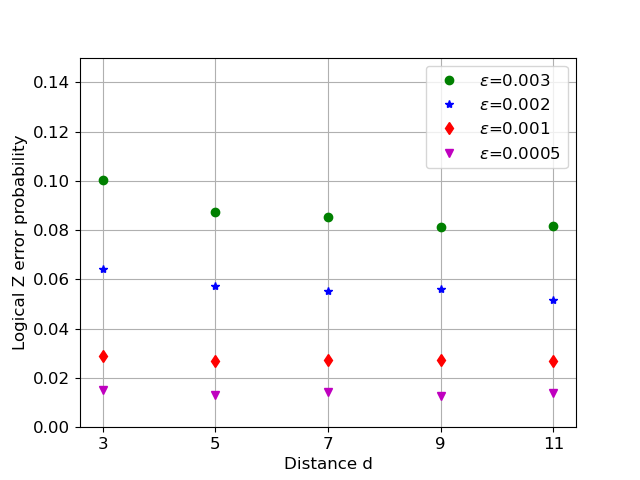}}
\caption{\textbf{Probability of a logical $Z$ error} 
$P_Z$
for the code switching $\mathcal{S}_1\to \mathcal{S}_2\to \mathcal{S}_1$
with $L=3$ syndrome  cycles for the asymmetric
surface code $\mathcal{S}_2$ and $d$ syndrome cycles for the
regular surface code $\mathcal{S}_1$.
To enable efficient simulation we skipped the
$T$-gate application in Algorithm~\ref{algo_1}. 
Our data suggests a scaling
$P_Z\approx 26 \varepsilon$ 
for a large code distance.
The large constant prefactor stems from leaving the
logical qubit unprotected from $Z$ errors for $L+1$
syndrome cycles when the full syndrome information for
$X$ stabilizers of the code $\mathcal{S}_1$ is unavailable.
\label{fig:plot2}
}
\end{figure}

Let us first discuss how to correct Pauli frame errors.
Suppose $U$ is a quantum circuit describing the first two steps of Algorithm~\ref{algo_1} with $L$ syndrome measurement cycles at step~2. 
The circuit outputs a quantum state $\ket{\psi_2}$ and a classical syndrome history $h$ specifying the syndrome of each stabilizer of $\mathcal{S}_2$ measured in each of $L$ syndrome cycles. 
In the absence of errors all stabilizers except for $G_1,\ldots,G_t$ have trivial syndromes.
Furthermore, the syndromes of $G_i$
measured in different cycles are the same. 
Let $\mathcal{F}(U)$ be the set of all possible faulty implementations of $U$,
as prescribed by the depolarizing noise model.
Each circuit $\tilde{U}\in \mathcal{F}(U)$ differs from $U$ by inserting Pauli errors
at some space-time locations (measurement and reset errors
can be modeled by inserting suitable Pauli errors before the ideal measurement
and after the ideal reset). 
We pick a faulty circuit $\tilde{U} \in \mathcal{F}(U)$ randomly, according to the depolarizing noise model.
A simulation of the circuit $\tilde{U}$ yields a noisy syndrome history  
$\tilde{h}$.
To perform error correction we used 
the maximum-weight matching (MWM) decoder proposed in \reference~\cite{bravyi2013_simulation}
which is closely related to the one used in \reference~\cite{dennis2002_topological,fowler2009_high}.
The MWM decoder
takes the syndrome history $\tilde{h}$ as an input
and outputs a candidate faulty circuit $\tilde{U}_{\mathrm{dec}}\in \mathcal{F}(U)$
consistent with  $\tilde{h}$.
By propagating all Pauli errors contained in 
$\tilde{U}_{\mathrm{dec}}$ towards the final time step, we 
compute error-corrected syndromes $\sigma_i$
and a Pauli correction $C$ acting on the data qubits
such that $CU\ket{\psi_1}$ is our guess of the noisy
state $\tilde{U}\ket{\psi_1}$, based on the available syndrome
information. Step~4 of Algorithm~\ref{algo_1} computes an error-corrected
phase $\sigma=(-1)^a \sigma_1 \cdots \sigma_t$,
where $a=1$ if $C$ contains an $X$ error
on the qubit $q_{\mathrm{loc}}$
and $a=0$ otherwise. To avoid the actual application
of the $T$-gate (which would require non-stabilizer
simulators) we chose the initial logical state
$\ket{\psi_1}$ as an eigenvector of $\overline{Z}$,
that is, $\ket{\psi_1}=\ket{\overline{b}}$, where
$b\in \{0,1\}$ is a random bit.
Error correction succeeds if $\tilde{U}\ket{\psi_1}$
is an eigenvector of $(\overline{T}_{\mathrm{loc}})^\sigma$
with the eigenvalue $(-1)^b$.
Otherwise, a Pauli frame error is declared.

Our numerical results for the probability of a Pauli frame
error $P_F$ are shown in Figure~\ref{fig:plot1}.
Each data point represents an empirical
estimate of $P_F$ obtained by the Monte Carlo method
(we used between $50,000$ and $200,000$
Monte Carlo trials per data point). 
We chose $L=3$ syndrome measurement cycles to ensure that 
any single fault in the measurement of $\sigma_i$
is correctable.
Recall that the MWM decoder detects errors by examining 
differences between syndromes measured in subsequent cycles.
Thus, $L=3$ provides two parity checks to diagnose errors
in the three measured values of $\sigma_i$,
in a way analogous to the $3$-bit repetition code.
In contrast, choosing $L=2$ would only allow us to detect
a faulty bit $\sigma_i$ but not to correct it
(if the two measured values of $\sigma_i$ disagree, there
is no way to select the correct value). 
Choosing $L>3$ is undesirable as this leaves the logical qubit
unprotected from $Z$ errors for a longer time.
Our data suggests a scaling
$P_F\approx 6\varepsilon + O(d \varepsilon^2)$.
We numerically observed that the dominant contribution
to $P_F$ stems from
$X$ errors occurring on the qubit $q_{\mathrm{loc}}$
in the last syndrome cycle as well as syndrome measurement errors on the $Z$ stabilizer acting
on the qubit $q_{\mathrm{loc}}$ in the last syndrome cycle.
The term $O(d\varepsilon^2)$ can be understood as a result
of a weight-two error affecting the measurement of some 
syndrome bit $\sigma_i$. Here $i=1,\ldots,t$
can be arbitrary. Such weight-two error cannot be
corrected using only $L=3$ syndrome cycles. 

Finally, we estimated the probability of logical $X$ and $Z$ errors  by simulating a simplified version of Algorithm~\ref{algo_1} without steps~3,4,5. In other words, we skipped the
application of the physical $T$-gate. This enables
efficient simulation using the stabilizer formalism.
The MWM decoder takes as input the combined syndrome
history measured at steps~2,6 and calculates a Pauli correction
to be applied after step~8. 
The standard implementation
for MWM decoder~\cite{bravyi2013_simulation}
was properly modified to
take into account that some stabilizers are turned off and on
during the code switching. 
Namely, syndromes of the operators
$G_i$ measured in the first cycle of step~2 as well as syndromes of the operators $F_i$ measured in the first cycle of
step~6 are not used to diagnose errors.  
As commonly done in the literature, we 
chose the number of syndrome cycles for the surface
code $\mathcal{S}_1$ equal to the code distance $d$
and added a noiseless syndrome cycle at the end of the protocol.

Our numerical results for the probability of 
logical $Z$ and $X$ errors $P_Z$ and $P_X$ are
shown on Figures~\ref{fig:plot2} and~\ref{fig:plot2X}, respectively.
Our data suggests a scaling
$P_{Z} \approx 26\varepsilon$ for a large
code distance. The large constant prefactor 
can be seen as the price we pay for
leaving the logical qubit unprotected from $Z$ errors during
$L$ syndrome cycles at step~2. In fact, since the syndromes
of $X$ stabilizers $F_i$ measured in the first cycle
of step~6 are random, these syndrome provide no information
about $Z$ errors that occurred in this cycle. 
Thus, the logical qubit is unprotected
from $Z$ errors for $L+1=4$ cycles. 
Since each cycle is a depth-$6$ circuit (four CNOT rounds,
reset round, and measurement round), the error correction
for $Z$ errors is effectively turned off for $24$ time steps.
This is in a good agreement with the scaling $P_{Z} \approx 26\varepsilon$.
Finally, we observed that the probability of logical $X$ errors is exponentially suppressed as one increases the code distance,
see Figure~\ref{fig:plot2X}.
Thus, the noise in the implemented logical $T$-gate is dominated
by (coherent) $Z$ errors, as expected. 

We anticipate the scaling of $P_Z$ can be improved by optimizing the code switching protocol.
For example, one can reduce $P_Z$ roughly by the factor $2/3$
by choosing the number of syndrome measurement cycles $L$ for the code $\mathcal{S}_2$ adaptively such that $L=2$ if the syndromes $\sigma_i$ measured in the first and the second cycles are the same for all $i=1,\ldots,t$ and $L=3$ otherwise. This would ensure that a single fault in the measurement of $\sigma_i$ can be corrected, similar to the $L=3$ implementation described above. However, the logical qubit remains unprotected from logical $Z$ errors for a shorter time (two cycles instead of three cycles, in the limit $\varepsilon\to 0$).

Under the assumption that we use a code with a sufficiently large distance, the logical $X$ errors are negligible. Hence, the logical error rate $\bar \epsilon$ is determined by the logical $Z$ and Pauli frame errors.
The latter are described by the channel
$(1-P_F){\idchan} + P_F \mathcal{S}^\dag$, where 
$\mathcal{S}^\dag(B)=S^\dag BS$ for any single qubit operator $B$.
Twirling over the $X$ gate gives
a $Z$-type noise channel $(1-P_F/2){\idchan} + (P_F/2){\mathcal{Z}} $. Assuming that Pauli frame and logical $Z$ errors are independent, the combined logical error rate is
therefore $\bar{\varepsilon}\approx P_Z+P_F/2\approx 30\varepsilon$
in the limit $\varepsilon\ll 1$.

\begin{figure}
\centerline{
\includegraphics[width=0.65\textwidth]{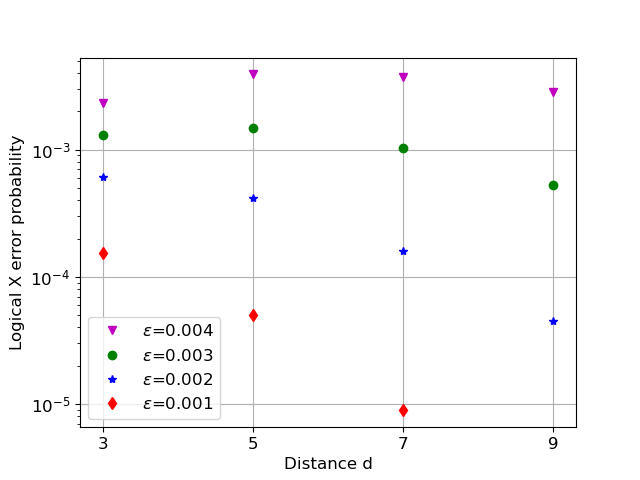}}
\caption{\textbf{Probability of a logical $X$ error} 
$P_X$
for the code switching $\mathcal{S}_1\to \mathcal{S}_2\to \mathcal{S}_1$
with $L=3$ syndrome  cycles for the asymmetric
surface code $\mathcal{S}_2$ and $d$ syndrome cycles for the
regular surface code $\mathcal{S}_1$.
To enable efficient simulation we skipped the
$T$-gate application in Algorithm~\ref{algo_1}. 
Our data suggests that logical $X$ errors
are exponentially suppressed as one increases
the code distance.
\label{fig:plot2X}
}
\end{figure}

\end{document}